\documentclass[12pt]{report}
\usepackage[utf8]{inputenc}
\usepackage{hyperref}
\usepackage{authblk}
\usepackage{fullpage}
\usepackage{pdfpages}
\usepackage{enumitem} 
\usepackage[normalem]{ulem}
\usepackage{soul}
\usepackage{upgreek}
\usepackage{blindtext}
\usepackage{boxhandler}
\usepackage[toc]{appendix}
\usepackage{pdfpages}
\usepackage{fourier}

\usepackage{graphicx}
\graphicspath{ {images/} }
\usepackage{float}
\usepackage{morefloats} 
\usepackage{wrapfig}
\usepackage{caption}
\usepackage{xcolor}

\usepackage{mathtools}
\usepackage{amsthm, amssymb}
\usepackage{mathrsfs}
\usepackage{braket} 
\usepackage{slashed}
\usepackage{cancel}

\usepackage{amsmath}
\usepackage{verbatim}
\usepackage{color}
\usepackage{subfig}
\usepackage[section]{placeins}
\edef\restoreparindent{\parindent=\the\parindent\relax}
\usepackage{parskip}
\restoreparindent
\usepackage[normalem]{ulem}
\usepackage{fullpage}
\usepackage{array}

\newtheorem{claim}{Claim}
\newcommand{\be}{\begin{equation}}
\newcommand{\ee}{\end{equation}}
\newcommand{\bea}{\begin{eqnarray}}
\newcommand{\eea}{\end{eqnarray}}
\newcommand{\e}{\epsilon}
\newcommand{\cof}{\Lambda}
\newcommand{\sgn}{\mathrm{sgn}}

\newcommand{\vac}{| 0 \rangle}
\newcommand{\cF}{\mathcal F}
\newcommand{\cH}{\mathcal H} 
\newcommand{\mL}{\mathcal L}
\newcommand{\li}{{\mathrm{Li}}}
\newcommand{\re}{{\mathrm{Re}}}
\newcommand{\kker}{\mathrm{Ker}}
\newcommand{\ha}{\hat a}
\newcommand{\hb}{\hat b}
\newcommand{\cS}{S}
\newcommand{\mink}{\mathbb M}
\newcommand{\Bkg}{\Box_{\mathrm{KG}}}

\newcommand{\cK}{\mathcal K}
\newcommand{\cO}{\mathcal O} 
\newcommand{\intm}{\int_{-}}
\newcommand{\intp}{\int_{+}}
\newcommand{\gOjsl}{\Omega^{l}_{js}} 
\newcommand{\djsa}{\Delta^{a}_{js}}
\newcommand{\dja}{\Delta^{a}_{j(n-j)}} 
\newcommand{\UAS}{U^{A/S}}
\newcommand{\FAS}{F^{A/S}}
\newcommand{\GAS}{G^{A/S}}
\newcommand{\HAS}{H^{A/S}}
\newcommand{\QAS}{Q^{A/S}}
\newcommand{\PAS}{P^{A/S}}
\newcommand{\uAS}{u^{A/S}}
\newcommand{\kAS}{k^{A/S}}
\newcommand{\Du}{\Delta u}
\newcommand{\Dv}{\Delta v}

\newcommand{\ka}{k_A}
\newcommand{\ks}{k_S}
\newcommand{\cc}{{\cal C}}
\newcommand{\id}{{\mathbb I}}
\newcommand{\diam}{\mathcal{D}}
\newcommand{\wsjc}{W^c_\text{SJ}}
\newcommand{\wsj}{W_\text{SJ}}
\newcommand{\wmink}{W^{\mathrm{mink}}_0}
\newcommand{\wminkm}{W^{\mathrm{mink}}_m}
\newcommand{\wrind}{W^{\mathrm{rind}}_0}
\newcommand{\wrindm}{W^{\mathrm{rind}}_m}
\newcommand{\wmirr}{W^{\mathrm{mirror}}_0}
\newcommand{\wmirrm}{W^{\mathrm{mirror}}_m}
\newcommand{\wa}{A_{\mathrm{I}}}
\newcommand{\waa}{A_{\mathrm{II}}}
\newcommand{\waaa}{A_{\mathrm{III}}}
\newcommand{\waaaa}{A_{\mathrm{IV}}}
\newcommand{\ws}{S_{\mathrm{I}}}
\newcommand{\wss}{S_{\mathrm{II}}}
\newcommand{\wsss}{S_{\mathrm{III}}}
\newcommand{\wssss}{S_{\mathrm{IV}}}
\newcommand{\emc}{\e_m^{center}}
\newcommand{\emcr}{\e_m^{corner}}

\newcommand{\M}{\mathcal{M}}

\newcommand{\R}{\mathcal{R}}
\newcommand{\mr}{\mathbb{R}}
\newcommand{\ms}{\mathbb{S}}
\newcommand{\mz}{\mathbb{Z}}
\newcommand{\mt}{\mathbb{T}}
\newcommand{\gt}{\tilde{g}}

\newcommand{\Rt}{\tilde{R}}

\newcommand{\hPhi}{\hat{\Phi}}
\newcommand{\Phit}{\Phi_{st}}
\newcommand{\itd}{i\tilde{\Delta}}
\newcommand{\hd}{{\Delta}}
\newcommand{\rew}{\text{Re}(\wsj)}
\newcommand{\vx}{\vec{x}}
\newcommand{\he}{{E}}
\newcommand{\ksj}{\left|0\right>_\text{SJ}}

\newcommand{\kconf}{\left|0\right>_c}

\newcommand{\spkg}{K_{sp}}
\newcommand{\spl}{\Delta_{\ms^{d-1}}}
\newcommand{\spac}{\Sigma}
\newcommand{\kcurv}{\kappa}
\newcommand{\degn}{{\cal D}}
\newcommand{\gp}{\mathcal{C}}

\newcommand{\usj}{{\bf{u}}^\text{SJ}}

\newcommand{\up}{{\bf{u}}^+}
\newcommand{\un}{{\bf{u}}^-}

\newcommand{\uconfp}{\tilde{\bf{u}}^+}
\newcommand{\uconfn}{\tilde{\bf{u}}^-}

\newcommand{\conf}{\Omega}

\newcommand{\hD}{\Delta}
\newcommand{\hW}{W}
\newcommand{\hR}{R}
\newcommand{\im}{\mathrm{Im}}

\newcommand{\pg}{\Psi}
\newcommand{\eul}{\gamma_e}
\newcommand{\taubyl}{\gamma}

\newcommand{\mathsym}[1]{{}}
\newcommand{\unicode}[1]{{}}
\newcommand{\cI}{\mathcal I}
\newcommand{\cC}{\mathcal C}

\newcommand{\hwt}{\hW_\tau}
\newcommand{\sinc}{{\rm sinc}}
\newcommand{\hr}{\widehat \rho}

\newcommand{\wz}{{W}^{(0)}}
 
\newcommand{\cm}{\mathrm{c}_m} 
\newcommand{\sm}{\mathrm{s}_m} 
\newcommand{\mx}{\mathrm{max}}
 
\newcommand{\ssee}{\mathcal S}
\newcommand{\taup}{{\bar{\tau}}}
 
\newcommand{\Ap}{\bar{A}}
\newcommand{\Bp}{\bar{B}}
\newcommand{\psit}{\psi^{(\tau)}}
\newcommand{\psitp}{\psi^{(\taup)}}

\newcommand{\psistp}{\psi^{*(\taup)}}
\newcommand{\mtau}{\cM_\tau}
\newcommand{\mtaup}{\cM_\taup}
\newcommand{\wtau}{W_\tau}

\newcommand{\cylev}{\varrho}
\newcommand{\ccL}{\mathcal L}
\newcommand{\hwdt}{\hW_{\tau}|_{\diam_s}}

\newcommand{\eq}[1]{\begin{align}#1\end{align}}

\newcommand{\bfig}{\begin{figure}}
\newcommand{\efig}{\end{figure}}
\newcommand{\cL}{\mathcal{L}}
\newcommand{\npl}{N_{pl}}

\newcommand{\tr}{\text{Tr}}
\newcommand{\mass}{\mathbf{m}}
\newcommand{\m}{m}

\newcommand{\hn}{\hat{n}}

\newcommand{\bv}{\mathbf \Phi}
\newcommand{\bk}{\mathbf k}

\newcommand{\bpp}{\mathbf p}
\newcommand{\bap}{{\bar {\mathbf p}}}

\newcommand{\bx}{\mathbf x}
\newcommand{\bu}{\mathbf \Psi}
\newcommand{\br}{\mathbf r}

\newcommand{\Righf}{I}
\newcommand{\fut}{III}
\newcommand{\Leff}{II}
\newcommand{\cM}{\mathcal{M}}
\newcommand{\cB}{\mathcal{B}}
\newcommand{\rs}{{r_*}}

\newcommand{\tip}{{\tilde{p}}}

\newcommand{\tl}{\tilde{l}}
\newcommand{\nd}{{\nu_d}}

\newcommand{\vacket}{\left|0\right>}
\newcommand{\vacbra}{\left<0\right|}
\newcommand{\oi}{I}
\newcommand{\link}{\prec\!\!*\;}
\newcommand{\hrho}{\widehat\rho}

\newcommand{\rin}[2]{{\langle #1 , #2\rangle}_R}
\newcommand{\htp}{\hat{p}}
\newcommand{\htq}{\hat{q}}
\newcommand{\htD}{{\Gamma}}
\newcommand{\pone}{\psi^{(1)}}
\newcommand{\ptwo}{\psi^{(2)}}
\newcommand{\tpone}{\varphi^{(1)}}
\newcommand{\tptwo}{\varphi^{(2)}}
\newcommand{\hK}{K}

\begin{document}
	
\begin{titlepage}
 \begin{center}
	
     \Huge
	\textbf{Quantum Fields from Causal Order}
	
	\vspace{1.5cm}
	
	\Large
	\textbf{Abhishek Mathur}
	
	\vspace{0.5cm}
	\textit{under the supervision of}
	
	\vspace{0.2cm} 
	
	\textbf{Prof. Sumati Surya}
	
	\vspace{0.2cm}
	
	Raman Research Institute
	
	\vspace{0.8cm}

	\centering{\includegraphics[height=5cm]{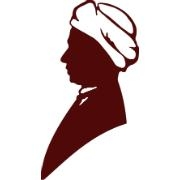} \hspace{1.5cm}
	\includegraphics[height=5cm]{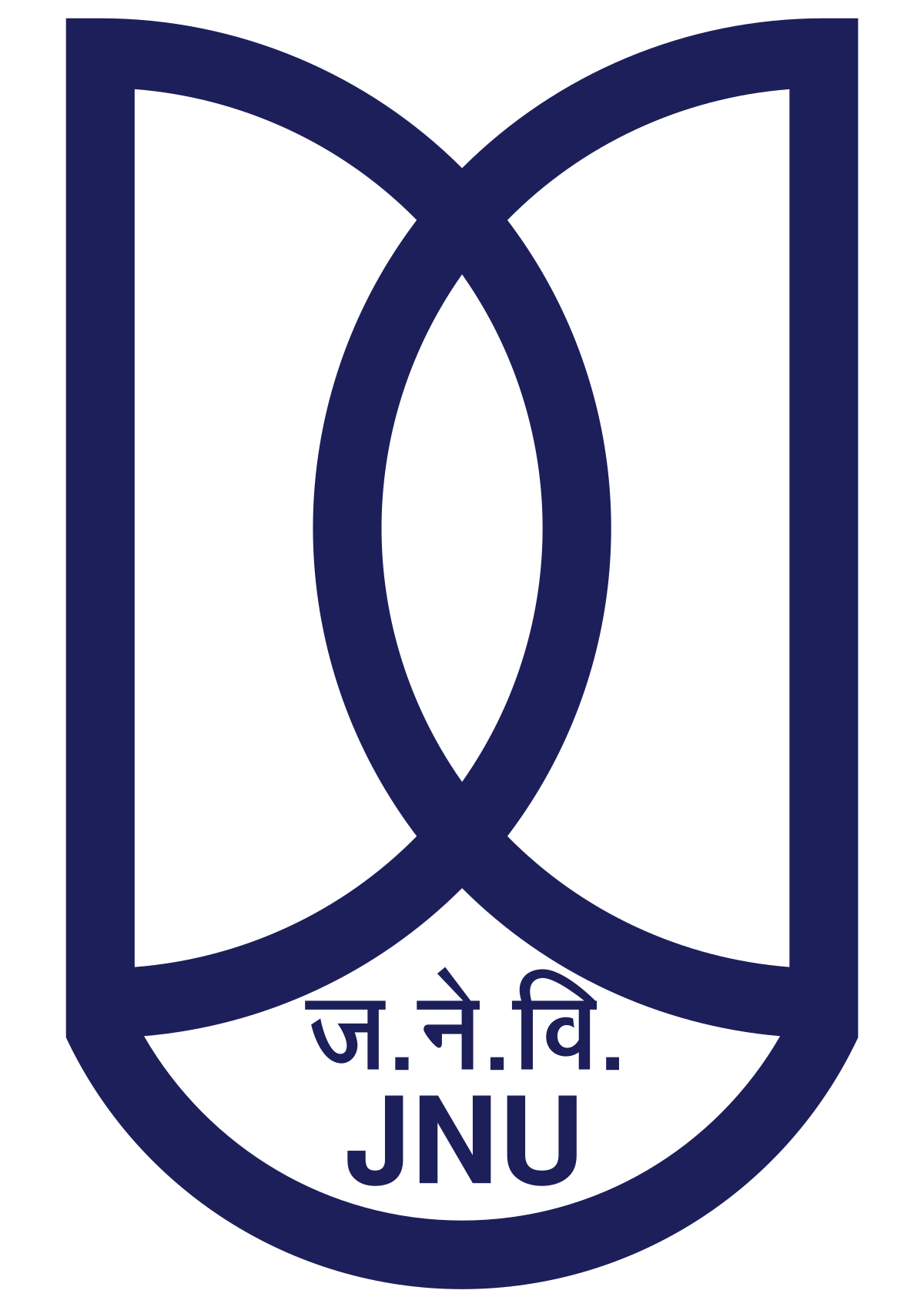}}
	
	\vfill
	
	A thesis submitted for the degree of\\
	Doctor of Philosophy\\
	to\\
	Jawaharlal Nehru University
	
\end{center}
\end{titlepage}






\chapter*{Synopsis}
Quantum Field Theory (QFT) in curved spacetime describes physical phenomena in a regime where quantum fields interact with the classical gravitational field. These conditions are generally satisfied either in the vicinity of a black hole or in the post Planckian very early universe. The standard approach to QFT relies heavily on the Poincare symmetry of Minkowski spacetime as a guide to a particular choice of mode expansion and therefore a preferred choice of vacuum. In contrast, in a generic curved spacetime which does not have any symmetry, the choice of mode expansion and hence the vacuum is arbitrary. This suggests that the notion of vacuum and particles is subsidiary. This is the approach taken in algebraic QFT (AQFT), where the primary role is played by the algebra of observables and the choice of state, which is primarily defined as a complex function on this algebra, is relegated to a choice of the representation of this algebra, and is therefore not unique. It is natural to ask if there exists an alternative formulation of the QFT vacuum which is inherently covariant. One such formalism was developed by Sorkin \cite{sorkin} and Johnston \cite{Johnston:2009fr} called the Sorkin-Johnston (SJ) vacuum. In this thesis we will explore various aspects of the SJ vacuum both in the continuum and in the causal set theory. QFT also plays an important role in understanding the entropy of black holes and other horizons. Entanglement entropy of quantum fields has been proposed as a candidate for these entropies. The standard formulation of the entanglement entropy is given by the von Neumann entropy formula, which requires the QFT states to be defined on a Cauchy hypersurface and is therefore dependent on the choice of the hypersurface. Sorkin, in his 2014 paper \cite{ssee}, proposed a spacetime formulation of the entanglement entropy, which is covariant.

In this thesis, we study examples of the SJ or observer independent QFT vacuum in curved spacetimes and the Sorkin's spacetime entanglement entropy (SSEE) for various horizons. The work in this thesis can be divided into two parts. The first is the SJ vacuum for a free real scalar field theory \cite{sorkin, johnston, Johnston:2009fr, Afshordi:2012ez}. It can be covariantly defined in any globally hyperbolic spacetime $(\cM,g)$ with finite spacetime volume without appealing to the spacetime symmetries and the associated ``preferred'' class of observers. Rather than using an arbitrary or symmetry-dictated mode decomposition of the Klein-Gordon solution space, the SJ vacuum is constructed from the covariantly defined eigenmodes of the integral Pauli-Jordan operator $i\hD$, whose integral kernel is the field commutator, i.e., $i\Delta(x;x')=[\hPhi(x),\hPhi(x')]$. The Pauli-Jordan function $\Delta$ itself is the difference between the retarded and the advanced Green's function, which encodes the causal structure of the spacetime. This formalism is motivated by the attempts to study QFT on causal sets, which are locally finite partially ordered sets obtained from a covariant discretisation of the spacetime. The ordering on the elements in the casual set corresponds to spacetime causal ordering. It is not easy to construct Killing vectors, which are associated with spacetime symmetries, on a causal set and therefore the SJ formalism, which is covariant, is most suitable for QFT on a causal set. However explicit computation of the SJ vacuum in a continuum spacetime manifold is a challenge because one need to solve the integral eigenvalue equations. Therefore only in a very few cases has the SJ vacuum explicitly been found. These cases include the massless scalar field in the 2d flat causal diamond \cite{johnston,Afshordi:2012ez} and in a patch of trousers spacetime \cite{Buck:2016ehk}. The SJ vacuum for a massive scalar field has been studied in ultrastatic slab spacetimes by Fewster and Verch \cite{fewster2012} and in conformally ultrastatic slab spacetimes with time dependent conformal factor by Brum and Fredenhagen \cite{Brum:2013bia}.
In this work we explicitly compute the SJ vacuum for a massive scalar field in the small mass limit to $\cO((mL)^4)$, in the 2d flat causal diamond of height $2\sqrt{2}L$. We also compute the SJ vacuum for a conformally coupled massless scalar field in a bounded time slab of cosmological spacetimes with compact spacelike hypersurfaces. We study the properties of the SJ spectrum and the two point correlation function also known as the Wightman function ($W(x;x')=\langle\hPhi(x)\hPhi(x')\rangle$), and compare with that of the known vacua in the respective cases. We also compare our results with that obtained in causal sets approximated by the respective spacetimes.

The second part of the thesis is the study of the spacetime formulation of the entanglement entropy developed by Sorkin \cite{ssee} which makes use of the Wightman function restricted to the region of interest and the spacetime field commutator (Pauli Jordan function). Entanglement entropy of quantum fields inside the event horizon of a black hole has been proposed as a candidate for the Bekenstein-Hawking entropy of the black hole. This has led to a significant interest in the entanglement entropy of relativistic quantum fields. The standard formula of the entanglement entropy of a quantum field restricted to a region of space (system) is given by the von Neumann entropy formula. It is given in terms of the reduced density matrix which is obtained by tracing out degrees of freedom in the complementary region (surrounding). In order to compute the von Neumann entropy, one need to define quantum states on a Cauchy hypersurface. This spatial approach is good for non-relativistic systems where space and time are treated separately but for a relativistic system one would like to treat space and time on the same footing. Also, the von Neumann entropy cannot be extended to discrete theories of quantum gravity like causal set theory, which lacks the notion of a Cauchy hypersurface. One therefore needs a spacetime formulation of the entanglement entropy to make it applicable in those contexts. One such spacetime formalism for the entanglement entropy was proposed by Sorkin \cite{ssee}. In this thesis, we use Sorkin's spacetime entanglement entropy (SSEE) formula to compute the entanglement entropy in a 2d causal diamond embedded in a 2d cylinder spacetime and in the static patch of the de Sitter and Schwarzschild de Sitter spacetime.

The thesis is planned as follows. In chapter~\ref{ch.int} we introduce the SJ formalism for quantum fields in curved spacetime manifold and draw a comparison between the SJ formalism and the standard formalism for quantum fields in curved spacetime. We see that in the standard formulation of QFT in curved spacetime, there is an ambiguity in the choice of the QFT modes and hence the vacuum. On the other hand in the SJ formulation of QFT we have a unique vacuum, which is defined as the positive semidefinite part of the Pauli-Jordan operator. We will see that this SJ vacuum lies in the collection of all possible vacua one can have in the standard approach to QFT. Apart from the usual properties of the scalar field vacua, the SJ vacuum uniquely satisfies an additional property namely that of orthogonal support i.e. $WW^*=0$. This is a property which we know is satisfied by the standard Minkowski vacuum in flat spacetime. We discuss the SJ formalism both in the continuum and in the causal set. We end the chapter with an introduction to the spacetime formulation of the entanglement entropy by Sorkin \cite{ssee}.

In chapter~\ref{ch.2ddiamsj}, which consists of our work in \cite{Mathur:2019yvl}, we construct the SJ vacuum for the massive scalar field in the limit of small mass in the 2d flat causal diamond of height $2\sqrt{2}L$. Since the integral eigenvalue equations involved are difficult to solve exactly, we restrict ourselves to the small mass case and solve for the SJ modes upto $\cO((mL)^4)$. This allows us to construct the Wightman function upto this order. We study the behaviour of the SJ Wightman function in the center and the corner of the 2d causal diamond by using a combination of analytical and numerical methods. We find that in the center of the diamond, the small mass SJ Wightman function behaves like the massless Minkowski Wightman function with a fixed infrared cut-off depending on the size of the diamond, whereas in the corner of the diamond it behaves like the mirror Wightman function and not like the Rindler Wightman function as one might expect. In order to have an insight into the SJ vacuum for large masses, we compute the SJ Wightman function in a causal set approximated by the 2d causal diamond, using the retarded Green's function found by Johnston \cite{Johnston:2008za}. We find that the for a small mass upto an ultraviolet cut-off, the SJ spectrum and the Wightman function behaves in the same way as their continuum counterpart. Here we see that for a large mass, the SJ Wightman function behaves like the massive Minkowski Wightman function in the center of the diamond, but below a certain critical mass it starts behaving like the massless Minkowski Wightman function, mimicking the behaviour of the SJ Wightman function obtained in the continuum. In the corner of the diamond, we see that for large masses the correlation between the Rindler and the mirror Wightman function is better than that of the SJ Wightman function with any of them. Therefore for a large mass, one cannot associate the SJ vacuum to the mirror or the Rindler vacuum on the basis of the correlation plots. We also show that for the massless scalar field, one can get the Rindler vacuum from the SJ prescription in an appropriate limit by redefining the integral operator and the inner product on the space of $\cL^2$ functions using a non-trivial weight function.

In chapter~\ref{ch.sjcosm} we look at the SJ modes and associated SJ vacuum for conformally coupled massless scalar fields in finite time slabs of various cosmological spacetimes with compact spacelike hypersurface. We use the work of Fewster and Verch on ultrastatic spacetimes as our starting point to solve the integral SJ eigenvalue equation. The spacetimes include the time-symmetric slab of global de Sitter spacetime, as well as  finite time FLRW spacetimes with toroidal Cauchy hypersurfaces. A generalisation of this construction for a massive scalar field can be found in the work of Brum and Fredenhagen \cite{Brum:2013bia}. We show that the SJ modes differ from the conformal modes for these spacetimes. We study the dependence of SJ modes on the size of the time slab and find that in the infinite time limit this difference vanishes for FLRW spacetimes and odd-dimensional de Sitter spacetimes. In even-dimensional de Sitter spacetimes however, the SJ modes themselves become ill-defined in this limit, which is in agreement with earlier work of Aslanbeigi and Buck \cite{Aslanbeigi:2013fga}. We compare the 2d and 4d de Sitter SJ spectrum with that obtained from numerical simulations on causal sets approximated by these spacetimes and find that for 2d, the SJ spectrum obtained in the continuum and the causal set are in agreement upto an ultraviolet cut-off whereas there is a slight discrepancy in 4d SJ spectrum which may be attributed to the fundamental difference between the discrete and the continuum Green's function. This discrepancy is yet to be understood in detail. 

In chapter~\ref{ch.sseeds}, which consists of our work in \cite{Mathur:2021ial}, we study SSEE in some special cases in which it can be evaluated analytically. In these cases the restriction of the Wightman function to a subregion $\cO\subset\cM$ is block diagonal in modes of $\cO$. We show that in those cases, one can easily evaluate the SSEE analytically, which is found to be dependent only on the Bogoliubov transformation between the modes in $\cM$ and the modes in $\cO$. We explicitly compute the SSEE for a massive scalar field in a static patch of de Sitter and for a massless scalar field in a static patch of Schwarzschild de Sitter spacetimes. We find that the SSEE in the static patch of de Sitter spacetime exactly matches with the von Neumann entropy obtained by Higuchi and Yamamoto \cite{Higuchi:2018tuk}. We also argue that the SSEE satisfies an area law as expected. In Schwarzschild de Sitter spacetime, we see that the SSEE satisfies the same area law as in de Sitter spacetime. What is surprising is that in 2d, instead of getting a logarithmic dependence on the UV cut-off, we find that the SSEE is a finite constant. This is also in keeping with the result of Higuchi and Yamamoto.

In chapter~\ref{ch.sseecyl}, which consists of our work in \cite{Mathur:2021zzl}, we compute the SSEE of a massless scalar field restricted to a causal diamond in the 2d cylinder spacetime using a combination of analytical and numerical methods. We consider the scalar field to be in the SJ vacuum state in a slab of 2d cylinder which is dependent on its height. This calculation can be seen as the spacetime version of the Calabrese-Cardy calculation of the entanglement entropy \cite{cc} of a quantum field restricted to an arc of a circle. Here the 2d cylinder is the domain of dependence of the circle and the causal diamond is the domain of dependence of an arc of the circle. The SSEE is dependent on three parameters, they are the UV cut-off, the relative size of the diamond and the height of the cylinder. We find that the SSEE obtained here is of the Calabrese-Cardy form. The UV cut-off dependent term is universal when we use a covariant UV cut-off, i.e, a cut-off in the SJ spectrum. The relative size dependent term depends on the choice of the pure state in the cylinder (which is dependent on the height of the cylinder) but asymptotes to the one obtained by Calabrese and Cardy in the large height limit.

In chapter~\ref{ch.concl} we conclude the thesis with a discussion on results and some open questions for future investigations.



\chapter*{\centering List of Publications}

\begin{enumerate}
\item {\bf{Sorkin-Johnston vacuum for a massive scalar field in the 2D causal diamond}}\\
Abhishek Mathur and Sumati Surya\\
Phys. Rev. D 100 (2019) 4, 045007\\
Erratum: Phys. Rev. D. 104 (2021) 8, 089901
\item {\bf{A spacetime calculation of the Calabrese-Cardy entanglement entropy}}\\
Abhishek Mathur, Sumati Surya and Nomaan X\\
Phys. Lett. B 820 (2021) 136567
\item {\bf{Spacetime Entanglement Entropy of de Sitter and Black Hole Horizons}}\\
Abhishek Mathur, Sumati Surya and Nomaan X\\
Class. Quant. Grav. 39 (2022),3, 035004
\item{\bf{Spacetime Entanglement Entropy: Covariance and Dicreteness}}\\
Abhishek Mathur, Sumati Surya and Nomaan x\\
{\sl{Submitted to General Relativity and Gravitation, Springer}}
\end{enumerate}

%
%
%

\chapter*{Acknowledgements}

I start with thanking my thesis supervisor Sumati Surya for her guidance and support. She introduced me to interesting problems and was always available for discussions. She also helped me a lot in improving my writing and presentation skills. I would also like to thank Nomaan X and Anup Anand Singh for collaborating with me and for interesting discussions we had, and Justin David and Aninda Sinha for being a part of my thesis committee, listening to my annual presentation and providing valuable feedback.

I would like to thank several visitors who visited RRI during my presence. First of all I would like to thank Rafael Sorkin for long discussions we had. These discussions helped me a lot in my work. I would then like to thank A. P. Balachandran for helpful discussions. I would also like to thank other visitors including Kasia Rejzner, Maximilian Ruep, Yasaman Yazdi, Stav Zalel, David Rideout, Dionigi Benincasa and Will Cunningham. I would also like to thank professors at RRI and IISc including Joseph Samuel, Sachindeo Vaidya and Supurna Sinha for helphul discussions and also for their help with references. I would also like to thank several students and post-docs at RRI including Sujit Kumar Nath, Anirudh Reddy, Kumar Shivam, Santanu Das, Deepak Gupta, Ion Santra, Simanraj Sadana for discussions with them on different topics in physics. We used to have discussions during evening tea on different interesting topics which has nothing to do with our work.

I would like to thank theoretical physics staff Manjunath, Gayathri, Chaithanya and Mahadeva along with other RRI administration staff for their support. I would also like to thank library staff for the kind of environment they created at library, where I like to spend most of my time.

I would like to thank Alok Laddha for allowing me to deliver a couple of lectures at Chennai Mathematical Institute and everyone at CMI for asking interesting questions and providing valuable feedback.

Finally I would like to thank my mother and sister for providing their full support at home.

\tableofcontents




\chapter{Introduction}\label{ch.int}

Quantum Field Theory (QFT) and General Relativity (GR) are two of the major areas of interest in modern theoretical physics. They together, as per our present understanding, provide the most fundamental description of our universe. 
QFT is used to understand the behaviour of fundamental particles and their interactions, and is invoked at small length scales. Einstein's GR on the other hand is a theory of spacetime itself. It is used to understand the influence of matter on the spacetime structure and vice versa, and is invoked either at large length scales or in the presence of high matter density leading to high gravitational effects.  
Despite their individual success in describing the physical reality in their respective regime, they are not fully compatible with each other.

In QFT quantum fields (degree of freedom specified at each point in space) are fundamental and particles are just the field excitations over the ground state. QFT is studied in an observer dependent way. In Minkowski spacetime we have Poincare symmetry to guide us towards a preferred family of inertial observers, which leads us to a preferred Poincare invariant vacuum. This Minkowski vacuum is considered as the bedrock of the theory and the Poincare invariance is used to explain many of its features.

Studying quantum field theory in curved spacetime backgrounds is seen as a step towards a theory of quantum gravity. It describes physics in a regime where both quantum and gravitational effects are prominent, for example in the presence of a black hole or in the post-Planckian early universe. 

In extending QFT to a curved spacetime with no spacetime symmetry to guide us towards a preferred observer, the choice of the vacuum and therefore the notion of particles is arbitrary. Even in de Sitter spacetimes, which are maximally symmetric, instead of a unique symmetry-dictated vacuum, there exists a two parameter family of non-equivalent de Sitter invariant $\alpha$-vacua \cite{Mottola:1984ar,Allen:1985ux}. In Algebraic Quantum Field Theory (AQFT)\cite{fewster2020algebraic} primary role is taken by the algebra of observables and the state is defined as a complex function on this algebra. The choice of vacuum is relegated to the choice of the representation, which is not unique. This non-uniqueness of the QFT vacuum leads to interesting phenomenological consequences like the Unruh effect \cite{Unruh:1976db}, the Hawking effect \cite{Hawking:1974rv,Hawking:1975vcx} etc. In these cases of physical interest, we work with a preferred notion of vacua which has some physical significance. For example in the ``Unruh effect" in flat spacetime, we express the standard Minkowski vacuum, which is the vacuum state in the representation adapted to the family of inertial observers, as a state in the representation adapted to an accelerating observer. But what about studying QFT in spacetimes which lack any symmetry and therefore a notion of vacuum associated with a physically interesting family of observers? It is therefore important to have a symmetry independent formulation of QFT leading to a preferred vacuum. One such proposal for a unique (observer independent) vacuum for a real free scalar QFT was made by Sorkin \cite{sorkin} and Johnston \cite{Johnston:2009fr} in which the vacuum, called as the Sorkin-Johnston(SJ) vacuum, is defined via its two point correlation function, also called as the Wightman function, as the positive semidefinite part of the field commutator\footnote{This is discussed in detail in section~\ref{sjint.sec}.}. This came out as a result of attempts of studying QFT on causal sets, where we have a notion of the retarded Green's function due to \cite{johnston,Nomaan:2017bpl}, which is then used to evaluate the field commutator and subsequently the SJ Wightman function. This procedure was then shown to be applicable to continuum spacetimes as well \cite{aas}. Computing the positive semidefinite part of the field commutator requires one to promote it to a convolution operator on the space of complex functions of compact support and then to find its eigenvalues and eigenfunctions\footnote{See section~\ref{sjcstint.sec} for details.}. This is a challenging task particularly in a spacetime without any symmetry to help us in guessing the eigenfunctions. As a result there are very few cases in which the SJ vacuum has been obtained explicitly. These include massless scalar field in a 2d causal diamond \cite{johnston,Afshordi:2012ez} and trouser spacetime \cite{Buck:2016ehk}, massive scalar field in ultrastatic slab spacetimes \cite{fewster2012} and conformally ultrastatic spacetimes with a time dependent conformal factor \cite{Brum:2013bia}. In \cite{fewster2012}, the authors further showed that the SJ vacuum is in general not Hadamard, which lead the authors of \cite{Brum:2013bia} to suggest a modification to the original SJ construction, to make it Hadamard in the ultrastatic slab spacetimes, by smoothening the field commutator at the boundaries of the spacetime. Though this modification cures the non-Hadamard behaviour of the vacuum, it comes at the cost of the uniqueness of the SJ vacuum. The SJ formalism being observer independent does not admit a straightforward description of physical phenomena like the Unruh effect, Hawking effect, cosmological particle production etc.

In this thesis we also study the entanglement entropy of quantum fields, which is an important quantity of interest in the study of QFT. Entanglement entropy of quantum fields across a black hole horizon is proposed as a candidate for the Bekenstein-Hawking entropy \cite{bkls,brink2019jacob}. The standard approach to compute the entanglement entropy is to extend the notion of the von Neumann entropy to the quantum field system. This approach requires one to define the density operator corresponding to the vacuum state on a Cauchy hypersurface and therefore cannot be extended to AdS or topology changing spacetimes or quantum gravity theories like causal set theory which does not admit a Cauchy hypersurface. A spacetime formulation of the entanglement entropy was therefore required which was developed by Sorkin \cite{ssee} and was shown to be in agreement with the von Neumann entropy for a system in a Gaussian state with finite degrees of freedom. As opposed to the von Neumann entropy which makes use of the density operator defined on a section of a Cauchy hypersurface, the Sorkin's spacetime entanglement entropy (SSEE) is evaluated in an open subset of the spacetime using the two point correlation function restricted to the subset\footnote{See section~\ref{sseeint.sec} for some discussion on this.}.

This thesis is broadly divided into two parts. In the first part of the thesis, which consists of chapter \ref{ch.2ddiamsj} and \ref{ch.sjcosm}, we will study the SJ formalism for a free real scalar field theory in a bounded (finite spacetime volume) globally hyperbolic spacetime. We explicitly compute the SJ vacuum for a small mass scalar field in a 2d causal diamond (chapter \ref{ch.2ddiamsj}) and for a conformally coupled massless scalar field in a slab of cosmological spacetimes (chapter \ref{ch.sjcosm}) and study their properties. We compare the SJ vacuum in these spacetimes with known vacua. In the second part of the thesis, which consists of chapter \ref{ch.sseeds} and \ref{ch.sseecyl}, we study the SSEE. 
In chapter \ref{ch.sseeds} we study a special class of spacetimes where the SSEE can be evaluated analytically. In particular we compute the SSEE in the static patch of de Sitter spacetime and compare our results with that of Higuchi and Yamamoto \cite{Higuchi:2018tuk}. We also compute the SSEE in the static region of Schwarzschild de Sitter spacetime. In chapter \ref{ch.sseecyl} we explicitly compute the SSEE for a massless scalar field in a causal diamond embedded in a cylinder spacetime using a combination of analytical and numerical techniques and compare our results with the entanglement entropy evaluated by Calabrese and Cardy \cite{cc}.

We start this chapter with a review of the standard approach to a QFT vacuum in Sec.~\ref{stdqft.sec}, followed by a review of the SJ approach in Sec.~\ref{sjint.sec}. We then review the extension of the SJ formalism to causal sets in Sec.~\ref{sjcstint.sec}, which was the original motivation of the formalism. We end this chapter with a brief review of the SSEE in Sec.~\ref{sseeint.sec}.

\section{The standard approach to QFT}\label{stdqft.sec}
We now discuss the standard approach to quantizing fields and defining a QFT vacuum. One can refer to the book by Birrell and Davies \cite{birrell} and Wald \cite{wald} for details on the subject. We start with identifying the degrees of freedom, which in our case is a scalar field $\Phi(x)$ defined on a globally hyperbolic spacetime $(\cM,g)$, where $\cM$ is the spacetime manifold and $g$ is the metric on it, and the field equation governing its classical dynamics, which in our case is the Klein-Gordon equation
\eq{\Bkg\Phi\equiv (\square-\m^2)\Phi = 0.\label{kg.eq}}
$\square = \nabla_\mu\nabla^\mu$ is the d' Alembertian operator and $\m$ is the mass of the field. We find the solution space $\kker(\Bkg)$ of the Klein-Gordon operator, since in a globally hyperbolic spacetime $(\cM,g)$ there is a one-one onto map between the space of solutions of $\Bkg$ in $(\cM,g)$ and phase space on a Cauchy hypersurface $\spac$. We write a general solution of the Klein-Gordon equation in terms of the Klein-Gordon orthonormal modes $\{\bv_k\}$
\eq{\Phi(x) = \sum_k\left(a_k \bv_k(x) + a_k^* \bv_k^*(x)\right),\label{phexp.eq}}
with
\eq{(\bv_j,\bv_k) = -(\bv_j^*,\bv_k^*) = \delta_{jk}\quad\text{and}\quad (\bv_j,\bv_k^*)=0,}
where
\eq{(f,g) = i\int_{\spac} d\spac^a \; \left(f^*\partial_a g - g\partial_a f^*\right), \label{kgip.eq}}
is the Klein-Gordon inner product on $\kker(\Bkg)$. $d\spac^a$ is the volume element on the spacelike hypersurface $\spac\in\cM$ with respect to a future pointing unit normal. It can be shown that if $f$ and $g$ are the solutions of the Klein-Gordon equation, then $(f,g)$ is independent of the choice of the Cauchy hypersurface and is therefore uniquely determined in the spacetime. The corresponding quantum theory is obtained by mapping the field $\Phi(x)$ for all $x\in\cM$ to an operator valued distribution $\hPhi(x)$ on an appropriate state space and imposing the following commutator condition
\eq{[\hPhi(x),\hPhi(x')] = i\Delta(x;x)\label{pbcond.eq}}
where $\Delta$ is the difference between the retarded and the advanced Green's function of the Klein-Gordon operator, and is called as the Pauli-Jordan (PJ) function
\eq{\Delta(x;x') \equiv G_R(x;x')-G_A(x;x') \label{pj.eq}.}
Here $G_A = G_R^T$. The coefficients $a_k$ and $a_k^*$ in Eqn.~\eqref{phexp.eq} then gets mapped to the annihilation and creation operators $\ha_k$ and $\ha_k^\dagger$ respectively satisfying the following commutation relation,
\eq{[\ha_j,\ha_k^\dagger] = \delta_{jk}\quad\text{and}\quad [\ha_j,\ha_k] = [\ha_j^\dagger,\ha_k^\dagger]=0.\label{caocc.eq}}
The vacuum $\vacket$ is then defined as a state which is annihilated by all annihilation operators, i.e., $\ha_k\vacket=0$ for all $k$. The state space, which is also called as the Fock space, is then generated by acting creation operators $\ha_k^\dagger$ on the vacuum. It can be checked that the basis states 
\eq{|(n_1)_{k_1},(n_2)_{k_2},\ldots\rangle = \frac{1}{\sqrt{n_1!n_2!\ldots}}(\ha_{k_1}^\dagger)^{n_1}(\ha_{k_2}^\dagger)^{n_2}\ldots\vacket}
 of the Fock space are the eigenstates of the number operator $N_k = \ha_k^\dagger \ha_k$ with non-negative integer eigenvalues, which can be thought of as the particle number. For example $N_j$ acting on $\ha_k^\dagger\vacket$ is
\eq{N_j\ha_k^\dagger\vacket = \ha_j^\dagger \ha_j \ha_k^\dagger\vacket = \ha_j^\dagger [\ha_j,\ha_k^\dagger]\vacket = \delta_{jk}\ha_k^\dagger\vacket.}
Therefore $\ha_k^\dagger\vacket$ is a single $k$-particle state. At this point we observe that the choice of the vacuum and the particle states depends on the choice of the mode decomposition in Eqn.~\eqref{phexp.eq}. In Minkowski spacetime plane waves modes serves as a preferred set of Klein-Gordon orthonormal modes, i.e., $\bv_k(x) = \frac{1}{\sqrt{4\pi\omega_k}}e^{-i\omega_k t +\vec{k}.\vx}$, where $\omega_k=\sqrt{|\vec{k}|^2+m^2}$. Therefore we have a preferred notion of a vacuum and particle states. In contrast, in a generic curved spacetime instead of having a preferred set of modes, we have a collection of all equally valid set of Klein-Gordon orthonormal modes
\eq{\hPhi(x) =  \sum_k\big(\ha_k \bv_k(x) + \ha_k^\dagger \bv_k^*(x)\big) =  \sum_p\big(\hb_p \bu_p(x) + \hb_p^\dagger \bu_p^*(x)\big)}
which are related to each other by a Bogoliubov transformation.
\eq{\bv_k(x) = \sum_p \big(\alpha_{kp}\bu_p(x) + \beta_{kp}\bu_p^*(x)\big),\quad \ha_k = \sum_p \big(\alpha_{kp}^*\hb_p - \beta_{kp}^*\hb_p^\dagger\big)}
where the Bogoliubov coefficients $\alpha_{kp}$ and $\beta_{kp}$ satisfies the following relations
\eq{
\sum_\bk \left(\alpha_{kp}\alpha_{kp'}^* - \beta_{kp}^*\beta_{kp'}\right) &= \delta_{pp'},\label{eq:bc1}\\
\sum_\bk \left(\alpha_{kp}\beta_{kp'}^* - \beta_{kp}^*\alpha_{kp'}\right) &=0,\label{bc2.eq}
}
which are obtained by imposing the commutation relation Eqn.~\eqref{caocc.eq} on both $\ha_k,\ha_k^\dagger$ and $\hb_p,\hb_p^\dagger$.

The vacuum corresponding to $\{\bv_k\}$ and $\{\bu_p\}$ modes are defined using $\ha_k\vacket_a=0$ and $\hb_p\vacket_b=0$ respectively which are different unless $\beta_{kp}=0$ for all $k$ and $p$.

So we see that in a curved spacetime, there are multiple choices of QFT vacua. In the next section we will look at a particular way of selecting a vacuum uniquely in an arbitrary bounded globally hyperbolic spacetime without appealing to any of its spacetime symmetry.

\section{The Sorkin-Johnston approach to QFT}\label{sjint.sec}

A free field vacuum $\vacket$ is a Gaussian state, i.e., it is completely determined by its two point correlation function $W(x;x')=\vacbra\hPhi(x)\hPhi(x')\vacket$ also called as the Wightman function. Higher order correlation functions can be obtained from Wightman function using Wick's rule which is as follows
\eq{\vacbra\hPhi\vacket=0,\quad \vacbra\hPhi\hPhi\ldots\hPhi\vacket = \sum  \vacbra\hPhi\hPhi\vacket\ldots\vacbra\hPhi\hPhi\vacket,}
where the sum is over all ways of pairing of $\hPhi's$ and within each pairing the original order is preserved. The Wightman function for a free real scalar field satisfies the following properties
\begin{itemize}
\item $W$ is Hermitian i.e. $W(x;x')=W^*(x';x)$.
\item $W-W^*=i\Delta$. This directly follows from the commutator condition Eqn.~\eqref{pbcond.eq} along with the Hermiticity of $W$.
\item $W$ is positive semidefinite, i.e., $\left<f, W\circ f\right>\geq 0$ for all test functions $f:\cM\rightarrow\mathbb{C}$, where $\left<.,.\right>$ denotes the $\cL^2$ inner product in the space of test functions with compact support $\cF(\cM,g)$,
\eq{\left<f,g\right> = \int_\cM dV_x f^*(x)g(x),\label{l2.eq}}
$dV_x$ is a volume element in $(\cM,g)$ and $\circ$ denotes an integral operation which is defined as follows,
\eq{A\circ f (x) = \int_\cM dV_x A(x;x')f(x').\label{intop.eq}}
Positive semidefiniteness of $W$ can be seen by writing $\left<f, W\circ f\right>$ as 
\eq{\left<f, W\circ f\right> =\int_\cM dV_x dV_{x'}f^*(x)f(x')\vacbra\hPhi(x)\hPhi(x')\vacket= \left|\left|\int_\cM dV_x f(x)\hPhi(x)\vacket\right|\right|^2,}
which is the norm of a state $\int_\cM dV_x f(x)\hPhi(x)\vacket$ and hence is expected to be greater than zero.
\end{itemize}
In terms of the modes $\{\bv_k\}$, the Wightman function takes the form
\eq{W(x;x') = \sum_k \bv_k(x)\bv_k^*(x').}

In a static spacetime $(\cM,g_s)$ where $\cM=\mr\times\spac$ and $g_s$ describes the spacetime metric leading to the invariant length element given by
\eq{ds^2 = -p(\vx)dt^2 + q_{ij}(\vx)dx^idx^j,\quad p(\vx)>0\;\;\forall\;x\in\spac,\quad t\in(-\infty,\infty)}
there exists a preferred set of positive frequency modes
\eq{\bv_k(x) = e^{-i\omega_k t}X_k(\vx)}
for which the Wightman function satisfies an additional condition $W\circ W^*=0$ given that it is well defined, which is true for a spacetime with compact spatial hypersurface. This is termed as the ground state condition. Further we claim that the ground state condition uniquely determines the Wightman function. An argument in favour of the above statement can be found in \cite{Afshordi:2012ez}, which is as follows. If possible let there be two vacua $W_1$ and $W_2$ satisfying the ground state condition, i.e, \footnote{We have omitted $\circ$ but it is understood to be there.} \eq{W_1W_1^* = W_2W_2^* = 0\label{gs2.eq}.} The commutator condition implies that
\eq{W_1-W_1^* = W_2-W_2^* \Rightarrow (W_1-W_1^*)^2 = (W_2-W_2^*)^2\label{cc2.eq}}
Eqn.~\eqref{gs2.eq} and \eqref{cc2.eq} implies that
\eq{(W_1+W_1^*)^2 = (W_2+W_2^*)^2.\label{re2.eq}}
Assuming that there is a unique positive definite square root of $(W+W^*)^2$, Eqn.~\eqref{re2.eq} implies
\eq{W_1+W_1^* = W_2+W_2^*.} This along with the commutator condition implies that $W_1=W_2$.

In the SJ formalism, we define the SJ vacuum via its Wightman function $W(x;x')$ by imposing the four defining conditions on it which are $(i)$ Hermiticity, $(ii)$ the commutator condition, $(iii)$ positive semidefinite condition and $(iv)$ the ground state condition. The first three conditions are true for any vacua whereas the fourth condition is motivated by a standard vacuum in static spacetimes. These conditions in a bounded globally hyperbolic spacetime leads to a unique vacuum called as the Sorkin-Johnston (SJ) vacuum, which is the positive part $\mathrm{Pos}(i\hD)$ of $i\hD$,
\eq{\wsj = \frac{1}{2}\left(\sqrt{(i\Delta)^2}+i\Delta\right) = \mathrm{Pos}(i\Delta),\label{idtosj.eq}}
where $\sqrt{(i\Delta)^2}$ is defined as a positive definite integral operator whose square is $(i\Delta)^2$, i.e., $\sqrt{(i\Delta)^2}\circ \sqrt{(i\Delta)^2} = (i\Delta)^2$. $i\Delta$ is imaginary and anti-symmetric and is therefore an Hermitian operator. Further in a bounded spacetime region, it is self-adjoint \cite{aas}. Eigenvalues of $i\Delta$ come in positive-negative pairs
\eq{i\Delta\circ u_k &= \lambda_k u_k,\label{cep.eq}\\ i\Delta\circ u_k^* &= -\lambda_k u_k^*.}
Note here that the boundedness (finite volume) of the spacetime $(\cM,g)$ ensures that $\lambda_k$'s are finite and $u_k$'s are $\cL^2$ normalisable\footnote{Check \cite{aas} for details}. $i\Delta$ can then be written as
\eq{i\Delta(x;x') = \sum_k \lambda_k \big(u_k(x)u_k^*(x') - u_k^*(x)u_k(x')\big),\label{idmode.eq}}
where $\lambda_k>0$ and $u_k$'s are assumed to be $\cL^2$ normal. The SJ Wightman function is the positive part of $i\Delta$, i.e.,
\eq{\wsj(x;x') = \sum_k \lambda_k u_k(x)u_k^*(x'),\label{eq:wsj}}
where $\lambda_k$'s are the positive eigenvalues of $i\Delta$ with $u_k$ being the corresponding eigenfunction. The problem of computing the SJ vacuum is now reduced to finding eigenvalues and eigenfunctions of the Pauli-Jordan operator. Since $\im(i\Delta)=\kker(\Bkg)$ \cite{wald,sorkin17}, the field operator $\hPhi$ can be expanded in terms of the SJ modes $\{\sqrt{\lambda_k}u_k\}$ as
\eq{\hPhi(x) = \sum_k \sqrt{\lambda_k}\big(\ha_k u_k(x) + \ha_k^\dagger u_k^*(x)\big),}
and the SJ vacuum $\ksj$ satisfies $\ha_k\ksj=0$.\footnote{One can check that eqn.~\eqref{idmode.eq} implies $[\ha_j,\ha_k]=[\ha_j^\dagger,\ha_k^\dagger]=0$ and $[\ha_j,\ha_k^\dagger]=\delta_{jk}$.} 

The algorithm for constructing the SJ vacuum can be summarised as follows
\eq{G_R(x;x')\rightarrow i\Delta(x;x')\rightarrow \wsj(x;x')=\mathrm{Pos}(i\Delta),\label{sjalgorithm.eq}}
i.e., we start with the retarded Green's function of the Klein-Gordon operator, which is unique in a globally hyperbolic spacetime. Starting from the retarded Green's function, one can construct the Pauli-Jordan function, as in Eqn.~\eqref{pj.eq}. The SJ vacuum using the Pauli-Jordan function can then be uniquely defined, as in Eqn.~\eqref{idtosj.eq}, which is always well defined in a bounded spacetime. SJ vacuum in an unbounded spacetime $(\cM,g)$ can be obtained by computing it in a bounded subspace $(\R,g)$ of $(\cM,g)$ and then taking the following limit
\eq{\wsj^{(\cM)}:=\lim_{\R\rightarrow\cM} \wsj^{(\R)},}
if the limit exist\footnote{This limit may not always exist. For example it does not exist in global de Sitter spacetime of even dimensions. We will see this in detail in chapter~\ref{ch.sjcosm}.}. Therefore to rigorously show that $\ksj$ reduces to a known vacuum in an unbounded spacetime one has to compare the full spacetime limit of the SJ Wightman function obtained in its bounded subset with that of the known vacuum. In \cite{aas} a mode comparison argument was used to show
that the SJ vacuum in Minkowski spacetime is the Minkowski vacuum. However, as argued in \cite{Surya:2018byh} a mode comparison may not indicate the equivalence of vacua. A more careful approach was adopted  in \cite{Afshordi:2012ez} where the massless SJ vacuum was calculated explicitly in a 2d causal diamond $\diam$ of length $2L$.  Evaluating $\wsj$ in the center of the diamond, i.e., with  $|\vec x -\vec x'| <<L $  and $|\vec x|, |\vec x'| <<L $ it  was shown that $\ksj \sim \vac_{\text{mink}}$.  Thus, away from the boundaries, the massless SJ vacuum is indeed the Minkowski vacuum. In chapter~\ref{ch.2ddiamsj} of this thesis, we perform a similar calculation for a massive scalar field with a small mass in a finite diamond, in which the SJ construction is well defined and study its behaviour in the center and the corner of the diamond.

As shown in \cite{fewster2012} the SJ vacuum is not necessarily Hadamard which implies that the nature of the UV divergence is different from that of flat spacetime. It was however shown by Brum and Fredanhagen \cite{Brum:2013bia} that an infinite family of Hadamard states can be constructed from the SJ state, by embedding the spacetime $(\cM,g)$ in a slightly larger spacetime, along with a deformation of the Pauli-Jordan operator, or as shown subsequently, by a deformation the $\cL^2$ measure on the space of test functions itself \cite{wingham}\footnote{The SJ vacuum was first constructed on a causal set, which can be thought of as a ``UV completion'' of continuum spacetime \cite{johnston,sorkin}. The  modification of the UV behaviour is therefore not unexpected. For example, an analysis of the causal set de Sitter SJ vacuum in  \cite{Surya:2018byh} suggests that it is significantly  different from the standard $\alpha$-vacuum of the continuum.}.
  

\section{SJ vacuum on a causal set}\label{sjcstint.sec}
In this section, we will go through the construction of the SJ vacuum on a causal set approximated by a bounded spacetime with a given spacetime dimension. On a causal set, due to lack of the continuum structure, it is difficult to construct an analogue of differential operators, killing vectors and therefore spacetime symmetry. It is however shown that in some special cases, one can construct the retarded Green's function on a causal set \cite{johnston,Nomaan:2017bpl,X:2021dsa}, which then directly leads to the SJ vacuum following the path shown in Eqn.~\eqref{sjalgorithm.eq}. This was the original motivation for the SJ formalism.

\subsection{A brief introduction to causal sets}\label{csrev.sec}
Causal set theory is an approach to Quantum Gravity which hypothesizes that the Lorentzian manifold structure of the spacetime is a large scale approximation to a more fundamental structure, which is that of a causal set \cite{Bombelli:1987aa}. A causal set $(\cC,\preceq)$ is a locally finite partially ordered set (poset) with the partial order $\preceq$ satisfying the following properties.
\begin{enumerate}
        \item {\sl{Reflexivity:}} For any $a\in\cC$, $a\preceq a$.
       \item {\sl{Transitivity:}} For any $a,b,c \in \cC$, $a\preceq b$ and $b\preceq c$ $\Rightarrow a\preceq c$.
       \item {\sl{Antisymmetry:}} For any $a, b \in \cC$, $a\preceq b$ and $b\preceq a$ implies $a=b$.
       \item {\sl{Local finiteness:}} For all $a,b\in\cC$, the set $\{z\in\cC|a\preceq z\preceq b\}$ is of finite size.
   \end{enumerate}
We further define $a\prec b$ iff $a\preceq b$ and $a\neq b$. Here the partial order $\prec$ represent causal relations, i.e., $a\prec b$ means that $b$ is in the causal future of $a$. Conditions 1-3 are satisfied by causal relations on any Lorentzian manifold whereas condition 4 implies spacetime discreteness. For any $a,b\in\cC$, we define the order interval $\oi[a,b]$ as
\eq{\oi[a,b] := \{z\in\cC|a\prec z\prec b\}}
which in the continuum approximation represents the Alexandrov interval (causal diamond) between $a$ and $b$. The cardinality of the order interval $|\oi[a,b]|$ can then be thought of as proportional to the volume of the causal diamond between $a$ and $b$. These information (causal structure + volume) is sufficient to determine the Lorentzian geometry of the spacetime, which is a result based on theorems by Hawking, King, McCarthy and Malament \cite{Hawking:1976fe,Malament:1977}.

A causal set can be characterised by its {\sl{causal matrix}} $C$ or {\sl{link matrix}} $L$, which are defined as
\eq{C_{ab} = \begin{cases} 1 &\text{if}\; b\prec a \\ 0 &\text{otherwise} \end{cases},\quad L_{ab} = \begin{cases} 1 &\text{if}\; b\link a \\ 0 &\text{otherwise} \end{cases}}
where $\link$ denotes a {\sl{link}}, which is defined as $a\link b$ iff $a\prec b$ and $|\oi[a,b]|=0$.

A causal set approximated by a given spacetime manifold is obtained by the process of sprinkling. It is the process of picking points from a spacetime $(\cM,g)$ using a random Poisson discretisation such that the probability of picking $n$ points from a region of spacetime volume $V$ is
\eq{P(n;V) = \frac{(\rho V)^n}{n!} e^{-\rho V},}
where $\rho$ is the sprinkling density. This gives $\left<n\right> = \rho V$, where $\left<.\right>$ denotes average over sprinklings. The causal ordering $\prec$ is then simply inherited from that of the spacetime $(\cM,g)$. $\rho^{-1}$ defines the discreteness scale and continuum limit is defined as $\rho\rightarrow\infty$. One may refer to \cite{Bombelli:1987aa,Surya:2019ndm} for more details on causal set theory.

\subsection{Scalar field Green's function on a causal set}
Retarded Green's function $G_c(x;x')$ for a scalar field on a causal set $(\cC,\preceq)$ approximated by $d=2$ and $d=4$ Minkowski spacetimes ($\mathbb{M}^2$ and $\mathbb{M}^4$) was studied by Johnston in \cite{johnston}, which was later extended to de Sitter and anti de Sitter spacetimes in \cite{Nomaan:2017bpl}, in terms of the quantities like causal matrix and link matrix which are used to characterise a causal set. This is done by requiring that the continuum limit ($\rho\rightarrow\infty$) of the causal set Green's function averaged over sprinklings on the corresponding continuum spacetime is equal to the retarded Green's function of the Klein-Gordon operator in that continuum spacetime, i.e.,
\eq{\lim_{\rho\rightarrow\infty} \left<G_c(x;x')\right> = G_R(x;x').}
Retarded Green's function is the solution of
\eq{\left(\square_x - m^2\right)G_R(x;x') = \frac{\delta(x-x')}{\sqrt{-g}}}
with the boundary condition $G_R(x;x') = 0$ for all $x$ and $x'$ such that $x_0<x'_0$, where $x_0$ is the time coordinate of $x$. Here $g$ is the determinant of the spacetime metric tensor.
In $\mathbb{M}^2$ the retarded Green's function for a massless scalar field is
\eq{G_{R,0}^{(2)}(x;x') = -\frac{1}{2}\theta(x_0-x'_0)\theta(\tau^2(x;x')),}
where $\tau^2(x;x') = (x_0-x'_0)^2-(x_1-x'_1)^2$ with $x_0$ and $x_1$ being the temporal and the spatial coordinates of $x$.
The corresponding causal set Green's function for massless scalar field is
\eq{G_{c,0}^{(2)}(x;x') = -\frac{1}{2}C_{xx'}.\label{csgf2d.eq}}
In $\mathbb{M}^4$ the retarded Green's function for a massless scalar field is
\eq{G_{R,0}^{(4)}(x;x') = -\frac{1}{2}\theta(x_0-x'_0)\delta(\tau^2(x;x')),}
where $\tau^2(x;x') = (x_0-x'_0)^2-(x_1-x'_1)^2-(x_2-x'_2)^2-(x_3-x'_3)^2$.
The corresponding causal set Green's function for massless scalar field is
\eq{G_{c,0}^{(4)}(x;x') = -\frac{1}{2\pi}\sqrt{\frac{\rho}{6}}L_{xx'}.\label{csgf4d.eq}}

We obtain the retarded Green's function for a massive scalar field from that of a massless scalar field as
\eq{G_{R,m} &= G_{R,0}-m^2 G_{R,0}\circ G_{R,0} + m^4 G_{R,0}\circ G_{R,0}\circ G_{R,0} + \ldots,\nonumber\\ &= G_{R,0}-m^2 G_{R,0}\circ G_{R,m}.\label{g0togm.eq}}
This can be checked by operating the massive Klein-Gordon operator on both sides of Eqn.~\eqref{g0togm.eq}. We can therefore define the massive causal set Green's function in terms of the massless one as
\eq{G_{c,m} &= G_{c,0}-\frac{m^2}{\rho} G_{c,0}*G_{c,0} + \frac{m^4}{\rho^2} G_{c,0}*G_{c,0}*G_{c,0} + \ldots,\nonumber\\ &= G_{c,0}-\frac{m^2}{\rho}G_{c,0}*G_{c,m},\nonumber\\&= \left(\mathbb{I} +\frac{m^2}{\rho} G_{c,0}\right)^{-1}*G_{c,0}. \label{g0togmcs.eq}}
Here $*$ denotes the matrix multiplication. The operator $\rho^{-1}*$ is the causal set analogue of the integral operator $\circ$ defined in Eqn.~\eqref{intop.eq}, where $\rho^{-1}$ is the causal set analogue of the volume element $dV$. We use Eqn.~\eqref{g0togmcs.eq} to get the massive causal set Green's function in terms of the massless one, which then acts as the starting point towards the computation of the QFT vacuum state for a massive scalar field on a causal set approximated by $d=2$ and $d=4$ Minkowski, de Sitter and anti de Sitter spacetimes.

\subsection{SJ vacuum}
Now we have the causal set Green's function, we just have to follow the algorithm of Eqn.~\eqref{sjalgorithm.eq} to get to the causal set SJ Wightman function. The Pauli-Jordan matrix can simply be written as the difference between the causal set Green's function and its transpose,
\eq{\hD_c = G_c - G_c^T.}
We then solve for the eigenvectors $u_k$ and eigenvalues $\lambda_k$ of $i\hD_c$ and evaluate the SJ Wightman function as
\eq{\wsjc = \mathrm{Pos}(i\Delta_c)= \sum_{\lambda_k>0}\lambda_k u_k u_k^\dagger,}
where the sum is over positive eigenvalues of $i\Delta_c$.

\section{Spacetime Entanglement Entropy of Quantum Fields}\label{sseeint.sec}
Entanglement entropy of quantum fields is usually studied using the von Neumann formula, which is in terms of the density matrix ($\hrho$)\footnote{Not to be confused with the density of a causal set.} of the state. For a state $\hrho$ in a bipartite system with two subsystems $A$ and $B$, the von Neumann entropy $\cS_A$ of the subsystem $A$ is defined as a trace
\eq{\cS_A = - \tr\big(\hrho_A\log\hrho_A\big),\label{vne.eq}}
where $\hrho_A$ is defined as the trace of $\hrho$ over the states of $B$,
\eq{\hrho_A = \tr_B(\hrho) = \sum_{\{\left|\psi_b\right>\}\subset \cF_B}\left<\psi_b\right|\hrho\left|\psi_b\right>,}
where the sum is over basis $\{\left|\psi_b\right>\}$ of the state space $\cF_B$ of the subsystem $B$. One can check that if the state $\hrho$ in $\cF_A\otimes\cF_B$ is not entangled, i.e., it is of the form $\hrho = \big|\psi^{(1)}_A\psi^{(2)}_B\big>\big<\psi^{(1)}_A\psi^{(2)}_B\big|$ for some $\big|\psi^{(1)}_A\big>\in\cF_A$ and some $\big|\psi^{(2)}_B\big>\in\cF_B$, we have $\cS_A = \cS_B = 0$. The quantum field state in the von Neumann entropy formula is defined on a Cauchy hypersurface (which is a constant time surface) and therefore may have dependence on the choice of the hypersurface.

In most definitions of entanglement entropy, including the von Neumann entropy, one considers the entanglement of the state at a moment of time between two {\it spatial} regions,  
which is restrictive in the context of quantum gravity, or even quantum fields in curved
spacetimes which may lack a preferred time. A  more global, covariant notion
of the entanglement entropy could be more  useful when  working with a covariant path-integral or histories based approach to quantum gravity.  In particular, the notion of a state at a moment of time might not survive, especially in theories where the manifold structure of spacetime breaks down in the deep UV regime, as in causal set theory \cite{Bombelli:1987aa}. Such a formulation is moreover in keeping with the broader framework of AQFT, where observables  are associated with spacetime regions rather than spatial hypersurfaces \cite{fewster2020algebraic}.

A spacetime formula for the entanglement entropy of quantum fields in a Gaussian state restricted to a globally hyperbolic region $\cO\subset\cM$ is recently proposed by Sorkin \cite{ssee} in terms of the Wightman function. Starting from a pure Gaussian state Wightman function in $\cM$ and restricting it to $\cO$, the Sorkin's spacetime entanglement entropy (SSEE) associated with this state is defined as
\eq{\cS_\cO = \sum_\mu \mu\log|\mu|,\quad W\big|_{\cO}\circ\chi=\mu i\Delta\circ\chi,\quad \chi\not\in\kker(i\Delta).\label{ssee.eq}}
The Wightman function in Eqn.~\eqref{ssee.eq} describes the state of the system which is inherently observer dependent, unless we fix it to be the SJ Wightman function. But this choice is not always possible, as we will see in chapter~\ref{ch.sjcosm} that the SJ vacuum in even dimensional de Sitter spacetime is ill defined. This expression was obtained in \cite{ssee} by observing that for discrete spaces the field operators, which are the generators of the free field algebra,  can be expanded in a ``position-momentum'' basis which renders $i\hD$ into a $2\times 2$ block diagonal form and simultaneously diagonalises the symmetric part of $\hW$. Each block is a single particle system for which the von Neumann entropy for a Gaussian state can be calculated using results of \cite{bkls}. The expression can be then rearranged in terms of the eigenvalues of $\hD^{-1} \hW$ \footnote{A derivation of the same can also be found in \cite{Chen:2020ild,Keseman:2021dkf}.}. Summing over all the blocks gives the SSEE expression Eqn.~\eqref{ssee.eq}. The equivalence between the von Neumann entropy Eqn.~\eqref{vne.eq} and the SSEE Eqn.~\eqref{ssee.eq} for a Gaussian state was also obtained in \cite{tommaso}. 

We now sketch the construction of the SSEE \cite{ssee} in the continuum for a compact region $\cO \subset M$ in a spacetime
$(M,g)$. The Wightman function $\hW$ is positive semi-definite and self-adjoint  with respect to the $\cL^2$ norm. Hence $\ker(\hR) \subseteq \ker(\hD)$, or equivalently $\im(\hR) \supseteq \im(\hD)$, where $R:=(W+W^*)/2$ is the real part of W.

In what follows we  restrict to $\im(\hR)$.  Notice that  $\hR$ is symmetric and  positive definite for $f \in \im(\hR)$ and can therefore be viewed as a ``metric'', with a  symmetric inverse 
\begin{equation}
(\hR^{-1} \circ \hR) (x,x')  = (\hR\circ \hR^{-1})(x,x') = \delta(x-x').  \label{Rinv.eq}
\end{equation}
We can use it to  define the new operators 
\begin{eqnarray} 
  i \htD &\equiv&   i \hR^{-1}\circ \hD \Rightarrow  i \hD= i \hR \circ \htD  \nonumber \\
 \Omega &\equiv& - \htD \circ \htD,  
\end{eqnarray}
Since $\hD \circ  f \in \im(\hD) \subseteq  \im(\hR) $, $\htD$ is well defined and  $\ker \htD=\ker \hD$. The operator
$\Omega$ is moreover  positive semi-definite with respect to the $\hR$ norm
\begin{equation}
\rin{f}{g} = \int dVdV'f^*(x) R(x,x') g(x'). 
  \end{equation} 
since 
\begin{equation} 
\rin{ f}{  \Omega \circ f}   = \rin{f}{ i \htD \circ i \htD f } = \rin{ i \htD \circ f}{  i \htD
 \circ f} \geq 0  
\end{equation}
for  $f \in \im(\hR)$. Thus the eigenvalues $\{\sigma_k^2 \}$  of $\Omega$ are all positive.  They are moreover degenerate
since the eigenfunctions come in pairs $(\psi, \psi^*)$, which in turn means that we can use a  real eigenbasis $\{\pone_k,
\ptwo_k\}$, where  $\ptwo_k= \sigma_k^{-1} \htD \circ  \pone_k$, satisfying
\begin{equation}
\rin{\pone}{\pone} = \rin{\ptwo}{\ptwo}=1,  \quad \rin{\pone}{\ptwo}=0
  \end{equation} 
  Since $\Omega$  is symmetric and therefore self adjoint, it admits a  
spectral decomposition 
\begin{equation}
  \Omega (x,x') = \sum_k  \sigma_k^2\biggl(  \pone_k(x) \tpone_k(x')+  \ptwo_k(x) \tptwo_k(x')\biggr)
  \end{equation} 
where $\tpone_k(x) \equiv (\hR \circ \pone_k)(x), \tptwo_k(x) \equiv (\hR \circ \ptwo_k)(x) $. Hence 
\begin{equation}
  i\htD (x,x') = \sum_k  \sigma_k\biggl( \pone_k(x) \tptwo_k(x') - \ptwo_k(x) \tpone_k(x')  \biggr),  
  \end{equation} 
with the related two point function $\hK = \hR^{-1}\circ \hW = {\mathbf 1} + \frac{i}{2} \htD$,  having  the ``block'' diagonal form
\eq{ K(x,x') \! = \!\sum_k   \pone_k(x) \tpone_k(x') +  \ptwo_k(x) \tptwo_k(x')  +\frac{i}{2} \sigma_k\biggl( \pone_k(x)
  \tptwo_k(x') - \ptwo_k(x) \tpone_k(x')\biggr).  \label{blockdiagK.eq}}
Since each block is  decoupled from all
others it can be viewed as a single pair of harmonic oscillators whose entropy can then be calculated relatively easily.

Since $\im(\Omega)=\im(\hD)= \ker \Box$, the eigenfunctions $\{\pone_k, \ptwo_k \}$  of $\Omega$ can be used for  a mode
decomposition for the quantum field. 
In  the harmonic oscillator basis,
\eq{\hat \Phi (x) = \sum_k \biggl(\htq_k \tpone_k(x) + \htp_k \tptwo_k(x) \biggr), \quad
  [\htq_k,\htp_{k'}]=i \sigma_k\delta_{k,k'}.\label{hphq.eq}}
which is consistent with the commutator $i \Delta(x,x') = [\hat \Phi(x), \hat \Phi(x')]$. In this basis 
\begin{eqnarray} 
K(x,x') & =&  < (\hR^{-1} \circ \hat \Phi(x))  \hat \Phi(x') >  \nonumber \\ &=&  \sum_{k} \biggl( <\htq_k \htq_{k}> 
             \pone_k(x)  \tpone_k(x') +  
  < \htq_k \htp_{k}> 
  \pone_k(x) \tptwo_k(x') \nonumber \\
  && + <\htp_k \htq_{k}>  
\ptwo_k(x) \tpone_k(x')+  <\htp_k \htp_{k}> 
\ptwo_k(x)  \tptwo_k(x') \biggr). 
\end{eqnarray}
where the block diagonalisation is obtained by comparing with $\mathbf 1=\frac{1}{2} (\hK + \hK^*)$ and $\htD=\frac{1}{2 i} (\hK
- \hK^*)$.  For each $k$ therefore we  have a harmonic oscillator with 
\begin{equation} 
  \hK_k \equiv
  \begin{pmatrix}
    <\htq_k \htq_k> &   <\htq_k \htp_k>\\
    <\htp_k \htq_k>& <\htp_k \htp_k>
  \end{pmatrix}  =  \begin{pmatrix}
   1 &  \frac{ i}{2} \sigma_k\\
   -\frac{ i}{2} \sigma_k & 1 
  \end{pmatrix}, \quad  i\htD_k \equiv \begin{pmatrix}
   0 &  i \sigma_k\\
   -i \sigma_k & 0  
 \end{pmatrix},
 \label{singleharm.eq}
\end{equation}
The problem thus reduces to a single degree  of freedom  with $\hK_k$ representing the two-point correlation functions
for a single harmonic oscillator in a given state, so that we can from now on drop  the index $k$.

Associated with any Gaussian state is a density matrix of the form 
\begin{equation}
\rho(q,q') = \sqrt{\frac{A}{\pi}} e^{-\frac{A}{2} (q^2 +q'^2) + i \frac{B}{2} (q^2-q'^2)-\frac{C}{2} (q-q')^2}.\label{gdm.eq}
  \end{equation} 
  One can use the replica trick \cite{Chen:2020ild,Keseman:2021dkf} to find the von Neumann entropy of $\rho$ to be  
  \begin{equation}
    S= - \frac{ \mu \ln \mu + (1-\mu) \ln (1-\mu) }{1-\mu},\label{rts.eq}
  \end{equation}
  where $\mu=\frac{\sqrt{1+2C/A} -1}{\sqrt{1+2C/A} + 1}$. The task at hand is to find the density matrix
  associated with the state Eqn. \eqref{singleharm.eq}. 

Consider the position  eigenbasis $\{\ket{q}\}$ of $\htq$, in which $\htq\ket{q} = q\ket{q}$, $\htp\ket{q}=-i\sigma\partial_q\ket{q}.$
The correlators $<\hat\eta_a \hat \eta_b > = \mathrm{Tr}(\hat \eta_a \hat \eta_b \hrho)$, for $(\hat \eta_1,
\hat \eta_2) \equiv (\htq,\htp)$ can then be explicitly calculated in this basis, using $\bra{q} \hrho \ket{q'} =\rho(q,q')$ as in Eqn
\eqref{gdm.eq}.  The integrals reduce to  the Gaussians $\int dq \, q^2\,  e^{-Aq^2}$ and are easily evaluated to
give 
\eq{<\htq\htq> = 1/(2A), \quad 
<\htq\htp> = \frac{ i}{2}\sigma,  \quad, 
<\htp\htq> =-\frac{ i}{2}\sigma,  \quad, 
<\htp\htp> = \sigma^2(A/2+C), }
where we have used the fact that $<qp>$ is purely imaginary and $<qq>$ and $ <pp>$ are purely real from Eqn
\eqref{singleharm.eq} to put $B=0$.    
Equating to Eqn. \eqref{singleharm.eq} gives  $A=1/2$ and $C=1/\sigma^2 - 1/4$, so that $\mu=\frac{2-\sigma}{2+\sigma}$
for the associated single oscillator  von Neumann entropy Eqn. \eqref{rts.eq}.  
Following the work of \cite{Chen:2020ild}, we notice that  the eigenvalues of $(i\htD)^{-1}\hK$  are 
\eq{\mu^\pm = \frac{1}{2}\pm\sqrt{\frac{1}{4}+\frac{C}{2A}}=\frac{1}{2} \pm \frac{1}{\sigma},}
in terms of which Eqn.~\eqref{rts.eq} simplifies to 
\eq{S= \mu^+\ln|\mu^+| + \mu^-\ln|\mu^-|.}
Since the eigenvalues of $(i\htD)^{-1}\hK$ are same as the generalised eigenvalues of $\hW\circ\chi = i\mu\hD\circ\chi$,
we come to the SSEE formula Eqn.~\eqref{ssee.eq} after summing over all $k$.

The SSEE formula generalises the calculation of the entanglement entropy for a state at a given time to that associated with a spacetime region. SSEE is a covariant entropy formula, depending on the choice of the vacuum and the UV cut-off required to compute the SSEE, and is also readily available for an extension to quantum gravity theories, such as causal set theory, where the Cauchy hypersurface may not be well defined. 

This formula has been applied in the continuum to the $d=2$ nested causal diamonds and shown to give the expected Calabrese-Cardy logarithmic  behaviour \cite{saravani2014spacetime}. It has also been calculated in the discrete setting, i.e., for causal sets approximated by causal diamonds in Minkowski and  de Sitter spacetimes: in $d=2$ where it shows the expected logarithmic behaviour and in $d=4$ where it shows the expected area behaviour 
~\cite{Sorkin:2016pbz,Belenchia:2017cex,Surya_2021,Duffy:2021dtc,Keseman:2021dkf}. 

In chapter \ref{ch.sseeds} and \ref{ch.sseecyl} we evaluate the SSEE in the static patch of de Sitter spacetime and find it to be in agreement with the result of Higuhi and Yamamoto \cite{Higuchi:2018tuk}, and in a causal diamond embedded in a $d=2$ cylinder spacetime and find it to be in agreement with the von Neumann entropy evaluated by Calabrese and Cardy in an arc of a circle \cite{cc}. Based on this and previous works in \cite{saravani2014spacetime} we conclude that the von Neumann entropy of a quantum field state defined on a spacelike hypersurface $\spac$ and restricted to a subset $A$ of $\spac$ is same as the SSEE of a pure state defined in the domain of dependence $D(\spac)$ of $\spac$ and restricted to the domain of dependence $D(A)$ of $A$.\footnote{Domain of dependence $D(A)\subset\cM$ of a spacelike hypersurface $A$ is a region such that for any $x\in D(A)$, all inextendible causal curves which pass through $x$ intersect $A$ exactly once. $D(A)$ is then a globally hyperbolic region. $\cM$ is globally hyperbolic iff there exist a spacelike hypersurface $\spac\in\cM$ such that $D(\spac)=\cM$.}. 
\chapter{Sorkin-Johnston vacuum in the 2d causal diamond}\label{ch.2ddiamsj}
As seen in Chapter~\ref{ch.int}, it is straight-forward to define the SJ vacuum in a bounded globally hyperbolic spacetime, however the integral form of the eigenvalue equation makes it a challenging task solve for the SJ vacuum. As a result there are very few cases in which it has been obtained explicitly. These include the  massless free scalar field in a 2d flat
causal diamond \cite{Afshordi:2012ez,johnston}, in a patch of trousers spacetime \cite{Buck:2016ehk}, in ultrastatic slab spacetimes \cite{fewster2012} etc.
In this chapter we study the SJ vacuum for a massive free scalar field in the 2d flat causal diamond $\diam$ of length $2L$,  both in the continuum and on a  causal set $\cc_\diam$ obtained from sprinkling into $\diam$.

A causal diamond $\diam_{pq}\subset\cM$ is the intersection of the interior of the future light cone $J^+(p)$ of a point $p\in\cM$ and that of the past light cone $J^-(q)$ of a point $q\in J^+(p)$, i.e, $\diam_{pq} = J^+(p)\cap J^-(q)$. The causal diamond $\diam$ which we are here interested in corresponds to $p$ and $q$ with the same spatial coordinate. In lightcone coordinates
\eq{u=\frac{t+x}{\sqrt{2}}\quad\text{and}\quad v=\frac{t-x}{\sqrt{2}},\label{lcc.eq}}
$\diam$ is defined as the region with $u\in[-L,L]$ and $v\in[-L,L]$, the diamond therefore is of length $2L$ or height $2\sqrt{2}L$.
\begin{figure}[h!]
\centering{\includegraphics[height=6cm]{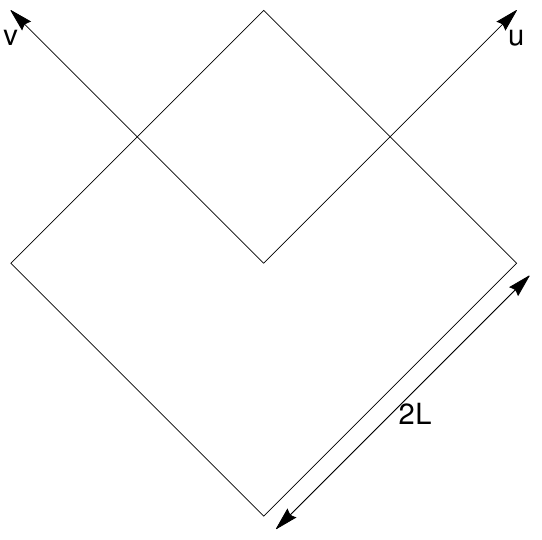}}
\caption{A causal diamond of length $2L$ with the lightcone $(u,v)$ coordinates shown.}
\label{diam1.fig}
\end{figure}

Since in the continuum causal diamond it is difficult to solve the central SJ eigenvalue problem Eqn.~\eqref{cep.eq} for a massive scalar field explicitly, we restrict our attention to a small mass scalar field by keeping terms only up to $\cO(m^4)$ (in dimensionless units, with $L=1$). The eigenfunctions and eigenvalues so obtained reduce  to their massless counterparts when $m=0$ \cite{Afshordi:2012ez}.  This allows us to formally construct the SJ Wightman function $\wsj$ in $\diam$ upto this order. As in \cite{Afshordi:2012ez} we consider two  regimes of interest: one in the center of the diamond, which resembles 2d Minkowski spacetime, and the other in one of the corners (right or left), which resembles a Rindler Wedge.  In a small central region $\diam_l$ of size $l$,   we find analytically that $\wsj$ resembles the {\it massless} Minkowski vacuum $\wmink$ with an IR cut-off up to a
small mass-dependent constant $\emc$, rather than the massive Minkowski vacuum $\wminkm$.  In the corner,  
$\wsj$ resembles the massive mirror vacuum  $\wmirrm$, with the difference depending on a small  mass-dependent constant $\emcr$,
rather than the expected agreement with the massive Rindler vacuum $\wrindm$.   Both $\emc$ and $\emcr$ are the errors
that arise in the approximation of quantization conditions which are solutions of a mass dependent transcendental equation,  and
are therefore non-trivial to calculate analytically.  

In order to find  $\e_m^{center}, \e_m^{corner}$ we evaluate 
$\wsj$ numerically using a convergent  truncation $\wsj^t$ of the mode-sum. The
calculations show that $\emc, \emcr$ contribute negligibly to $\wsj$ both in the center and the corner.   
This confirms that for a small mass $\wsj$ corresponds to the massless  
Minkowski vacuum. 
In the corner, again $\e_m^{corner}$ is found to be small, and confirms that  $\wsj$ resembles  $\wmirrm$ rather than
$\wrindm$.   

We then examine the behavior of this truncated $\wsj^{t}$ in a slightly enlarged  region in the center. We find that it continues to  differ  from $\wminkm$, while agreeing with $\wmink$ at least up to $l\sim 0.1$.   In an
enlarged corner region  $\wsj$ shows a marked deviation from $\wmirrm$, but it still does not resemble the Rindler
vacuum.

Later in this chapter we obtain the SJ Wightman function $\wsjc$ numerically for a causal set $\cc_\diam$ obtained by sprinkling into $\diam$, for a range of masses. We find that in the small mass regime $\wsjc$ agrees with our analytic calculation of $\wsj$ in the center of the diamond and therefore resembles $\wmink$. This means that it differs from $\wminkm$ in the small mass
regime. However, as the mass is increased, there is a cross-over
point at which the massless and massive Minkowski vacuum coincide. This occurs when the mass $m_c\equiv 2\cof \sim
0.924$, where $\cof \sim 0.462$ is the IR cut-off for the massless vacuum calculated in \cite{Afshordi:2012ez}.  For $m \geq
m_c$, $\wsjc$ then tracks the massive Minkowski vacuum instead of the massless Minkowski vacuum.   
In the corner of the diamond, the causal set $\wsjc$ looks like the mirror vacuum for small masses, whereas for large masses it behaves like the mirror and the Rindler vacuum equally well and therefore we can't make a distinction. In fact for large masses the correspondence between the mirror and the Rindler vacuum is better than the correspondence between $\wsjc$ and any of them.

Our calculations suggest that the massive $\wsj$ has a well defined
$m\rightarrow 0$ limit, unlike  $\wminkm$.  A possible reason for this is that
the SJ vacuum is built from the Green's function which is a continuous function of $m$ even as $m\rightarrow 0$. The behavior of
$\wsj$ for $m>0$  is also curious.   For  $\wmink$, $\cof$   sets a scale and 
dominates in the small $m$ regime, while for large $m$, the opposite is true.  At $m_c$, $\wmink$ and $\wminkm$
coincide at small distance scales, so that  $\wsj$ tracks  $\wmink$ for $m<m_c$ and  $\wminkm$ 
for $m>m_c$ in a continuous fashion.

Whether this  unexpected small mass behavior of $\wsj$ is the result of finiteness of $\diam$ or an intrinsic feature
of the 2d SJ vacuum is unclear at the moment.  Further examination of the massive SJ vacuum in different spacetimes
should shed light on these questions.  

We begin in Sec.~\ref{sec:soln} with setting up the SJ eigenvalue problem for the massive scalar field in $\diam$ and find the SJ spectrum in the small mass limit to $\cO(m^4)$.  Sec.~\ref{sec:wightmann} contains the analytic and numerical calculations of $\wsj$ in different regions of $\diam$. In Sec.~\ref{sec:causet} we show the results of simulations of the causal set SJ vacuum $\wsjc$ for a range of masses.  We then compare  $\wsjc$ with the analytical calculation $\wsj$ in the small mass regime, as well as  with the standard
vacua  in the large mass regime, both in the center and the corner of the
diamond for small and large values of $m$. In Sec.~\ref{sec:rindler} we present a trick to get the massless 2d Rindler vacuum from the SJ prescription. We end with a brief discussion of our results in Sec.~\ref{sec:conclusions}.  Appendixes \ref{expressions}, \ref{app:speigel} and \ref{sec:wight-app} contain the details of many of the calculations.

\section{The Spectrum of the Pauli Jordan Function: The small mass limit}
\label{sec:soln}

As we have stated earlier, the SJ modes are also solutions of the KG equation. A natural starting  point for
constructing these modes is therefore to start with a complete set of solutions $\{ s_k\}$ in the space $
\kker(\Bkg)$ where $\Bkg\equiv \Box-\m^2$,  and to find the action of $i \hD$ on this set. In light-cone coordinates the 2d Klein Gordon equation in Minkowski spacetime takes the simple form 
\be
\Bkg\Phi(u,v) \equiv \left(2\partial_u\partial_v+\m^2\right)\Phi(u,v)=0.  \label{kgeqn}
\ee
where $u$ and $v$ are the lightcone coordinates given by Eqn.~\eqref{lcc.eq}.
Thus, for $\m=0$ any differentiable  function $\psi(u)$ or $\xi(v)$ is in $\ker(\Bkg)$. 

One can generate a larger class of solutions starting  from a given differentiable function $\psi(u)$. The infinite sum 
\be
\Phi(u,v)\equiv \sum_{n=0}^\infty\frac{(-1)^nm^{2n}}{2^nn!}v^n\int^n\psi(u),
\ee
with  $\int^n\psi(u)\equiv\int du\int du\dots\int du\psi(u)$, can be seen to belong to $\ker(\Bkg)$. Similarly one can generate solutions starting with a differentiable function $\xi(v)$.  Different choices of $\psi(u),\xi(v)$ gives
different $\Phi(u,v)$. 

From the { Weierstrass theorem}, we know that any continuous function $\psi(u)$ in a bounded interval in $u$ can be
written as $\psi(u)=\sum_{n}a_nu^n$ for some $a_n's$. Hence  a natural class of solutions is generated by  $\psi(u)=u^l$,
\be
Z_l(u,v) \equiv
\sum_{n=0}^\infty\frac{(-1)^nm^{2n}l!}{2^nn!(n+l)!}u^{n+l}v^n=\frac{2^{l/2}l!}{m^l}\left(\frac{u}{v}\right)^{l/2}J_l\left(m\sqrt{2uv}\right)
\label{eq:serieszl},
\ee
for $l$ a whole number. Thus the SJ modes,  can in general be written as a sum over $Z_l(u,v)$
and $Z_l(v,u)$ for an appropriate set of $l$ values.  
Since plane waves are an important class of solutions, we note that starting from a function $\psi(u)=e^{au}$ for some
constant $a$ the plane wave solutions 
\be
U_a(u,v) \equiv \sum_{n=0}^\infty\frac{(-1)^nv^nm^{2n}}{2^nn!a^n}e^{au}=e^{au-\frac{m^2}{2a}v} \label{eq:seriesua}
\ee
and similarly, $U_a(v,u)$, can be obtained.

Before we proceed with the construction of the SJ modes, it will be useful to look at its following property.
\vskip 0.1in
\begin{claim}
In $\diam$ the SJ modes can be arranged into a complete set of eigenfunctions, each of which is either symmetric or
antisymmetric under the interchange of $u$ and $v$ coordinates.
\end{claim}
\begin{proof}
Let $u_k$ be an eigenfunction of $i\hD$ with eigenvalue $\lambda_k\neq 0$ i.e.
\be
i\hD\circ u_k=\lambda_ku_k \label{eqn:proof1}.
\ee
Define an operator $\hD'$ with integral kernel $\Delta'(u,v;u',v')=\Delta(v,u;v',u')$ and let $v_k$ such that $v_k(u,v)=u_k(v,u)$. Interchanging $u$ and $v$ since $u,v\in[-L,L]$, Eqn.~(\ref{eqn:proof1}) can be rewritten as
\be
i\hD'\circ v_k=\lambda_kv_k \label{eqn:proof2}.
\ee
Since $\Delta(u,v;u',v')$ is symmetric under $\{u,u'\}\leftrightarrow \{v,v'\}$, this implies that
\be
i\hD\circ v_k=i\hD'\circ v_k=\lambda_kv_k \label{eqn:proof3}.
\ee
Therefore $v_k$ is also an eigenfunction of $i\hD$ with same eigenvalue $\lambda_k$. This means that, the symmetric combination $u^S_k(u,v)=u_k(u,v)+u_k(v,u)$ and the antisymmetric combination $u^A_k(u,v)=u_k(u,v)-u_k(v,u)$ are also eigenfunctions of $i\hD$ with eigenvalue $\lambda_k$.
\end{proof}

In $\mink^2$ for $m=0$ the natural choice of solutions is the set of plane wave modes $\{e^{iku}, e^{ikv}\}$. However, in the
finite causal diamond, the constant function is also a solution. The explicit form of the corresponding SJ modes are
given  in Johnston's thesis \cite{johnston}. There are two
sets of eigenfunctions. The first set found by Johnston are the $f_k=e^{-iku}-e^{-ikv}$ modes with $k =n \pi/L$ and are
antisymmetric with respect to $u\leftrightarrow v$. The
second set  $g_k=e^{-iku}+e^{-ikv}-2\cos(kL)$, were found by Sorkin and satisfy the more complicated quantization condition $\tan(kL)=2kL$. These are symmetric with respect to $u\leftrightarrow v$.  The eigenvalues for each set are $L/k$.

We now proceed to set up the calculation for the central SJ eigenvalue problem.  We will find it useful to work with the
dimensionless quantities.
\be
  m L \rightarrow m,\;kL\rightarrow k,\;\frac{u}{L}\rightarrow u,\;\frac{v}{L}\rightarrow v,\;\frac{u'}{L}\rightarrow u',\;\frac{v'}{L}\rightarrow v'.
\ee
The  massive Pauli Jordan function in $\mink^2$ is  
\be
i\Delta(u,v;u',v')=-\frac{i}{2}J_0\left(m\sqrt{2\Du\Dv}\right)\left(\theta (\Du)+\theta(\Dv)-1\right) \label{idelta}
\ee
where $\Du=u-u', \Dv=v-v'$ and $\theta(x)$ is the Heaviside function.  
The SJ modes are thus given by (Eqn.~\ref{cep.eq}) 
\be
-\frac{iL^2}{2}\int_{-1}^1du'dv'J_0\left(m\sqrt{2\Du\Dv}\right)\biggl(\theta(\Du)+\theta(\Dv)-1\biggr) u_k(u',v')=\lambda_ku_k(u,v).
\ee
We will find it useful to make the change of variables $\Du=p, \Dv=q$ so that the above expression becomes  
\be
\frac{iL^2}{2}\left(\intm dp dq -\intp dp dq \right)  J_0\left(m\sqrt{2pq}\right)u_k(u-p,v-q)=\lambda_ku_k(u,v),  
\ee
where we have  used the short-hand $\intm dp dq \equiv \int_0^{u-1}dp\int_0^{v-1}dq$ and $\intp dp dq \equiv
\int_0^{u+1}dp\int_0^{v+1}dq$. 
Our strategy is to begin with the action of $i\hD$ on the   symmetric and antisymmetric combinations of the
$Z_l(u,v)$ and $U_a(u,v)$ solutions defined above,
\begin{eqnarray} 
  U^A_{a}(u,v)\equiv U_{a}(u,v)-U_{a}(v,u), &\;\;\;&  U^S_{a}(u,v)\equiv U_{a}(u,v)+U_{a}(v,u),\nonumber \\ 
        Z^A_{l}(u,v)\equiv Z_{l}(u,v)-Z_{l}(v,u), &\;\;\;&  Z^S_{l}(u,v)\equiv Z_{l}(u,v)+Z_{l}(v,u).     \label{eq:symasym}                                            
\end{eqnarray} 
so that the  general form for the two sets $u^{A/S}$ of SJ modes is given by 
\begin{equation} 
 u^{A/S}_{\vec{a},\vec{l}}(u,v) \equiv  \sum_{a \in \vec{a}} \alpha^{A/S}_aU^{A/S}_a(u,v) + \sum_{l \in \vec{l}}
 \beta^{A/S}_l Z^{A/S}_l(u,v).
\end{equation}
Here $\vec{a}, \vec{l}$ denote  set of values for  $a$ and $l$ which satisfy  quantization conditions. Of course each $U_a(u,v)$ is itself an infinite sum over  $Z_l(u,v)$, but
we nevertheless consider it separately, taking our cue from the massless calculation. 

The expressions 
\begin{eqnarray} 
i\hD\circ U_a(u,v) &=&\frac{iL^2}{2}\left( \intm dp dq-\intp dp dq\right) J_0\left(m\sqrt{2pq}\right)
                          U_a^*(p,q)U_a(u,v), \nonumber \\
i\hD\circ Z_l(u,v) &=&  \frac{iL^2}{2}\left( \intm dp dq-\intp dp dq\right) J_0\left(m\sqrt{2pq}\right)
                         Z_l(u-p,v-q) \label{integration2} 
\end{eqnarray} 
are in general not easy to evaluate and  subsequently manipulate in order to obtain the SJ modes.  We instead begin by
looking for solutions order by order in $m^{2}$ assuming that for some $n$, $m^{2n}<<1$.\footnote{The series expansion
  of $U^{A/S}_{ik}$ in the SJ modes for small $m$ can be truncated to a finite order of $m^2$ if and only if $k$ is of
  the order of unity or higher. However, this is the case for small $m$, since small $k$ corresponds to wavelengths much
  larger than the size of the diamond.}
We use the series form of $Z_l(u,v)$ in Eqn.~(\ref{eq:serieszl}) and $U_a(u,v)$ in Eqn.~(\ref{eq:seriesua})
  as well as 
\be
J_0\left(m\sqrt{2pq}\right) =\sum_{n=0}^\infty\frac{(-1)^nm^{2n}}{2^n(n!)^2}p^nq^n. 
\label{eq:seriesbess}
\ee
As we will show, for $n=4$, we find that, to $\cO(m^4)$ the two families of eigenfunctions, antisymmetric and symmetric are 

{\bf Antisymmetric:}
\be
u^A_k(u,v)=\left[U^A_{ik}(u,v)-\cos(k)\left(\left(\frac{im^2}{2k}-\frac{im^4(6+k^2)}{24k^3}\right)Z^A_1(u,v)-\frac{m^4}{4k^2}Z^A_2(u,v)\right)\right]+\cO(m^6), \label{eq:asym}
\ee
with eigenvalue $-\frac{L^2}{k}$ with $k \in \cK_A$ satisfying the quantization condition 
\be
\sin(k)=\left(\frac{m^2}{k}+\frac{m^4}{12k}\left(1-\frac{3}{k^2}\right)\right)\cos(k)+\cO(m^6).\label{eq:qcondA}
\ee
Solving for $k$, order by order in $m^2$ up to $\cO(m^4)$, as shown in Sec. \ref{sec:compl},  gives $k=\ka(n)$, where
\be
\ka(n)\equiv n\pi+\frac{m^2}{n\pi}+m^4\left(\frac{1}{12n\pi}-\frac{5}{4n^3\pi^3}\right)+\cO(m^6), \label{eq:qcondka}
\ee
where $n\in{\mathbb Z}$ and $n \neq 0$.

{\bf Symmetric:}
\bea
u^S_k(u,v)&=&\left[U^S_{ik}(u,v)-\cos(k)\left(\left(1+\frac{m^2}{2}-\frac{m^4}{8k^2}(2-9k^2)\right)Z^S_0(u,v)\right.\right.\nonumber\\
&&\left.\left.+\left(\frac{3im^2}{2k}-\frac{im^4}{24k^3}(6-31k^2)\right)Z^S_1(u,v)-\frac{m^4}{8k^2}(4-k^2)Z^S_2(u,v)\right)\right]+\cO(m^6), \label{eq:sym}
\eea
with eigenvalue $-\frac{L^2}{k}$, where $k \in \cK_S$ satisfies
\be
\sin(k)=\left(2k-\frac{m^2}{k}(1-2k^2)+\frac{m^4}{12k^3}(3-29k^2+28k^4)\right)\cos(k)+\cO(m^6). \label{eq:qcondS}
\ee
Solving for $k$, order by order in $m^2$ up to $\cO(m^4)$, as shown in Sec. \ref{sec:compl}, gives $k=\ks(k_0)$, where
\be
\ks(k_0)\equiv k_0+m^2\frac{1-2{k_0}^2}{k_0(1-4{k_0}^2)}+m^4\frac{(3-4{k_0}^2)(-5+35{k_0}^2-40{k_0}^4+16{k_0}^6)}{12{k_0}^3(1-4{k_0}^2)^3}+\cO(m^6) \label{eq:qcondks},
\ee
where $k_0$ are the solutions of $\sin(k)=2k\cos(k)$.

We plot these eigenvalues in Fig.~\ref{fig:eigenvalue} for m=0, 0.2 and 0.4. In the expressions for the eigenfunctions,  Eqns.~(\ref{eq:asym}) and (\ref{eq:sym}), it is to be noted that we have kept 
$U^{A/S}_{ik}$ and $Z^{A/S}_l$ as they are, rather than use their expansion to $\cO(m^4)$.   The reason for this is to
remind ourselves that they are solutions of the Klein Gordon equation. Note that in  Eqn.~(\ref{eq:asym}) and
Eqn.~(\ref{eq:sym}), we keep terms only up to $\cO(m^4)$ within the square bracket. In Sec. \ref{sec:compl} we show that
these form a complete set of orthonormal modes.

Here we have moved away from the $f_k$ and $g_k$ notation of \cite{Afshordi:2012ez,johnston} to $u^A_k$ and $u^S_k$ for
the antisymmetric and symmetric SJ modes respectively.

\begin{figure}[h]
\centerline{\begin{tabular}{ccc}
\includegraphics[height=5cm]{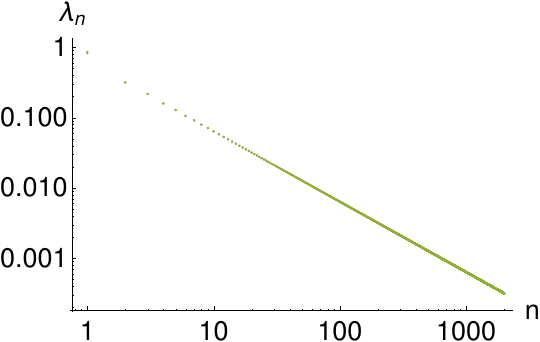}&
\includegraphics[height=5cm]{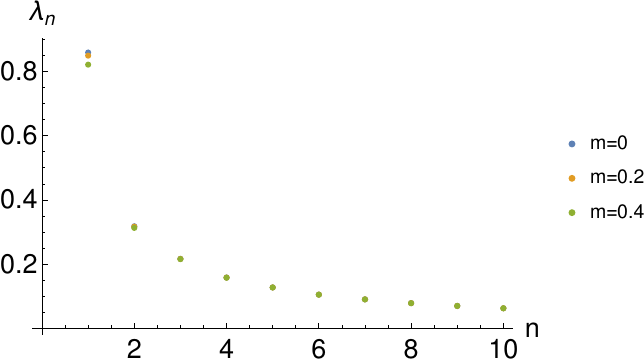}\\
(a)&(b)
\end{tabular}}
\caption{(a):A log-log plot of the SJ eigenvalues $\lambda_n$ vs $n$ for $m=0, 0.2$ and $0.4$, (b): a plot of $\lambda_n$ vs $n$ for small $n$. As one can see, the eigenvalues for $m=0.2$ and $0.4$ are barely distinguishable from $m=0$, except for the very smallest $n$ values.}
\label{fig:eigenvalue}
\end{figure}

\subsection{Details of the calculations of SJ modes}
We now show the calculation in broad strokes below, leaving some of the details to the Appendix \ref{expressions}. We begin by  reviewing the massless case. Here  $Z_l(u,v)$ reduces to $u^l$ and $U_a(u,v)$ to $e^{au}$.

Operating $i \hD$ on $u^l$ or $v^l$  we find that
\begin{eqnarray} 
  i\hD_{m=0}\circ u^l&=&\frac{iL^2}{2(l+1)}\left(\left(1+(-1)^{l+1}\right)-v\left(1-(-1)^{l+1}\right)-2u^{l+1}\right), \nonumber \\
  i\hD_{m=0}\circ v^l&=&\frac{iL^2}{2(l+1)}\left(\left(1+(-1)^{l+1}\right)-u\left(1-(-1)^{l+1}\right)-2v^{l+1}\right), 
\end{eqnarray}
while on the plane wave modes  
\bea
i\hD_{m=0}\circ e^{iku}&=&
-\frac{L^2}{k}\left(e^{iku}-\cos(k)+iv\sin(k)\right), \nonumber \\
i\hD_{m=0}\circ e^{ikv}&=&
-\frac{L^2}{k}\left(e^{ikv}-\cos(k)+iu\sin(k)\right).
\eea 

Here, $k$ takes on all values including $k=0$, which is the constant solution. From the antisymmetric combination
\be
i\hD_{m=0}\circ\left(e^{iku}-e^{ikv}\right)=-\frac{L^2}{k}\left(e^{iku}-e^{ikv}-i\sin(k)(u-v)\right), 
\ee
we find the first set of massless eigenfunctions
\be
u_k^{A(0)}(u,v)\equiv e^{iku}-e^{ikv}
\ee
with  $k \in \cK_f$ satisfying the quantization condition 
\be
\sin(k)=0\;\text{or}\;k=n\pi.  \label{f condition m0}
\ee
with eigenvalues $-\frac{L^2}{k}$. The symmetric  combination
on the other hand gives 
\begin{equation} 
i\hD_{m=0}\circ\left(e^{iku}+e^{ikv}\right)=
-\frac{L^2}{k}\left(e^{iku}+e^{ikv}-2\cos(k)\right)-\frac{iL^2}{k}\sin(k)(u+v).
\end{equation} 
Since the  symmetric eigenfunction can include a  constant piece and noting that 
\begin{equation}
 \hD_{m=0}\circ c =  -icL^2(u+v), 
  \end{equation} 
we find the second set of eigenfunctions 
\be
u_k^{S(0)}(u,v)\equiv e^{iku}+e^{ikv}-2\cos(k)
\ee
with eigenvalue $-\frac{L^2}{k}$, where $k \in \cK_g$ satisfies
\be
\sin(k)=2k\cos(k). \label{g condition m0}
\ee
$\{u_k^{A(0)}\}$ and $\{u_k^{S(0)}\}$ together form a complete set of eigenfunctions of $i\Delta$ as can be shown by \cite{johnston}.

This sets the stage for the calculation of the massive SJ modes. We begin by again looking the action of $i\hD$ on the
solutions $Z_l(u,v)$ and $U_a(u,v)$,  
\begin{eqnarray} 
i\hD\circ Z_l(u,v)&=&\frac{iL^2}{2}\sum_{j,s=0}^\infty\frac{(-1)^{j+s}m^{2(j+s)}l!}{2^{l+s}(j!)^2s!(s+l)!} \gOjsl ,\label{eq:ideltaz}   \\ 
i\hD\circ U_{a}(u,v) &=&\frac{iL^2}{2} U_a(u,v) \sum_{j,s=0}^\infty\frac{(-1)^j
                            m^{2(j+s)}}{2^{j+s}(j!)^2s!a^s}\djsa(u,v),
\label{eq:genzluaexp}                             
\end{eqnarray}
where
\begin{eqnarray}
\gOjsl(u,v) &\equiv & \biggl(\intm dp\, dq - \intp dp\, dq\biggr)p^jq^j(u-p)^{l+s}(v-q)^s, \nonumber \\  
\djsa(u,v) & \equiv &\biggl(\intm dp\, dq - \intp dp\, dq\biggr)p^jq^{j+s}e^{-a p}. \label{deltajs}
\end{eqnarray}
It is useful to re-express Eqn.~(\ref{eq:genzluaexp})  as 
\begin{equation} 
i\hD\circ U_{a}(u,v) = \frac{iL^2}{2} U_a(u,v) \sum_{n=0}^\infty m^{2n} A_{a,n}(u,v), \label{eq:ideltau} 
\end{equation}
where
\begin{equation}
A_{a,n}(u,v)\equiv \sum_{j=0}^{n}\frac{(-1)^j}{2^{n}(j!)^2(n-j)!a^{(n-j)}}\dja(u,v). 
\end{equation}
This gives
\begin{equation}
  i\hD\circ U_{a}(u,v) = -\frac{iL^2}{a}U_{a}(u,v) - \frac{iL^2}{a} \sum_{n=0}^\infty m^{2n} \cF_{a,n}(u,v),
  \label{eq:genexp} 
  \end{equation} 
  where
  \begin{equation}
\cF_{a,n}(u,v) \equiv F_{a,n}(u,v)\sinh(a)+G_{a,n}(u,v)\cosh(a) ,
    \end{equation} 
with
\bea
F_{a,n}(u,v) &\equiv & \sum_{s=0}^n\sum_{j=0}^s\sum_{l=0}^j \frac{(-1)^{n-s+j}v^{n-s}\left((u+1)^{j-l}(v+1)^{s+1}+(u-1)^{j-l}(v-1)^{s+1}\right)}{2^{n+1}a^{n-j+l}(n-s)!j!(s-j)!(j-l)!(s+1)}, \nonumber\\
G_{a,n}(u,v) &\equiv & \sum_{s=0}^n\sum_{j=0}^s\sum_{l=0}^j \frac{(-1)^{n-s+j}v^{n-s}\left((u-1)^{j-l}(v-1)^{s+1}-(u+1)^{j-l}(v+1)^{s+1}\right)}{2^{n+1}a^{n-j+l}(n-s)!j!(s-j)!(j-l)!(s+1)}. \label{f and g}
\eea

Our first guess, inspired by the massless calculation, is that in order to  find the SJ modes,  we will need the antisymmetrised and symmetrised versions
of Eqns.~(\ref{eq:ideltaz}) and (\ref{eq:ideltau}), which we denote
by $A/S$.  As noted above, and is evident from Eqn.~(\ref{eq:genexp}), in order to obtain the SJ modes,
$\UAS_a(u,v)$  must be supplemented by a function $\HAS_a(u,v)$ made from the $Z_l(u,v)$.

Taking our cue from the massless case,  let us assume that such a function
exists, i.e., 
\be
i\hD\circ\left(U^{A/S}_{a}(u,v)+\HAS_a(u,v)\right)=-\frac{i L^2}{a}\left(U^{A/S}_{a}(u,v)+\HAS_a(u,v)\right) \label{eq:u,h},
\ee
where $k$ satisfies an appropriate quantization condition $\cK^{A/S}$.  Then, from Eqn.~(\ref{eq:genexp}) $\HAS_a(u,v)$
must satisfy 
\begin{equation}
i\hD\circ \HAS_a(u,v)+\frac{i L^2}{a}\HAS_a(u,v)-\frac{iL^2}{a} \sum_{n=0}^\infty m^{2n} \cF^{A/S}_{a,n}(u,v)=0 .
  \end{equation} 

Up to now the discussion has been general. If the expressions above could be calculated in closed form,  then one would be able to
solve the SJ mode problem for any mass $m$.  It is unclear how to proceed to do this, except order by order in $m^2$.

We now demonstrate this explicitly up to $\cO(m^4)$. We begin by taking $a=ik$ and writing Eqn.~(\ref{eq:genexp}) as
\begin{equation} 
i\hD\circ \UAS_{ik}(u,v) \approx -\frac{L^2}{k}\UAS_{ik}(u,v)-\frac{L^2}{k}\left( i\sin(k) \sum_{n=0}^\infty m^{2n} \FAS_{ik,n}(u,v)
     + \cos(k) \sum_{n=0}^\infty m^{2n} \GAS_{ik,n}(u,v)\right),\label{ideltaua}
\end{equation}
where the expressions for $F_{ik,n}(u,v)$ and $G_{ik,n}(u,v)$ for different $n$ have been calculated in  Appendix \ref{expressions}. 
The function $H^{A/S}_k(u,v)$ must therefore satisfy
\begin{equation} 
i\hD\circ H^{A/S}_{ik}(u,v)+\frac{L^2}{k}\left(H^{A/S}_{ik}(u,v)-i\sin(k) \sum_{n=0}^\infty m^{2n} \FAS_{ik,n}(u,v)-\cos(k) \sum_{n=0}^\infty m^{2n} \GAS_{ik,n}(u,v)\right)=0. \label{eq:has}
\end{equation} 

From the result for the massless case, we expect the quantization condition for $k$ to be of the general form
\be
\sin(k)=\cos(k) \sum_{n=0}^\infty m^{2n} \QAS_n(k), \label{f condition}
\ee
with $Q^A_0(k)=0$ and $Q^S_0(k)=2k$. Inserting this into Eqn.~(\ref{eq:has}) gives
\begin{equation} 
  i\hD\circ \HAS_{ik}(u,v)+\frac{L^2}{k}\HAS_{ik}(u,v)-\frac{L^2}{k}\cos(k) \left( \sum_{n=0}^\infty m^{2n} \PAS_n(u,v) \right)
  =0,
  \label{eq:hask}
  \end{equation} 
  where
  \begin{equation}
    \PAS_n(u,v) \equiv \GAS_n(u,v) + i\sum_{j=0}^n \QAS_j(k) \FAS_{n-j}(u,v).  \label{eq:pas}
    \end{equation} 
The challenge is therefore to obtain the explicit form for these expressions. Finding a general expression in this
 manner is very
 challenging, but we will now show that it can be found to $\cO(m^4)$. 

Since the $\HAS_a(u,v)$ must be constructed from the $Z_l(u,v)$, we are interested in the action of $i\hD$ on $Z_l(u,v)$
up to $\cO(m^4)$ i.e.,
\begin{equation}
i\hD\circ Z_l(u,v)=  \frac{iL^2}{2}\sum_{j,s, j+s \leq 2 } \frac{(-1)^{j+s}m^{2(j+s)}l!}{2^{l+s}(j!)^2s!(s+l)!} \gOjsl+\cO(m^6).
  \end{equation}   
We calculate this expression for  $l=0,1,2$, up to $\cO(m^4)$  in the Appendix \ref{expressions}. Using the expression of
$P_n^A(u,v)$ given in Appendix \ref{expressions}, we find that up to  $\cO(m^4)$ the antisymmetric version of  Eqn.~(\ref{eq:hask}) reduces to  
\be
\left(i\hD+\frac{L^2}{k}\right)\circ\left(H^A_{ik}(u,v)+\cos(k)\left(\left(\frac{im^2}{2k}-\frac{im^4(6+k^2)}{24k^3}\right)Z^A_1(u,v)-\frac{m^4}{4k^2}Z^A_2(u,v)\right)\right)\approx 0. 
\ee
Therefore  
\be
u^A_k(u,v)=U^A_{ik}(u,v)-\cos(k)\left(\left(\frac{im^2}{2k}-\frac{im^4(6+k^2)}{24k^3}\right)Z^A_1(u,v)-\frac{m^4}{4k^2}Z^A_2(u,v)\right)+\cO(m^6),
\ee
with eigenvalue $-\frac{L^2}{k}$ with $k \in \cK_A$ satisfying the quantization condition 
\be
\sin(k)=\left(\frac{m^2}{k}+\frac{m^4}{12k}\left(1-\frac{3}{k^2}\right)\right)\cos(k)+\cO(m^6).
\ee
Similarly using the expression of
$P_n^S(u,v)$ given in Appendix \ref{expressions} and after more painstaking algebra, 
we find that Eqn.~(\ref{eq:hask}) can  be written as
\bea
\left(i\hD+\frac{L^2}{k}\right)\circ\left(H^S_{ik}(u,v)+\cos(k)\left(\left(1+\frac{m^2}{2}-\frac{m^4}{8k^2}(2-9k^2)\right)Z^S_0(u,v)\right.\right.\nonumber\\
\left.\left.+\left(\frac{3im^2}{2k}-\frac{im^4}{24k^3}(6-31k^2)\right)Z^S_1(u,v)-\frac{m^4}{8k^2}(4-k^2)Z^S_2(u,v)\right)\right)
\approx 0 .
\eea
Therefore the symmetric eigenfunction is 
\bea
u^S_k(u,v)&=&U^S_{ik}(u,v)-\cos(k)\left(\left(1+\frac{m^2}{2}-\frac{m^4}{8k^2}(2-9k^2)\right)Z^S_0(u,v)\right.\nonumber\\
&&\left.+\left(\frac{3im^2}{2k}-\frac{im^4}{24k^3}(6-31k^2)\right)Z^S_1(u,v)-\frac{m^4}{8k^2}(4-k^2)Z^S_2(u,v)\right)+\cO(m^6),
\eea
with eigenvalue $-\frac{L^2}{k}$, where $k \in \cK_S$ satisfies
\be
\sin(k)=\left(2k-\frac{m^2}{k}(1-2k^2)+\frac{m^4}{12k^3}(3-29k^2+28k^4)\right)\cos(k)+\cO(m^6).
\ee

Unfortunately, the structure of neither the coefficients in $\uAS_k$  nor the quantization condition are enough to
suggest a generalization to all orders. One could of course proceed to the next order $\cO(m^6)$ but the calculation  gets
prohibitively more complex.  

\subsection{Completeness of the eigenfunctions}\label{sec:compl}

We now show that the eigenfunctions $\{ u_k^A| k \in \cK_A\}$ and  $\{ u_k^S| k \in \cK_S\}$  form a complete set of
eigenfunctions of $i\Delta$. If this is the case, then we can decompose $i\Delta$ as  
\be
i\Delta(u,v;u',v')=\sum_{k\in \cK_A}  -\frac{L^2}{k}  \frac{u^A_k(u,v){u^A_k}^*(u',v')}{||u^A_k||^2} + \sum_{k\in \cK_S}  -\frac{L^2}{k}
\frac{u^S_k(u,v){u^S_k}^*(u',v')}{||u^S_k||^2}+\cO(m^6), 
\ee
which implies that 
\be
\int_{S} du\, dv\,du'\, dv' |\Delta(u,v;u',v')|^2=\sum_{k\in \cK_A} \biggl(\frac{L^2}{k}\biggr) ^2 + \sum_{k\in \cK_S} \biggl(\frac{L^2
}{k}\biggr) ^2+\cO(m^6).\label{eq:trace}
\ee 

To  $\cO(m^4)$  the LHS  of Eqn.~(\ref{eq:trace}) reduces to 
\bea
&& \frac{L^4}{4}\int_{-1}^1dudv\left(\intm dp \, dq+\intp dp\, dq\right) J_0^2\left(m\sqrt{2pq}\right) \nonumber \\ 
&=&
\frac{L^4}{4}\int_{-1}^1dudv\left( \intm dp \, dq+\intp dp\, dq\right)\left(1-m^2pq+\frac{3}{8}m^4p^2q^2\right)+\cO(m^6)\nonumber\\
&=&2L^4\left(1-\frac{4}{9}m^2+\frac{1}{6}m^4\right)+\cO(m^6). 
\eea
For the RHS  $k\in \cK_{A/S}$, we make use of the expansion $\kAS \approx \kAS_0+m^2\kAS_1+m^4\kAS_2$.
For the antisymmetric quantization condition Eqn.~(\ref{eq:qcondA}) since $k_0^A=n \pi$ this gives, up to $\cO(m^4)$  
\be
m^2k_1^A+m^4k_2^A=\frac{m^2}{k_0^A}\left(1-m^2\frac{k_1^A}{k_0^A}\right)-\frac{m^4}{4{k_0^A}^3}+\frac{m^4}{12k_0^A}+\cO(m^6).
\ee
Solving the above equation for different orders of $m^2$, we get
\bea
k_1^A&=&\frac{1}{n\pi}, \\
k_2^A&=&\frac{1}{12n\pi}-\frac{5}{4n^3\pi^3},
\eea
so that 
\bea
\sum_{k \in \cK_A} L^4 \frac{1}{k^2}
&=&2L^4\sum_{n=1}^\infty\frac{1}{n^2\pi^2}\left(1-2m^2\frac{1}{n^2\pi^2}-m^4\left(\frac{1}{6n^2\pi^2}-\frac{11}{2n^4\pi^4}\right)\right)+\cO(m^6)\nonumber\\
&=&2L^4\left(\frac{1}{6}-\frac{m^2}{45}+\frac{m^4}{252}\right)+\cO(m^6).
\eea
For the symmetric contribution  Eqn.~(\ref{eq:qcondS})   up to $\cO(m^4)$ we have 
\be
\sum_{n=0}^2 m^{2n}K_n(k^S_0,k^S_1,k^S_2)+\cO(m^6)=0,
\ee
where
\bea
K_1(k^S_0,k^S_1,k^S_2)&=&\sin(k^S_0)-2k^S_0\cos(k^S_0), \nonumber\\
K_2(k^S_0,k^S_1,k^S_2)&=&\left(\frac{2{k^S_0}^2-1+k^S_1k^S_0}{k^S_0}\right)\cos(k^S_0)-2k^S_1k^S_0\sin(k^S_0),\nonumber\\
K_3(k^S_0,k^S_1,k^S_2)&=&\left(\frac{3-29{k^S_0}^2+28{k^S_0}^4+12k^S_1k^S_0}{12{k^S_0}^3}+2k^S_1+k^S_2-{k^S_1}^2k^S_2\right)\cos(k^S_0)\nonumber\\
&&+\left(\frac{k^S_1-2k^S_1k^S_0-2{k^S_0}^3}{k^S_0}-\frac{3}{2}{k^S_1}^2\right)\sin(k^S_0).
\eea
Equating the above order by order in $m^2$, we get
\bea
\sin(k^S_0)&=&2k^S_0\cos(k^S_0),\\
k^S_1&=&\frac{1-2{k^S_0}^2}{k^S_0(1-4{k^S_0}^2)},\\
k^S_2&=&\frac{(3-4{k^S_0}^2)(-5+35{k^S_0}^2-40{k^S_0}^4+16{k^S_0}^6)}{12{k^S_0}^3(1-4{k^S_0}^2)^3}.
\eea
\bea
\sum_{k \in \cK_S} L^4 \frac{1}{k^2}
&=&2L^4\sum_{k^S_0\in \cK_g}\left(\frac{1}{{k^S_0}^2}-2m^2\left(\frac{1}{{k^S_0}^4}+\frac{2}{{k^S_0}^2}-\frac{8}{4{k^S_0}^2-1}\right)\right.\nonumber\\
&&\left.+m^4\left(\frac{11}{2{k^S_0}^6}+\frac{127}{6{k^S_0}^4}+\frac{280}{3{k^S_0}^2}+\frac{32}{(4{k^S_0}^2-1)^3}+\frac{32}{(4{k^S_0}^2-1)^2}-\frac{1120}{3(4{k^S_0}^2-1)}\right)\right)+\cO(m^6).\nonumber\\ \label{eqn:gcomp}
\eea
We evaluate the above series by using the method developed in \cite{speigel} and used in \cite{Afshordi:2012ez,johnston}, details of which can be found in Appendix \ref{app:speigel}. This leads to
\be
\sum_{k^S_0\in\cK_g}\frac{1}{{k^S_0}^2}=\frac{5}{6}\;\;,\;\;\;\sum_{k^S_0\in\cK_g}\frac{1}{{k^S_0}^4}=\frac{49}{90}\;\;\text{and}\;\;\;\sum_{k^S_0\in\cK_g}\frac{1}{{k^S_0}^6}=\frac{377}{945} \label{eq:ser1}
\ee
and
\bea
\sum_{k^S_0\in\cK_g}\frac{1}{4{k^S_0}^2-1}&=&\frac{1}{4},\nonumber\\
\sum_{k^S_0\in\cK_g}\frac{1}{(4{k^S_0}^2-1)^2}&=&-\frac{1}{4}\left(\frac{\cos(1/2)-2\sin(1/2)}{\cos(1/2)-\sin(1/2)}\right),\nonumber\\
\sum_{k^S_0\in\cK_g}\frac{1}{(4{k^S_0}^2-1)^3}&=&\frac{1}{64}\left(1+\frac{19\cos(1/2)-35\sin(1/2)}{\cos(1/2)-\sin(1/2)}\right).\label{eq:ser2}
\eea
This simplifies Eqn.~(\ref{eqn:gcomp}) to 
\be
\sum_{k\in \cK_S} 2L^4\frac{1}{k^2}=2L^4\left(\frac{5}{6}-\frac{19}{45}m^2+\frac{41}{252}m^4\right)+\cO(m^6).
\ee
Adding the contributions from the antisymmetric and symmetric eigenfunctions the RHS of Eqn.~(\ref{eq:trace}) reduces to 
\be
\sum\lambda_k^2=2L^4\left(1-\frac{4}{9}m^2+\frac{1}{6}m^4\right)+\cO(m^6),
\ee
which is same as its LHS. Thus, to $\cO(m^4)$  the $\uAS_k$ are a complete set of eigenfunctions of $i\hD$.

\section{The Wightman function: the small mass limit}
\label{sec:wightmann}
We can now write down the formal expression for the SJ  Wightman function to $\cO(m^4)$  using the SJ modes
obtained above, as 
\begin{equation}
\wsj(u,v,u',v')=\sum_{k\in\cK_{A},\, k<0}-\frac{L^2}{k}\frac{u^{A}_k(u,v){u^{A}_k}^*(u',v')}{||u_k^{A}||^2}
+\sum_{k\in\cK_{S}, \, k<0}-\frac{L^2}{k}\frac{u^{S}_k(u,v){u^{S}_k}^*(u',v')}{||u_k^{S}||^2}+\cO(m^6),  \label{eq:fullw}
\end{equation}
where $\cK_{A/S}$ denote  the positive SJ eigenvalues.  In particular $k=-\ka(n)$ with $n\in{\mathbb Z}^+$ (Eqn.~(\ref{eq:qcondka})) and  $k=-\ks(k_0)$ with
$k_0$ satisfying  $\tan(k_0)=2k_0$ (Eqn.~(\ref{eq:qcondks})). 
Here $||u^{A/S}_k||$ denotes  the $\mL^2$ norm of the modes $u^{A/S}_k$ 
\be
||u^{A/S}_k||^2=L^2\int_{-1}^1 du \int_{-1}^1 dv u^{A/S}_k(u,v){u^{A/S}_k}^*(u,v). 
\ee
For $k=-\ka(n)$ 
\be
||u^A_k||^2 =
8L^2\left(1+\frac{m^2}{n^2\pi^2}+\frac{m^4}{n^2\pi^2}\left(\frac{1}{12}-\frac{11}{4n^2\pi^2}\right)\right)+\cO(m^6). \label{eq:norma} 
\ee
In the symmetric case, $k =-\ks(k_0) $ the quantization condition is complicated. Following \cite{Afshordi:2012ez}, we
make the approximation
\be
\ks(n) \approx \left(n-\frac{1}{2}\right)\pi, \, \, n\in {\mathbb Z}^+.   \label{eq:appquant}
\ee
As shown in  Fig.~\ref{fig:quant}, we see that except for the first few modes this is a good approximation, and in fact
improves with increasing mass\footnote{Of course, at the same time, our approximation of the SJ modes becomes worse with
  increasing mass.}.  This approximation in the quantization condition makes $\cos(\ks)=0$, thus simplifying
$u^S_k(u,v)$ to 
\be
u^S_{-\ks}(u,v)=U^S_{-i\ks}(u,v) \Rightarrow ||u^S_{\ks}||=8L^2. \label{eq:norms} 
\ee
\begin{figure}[h]
\centerline{\includegraphics[height=6cm]{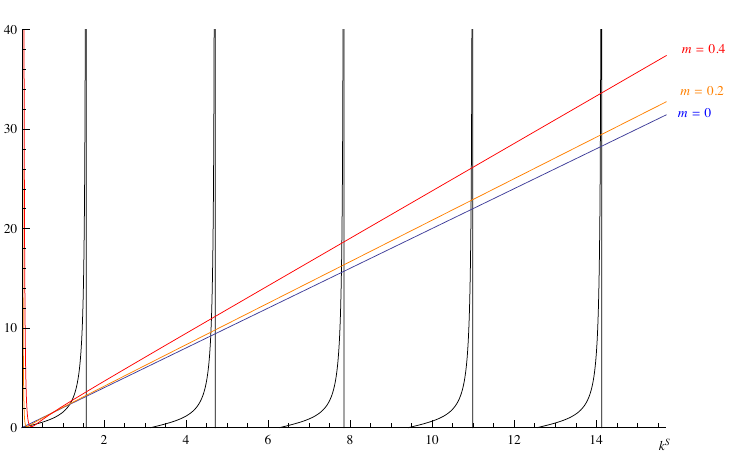}} 
\caption{Plot of the quantization condition, Eqn.~(\ref{eq:qcondS}) for the symmetric SJ modes for m=0,0.2 and 0.4,
  where $k_S>0$. }
\label{fig:quant}
\end{figure}

We examine the antisymmetric and symmetric contributions to $\wsj$ separately 
\be
\wsj=\wsj^A+\wsj^S. \label{eq:wsjsplit} 
\ee

For the antisymmetric contribution, using the quantization condition $k=-k_A(n)$ and the simplification Eqn.~(\ref{eq:norma}) for
the norm
\begin{equation}
  \wsj^A(u,v,u',v') = \sum_{n=1}^\infty \frac{1}{8n\pi}\left(1-\frac{2m^2}{n^2\pi^2}+\frac{m^4}{n^2\pi^2}\left(\frac{7}{n^2\pi^2}-\frac{1}{6}\right)\right) u^A_{k}(u,v) u_{k}^{A*}(u',v') + \cO(m^6).
\end{equation}
To leading order $u^A_{k}$ can be re-expressed as 
\begin{eqnarray} 
  u^A_{k}(u,v) &=& e^{-in\pi u}-e^{-in\pi v}+ \Psi_A(n,u,v) + \cO(m^6),\nonumber  \\ 
\Psi_A(n,u,v) &=&  \sum_{j=1}^3\left(\frac{(-1)^nf_j(m;u,v)}{n^j}+\frac{g_j(m;u,v)e^{-in\pi
      u}}{n^j}-\frac{g_j(m;v,u)e^{-in\pi v}}{n^j}\right), 
\end{eqnarray}
where 
\bea
f_1(m;u,v)\equiv \frac{im^2}{2\pi}(u-v)-\frac{im^4}{24\pi}(u-v)(1+3uv), &\quad& g_1(m;u,v)\equiv -\frac{im^2(2u+v)}{2\pi}-\frac{im^4u}{12\pi},\nonumber\\
f_2(m;u,v)\equiv -\frac{m^4}{4\pi^2}(u^2-v^2), &\quad& g_2(m;u,v) \equiv -\frac{m^4(2u+v)^2}{8\pi^2},\nonumber\\
f_3(m;u,v)\equiv-\frac{3im^4}{4\pi^3}(u-v), &\quad& g_3(m;u,v)\equiv \frac{im^4(15u+6v)}{12\pi^3}. 
\label{eq:gjfj}
\eea
We further split
\be
\wsj^A=\wa+\waa+\waaa+\waaaa+\cO(m^6),   \label{eq:splitwsja}
\ee
where 
\bea
\wa&\equiv& \sum_{n=1}^\infty\frac{1}{8n\pi}\left(1-\frac{2m^2}{n^2\pi^2}+\frac{m^4}{n^2\pi^2}\left(\frac{7}{n^2\pi^2}-\frac{1}{6}\right)\right)\left(e^{-in\pi u}-e^{-in\pi v}\right)\left(e^{in\pi u'}-e^{in\pi v'}\right),\nonumber\\
\waa&\equiv& \sum_{n=1}^\infty\frac{1}{8n\pi}\left(1-\frac{2m^2}{n^2\pi^2}\right)\left(e^{-in\pi u}-e^{-in\pi v}\right)\Psi_A^*(n,u',v'),\nonumber\\
\waaa&\equiv& \sum_{n=1}^\infty\frac{1}{8n\pi}\left(1-\frac{2m^2}{n^2\pi^2}\right)\Psi_A(n,u,v)\left(e^{in\pi u'}-e^{in\pi v'}\right),\nonumber\\
\waaaa&\equiv& \sum_{n=1}^\infty\frac{1}{8n\pi}\Psi_A(n,u,v)\Psi_A^*(n,u',v').\nonumber\\
\label{eq:aone}
\eea
These terms can be further simplified to $\cO(m^4)$ as we have shown in  Appendix.~\ref{sec:wight-app}. 

For the symmetric contribution $\wsj^S$ we use the simplification Eqns.~(\ref{eq:appquant}) and (\ref{eq:norms}) to express   
\be
\wsj^S = \sum_{n=1}^\infty\frac{1}{4\pi(2n-1)}U^S_{-i\ks}(u,v){U^S}_{-i\ks}^*(u',v')+\e_m(u,v,u',v') + \cO(m^6). 
\ee
Here $\e_m(u,v;u',v')$ is the correction term coming from the approximation of the quantization condition Eqn.~(\ref{eq:appquant}). This is
analytically difficult to obtain and in Sec.~\ref{sec:numerical}, we will evaluate it numerically for different values of $m$.   

Using the $\cO(m^4)$ expansion of $U_{-ik}$ from Eqn.~(\ref{eq:seriesua}), we write $U_{-ik_S}^S$ as
\bea
U_{-ik_S(n)}^S(u,v)&=&\left(e^{-i\left(n-\frac{1}{2}\right)\pi u}+e^{-i\left(n-\frac{1}{2}\right)\pi v}\right)+\Psi_S(n,u,v)+\cO(m^6), \nonumber\\
\Psi_S(n,u,v)&=&-\frac{im^2}{(2n-1)\pi}\left(v e^{-i\left(n-\frac{1}{2}\right)\pi u}+ue^{-i\left(n-\frac{1}{2}\right)\pi v}\right)\nonumber\\
&&\quad\quad\quad-\frac{m^4}{4(2n-1)^2\pi^2}\left(v^2e^{-i\left(n-\frac{1}{2}\right)\pi u}+u^2e^{-i\left(n-\frac{1}{2}\right)\pi v}\right).
\eea
Again for the symmetric part, we can write
\be
\wsj^S=\ws+\wss+\wsss+\wssss+\e_m(u,v,u',v')+\cO(m^6),
\label{eq:splitwsjs} 
\ee
where
\bea
\ws &\equiv& \frac{1}{4\pi}\sum_{n=1}^\infty\frac{1}{2n-1}\left(e^{-i\left(n-\frac{1}{2}\right)\pi u}+e^{-i\left(n-\frac{1}{2}\right)\pi v}\right)\left(e^{i\left(n-\frac{1}{2}\right)\pi u'}+e^{i\left(n-\frac{1}{2}\right)\pi v'}\right),\nonumber\\
\wss &\equiv& \frac{1}{4\pi}\sum_{n=1}^\infty\frac{1}{2n-1}\left(e^{-i\left(n-\frac{1}{2}\right)\pi u}+e^{-i\left(n-\frac{1}{2}\right)\pi v}\right)\Psi_S^*(n,u',v'),\nonumber\\
\wsss &\equiv& \frac{1}{4\pi}\sum_{n=1}^\infty\frac{1}{2n-1}\Psi_S(n,u,v)\left(e^{i\left(n-\frac{1}{2}\right)\pi u'}+e^{i\left(n-\frac{1}{2}\right)\pi v'}\right),\nonumber\\
\wssss &\equiv& \frac{1}{4\pi}\sum_{n=1}^\infty\frac{1}{2n-1}\Psi_S(n,u,v)\Psi_S^*(n,u',v').
\label{eq:sone} 
\eea
Using the following result
\be
\sum_{n=1}^\infty \frac{e^{i\left(n-\frac{1}{2}\right)\pi x}}{(2n-1)^j} =
\li_j\left(e^{i\pi\frac{x}{2}}\right)-\frac{1}{2^j}\li_j\left(e^{i\pi x}\right),
\label{eq:soneexp}
\ee
$\ws,\wss,\wsss$ and $\wssss$ can further be simplified up to $\cO(m^4)$ as we have shown in Appendix \ref{sec:wight-app}.
In particular, $\ws$ can be written as
\be
\ws=\frac{1}{4\pi}\left(\tanh^{-1}\left(e^{-\frac{i\pi(u-u')}{2}}\right)+\tanh^{-1}\left(e^{-\frac{i\pi(v-v')}{2}}\right)+\tanh^{-1}\left(e^{-\frac{i\pi(u-v')}{2}}\right)+\tanh^{-1}\left(e^{-\frac{i\pi(v-u')}{2}}\right)\right).\label{eq:sonetanh}
\ee

Despite these simplifications in $\wsj$, it  is difficult to find a general closed form expression for $\wsj$. Instead, as  was done in \cite{Afshordi:2012ez}, we
focus on two subregions of $\diam$, as shown in Fig.~\ref{fig:cd}. In the center, far away from the boundary,  one expects to obtain the Minkowski
vacuum, while in the corner, one expects the Rindler vacuum. In  the massless case studied by \cite{Afshordi:2012ez} the
former expectation was shown to be the case. However, in the corner, instead of the Rindler vacuum, they found that 
that $\wsj$ looks like the massless mirror vacuum. One of the main motivations to look at the massive case, is to compare with
these results. 

\begin{figure}[h]
\centering{\includegraphics[height=4.5cm]{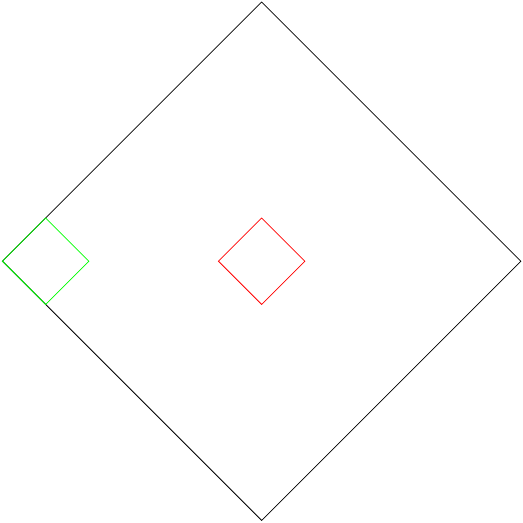}
\caption{The  center and corner regions in the causal diamond $\diam$.}
\label{fig:cd}}
\end{figure}

We now write down the expressions for the various vacua that we wish to compare with: 
\begin{eqnarray}
  \wmink(u,v;u',v')&=&-\frac{1}{4\pi}\ln\left(\cof^2e^{2\gamma}|2\Du\Dv|\right)-\frac{i}{4}\sgn(\Du+\Dv)\theta(\Du\Dv),\label{eq:minkm0}\\
  \wminkm(u,v;u',v')&=&\frac{1}{2\pi}K_0\left(m\sqrt{-2\Du\Dv+i(\Du+\Dv)\e}\right), \label{eq:massmink} \\ 
 \wrind(\eta,\xi,
  \eta',\xi')&=&-\frac{1}{4\pi}\ln\left(\cof^2e^{2\gamma}|\Delta\eta^2-\Delta\xi^2|\right)-\frac{i}{4}\sgn(\Delta\eta)\theta(\Delta\eta^2-\Delta\xi^2), \label{eq:rindm0}
  \\
  \wrindm(\eta,\xi, \eta',\xi') &=& \wminkm(u,v,u',v')-\frac{1}{2\pi}\int_{-\infty}^\infty
                                    \frac{dy}{\pi^2+y^2}K_0(m\gamma_1), \label{eq:rindm} \\
 \wmirr(u,v,u',v')&=&\wmink(u,v;u',v')-\wmink(u,v;v',u'), \label{eq:mirror} \\
\wmirrm(u,v,u',v')&=&\wminkm(u,v;u',v')-\wminkm(u,v;v',u').   \label{eq:mirrorm}                                 
\end{eqnarray}
In the expression Eqn.~(\ref{eq:minkm0})   for the massless Minkowski vacuum,  $\gamma$ is the Euler-Mascheroni constant
and $\cof=0.462$ (obtained in \cite{Afshordi:2012ez} by comparing $\wsj$ with $\wmink$). In the expression  Eqn.~(\ref{eq:massmink})  for the
massive Minkowski vacuum \cite{abdallah}, $K_0$ is the modified Bessel function of the second kind, with $\e$ a constant
such that that $0< \e \ll 1$.
In the expressions Eqn.~(\ref{eq:rindm0}) and Eqn.~(\ref{eq:rindm}) (see \cite{candelas:1976})
 for the Rindler vacua, $\alpha$ is the acceleration
parameter, with
\begin{eqnarray} 
\eta = \frac{1}{\alpha}\tanh^{-1}\left(\frac{u+v}{u-v}\right),&\quad & 
\xi= \frac{1}{2\alpha}\ln\left(-2\alpha^2uv\right), \nonumber \\ 
\Delta\eta=\eta-\eta', \quad \Delta\xi=\xi-\xi',  &\quad& \gamma_1=\sqrt{\zeta^2+\zeta'^2+2\zeta\zeta'\cosh(y-\alpha(\eta-\eta'))},\nonumber\\
\zeta=\sqrt{-2uv}.
\end{eqnarray} 

\subsection{The center}

We now consider a small diamond ${\diam}_l$ at the center of  $\diam$ with $l \ll 1$  where one expects  $\wsj$ to resemble  
$\wminkm$.  For small $\Delta u,
\Delta v$, $\wminkm$ can be written as
\be
\wminkm(u,v;u',v')\approx
-\frac{1}{4\pi}\ln\left(\frac{m^2e^{2\gamma}}{2}\left|\Du\Dv\right|\right)-\frac{i}{4}\sgn(\Du+\Dv)\theta(\Du\Dv)J_0\left(m\sqrt{2\Du\Dv}\right).\label{eq:centerminkm}
\ee
To leading logarithmic order this is similar in form to $\wmink$ (Eqn.~(\ref{eq:minkm0})), with $m$ replaced by $2 \cof$.
We plot these functions in Fig.~\ref{fig:wmink}. For $m \ll \cof$ the real part of $\wminkm$ is larger than
$\wmink$ and for $m \gg \cof$ it is smaller. When $m_c=2\cof$, the two coincide in this approximation. 

\begin{figure}[h]
\centerline{\includegraphics{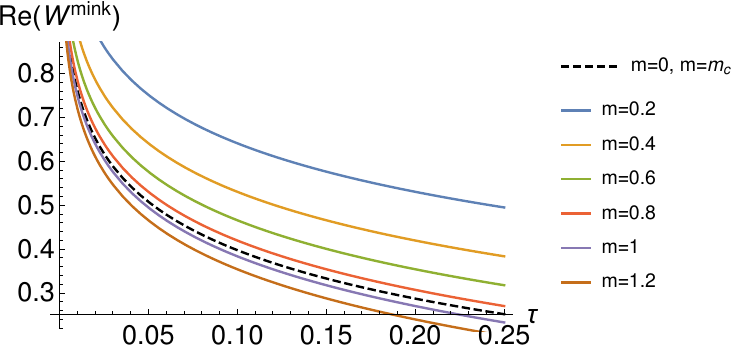}}
\caption{Plot of $\re(\wmink)$ and $\re(\wminkm)$ vs the proper time ($\tau$)}
\label{fig:wmink}
\end{figure}

Let us begin with $\wsj^A$, Eqns.~(\ref{eq:splitwsja})  and (\ref{eq:aone}). As shown in  Appendix
\ref{sec:wight-app},  the expressions for $\wa,\waa,\waaa$ and $\waaaa$ can be written in terms of Polylogarithms 
$\li_s(x)$.  For small $x$, i.e., near the center of $\diam$ they simplify for the  $s=1,3$ and $5$ to 
\bea
\li_1\left(e^{i\pi x}\right)&=& -\ln(-i\pi x)-\frac{i\pi x}{2}+\frac{\pi^2x^2}{24}+\cO(x^3), \label{eq:li1}\\
\li_3\left(e^{i\pi x}\right)&=& \zeta(3)+\frac{i\pi^3x}{6}+\left(-\frac{3\pi^2}{4}+\frac{\pi^2}{2}\ln(-i\pi
  x)\right)x^2+\cO(x^3), \\
\li_5\left(e^{i\pi x}\right)&=& \zeta(5)+\frac{i\pi^5x}{90}-\frac{\pi^2\zeta(3) x^2}{2}+\cO(x^3),
\label{eq:polylog} 
\eea
where $\zeta$ are the Riemann Zeta function and $x$ denotes $u$ or $v$.  In the expression for $\wa$, the
constant and linear terms in $x$ cancel out, so that
\bea
\wa &=& -\frac{1}{8\pi}\left(\ln(|u-u'||v-v'|)-\ln(|u-v'||v-u'|)-C_1\frac{i\pi}{2}\right)\nonumber\\
&&-\left(\frac{\pi}{96}+\frac{3m^2}{8\pi}+\frac{m^4}{8\pi}\left(\frac{1}{4}-\frac{7\zeta(3)}{\pi^2}\right)\right)(u-v)(u'-v')\nonumber\\
&&-\frac{m^2}{8\pi}\left(1+\frac{m^2}{12}\right)\left[(u-u')^2\ln\left(-i\pi(u-u')\right)+(v-v')^2\ln\left(-i\pi(v-v')\right)\right.\nonumber\\
&&\left.-(u-v')^2\ln\left(-i\pi(u-v')\right)-(v-u')^2\ln\left(-i\pi(v-u')\right) \right] +\cO(\Delta^3), 
\eea
where $C_1=\sgn(u-u')+\sgn(v-v')-\sgn(u-v')-\sgn(v-u')$ and $\Delta$ collectively denotes either $u-u',v-v',v'-u$ or $ v-u'$.   
For sufficiently small $x$, the logarithmic term dominates significantly over other terms,  and hence in $\diam_l$ 
\be
\wa=-\frac{1}{8\pi}\left(\ln(|u-u'||v-v'|)-\ln(|u-v'||v-u'|)-C_1\frac{i\pi}{2}\right) + \cO(m^2,\Delta^2), \label{eq:wa-cent}
\ee
where we have hidden all the mass dependence in the correction.

Next, $\waa,\waaa$ and $\waaaa$ also involve another set of Polylogarithms of the type $\li_s(-e^{i\pi x})$ for $ s \geq 2$
as well as $\li_s(e^{ i\pi x})$ for $s=2,3,4$, which are multiplied to the functions $g_j(m;u,v)$ and $f_j(m;u,v)$ given in Eqn.~(\ref{eq:gjfj}). The $g_j(m;u,v)$ and $f_j(m;u,v)$ themselves go to zero  either linearly or quadratically with $u,v$. This second set of  Polylogarithms, unlike the first in Eqn.~(\ref{eq:polylog}), are strictly convergent as $x \rightarrow 0$. Hence the $\waa,\waaa $ and $\waaaa$ are  
strongly sub-dominant with respect to $\wa$ so that 
\be
\wsj^A(u,v,u',v') = -\frac{1}{8\pi}\left(\ln(|u-u'||v-v'|)-\ln(|u-v'||v-u'|)-C_1\frac{i\pi}{2}\right)+\cO(m^2,\Delta^2).
\ee
Here we note that while the mass correction is significant in the antisymmetric SJ modes, it becomes  insignificant in $\wsj^A$ in the center of the diamond, compared to the dominating logarithmic term.  Thus we see that in the center of $\diam$,  $\wsj^A$ is identical to the massless case found in \cite{Afshordi:2012ez}. 

We now turn to the symmetric part $\wsj^S$, Eqns.~(\ref{eq:splitwsjs})  and (\ref{eq:sone}). The expressions for
$\ws,\wss,\wsss$  and $\wssss$ can again be written in terms of Polylogarithms $\li_s(x)$  as shown in Appendix
\ref{sec:wight-app}. For $\ws$ however, the form given in Eqn.~(\ref{eq:sonetanh}) is easier to analyze.  
Noting that for small $x$ 
\be
\tanh^{-1}\left(e^{i\pi x/2}\right)=-\frac{1}{2}\ln\left(\frac{-i\pi x}{4}\right)-\frac{\pi^2x^2}{96}+\cO(x^3),
\ee
near the center of $\diam$ we see that
\bea
\ws&=&-\frac{1}{8\pi}\left[\ln(|u-u'||v-v'|)+\ln(|u-v'||v-u'|)+4\ln\left(\frac{\pi}{4}\right)-C_2\frac{i\pi}{2}\right]\nonumber\\
&&-\frac{\pi}{384}\left((u-u')^2+(u-v')^2+(v-u')^2+(v-v')^2\right)+\cO(\Delta^3),
\eea
where $C_2=\sgn(u-u')+\sgn(v-v')+\sgn(u-v')+\sgn(v-u')$. Since the logarithmic term dominates,
\be
\ws=-\frac{1}{8\pi}\left[\ln(|u-u'||v-v'|)+\ln(|u-v'||v-u'|)+4\ln\left(\frac{\pi}{4}\right)-C_2\frac{i\pi}{2}\right]+\cO(\Delta^2).  
\ee
Next, we see that $\wss,\wsss$ and $\wssss$ involve a set of Polylogarithms of the type $\li_s(e^{i\pi x}),$ for $
s=2,3$, multiplied by linear and quadratic functions of $u,v,u'$ and $v'$. This set of  Polylogarithms are in fact
strictly convergent as $x \rightarrow 0$. Hence the $\wss,\wsss $ and $\wssss$ are strongly sub-dominant, with respect
to $\ws$, so that
\bea
\wsj^S(u,v,u',v')&=&-\frac{1}{8\pi}\left[\ln(|u-u'||v-v'|)+\ln(|u-v'||v-u'|)+4\ln\left(\frac{\pi}{4}\right)-C_2\frac{i\pi}{2}\right]\nonumber\\
&&\quad+\emc+\cO(m^2,\Delta^2), \label{eq:emc} 
\eea
where $\emc$ is the correction in the center coming from the
approximation to the quantization condition Eqn.~(\ref{eq:appquant}). We will determine this 
numerically in Section \ref{sec:numerical}.
Up to this mass correction $\wsj^S$ resembles the massless case found in \cite{Afshordi:2012ez}.

Putting these pieces together we find that
\be
\wsj^{center}(u,v,u',v') \approx 
-\frac{1}{4\pi}\ln|\Du\Dv|-\frac{i}{4}\sgn(\Du+\Dv)\theta(\Du\Dv)-\frac{1}{2\pi}\ln\left(\frac{\pi}{4}\right)+\e_m^{center}. \label{eq:wcent}
\ee
A direct comparison with $\wmink$ gives 
\begin{equation}
\wsj^{center}(u,v,u',v') - \wmink(u,v,u',v') \approx -\frac{1}{2\pi}\ln\left(\frac{\pi}{4}\right)+\e_m^{center} +
\frac{1}{4\pi}\ln\left(2 \cof^2e^{2\gamma}\right),  \label{eq:compcenter} 
  \end{equation} 
where $\cof \approx 0.462$ is fixed by comparing the massless $\wsj$ with $\wmink$ as in \cite{Afshordi:2012ez}.

\subsection {The corner}

We now consider either of the two spatial corners of the diamond, $\diam_c \subset \diam$ as shown in Fig.~\ref{fig:cd}. We use the small $\Delta u, \Delta v$ form of $\wminkm$ to
express
\be
\wmirrm \approx -\frac{1}{4\pi}\ln\left|\frac{\Du\Dv}{(u-v')(v-u')}\right|-\frac{i}{4}\sgn(\Du+\Dv)\left(\theta(\Du\Dv)-\theta((u-v')(v-u'))\right) \label{eq:approxmirrm} .
 \ee
 As in \cite{Afshordi:2012ez} we make the coordinate transformation
\be
\{u,u',v,v'\}\rightarrow \{u-1,u'-1,v+1,v'+1\}, \label{eq:corner}
\ee
which brings the origin $(0,0)$ to the left corner of the diamond. 

For $\wsj^A$ (Eqn.~(\ref{eq:splitwsja}) and Eqn.~(\ref{eq:aone})),  we note that $\wa$ is invariant under this coordinate
transformation and hence given by Eqn.~(\ref{eq:wa-cent}) near the origin of $\diam_c$. In $\waa,\waaa$ and $\waaaa$ the
constant terms cancel out and, similar to the center calculation, they goes to zero linearly with $u,v$ and hence are
strongly sub-dominant with respect to $\wa$. 
Therefore,  in the corner, $\wsj^A$ simplifies to
\be
\wsj^A(u,v,u',v') = -\frac{1}{8\pi}\left(\ln(|u-u'||v-v'|)-\ln(|u-v'||v-u'|)-C_1\frac{i\pi}{2}\right)+\cO(m^2,\Delta),
\ee
and the sub-dominant part is now linear in $\Delta$.

For $\wsj^S$ (Eqn.~(\ref{eq:splitwsjs}) and Eqn.~(\ref{eq:sone})), under the coordinate transformation
\be
\ws=\frac{1}{4\pi}\left(\tanh^{-1}\left(e^{-\frac{i\pi(u-u')}{2}}\right)+\tanh^{-1}\left(e^{-\frac{i\pi(v-v')}{2}}\right)-\tanh^{-1}\left(e^{-\frac{i\pi(u-v')}{2}}\right)-\tanh^{-1}\left(e^{-\frac{i\pi(v-u')}{2}}\right)\right).
\ee
In the corner $\diam_c\subset\diam$ this simplifies to
\bea
\ws&=&\frac{1}{8\pi}\left[-\ln(|u-u'||v-v'|)+\ln(|u-v'||v-u'|)+C_1\frac{i\pi}{2}\right]\nonumber\\
&&-\frac{\pi}{384}\left((u-u')^2+(v-v')^2-(u-v')^2-(v-u')^2\right)+\cO(\Delta^3).
\eea
For sufficiently small $\Delta$, the logarithmic term dominates the other terms so that
\be
\ws = \frac{1}{8\pi}\left[-\ln(|u-u'||v-v'|)+\ln(|u-v'||v-u'|)+C_1\frac{i\pi}{2}\right]+\cO(\Delta^2).
\ee
As in the center, $\wss$ and $\wsss$ go  to zero while 
\be
\wssss=\frac{7\zeta(3)m^4}{8\pi^3}+\cO(\Delta)\approx 0.034m^4.
\ee
Therefore in the corner  we see that 
\be
\wsj^S \approx  \frac{1}{8\pi}\left[-\ln(|u-u'||v-v'|)+\ln(|u-v'||v-u'|)+C_1\frac{i\pi}{2}\right]+0.034m^4+\emcr  \label{eq:emcr}
\ee
i.e., there is a mass correction to the massless $\wsj^S$. $\emcr$ is, as in the center calculation, a small but finite term coming from the
approximation to the quantization condition Eqn.~(\ref{eq:appquant}), which we will  evaluate numerically in
Sec.~\ref{sec:numerical}. 

Putting these pieces together we find that in the corner $\wsj$ takes the form 
\bea
\wsj^{corner}(u,v,u',v') &\approx& -\frac{1}{4\pi}\ln\left|\frac{\Du\Dv}{(u-v')(v-u')}\right|-\frac{i}{4}\sgn(\Du+\Dv)\left(\theta(\Du\Dv)-\theta((u-v')(v-u'))\right)\nonumber\\
&&+0.034m^4+\emcr.
\eea
A direct comparison with $\wmirrm$ Eqn.~(\ref{eq:approxmirrm})  gives 
\be
\wsj^{corner}(u,v,u',v') - \wmirrm(u,v,u',v') \approx 0.034m^4+\emcr  \label{eq:compcorner} .
\ee 

\subsection {Numerical simulations for determining $\e_m$} \label{sec:numerical}
The formal expansion of $\wsj $ in terms of the SJ modes Eqn.~(\ref{eq:fullw}) can be truncated and evaluated numerically
in $\diam$. Here we do not need to use the approximation of the quantization condition Eqn.~(\ref{eq:appquant}). This
allows us to evaluate the ensuing corrections $\e_m^{center},\e_m^{corner}$ numerically,  and thus quantify the
comparisons of $\wsj$ obtained in the center and  corner of $\diam$ with the standard vacua. 

We begin with the $N^{\mathrm{th}}$ truncation $\wsj^t$ of the series form of $\wsj$ Eqn.~(\ref{eq:fullw}) in the full
diamond $\diam$ for $N=100, 200, \ldots 1000$.  Fig.~\ref{fig:truncations}  shows an explicit convergence of $\wsj^t$ for
these values of $N$. 
For the plot we considered  the pairs  $(u,v)=(x,x)$ and $(u',v')=(-x,-x)$ for timelike separated points, and $(u,v)=(x,-x)$ and $(u',v')=(-x,x)$
for spacelike separated points.  From this point onwards, we will consider $\wsj^t$ for $N=1000$.
\begin{figure}[h]
  \centering{\begin{tabular}{cc}
\includegraphics[height=4.2cm]{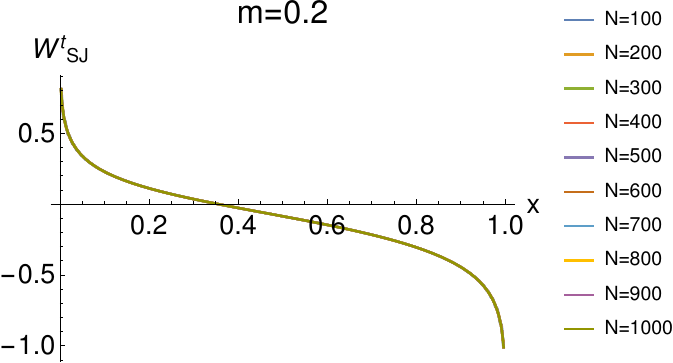} &\hskip 2cm 
                                                            \includegraphics[height=4.2cm]{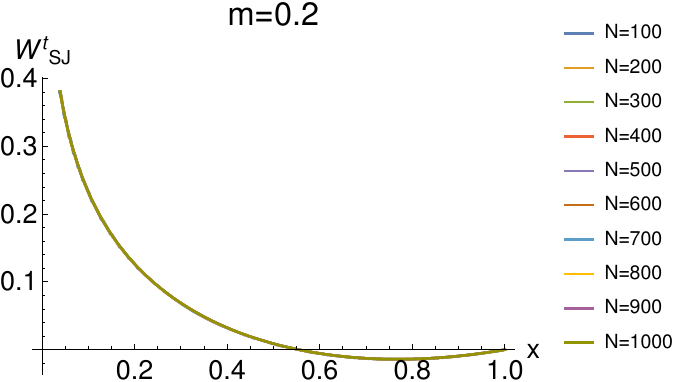}\\
                \includegraphics[height=4.2cm]{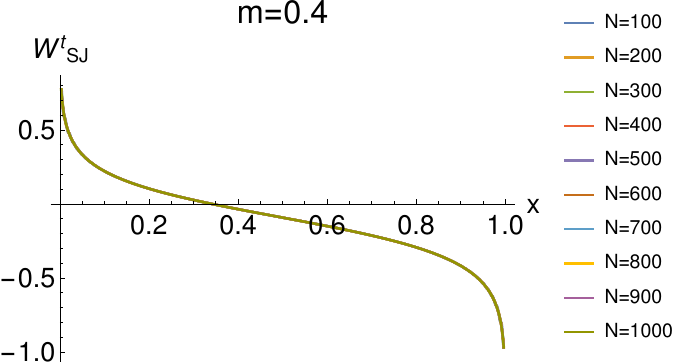} &\hskip 2cm 
\includegraphics[height=4.2cm]{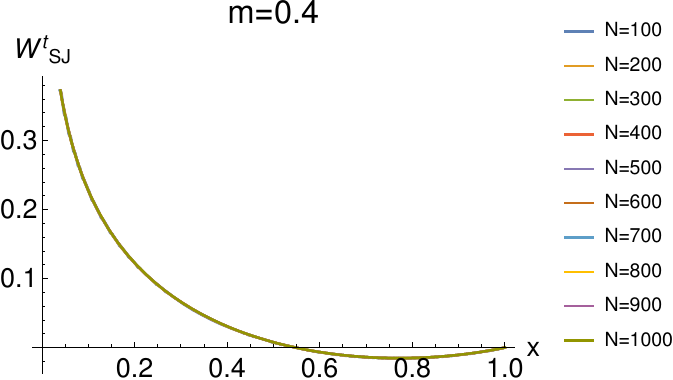} \\
\end{tabular}
\caption{We show the convergence of the truncation of the series $\wsj^t$ with $N$ for $m=0.2,0.4$ for timelike separated points (left) and spacelike separated points (right).}
\label{fig:truncations}}
\end{figure} 
  
Next, we consider the difference $\wsj^t-\wsj^{t,approx}$ where the latter uses the approximation Eqn.~
(\ref{eq:appquant}), both in the center and the corner of $\diam$ in order to obtain $\emc,\emcr$.  It suffices to
look at their  symmetric parts $\wsj^{S,t}$ since only these contribute (see Eqns.~(\ref{eq:emc}),
(\ref{eq:emcr})). $\emc$ and $\emcr$ are {\it not} strictly constants. However, as we will see, they are approximately
so. As in \cite{Afshordi:2012ez}, they are evaluated by taking a set of randomly selected points in a small diamond in
the center as well as in the corner. Here we take $10$ points and consider all $55$ pairs between them to calculate
$\emc,\emcr$. What we find in Fig.~\ref{fig:em} is that they are very nearly equal and hence we can consider their
average. The explicit averages for these masses are tabulated in Table \ref{tab:em} for future reference.  
\vskip 1cm 
\begin{figure}[h]
\centerline{\begin{tabular}{cc}
\includegraphics[height=4.2cm]{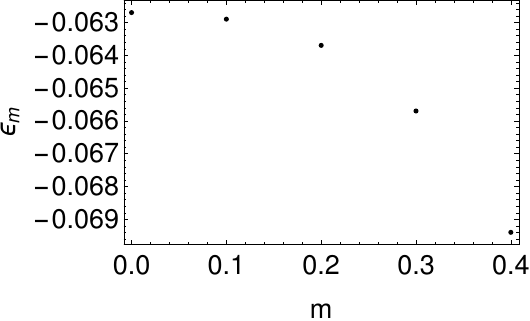}& \hskip 1cm 
\includegraphics[height=4.2cm]{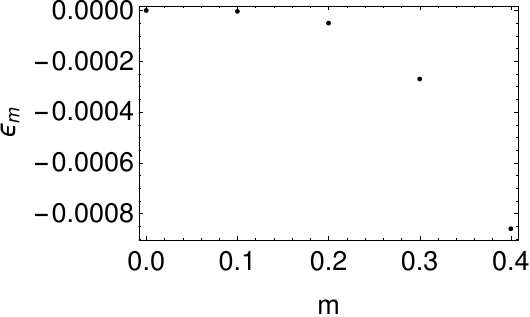}\\
\end{tabular}}
\caption{$\emc$ and $\emcr$  evaluated in a small diamond of $l=10^{-5}$ in the center and the corner of $\diam$, for
  m=0,0.1,0.2,0.3 and 0.4. The standard deviation is very small and hence we can take $\emc$ and $\emcr$ to be
  approximately constant.}
\label{fig:em}
\end{figure} 
\vskip 1cm 
\begin{table}[htb]
\centering{
\begin{tabular}{|c|c|c|}
\hline
mass & $\e_m^{center}$ & $\e_m^{corner}$\\
\hline
0 & -0.0627 & 0\\
0.1 & -0.0629 & $-3.5\times 10^{-6}$\\
0.2 & -0.0637 & -0.00005\\
0.3 & -0.0657 & -0.00027\\
0.4 & -0.0694 & -0.00086\\
\hline
\end{tabular}
\caption{A tabulation of $\emc, \emcr$  for different $m$}
\label{tab:em}}
\end{table}
\vskip 0.1in
This allows us to now compare $\wsj$ calculated in the center Eqn.~(\ref{eq:wcent}) with $\wmink, \wminkm$. 
The difference with $\wmink$ given in  Eqn.~(\ref{eq:compcenter}) is indeed very small. For  $m=0.2$,  for example,  
\begin{equation}
\wmink - \wsj^{center}  \simeq -\frac{1}{4 \pi} \log(2 \times 0.462^2)- 
 \frac {\gamma}{2\pi} - (-\frac{1}{2 \pi} \log(\frac{\pi}{4}) - \emc) \simeq 0.001. 
\end{equation}
Similarly, in the corner, the difference with $\wmirrm$ is again very small. For example for $m=0.2$ it gives
\begin{equation} 
\wmirrm - \wsj^{corner}   \simeq   0.034\times (0.2)^4+\emcr \simeq 4\times 10^{-6} 
 \end{equation} 
Thus, we see that in the small mass limit, $\wsj$ does not differ from the massless Minkowski vacuum in the center region, and
continues to mimic the mirror vacuum in the corner.

Since our analytical calculation is restricted to a very small $\Delta u, \Delta v$, where perhaps the effect of a
small mass is small, we can use the truncation $\wsj^t$ for comparisons with the standard vacuum in larger regions of $\diam$. 
This is shown in the residue plots in Figs.~\ref{fig:res-m0}. In the full diamond, we consider 
 the pairs  $(u,v)=(x,x)$ and $(u',v')=(-x,-x)$ for timelike separated points, and $(u,v)=(x,-x)$ and $(u',v')=(-x,x)$
 for spacelike separated points. We find that for $m=0.2$, $l\sim 0.02$, $\wsj^t$ differs very
 little from the 
 massless Minkowski vacuum, while as the mass increases, so does the discrepancy. 
 \begin{figure}[h]
\centering{\begin{tabular}{cc}
              \includegraphics[height=4cm]{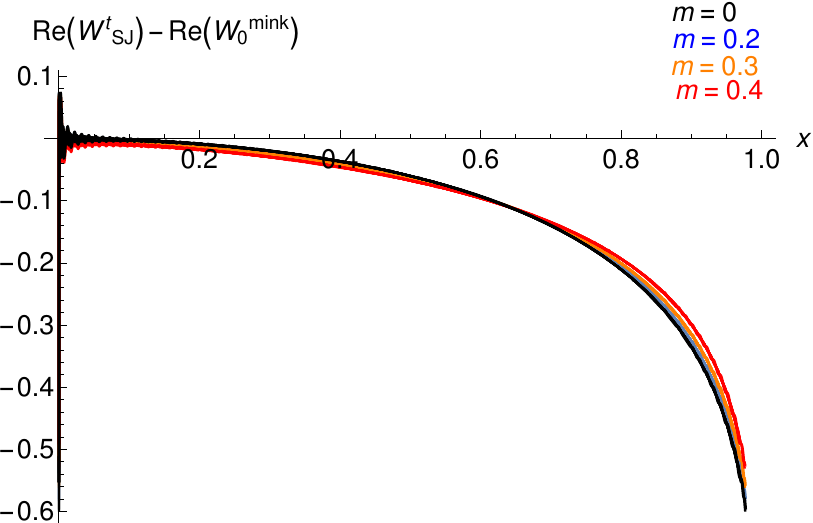}&\hskip 1cm
                                                                     \includegraphics[height=4cm]{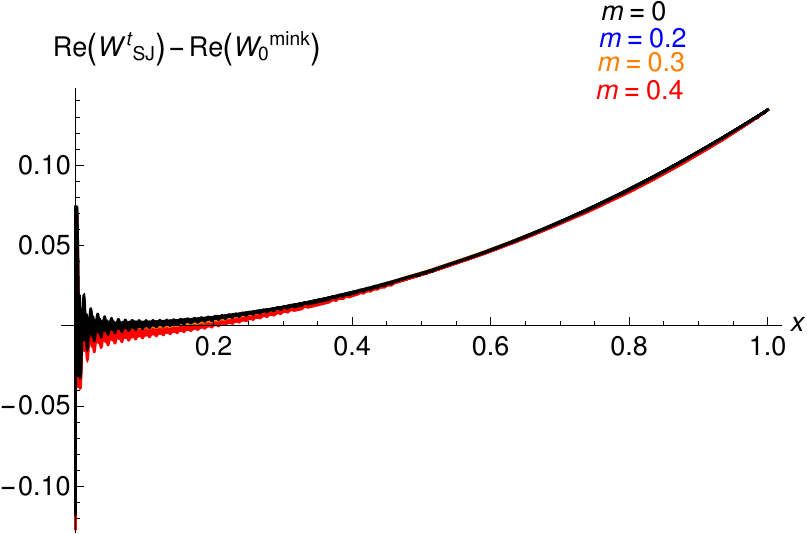}\\
\includegraphics[height=4cm]{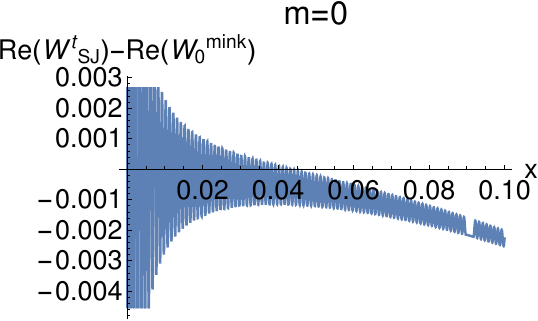}&\hskip 2cm
                                                            \includegraphics[height=4cm]{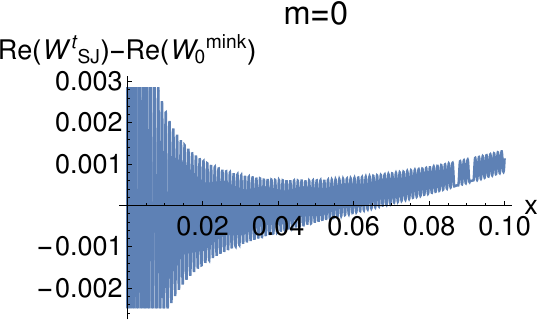}\\

         \includegraphics[height=4cm]{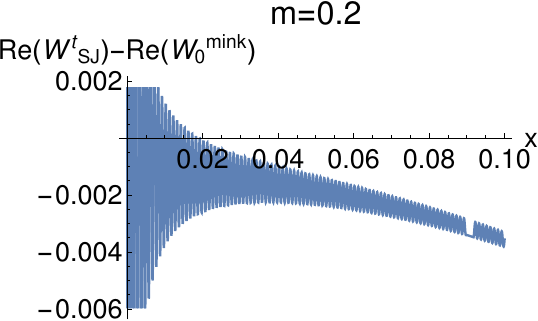}&\hskip 2cm
                                                            \includegraphics[height=4cm]{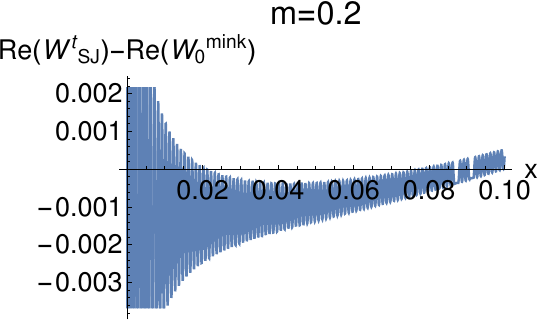}\\

              \includegraphics[height=4cm]{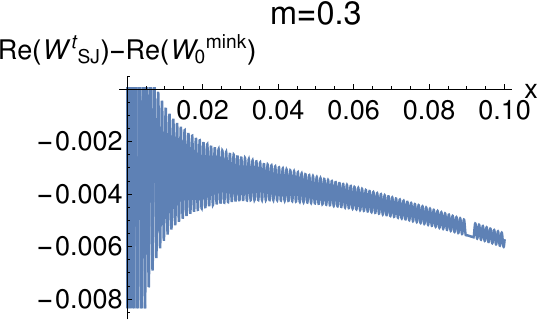}&\hskip 2cm
                                                            \includegraphics[height=4cm]{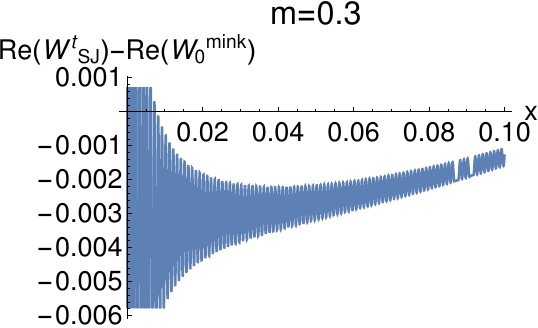}\\

              \includegraphics[height=4cm]{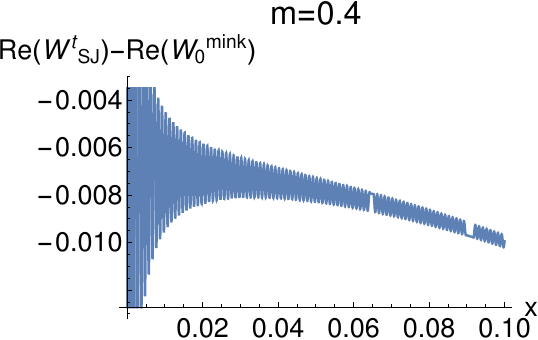}&\hskip 2cm
                                                            \includegraphics[height=4cm]{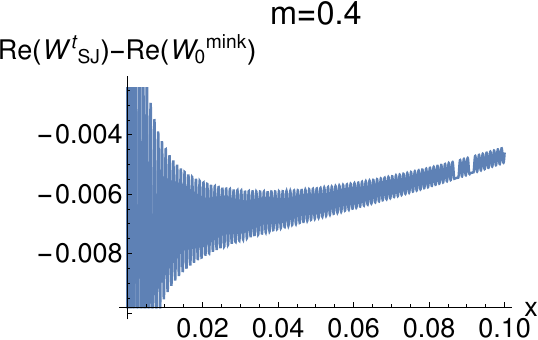}
            \end{tabular}
\caption{Residue plot of $\re(\wsj^t-\wmink)$ for timelike and spacelike separated points respectively, for the full
  diamond, as well as in a center region of size $l\sim 0.1$.}
\label{fig:res-m0}}
\end{figure}
On the other hand, as we see in  Figs.~\ref{fig:res-m}  we find that $\wsj^t$ clearly does  {\it not} agree with the massive Minkowski vacuum, in this small mass
limit.
\begin{figure}[h]
\centerline{\begin{tabular}{cc}
\includegraphics[height=5cm]{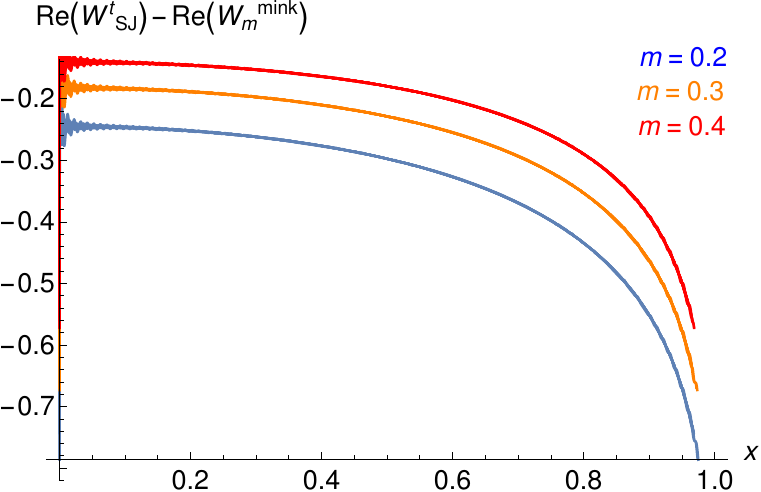}&\hskip 1cm
\includegraphics[height=5cm]{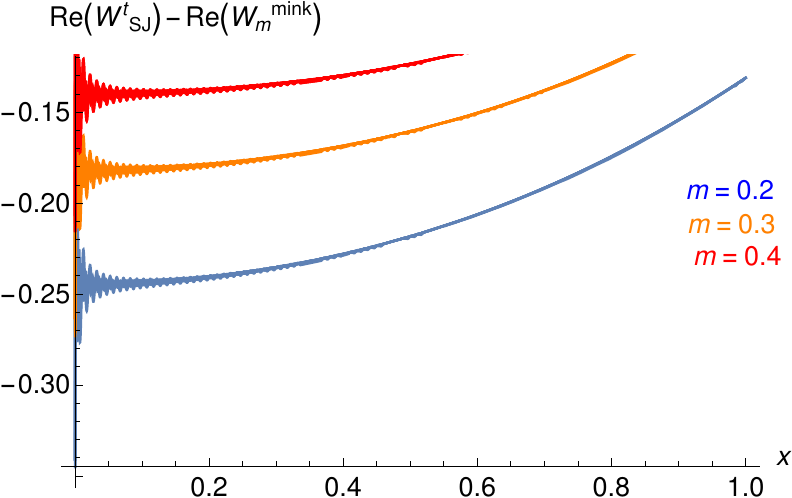}
\end{tabular}}
\caption{Residue plot of $\re(\wsj^t-\wminkm)$ for timelike and spacelike separated points respectively, for the full
  diamond. The discrepancy is obvious. }
\label{fig:res-m}
\end{figure}
 
A similar calculation in the corner shows that $\wsj^t$ looks like the massive mirror vacuum rather than the Rindler vacuum. Here, we consider pairs of points: $(u,v)=(l+x,-l+x)$ and $(u',v')=(l-x,-l-x)$ for timelike separation  and
$(u,v)=(l+x,-l-x)$ and $(u',v')=(l-x,-l+x)$ for spacelike separation, where the origin $(0,0)$ is at the left corner of the diamond $\diam$ and $2l$ is the length of the corner diamond $\diam_c$. This is shown in the residue plots in Figs.~\ref{fig:corres01} and \ref{fig:corres}.

\begin{figure}[h]
\centering{\begin{tabular}{cc}
\includegraphics[height=4cm]{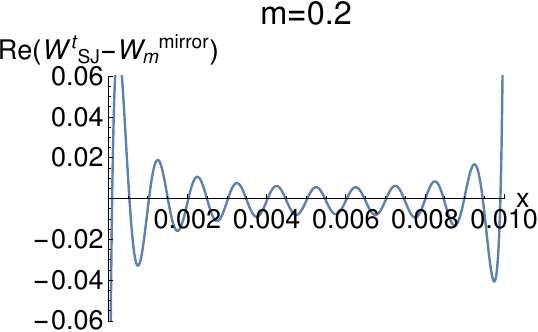}&\hskip 2cm
\includegraphics[height=4cm]{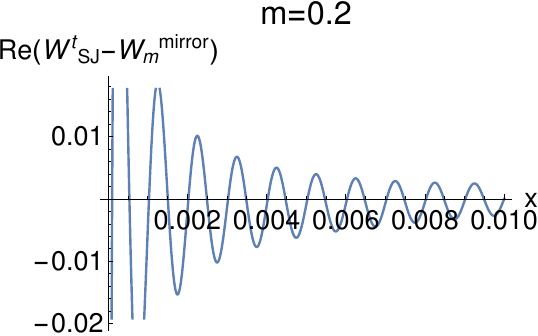}\\
\includegraphics[height=4cm]{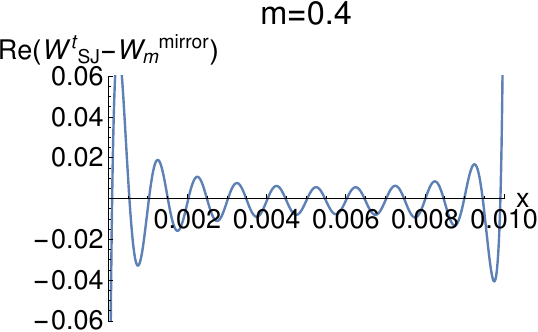}&\hskip 2cm
\includegraphics[height=4cm]{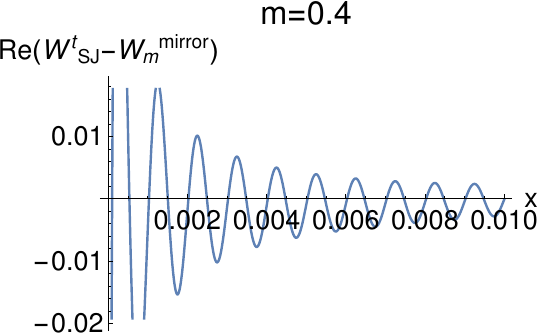}
\end{tabular}
\caption{Residue plot of $\re(\wsj^t-\wmirrm)$ for timelike and spacelike separated points respectively in the corner
  region,  $l\sim 0.01$.}
\label{fig:corres01}}
\end{figure}
\begin{figure}[h]
\centering{\begin{tabular}{cc}
\includegraphics[height=4cm]{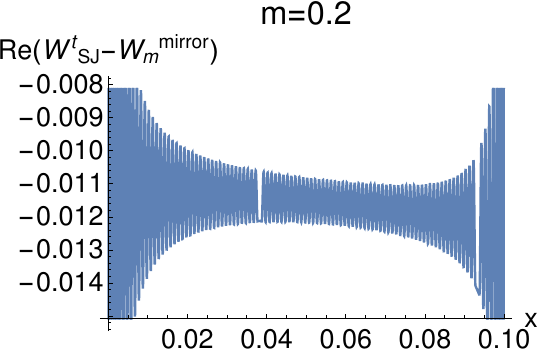}&\hskip 2cm
\includegraphics[height=4cm]{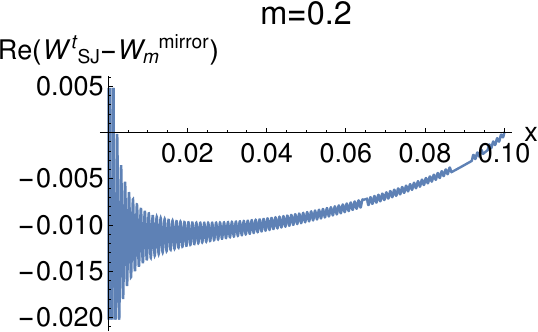}\\
\includegraphics[height=4cm]{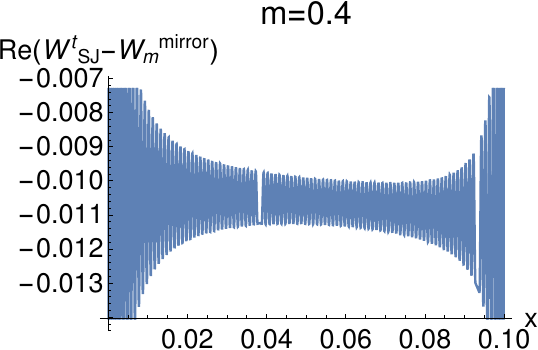}&\hskip 2cm
\includegraphics[height=4cm]{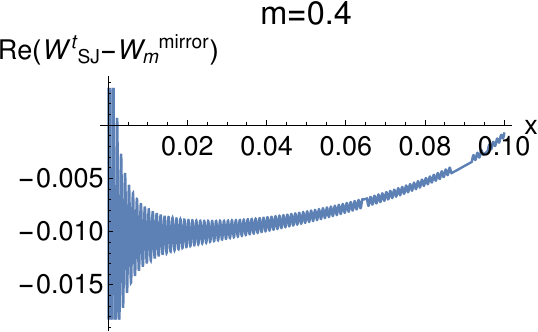}
\end{tabular}
\caption{Residue plot of $\re(\wsj^t-\wmirrm)$ for timelike and spacelike separated points respectively in the corner
  region, $l\sim 0.1$.}
\label{fig:corres}}
\end{figure}
\begin{figure}[h]
\centerline{\begin{tabular}{ccc}
\includegraphics[height=3.5cm]{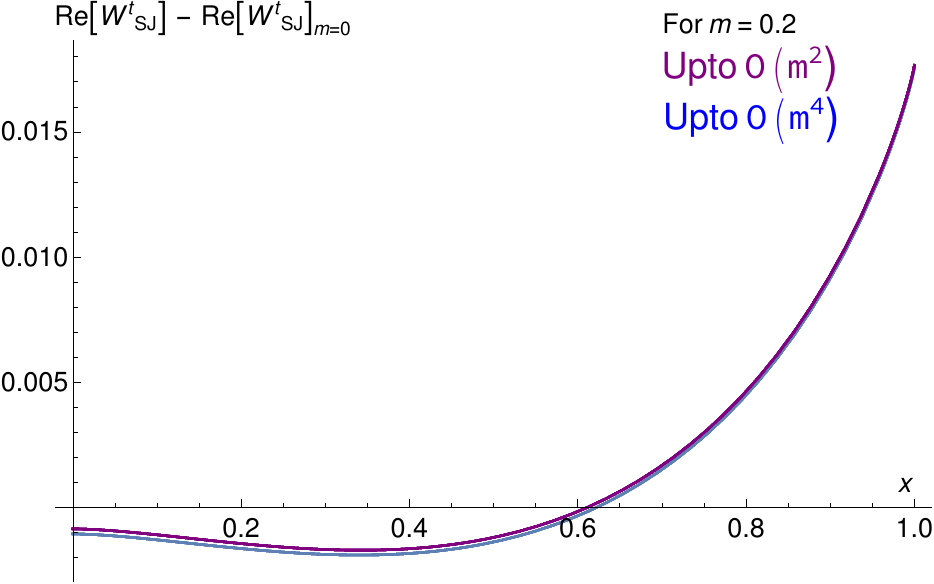}&
\includegraphics[height=3.5cm]{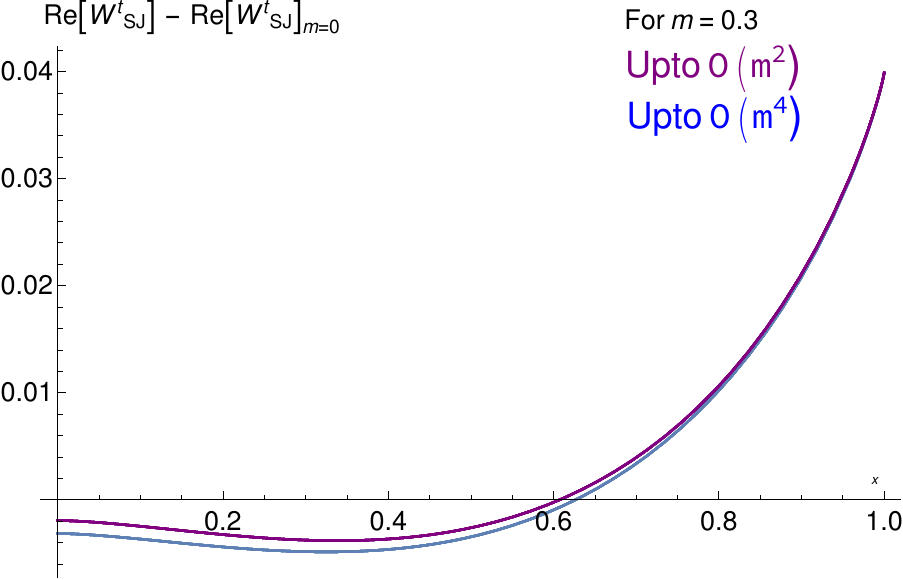}&
\includegraphics[height=3.5cm]{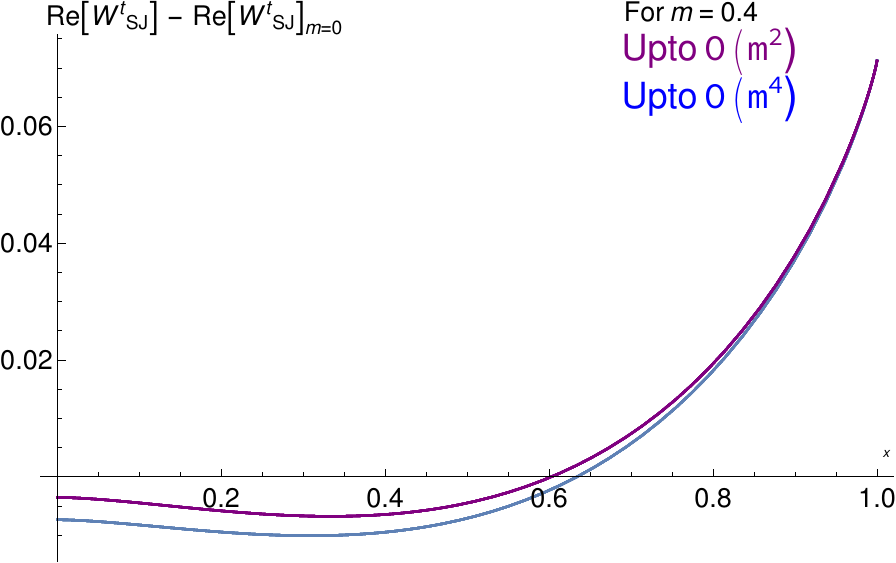}\\
\includegraphics[height=3.5cm]{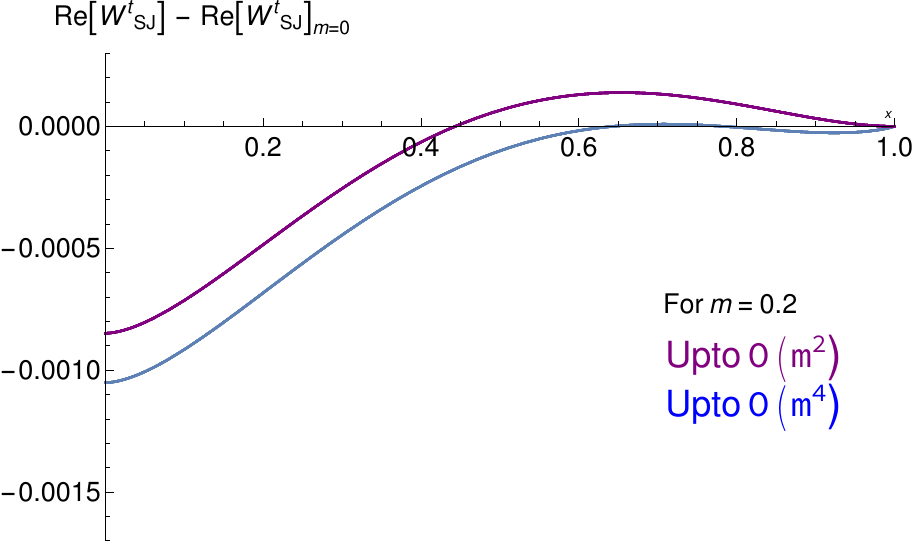}&
\includegraphics[height=3.5cm]{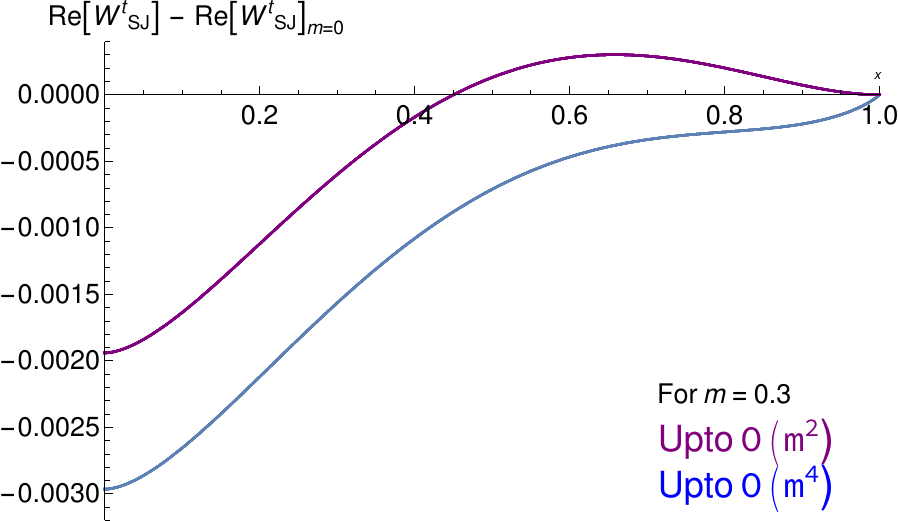}&
\includegraphics[height=3.5cm]{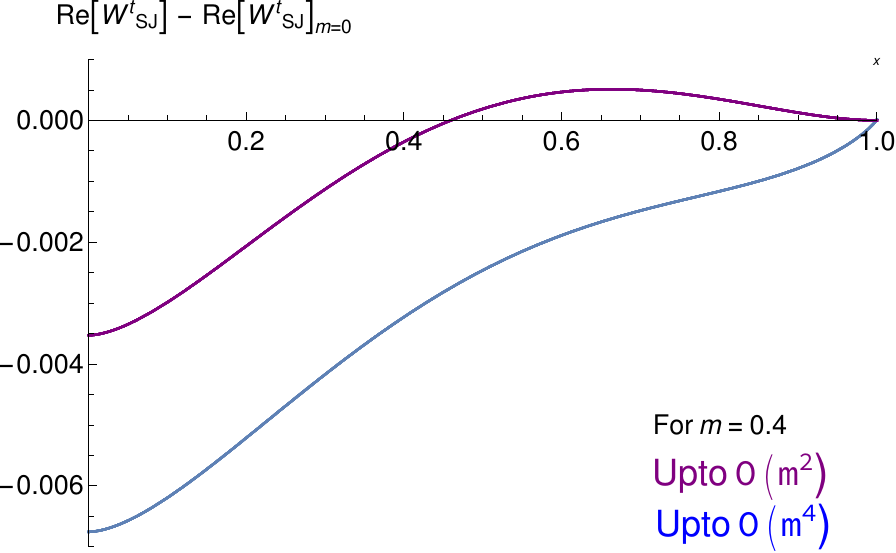}
\end{tabular}}
\caption{Plot of $\re(\wsj^t)-\re(\wsj)_{m=0}$ vs $x$ for $\cO(m^2)$ and $\cO(m^4)$ corrections. The plots in the first line are all for timelike separated points while those in the second line are for spacelike separated points.}
\label{fig:diff-re}
\end{figure}

Our calculation suggest that the $\cO(m^4)$ corrections are largely irrelevant to $\wsj$ in the center and the corner
of $\diam$.  A question that occurs is whether increasing the order of the correction makes a significant difference. In Fig.~\ref{fig:diff-re}
we show the sensitivity of  the difference in $\wsj^t$ with $\wmink$,   to $\cO(m^2)$ and $\cO(m^4)$.   As we can see,
the $\cO(m^4)$  corrections while  not negligible, are relatively small for $m\sim 0.2$.

What we have seen from our calculations so far is that in the small mass approximation, $\wsj$  continues to behave in
the center like the massless Minkowski vacuum, and in the corner as the massive Mirror vacuum. This behavior is very
curious since it suggests an unexpected mass dependence in  $\wsj$, not seen in the standard vacuum. In order to explore
this we must examine $\wsj$ for large masses. Because  we are limited in our analytic calculations, we now proceed to
a fully numerical calculation of  $\wsj$ in a causal set for comparison.  

\section{The massive SJ Wightman function in the causal set }\label{sec:causet}

This curious behavior of the SJ vacuum seems to be a result of our small mass approximation. Since we do not know how
to evaluate it analytically for finite mass we look for a numerical evaluation on a causal set $\cc_\diam$ that is approximated
by $\diam$ (see Sec.~\ref{csrev.sec} for a brief introduction to causal sets). 

$\cc_\diam$ is obtained via a Poisson sprinkling into  $\diam$ at density $\rho$. The expected total number of elements is
then $\langle N \rangle =\rho V_\diam$, where $V_\diam$ is the total volume of the diamond. The partial order is then determined by the causal relation among the elements i.e. $x_i \prec x_j$ iff $x_j$ is in the causal future of $x_i$.

The causal set SJ Wightman function $\wsjc$ is constructed using the same procedure as in the continuum,
namely starting from the  causal set retarded Green's function. The massive Green's function  in $\diam$ is
\cite{johnston,Johnston:2008za}
\be
G_m=\left(\id+\frac{m^2}{\rho}G_0\right)^{-1}G_0,
\ee
where $\id$ is the $N\times N$ identity matrix and $G_0$ is the massless retarded Green's function which is given by $G_0 = -C/2$, where $C$ is the causal matrix (see Eqn.~\eqref{csgf2d.eq}) 

We sprinkle $N=10,000$ elements in $\diam$ of length $2$, i.e., of  density $\rho=2500$ for different values of mass.   
In Fig.~\ref{fig:dis-ev}  we plot the SJ eigenvalues for these various masses. We find that the eigenvalues for small
masses are very close to the massless eigenvalues, especially for small $n$. As $n$ increases, they become
indistinguishable. In Fig.~\ref{fig:scatterwsjc} we show the scatter plot of $\wsjc$. For the smaller masses, $\wsjc$
tracks the massless case closely, but at larger masses $m\sim 10$ it shows the characteristic behavior expected of the
massive Minkowski  vacuum \cite{Johnston:2009fr}.  
\begin{figure}[h]
\centering{\begin{tabular}{cc}
\includegraphics[height=4cm]{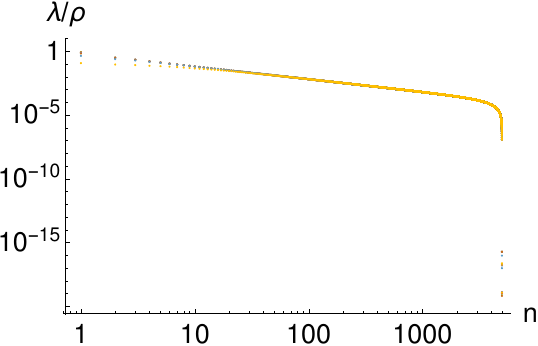}& \hskip 1cm
\includegraphics[height=4cm]{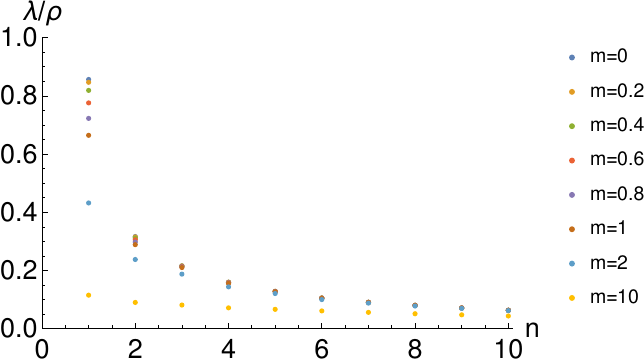}\\
(a)&(b)
\end{tabular}
\caption{(a):A log-log plot of the SJ eigenvalues $\lambda$ divided by density $\rho$ vs $n$ for $m=0,0.2,0.4,0.6,0.8,1,2$ and $10$, (b): a plot of $\lambda/\rho$ vs $n$ for small $n$.}
\label{fig:dis-ev}}
\end{figure}
\vskip 0.5cm
\begin{figure}[h]
\centering{\begin{tabular}{cc}
\includegraphics[height=4cm]{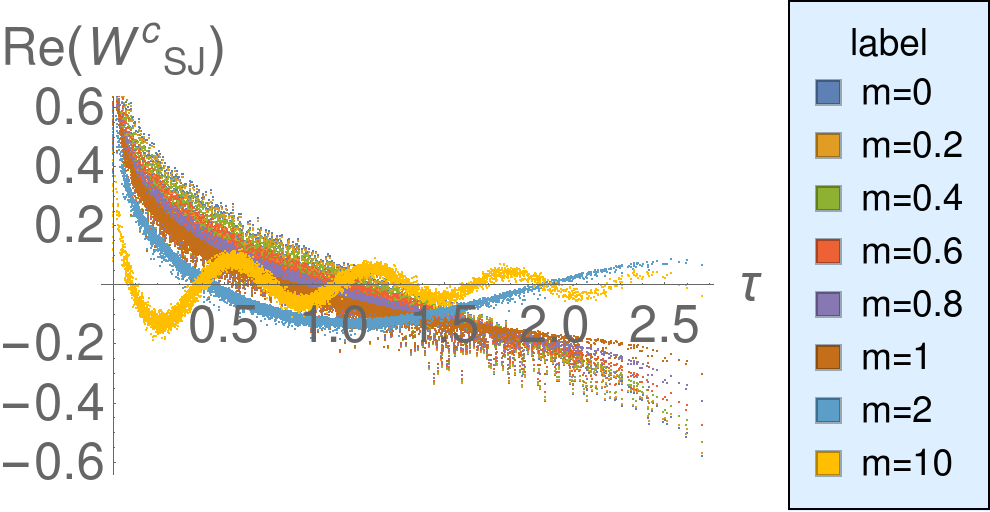}& \hskip 1cm
\includegraphics[height=4cm]{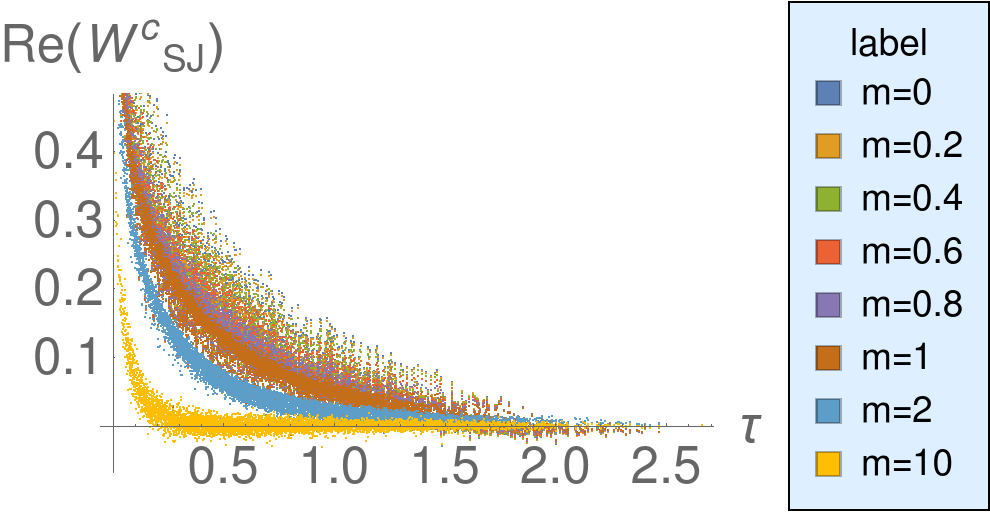}\\
(a)&(b)
\end{tabular}
\caption{$\wsjc$ for $m=0,0.2,0.4,0.6,0.8,1,2$ and $10$ vs (a) proper time for timelike separated points and (b) proper distance for spacelike separated points.}
\label{fig:scatterwsjc}}
\end{figure}

Next, we focus our attention to the center of the diamond so that we can compare with our analytic results.  We consider
a central region $\diam_l$ with $l =0.1$. Figs ~\ref{fig:smallmass-sj-tau} and ~\ref{fig:largemass-sj-tau} shows $\wsjc$ vs proper time and proper distance for  timelike
and spacelike separated pairs, respectively for small and large masses.  The comparisons with the massless and massive
Minkowski vacuum show a curious behavior. For the small $m$ values $\wsjc$  agrees perfectly with our analytic results
above, namely that $\wsj$ is more like $\wmink$ than $\wminkm$. However, as $m$ increases,  $\wminkm$ approaches
$\wmink$,  coinciding with it at $m=2\cof$. After this value of $m$,
$\wsjc$ then tracks $\wminkm$ rather than $\wmink$. This transition is continuous, and suggests that the small $m$ behavior
of $\wsjc$ goes continuously over to $\wmink$, unlike $\wminkm$. 

Next we compare $\wsjc $ in the  corner of the diamond with  $\wmirrm$
and $\wrindm$ for all pair of spacetime points in the left corner of the diamond for a range of masses. Instead of
plotting the actual functions, we consider the correlation plot as was done in \cite{Afshordi:2012ez}. 
To generate these plots we considered a small causal diamond in the corner of length $l=0.2$ which contained  118
elements. $\wmirrm$ and $\wrindm$ were calculated for each pair of elements and compared with $\wsjc$ (see Figs.~\ref{fig:mirrsj} and \ref{fig:rindsj}). 
We observe that for small masses there is much better
correlation between $\wsj$ and $\wmirrm$ as compared to $\wrind$ which is in
agreement with our analytic calculations. Figs.~\ref{fig:mirrsj} and \ref{fig:rindsj} show that while the correlation of $\wsjc$ and $\wmirrm$ remains largely unchanged with mass, that with $\wrindm$ increases with mass.  This can be traced to the  increased correlation between $\wrindm$ and $\wmirrm$ with mass as shown in Fig.~\ref{fig:rindmirr}.  This  in turn is related to the dominance of $\wminkm$  in the expressions for $\wrindm$ and $\wmirrm$ (Eqn.~\eqref{eq:rindm} and \eqref{eq:mirrorm})  for large mass, as shown in  Fig.~\ref{fig:mrm}. 


\begin{figure} 
\centerline{\begin{tabular}{cc}
\includegraphics[height=4cm]{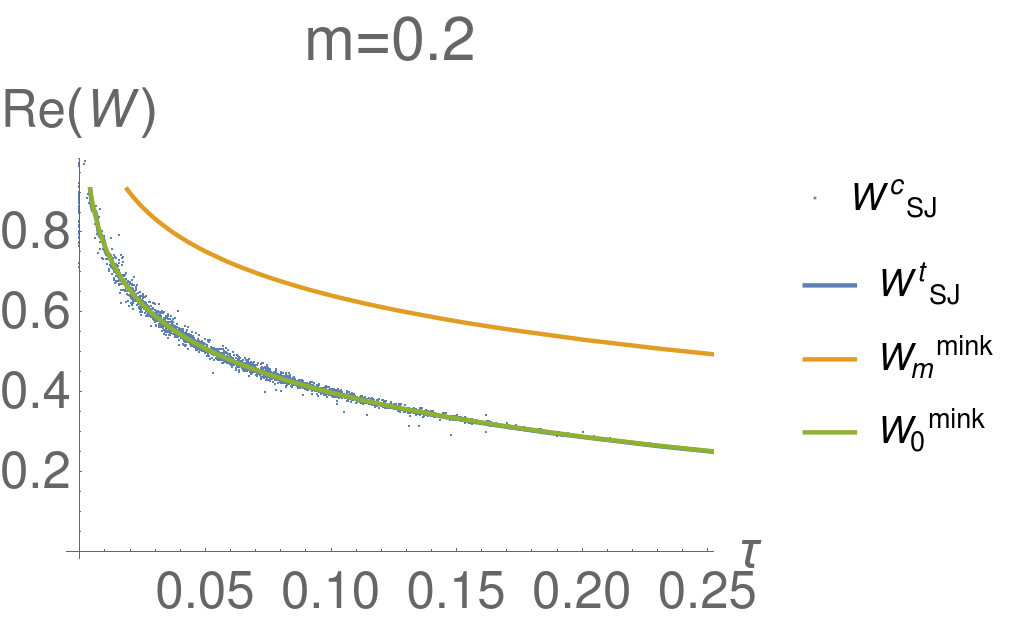}&
\includegraphics[height=4cm]{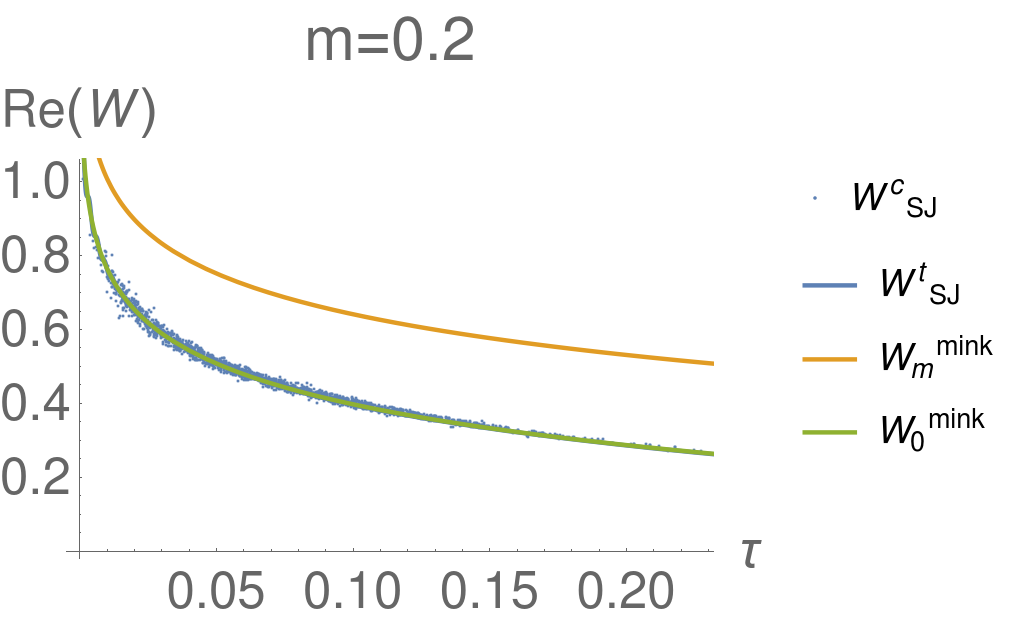}\\
\includegraphics[height=4cm]{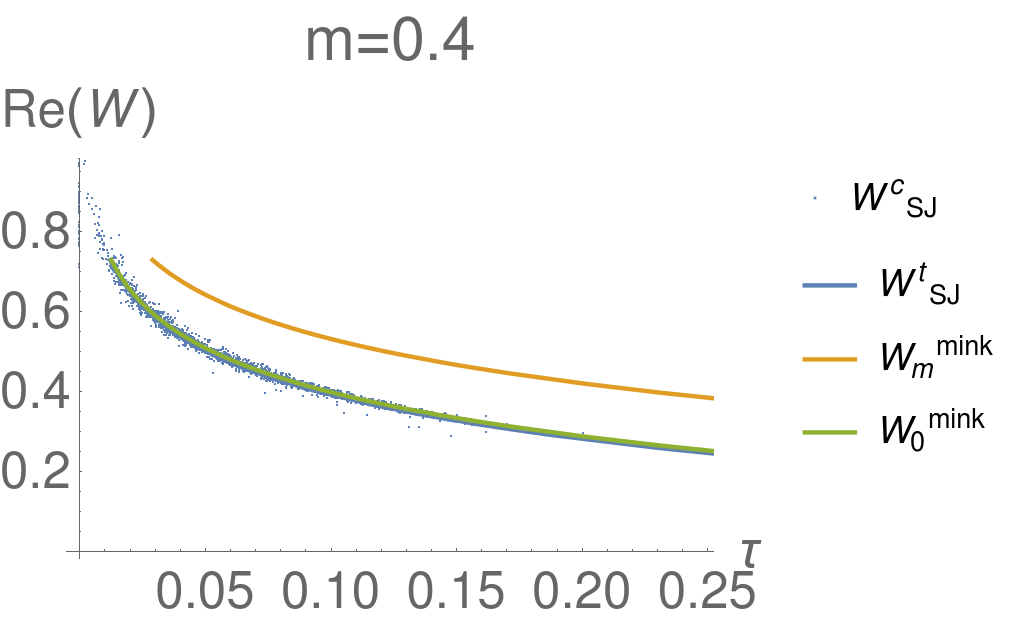}&
 \includegraphics[height=4cm]{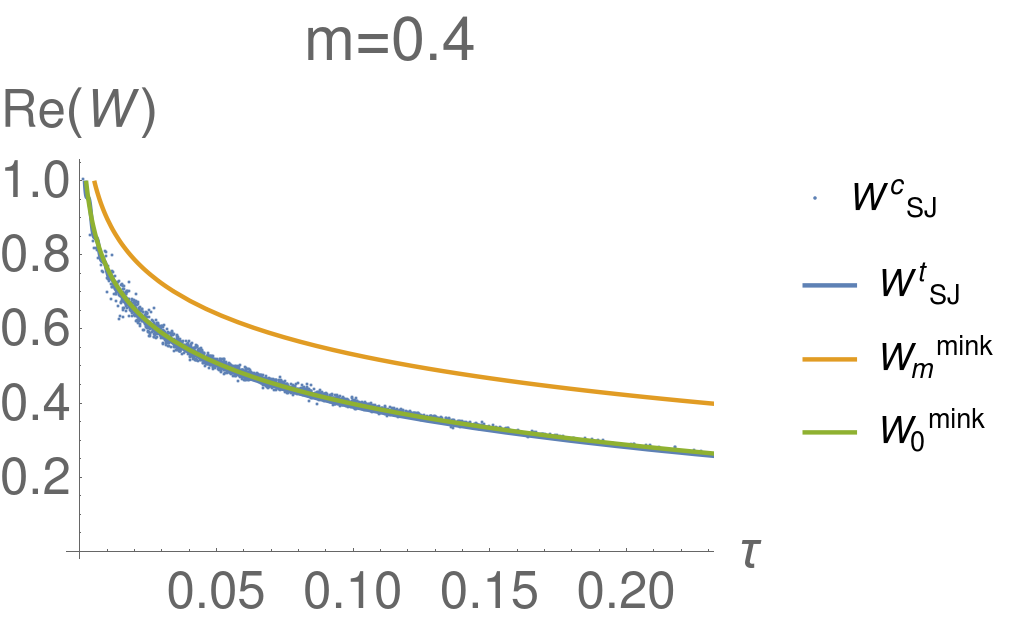}\\
            \end{tabular}}
\caption{$\wsjc$ (blue dots) vs proper time ($\tau$) in the center of the diamond. The plots on the left are for
  timelike separated points and those on the right are for spacelike separated points, for the small mass regime, $m=0.2$ and $ 0.4$. We show
  $\wmink$ (green), $\wminkm$ (orange) and our previous analytic calculation of $\wsj$ (blue
  line). The scatter plot clearly follows the massless green curve for these masses.}
        \label{fig:smallmass-sj-tau}\end{figure}
        \begin{figure} 
\centerline{\begin{tabular}{cc}
\includegraphics[height=4cm]{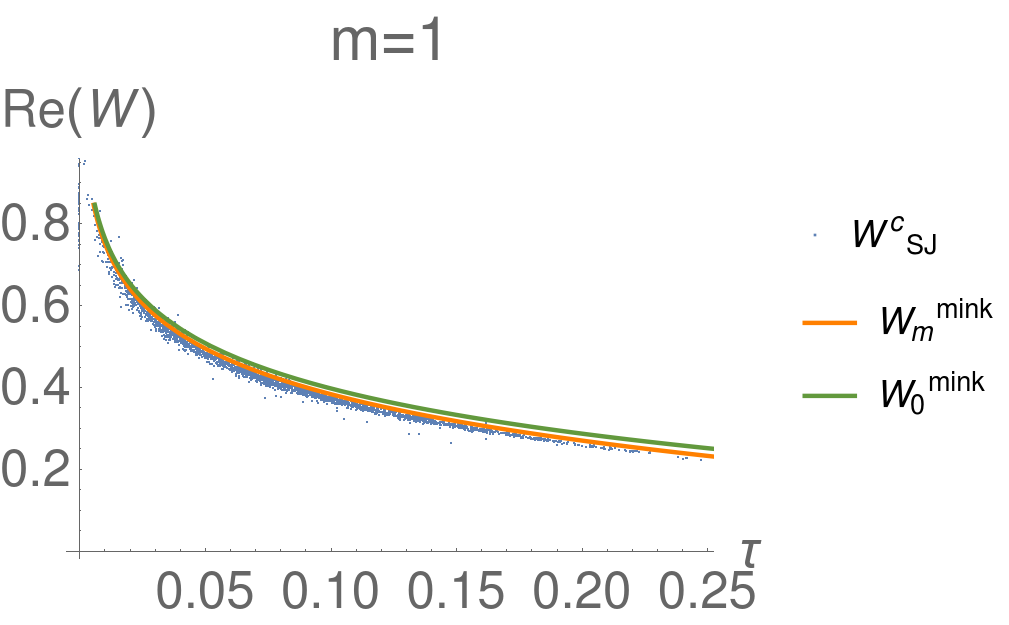}&
\includegraphics[height=4cm]{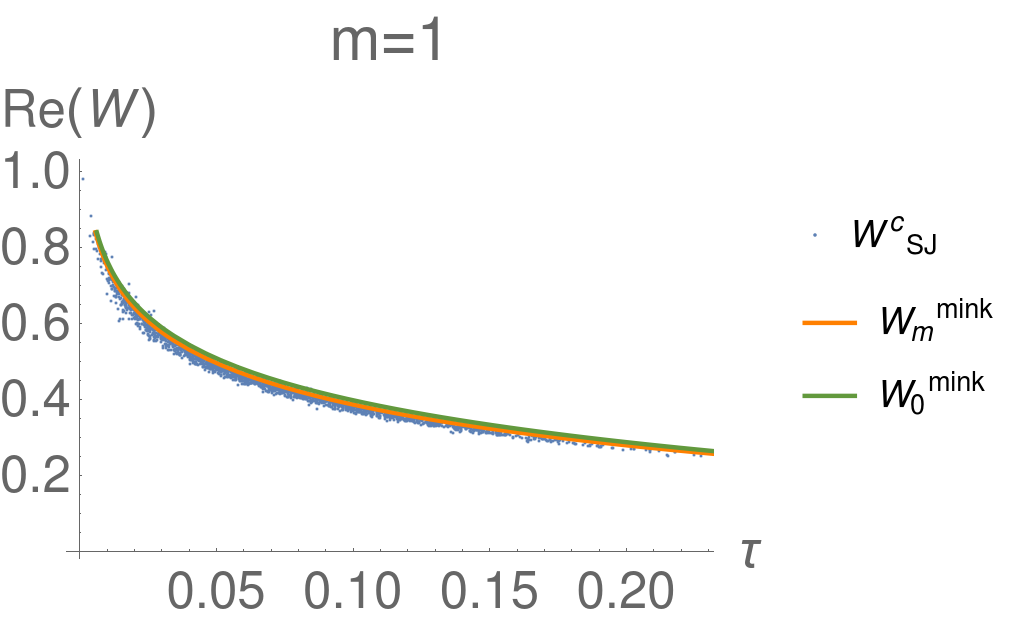}\\
\includegraphics[height=4cm]{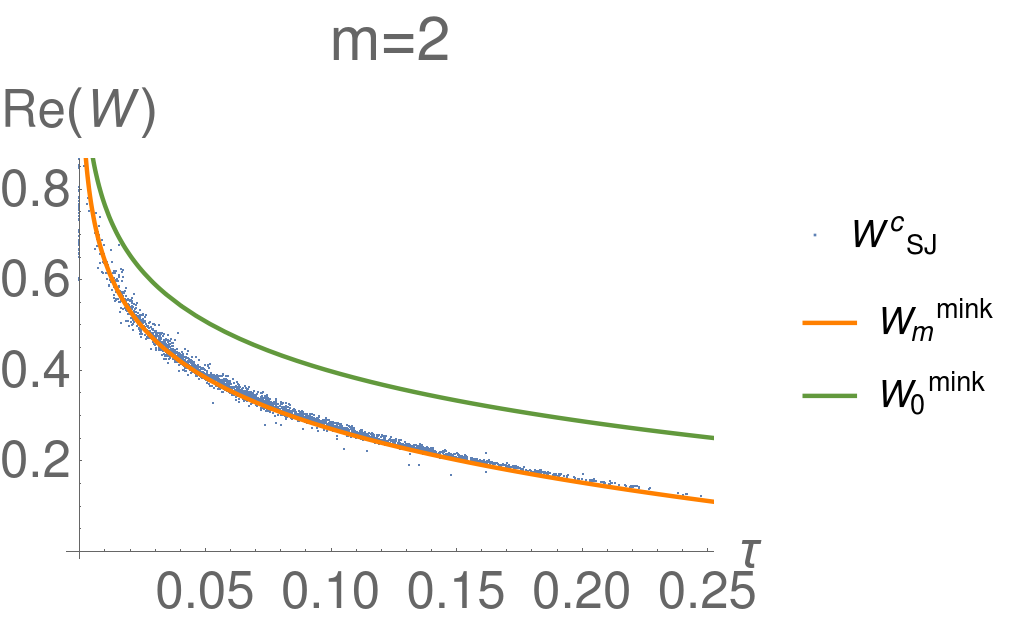}&
\includegraphics[height=4cm]{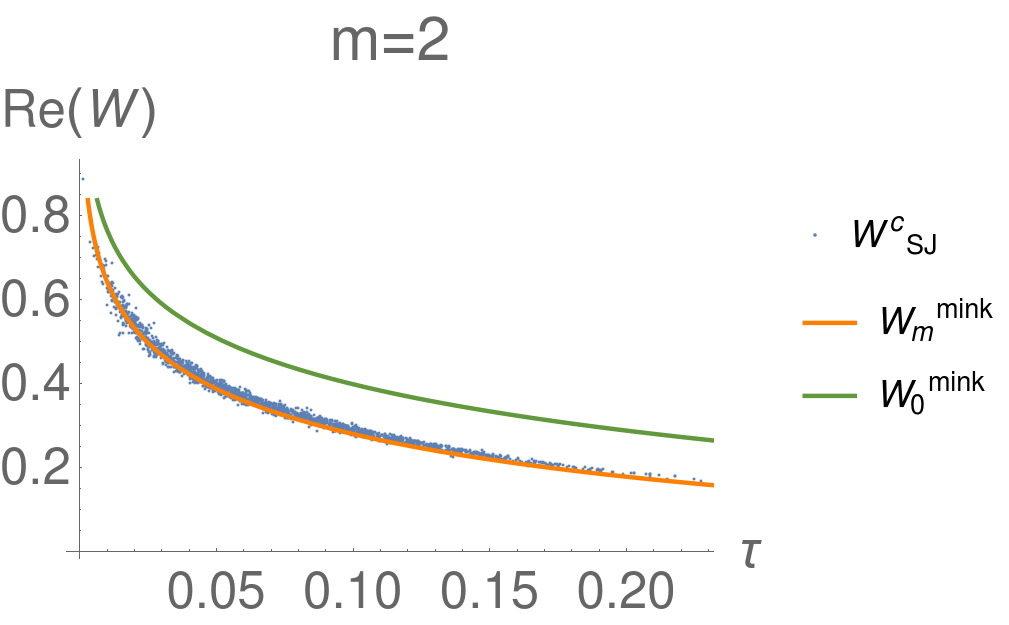}\\
\end{tabular}}
\caption{ The same plots as in Fig.~\ref{fig:smallmass-sj-tau} but for  $m=1$ and $m=2$. The scatter plot follows 
  the massive orange curve for $m\geq m_c$.}
\label{fig:largemass-sj-tau}
\end{figure}
\vskip 0.1in
\begin{figure}[h]
 \centering
   \subfloat[$m=0$]{\includegraphics[height=2.3cm]{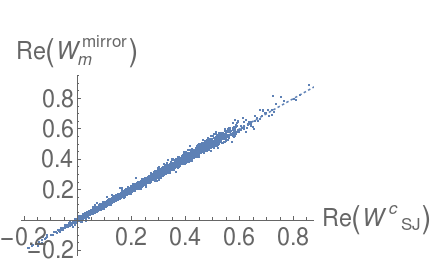} }\, 
  \subfloat[$m=0.1$]{\includegraphics[height=1.9cm]{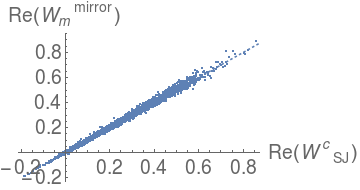} }\, 
  \subfloat[$m=0.2$]{\includegraphics[height=1.9cm]{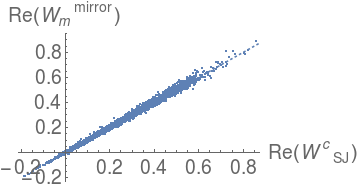}} \, 
  \subfloat[$m=0.3$]{\includegraphics[height=1.9cm]{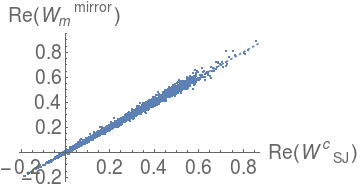}} \,
  \subfloat[$m=0.4$]{\includegraphics[height=1.9cm]{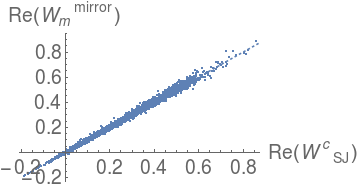} }\, 
  \subfloat[$m=1$]{\includegraphics[height=1.9cm]{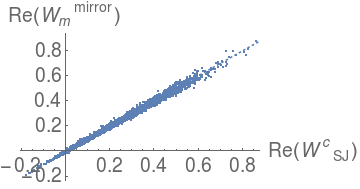}} \, 
  \subfloat[$m=2$]{\includegraphics[height=1.9cm]{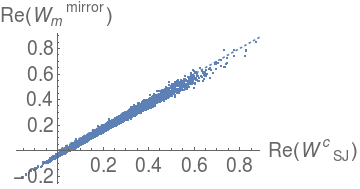}} \,
 \subfloat[$m=5$]{\includegraphics[height=1.9cm]{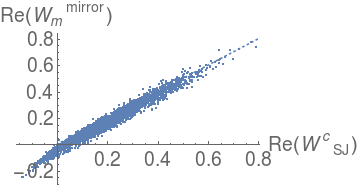}} \,
 \subfloat[$m=8$]{\includegraphics[height=1.9cm]{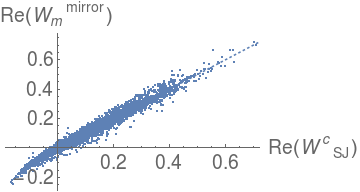}} \,
  \subfloat[$m=10$]{\includegraphics[height=1.9cm]{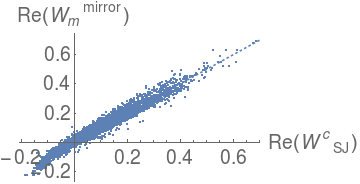}}\,
 \subfloat[$m=12$]{\includegraphics[height=1.9cm]{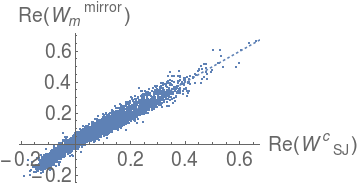}} \,
 \subfloat[$m=15$]{\includegraphics[height=1.9cm]{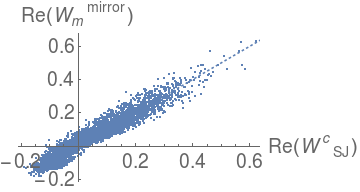}} 
  \caption{A correlation plot of the real parts of $\wsjc$ vs $\wmirrm$ in the left-hand corner of the 2d causal diamond for a range of masses. The
    diagonal is denoted by a dotted line. As is evident, the correlation remains largely unchanged with mass. The increase in scatter with mass is related to the fact that the density of sprinkling is left
 unchanged. }
   \label{fig:mirrsj} 
 \end{figure}
\begin{figure}[h]
 \centering
   \subfloat[$m=0$]{\includegraphics[height=2.3cm]{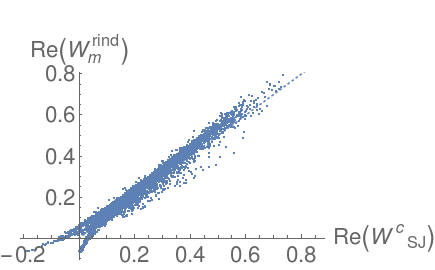} }\, 
  \subfloat[$m=0.1$]{\includegraphics[height=1.9cm]{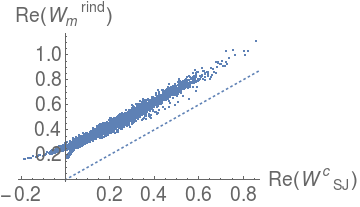} }\, 
  \subfloat[$m=0.2$]{\includegraphics[height=1.9cm]{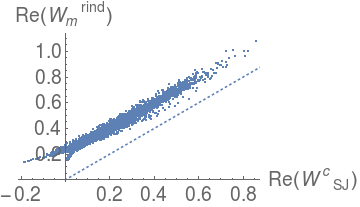}} \, 
  \subfloat[$m=0.3$]{\includegraphics[height=1.9cm]{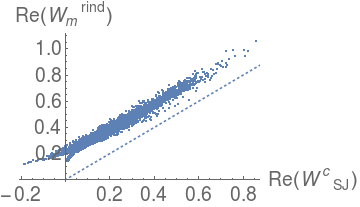}} \,
  \subfloat[$m=0.4$]{\includegraphics[height=1.9cm]{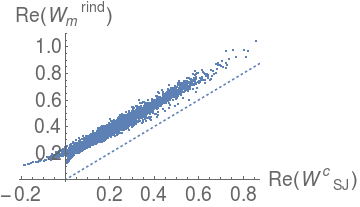} }\, 
  \subfloat[$m=1$]{\includegraphics[height=1.9cm]{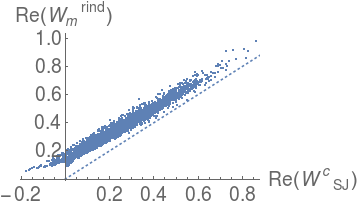}} \, 
  \subfloat[$m=2$]{\includegraphics[height=1.9cm]{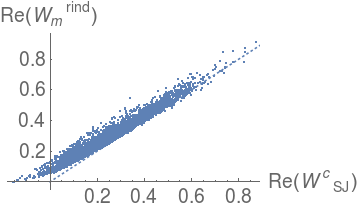}} \,
 \subfloat[$m=5$]{\includegraphics[height=1.9cm]{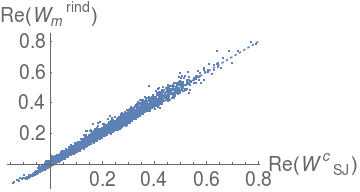}} \,
 \subfloat[$m=8$]{\includegraphics[height=1.9cm]{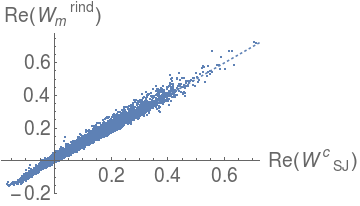}} \,
  \subfloat[$m=10$]{\includegraphics[height=1.9cm]{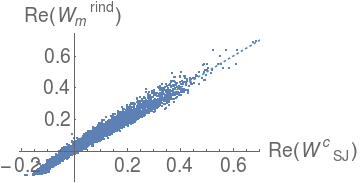}}\,
 \subfloat[$m=12$]{\includegraphics[height=1.9cm]{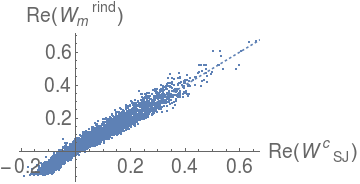}} \,
 \subfloat[$m=15$]{\includegraphics[height=1.9cm]{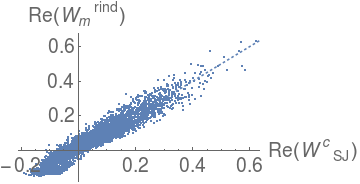}} 
  \caption{A correlation plot of the real parts of $\wsjc$ vs $\wrindm$ for the same
    range of masses. For small masses, the correlation is poor but improves
    with mass.}
   \label{fig:rindsj} 
 \end{figure}    
\begin{figure}[h]
    \centering
    \subfloat[$m=0.1$]{\includegraphics[height=1.9cm]{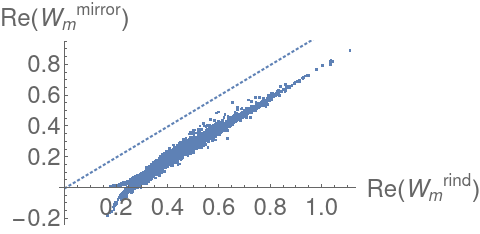}}\,
    \subfloat[$m=0.2$]{\includegraphics[height=1.9cm]{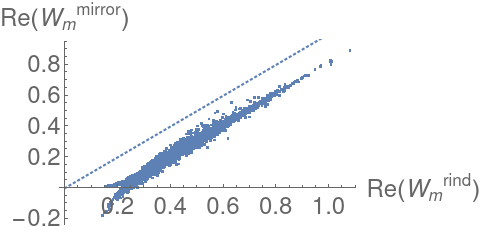}}\,
    \subfloat[$m=0.3$]{\includegraphics[height=1.9cm]{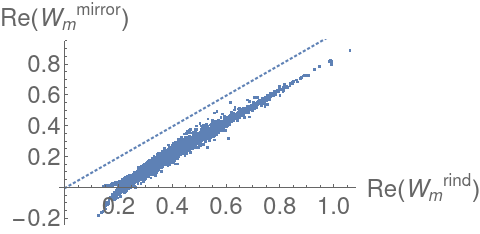}}\,    
    \subfloat[$m=0.4$]{\includegraphics[height=1.9cm]{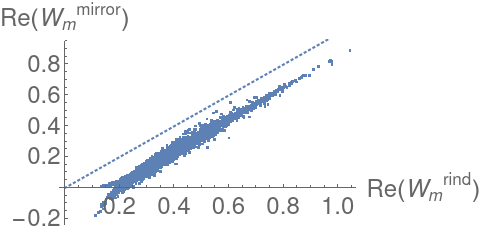}}\,
    \subfloat[$m=1$]{\includegraphics[height=1.9cm]{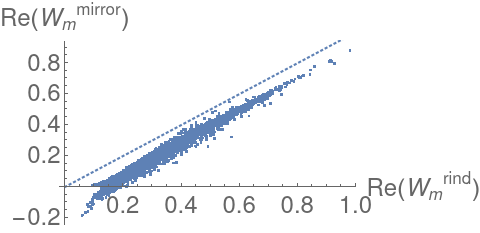}}\,
    \subfloat[$m=2$]{\includegraphics[height=1.9cm]{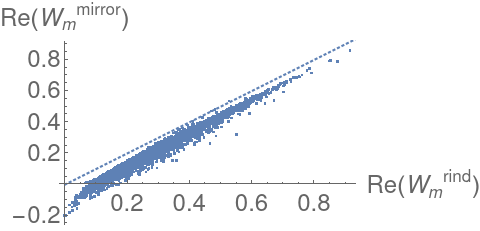}}\,
    \subfloat[$m=5$]{\includegraphics[height=1.9cm]{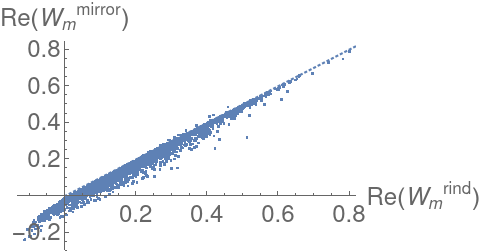}}\,
 \subfloat[$m=8$]{\includegraphics[height=1.9cm]{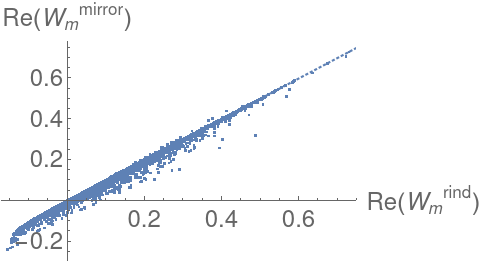}}\,
 \subfloat[$m=10$]{\includegraphics[height=1.9cm]{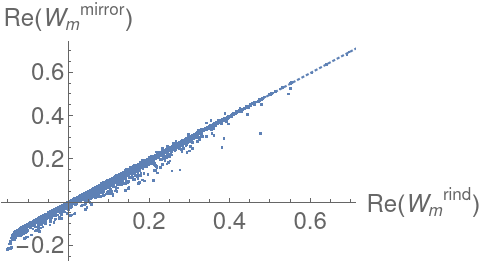}}\,
 \subfloat[$m=12$]{\includegraphics[height=1.9cm]{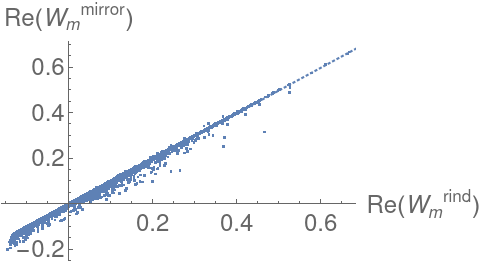}}\,
 \subfloat[$m=15$]{\includegraphics[height=1.9cm]{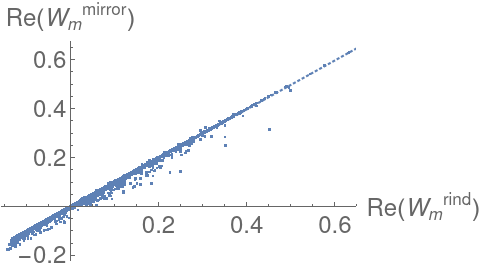}}\,
    \caption{A correlation plot of the real parts of $\wmirrm$ vs $\wrindm$ for the same
    range of masses. For small masses, the correlation is poor but improves
    with mass.}
    \label{fig:rindmirr}
  \end{figure}
 
  \begin{figure}[h]
    \centering
    \subfloat[$m=0.1$]{\includegraphics[height=2.4cm]{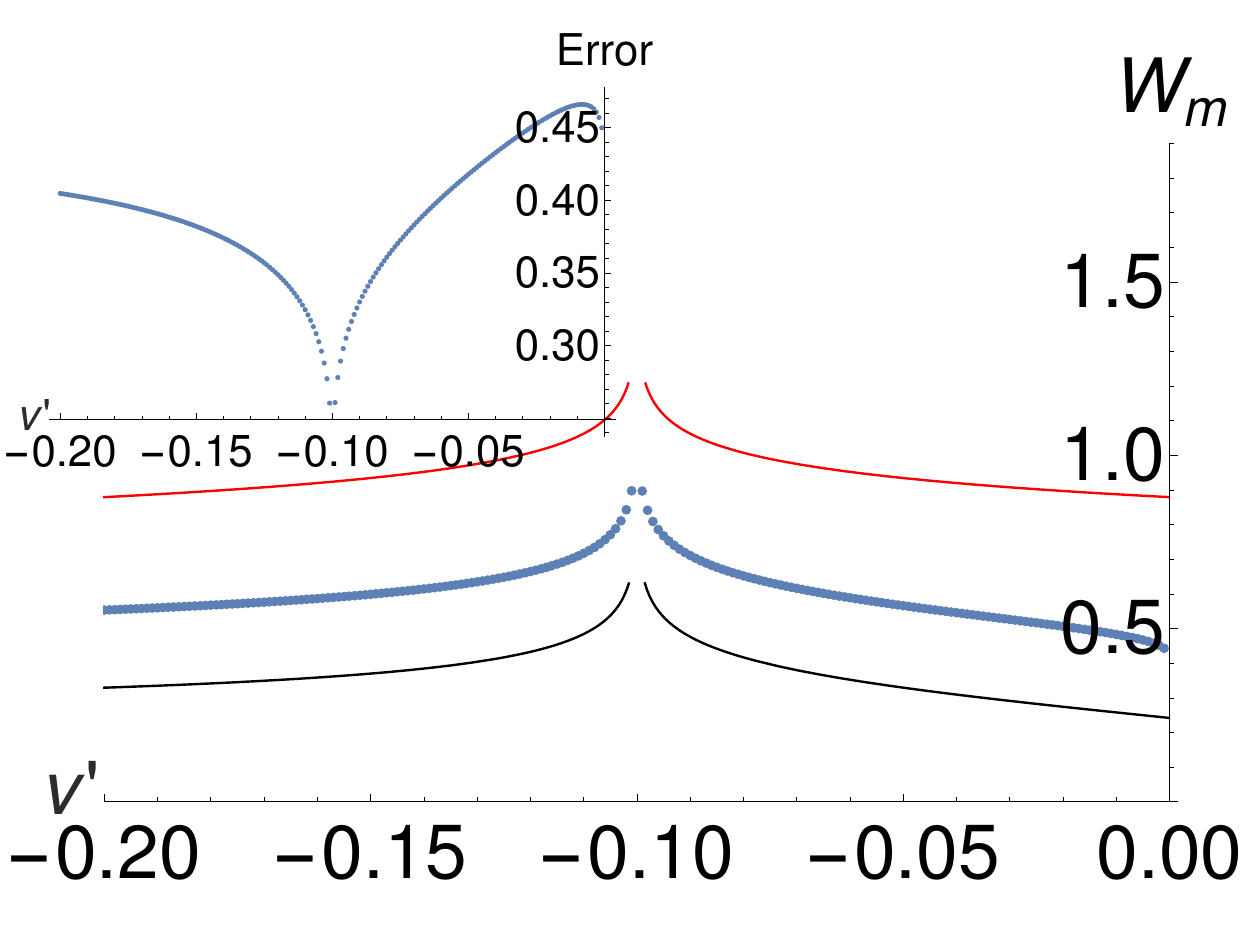}}\hskip 0.1in\,
    \subfloat[$m=0.2$]{\includegraphics[height=2.4cm]{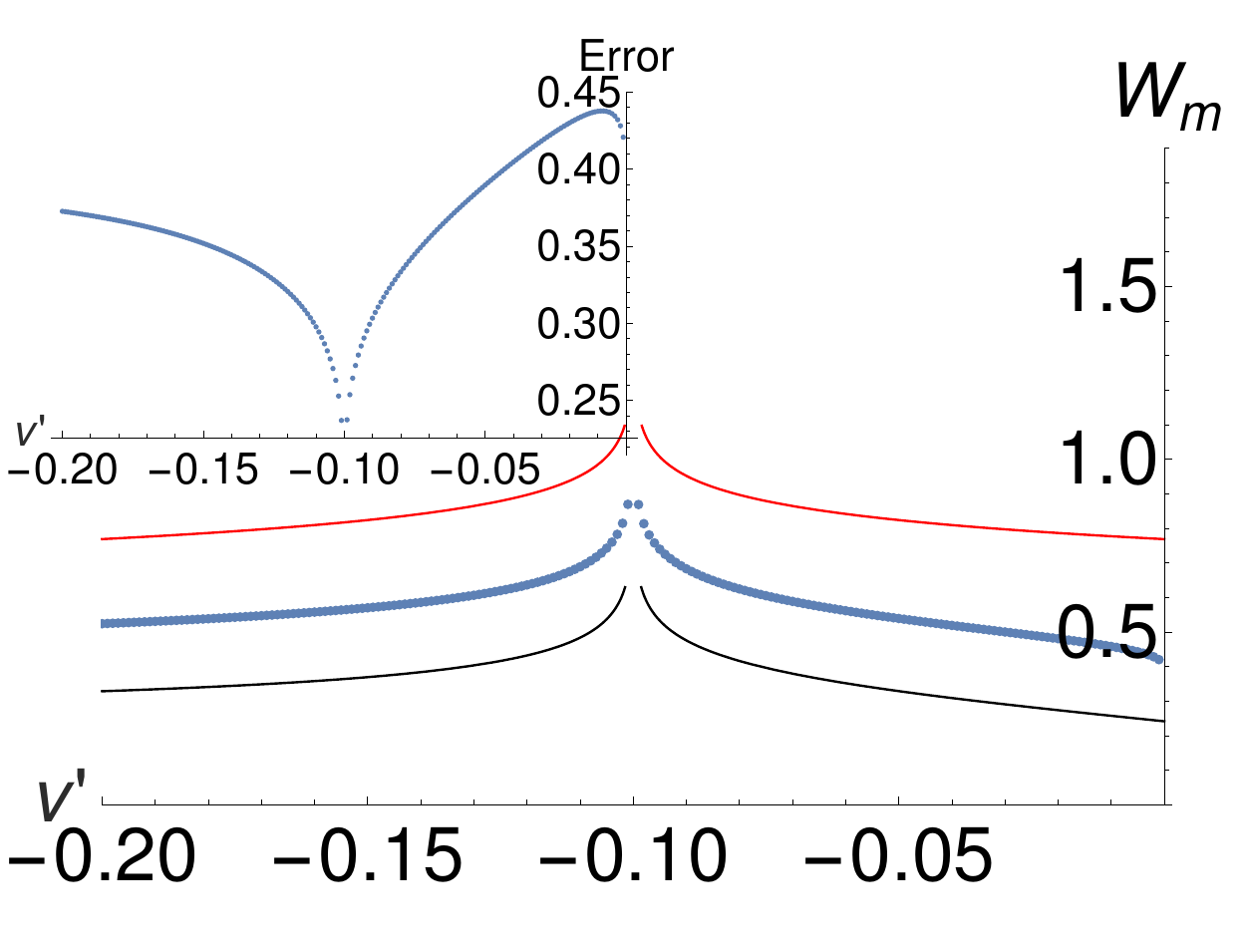}}\hskip 0.1in\,
    \subfloat[$m=0.3$]{\includegraphics[height=2.4cm]{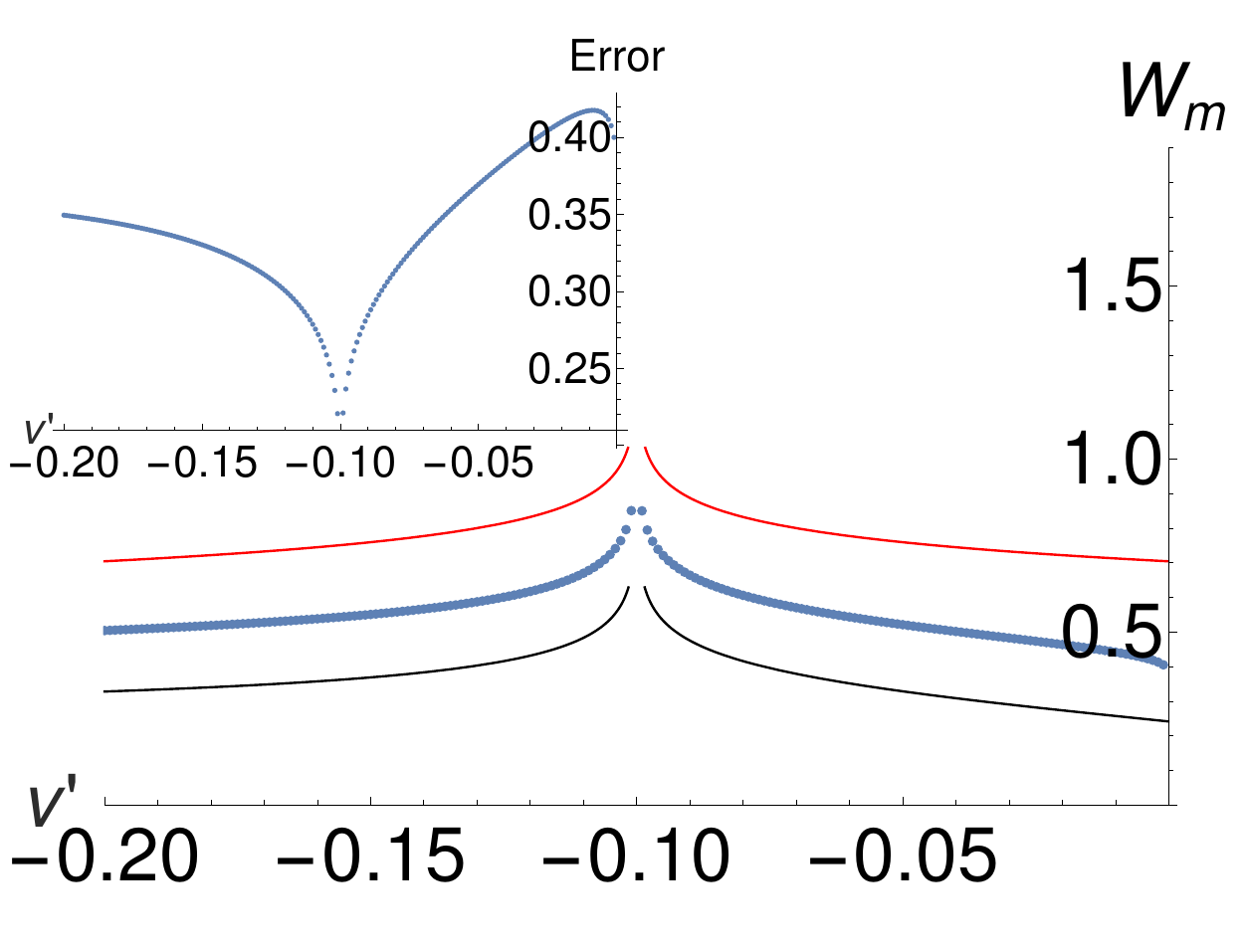}}\hskip 0.1in\,    
    \subfloat[$m=0.4$]{\includegraphics[height=2.4cm]{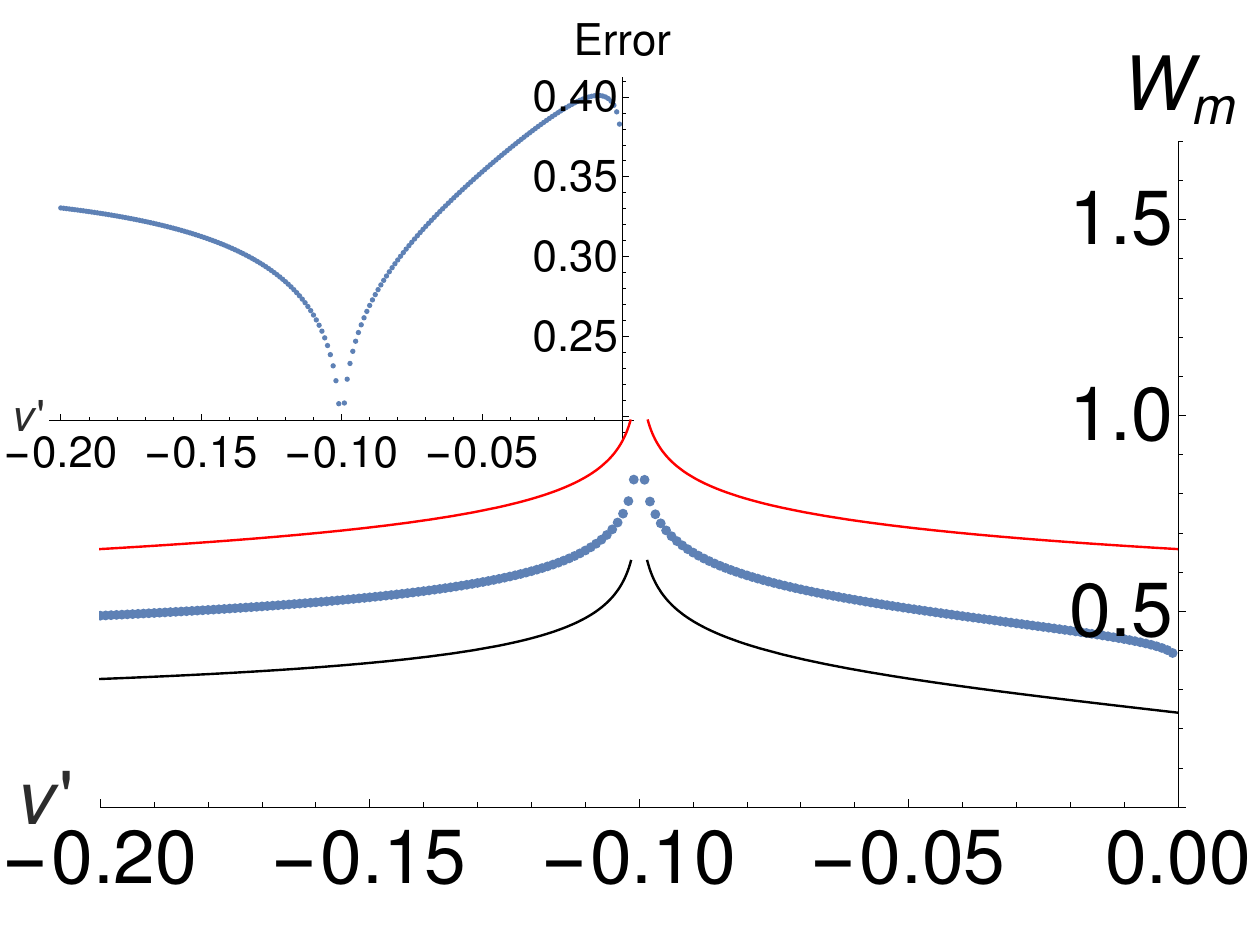}}\hskip 0.1in\,
 \subfloat[$m=1$]{\includegraphics[height=2.4cm]{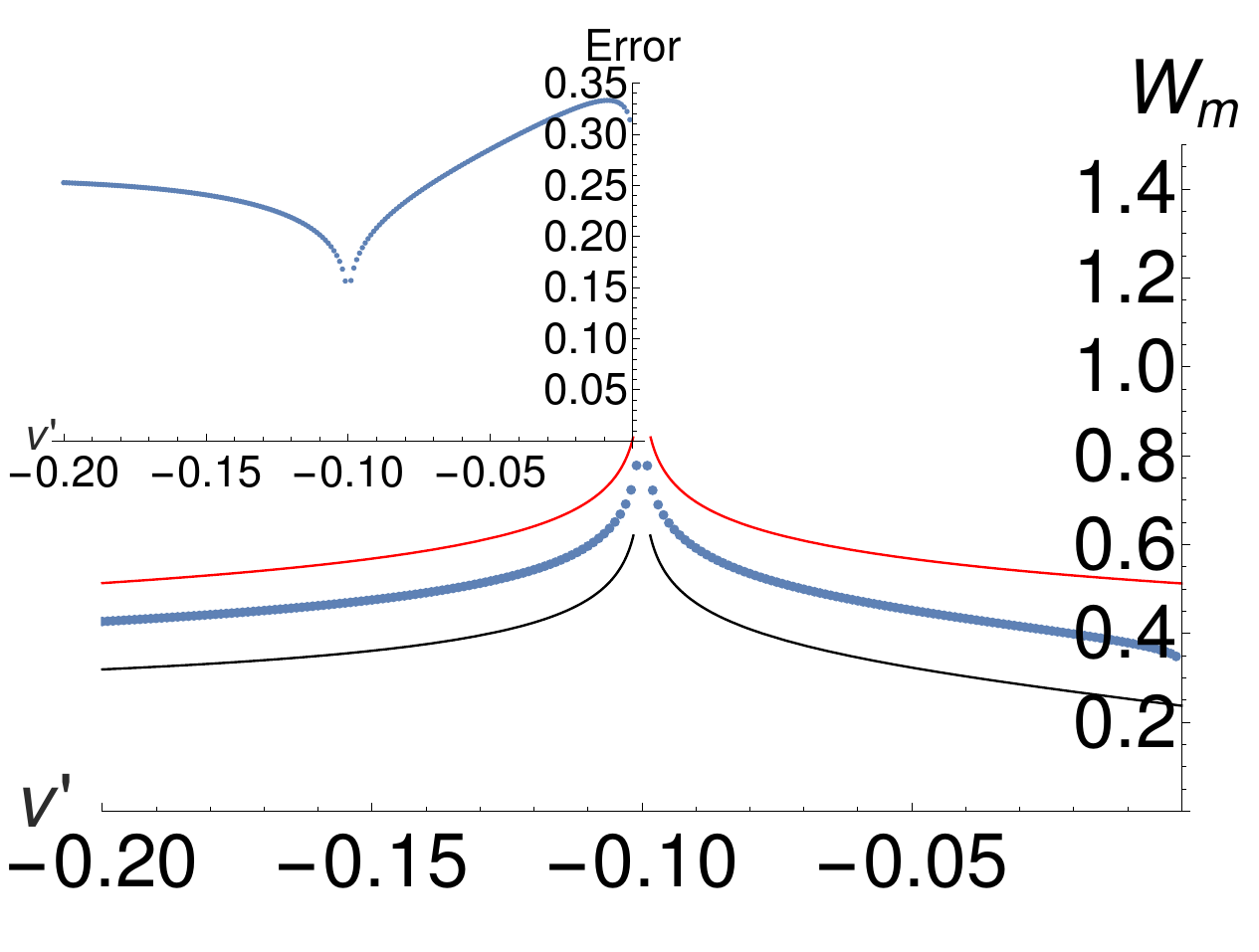}}\hskip 0.1in\,
    \subfloat[$m=2$]{\includegraphics[height=2.4cm]{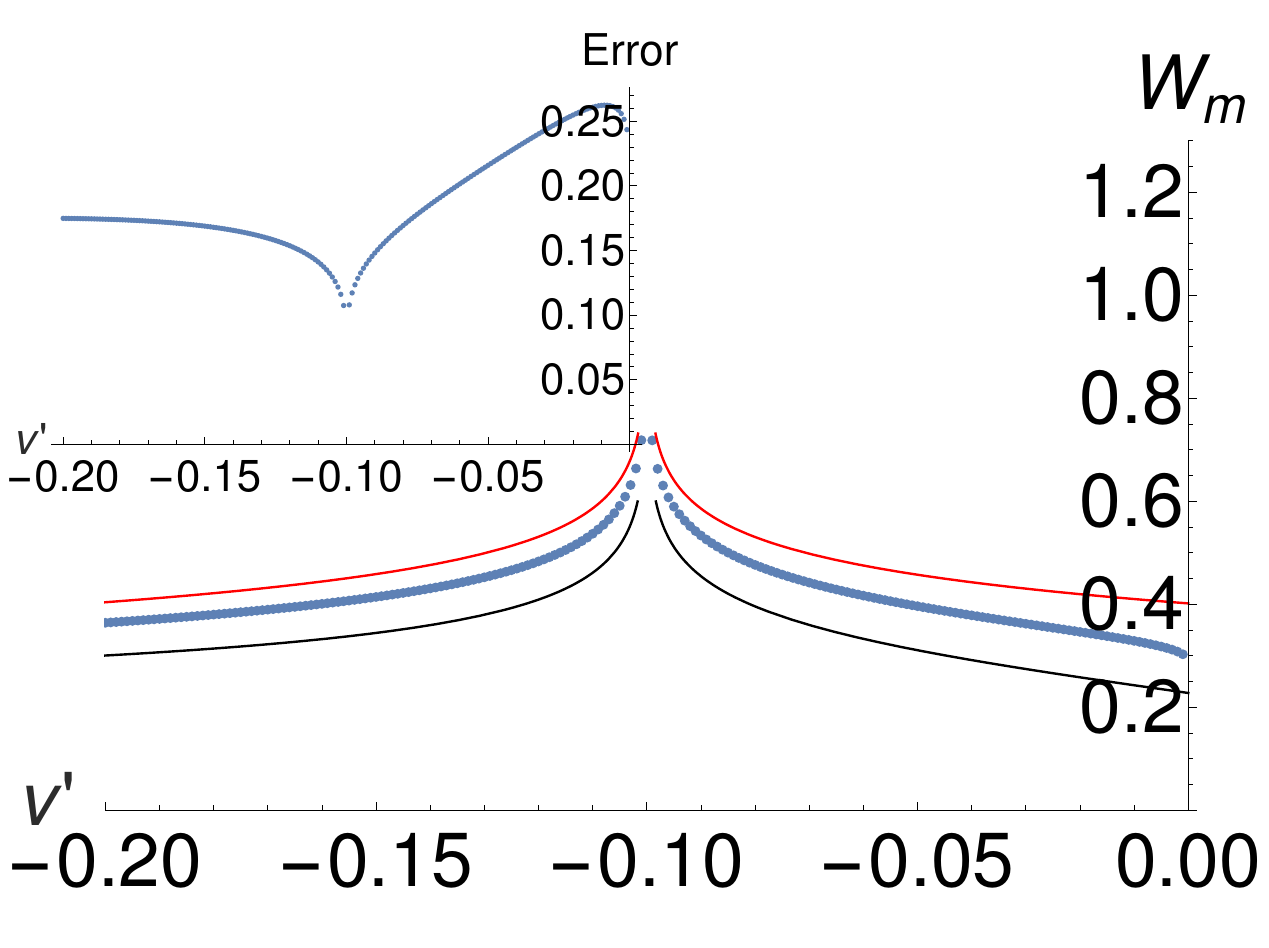}}\hskip 0.1in\,
    \subfloat[$m=5$]{\includegraphics[height=2.4cm]{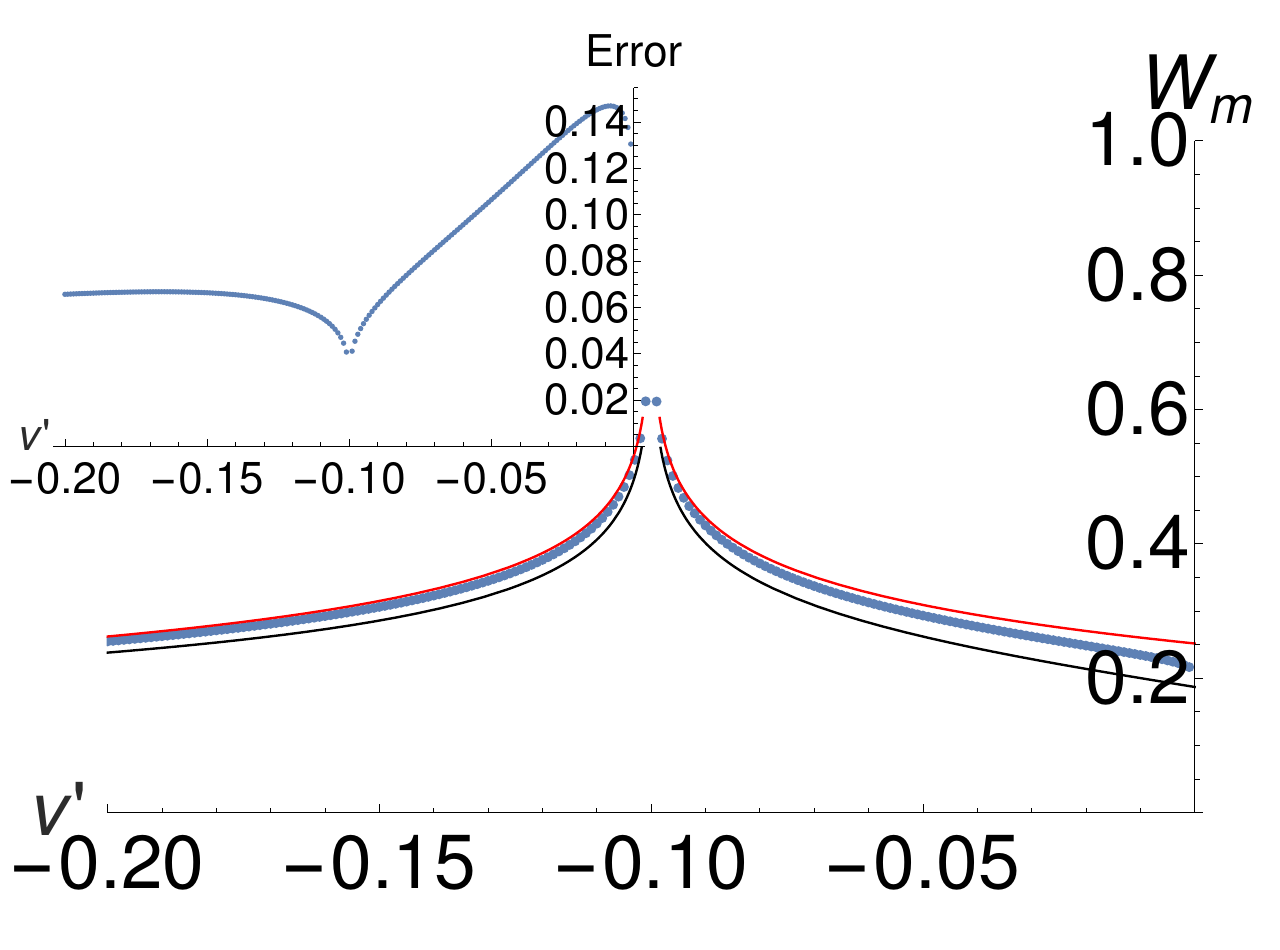}}\hskip 0.1in\,    
    \subfloat[$m=8$]{\includegraphics[height=2.4cm]{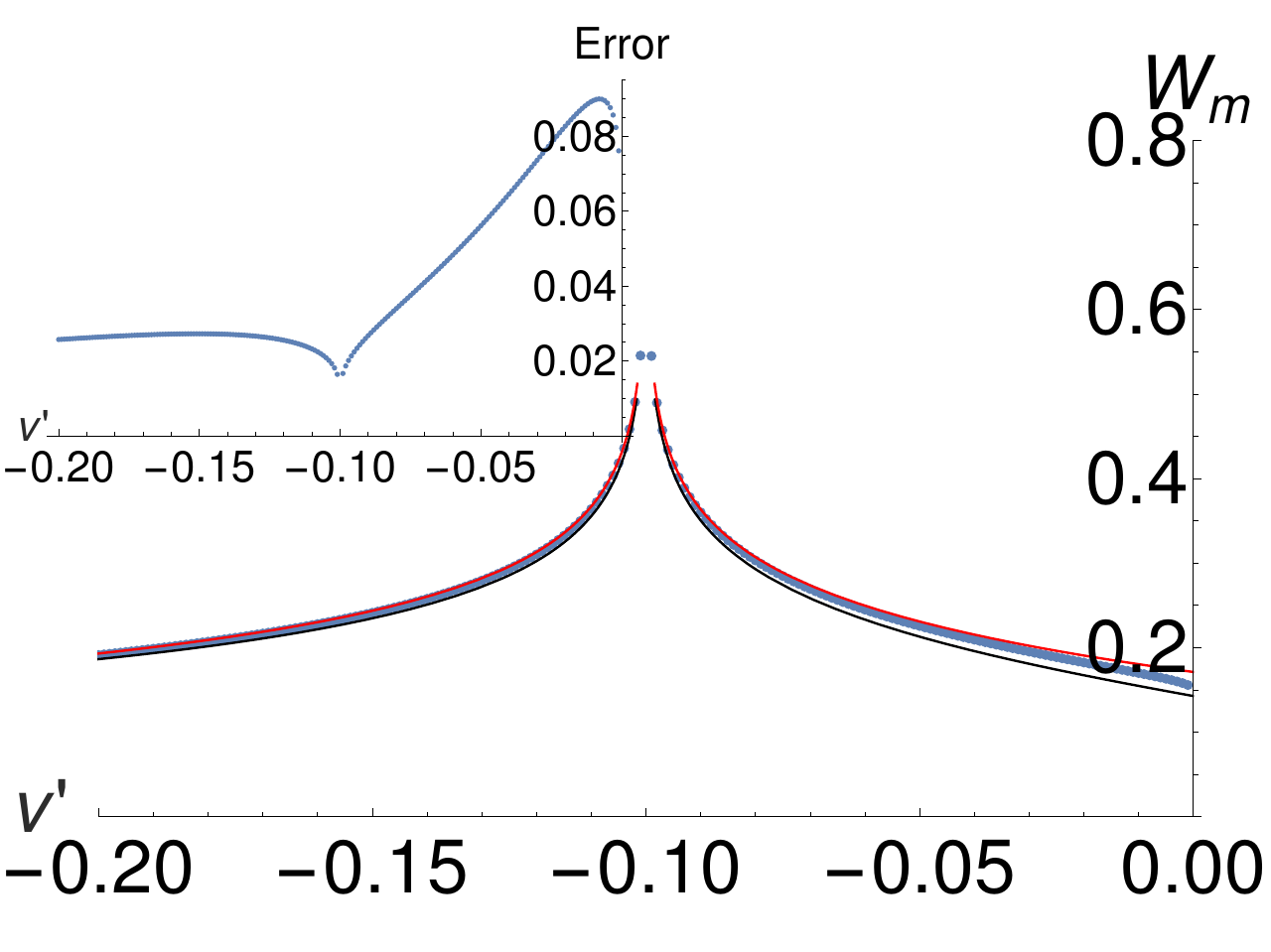}}\hskip 0.1in\,
 \subfloat[$m=10$]{\includegraphics[height=2.4cm]{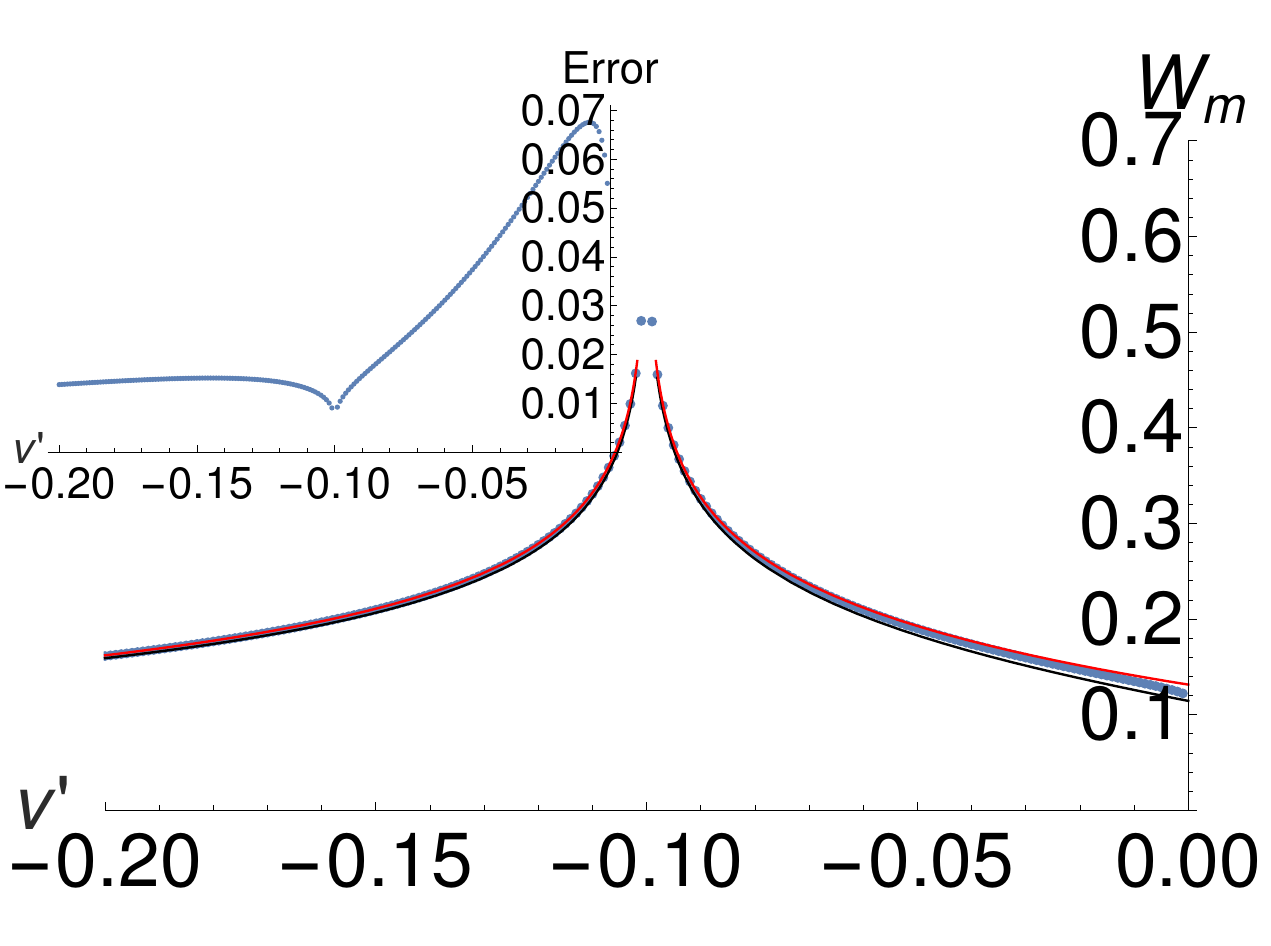}}\hskip 0.1in\,
    \subfloat[$m=12$]{\includegraphics[height=2.4cm]{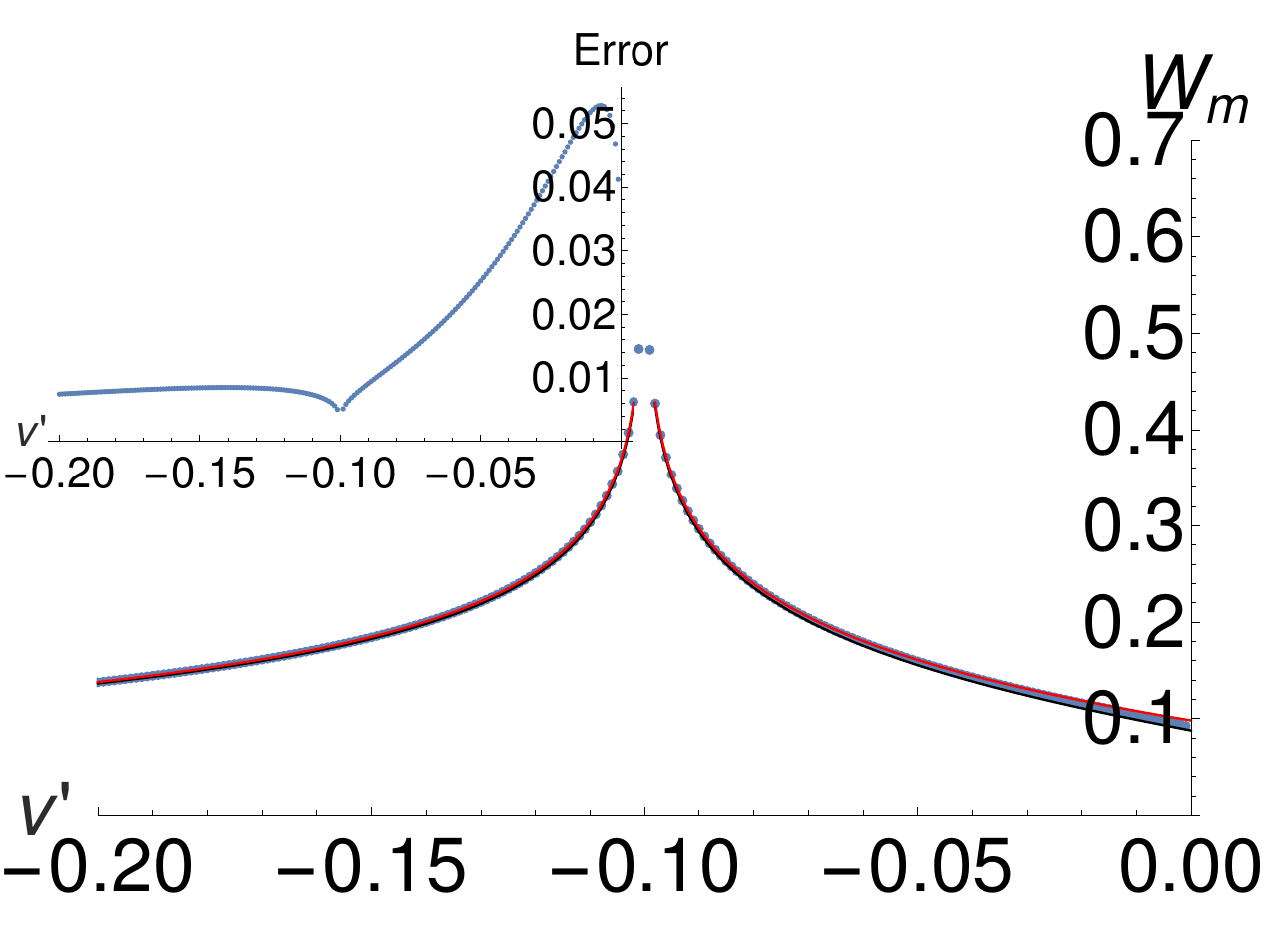}}\hskip 0.1in\,
    \subfloat[$m=15$]{\includegraphics[height=2.4cm]{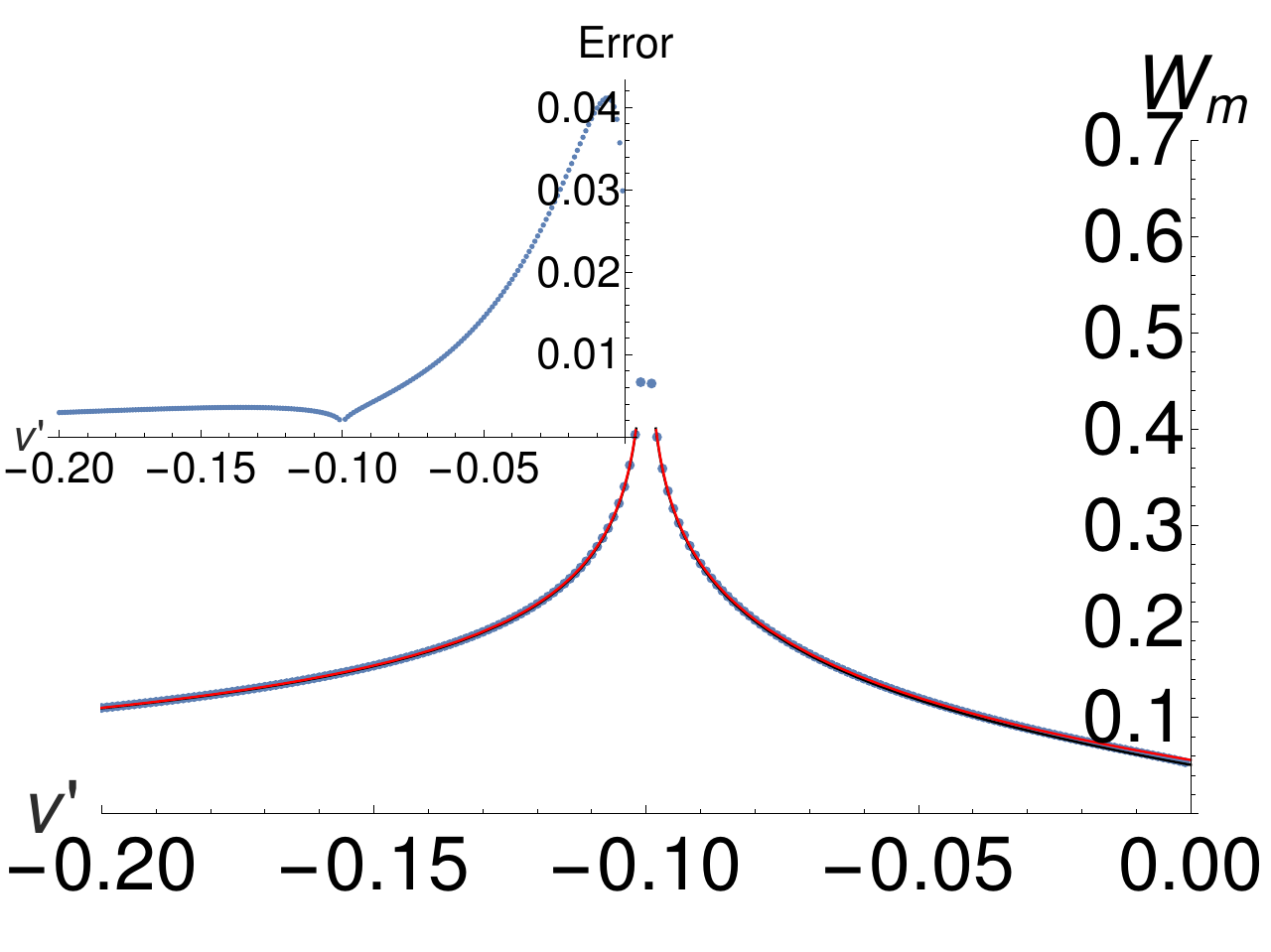}}
     \caption{Real parts of $\wminkm$ (red), $\wrindm$ (blue) and $\wmirrm$ (black) for a pair of points
      $(u=0.1,v=-0.1)$ and $(u'=0.11,v')$ with  varying $v'$. 
      As the mass increases all three converge to a common value.  To make the comparison explicit, the inset figure shows the relative error  between
    the real parts of   $\wrindm$ and $\wmirrm$ as a function of $v'$. } 
\label{fig:mrm}
\end{figure}

\section{Modifying the inner product to get the 2d Rindler Vacuum}
\label{sec:rindler}
In this section we obtain the massless Rindler Wightman function in the right Rindler Wedge as a particular limit of the
massless SJ Wightman function in 2d causal diamond. We achieve this by deviating from the standard $\mL^2$ inner product
on the function space $\cF(\cM,g)$, by  introducing a suitable non-trivial weight function $w(\bx)$, 
\be
(f,g)_w=\int_M f^*(\bx)g(\bx) \,  w(\bx) dV_\bx \label{eq:wip}
\ee
where $dV_X$ is the spacetime volume element. $w(\bx)$ takes real, positive and finite value for all $\bx\in\cM$. The inner product
defined in Eqn.~(\ref{eq:wip}) is well defined in $(\cM,g)$ and  satisfies the defining properties of an  inner product: 
\begin{itemize}
\item $(f,g)_w$ is linear in $g$.
\item $(f,g)_w$ is anti-linear in $f$.
\item $(f,f)_w\geq 0$. Equality holds iff $f=0$.
\end{itemize}
Similarly,  we redefine the integral operator  $i\hD$  to make it consistent with this inner product 
\be
\left(i\hD\circ_w f\right)(\bx)=\int_M i\Delta(\bx,\bx')f(\bx') \,  w(\bx') dV_{\bx'}. \label{eq:ihdw}
\ee
It is straightforward to check that even with this modification, $i\hD$ is Hermitian: 
\be
\left(f,i\hD\circ_w g\right)_w = \left(i\hD\circ_w f,g\right)_w. 
\ee
Next, we see that: 
\begin{claim}
$\kker(\Bkg)=\im_w(i\hD)$ for $w(\bx)$ real, positive and finite valued  in  $\bx$.
\end{claim}
\begin{proof}
For any  $\chi\in\im_w(i\hD)$,  there exists a $\psi\in\cF(\cM,g)$ such that $\chi=i\hD\circ_w\psi$. Since 
\be
i\hD\circ_w(\psi) = i\hD\circ (w\psi)
\ee
this implies that $\chi=i\hD\circ (w\psi) \in\im (i\hD)$, since  $w\psi\in\cF(\cM,g)$. Thus $\im_w (i\hD)\subseteq  \im (i\hD)$. Conversely, for
any $\chi'\in\im (i\hD)$, there exists a $\psi'\in\cF(\cM,g)$ such that $\chi'=i\hD\circ\psi'$. Since  $w$ is  real,
positive and finite valued  in  $\bx$,  $\psi/w \in \cF(\cM,g)$ and hence  $\chi'=i\hD\circ_w (\psi/w) \in \im_w
(i\hD)$. Hence $\im_w (i\hD) = \im (i\hD) = \kker(\Bkg)$.
\end{proof}

The 2d Minkowski metric in Rindler coordinates is 
\be
ds^2=e^{2a\xi}\left(-d\eta^2+d\xi^2\right)
\ee
where  
\be
t=a^{-1}e^{a\xi}\sinh(a\eta)\;,\quad x=a^{-1}e^{a\xi}\cosh(a\eta)
\ee
and $a>0$ is the acceleration parameter. Consider a causal diamond of length $2l$ centered at $(0,0)$ in 
$(\eta,\xi)$ coordinates. The center of the diamond  $(u,v)=(0,0)$ in the $u-v$
plane is at $(t,x)=(0,a^{-1})$,  and thus to the corner of the diamond in the $t-x$ plane as shown in Fig.~\ref{fig:rind}  
\begin{figure}[h]
\centering{\begin{tabular}{cc}
\includegraphics[height=4.5cm]{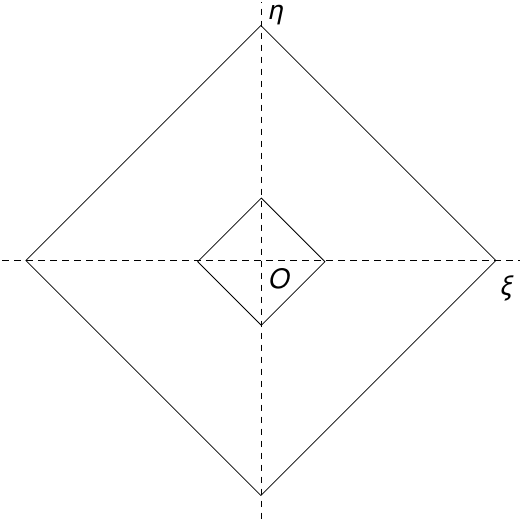}& \hskip 1.5cm
\includegraphics[height=4.5cm]{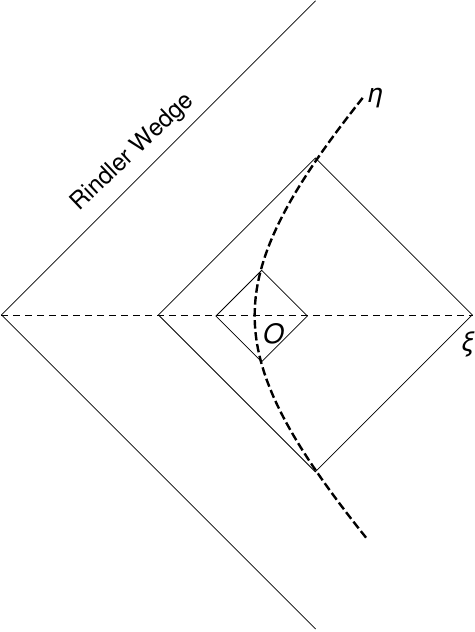}\\
(a)&\hskip 1.5cm (b)
\end{tabular}
\caption{A small causal diamond centered in a causal diamond $\diam$ in the $\eta-\xi$ plane is shifted to the corner of
  $\diam$ in the $t-x$ plane.}
\label{fig:rind}}
\end{figure}
The Pauli Jordan function is then similar to that in Minkowski coordinates
\be
i\Delta(u,v;u',v')=-\frac{i}{2}\left(\theta(u-u')+\theta(v-v')-1\right), 
\ee
where we have used the new light cone coordinates  $u=\frac{1}{\sqrt{2}}(\eta+\xi)$ and $ v=\frac{1}{\sqrt{2}}(\eta-\xi)$. 
The ``w-SJ''  modes  $u_k^w$ are then given by 
\bea
\int_{-L}^Li\Delta(u,v;u',v')u^w_k(u',v') w(u',v') e^{2a\xi'}du'dv'=\lambda_ku^w_k(u,v) \label{eq:rindev}
\eea
If  we now choose $w(u,v)=e^{-2a\xi}$, Eqn.~(\ref{eq:rindev}) is exactly the same as the eigenfunction equation for the
massless SJ modes in $\diam$ and hence  $\wsj$  is the same as the massless SJ function of \cite{Afshordi:2012ez}. Thus,
at the center of this diamond $\wsj$ takes the same form as Eqn.~(\ref{eq:minkm0}). The critical difference is that in
this case the $u$ and $v$ are lightcone coordinates for a  Rindler observer instead of an  inertial observer. Thus, in
$(t,x)$ coordinates, $\wsj$ {\it is} the Rindler vacuum (see Eqn.~(\ref{eq:rindm0})). The small diamond at the center
of $\diam$  the $\eta-\xi$ plane is a small diamond near (but not at)  the corner of  $\diam$ in the $t-x$
plane. Here, $\wsj$ then resembles the Rindler vacuum.   

Of course, the question is whether $\wsj$ will also look like $\wmink$ near the center of the diamond in the $t-x$ plane,
i.e. at $(t,x)=(0,a^{-1}\cosh(\sqrt{2}L a))$, which is $(0,a^{-1}\ln(\cosh(\sqrt{2}La)))$ in the $\eta-\xi$ plane. This is the
mirror vacuum, $\wmirr$ which rather than corresponding to  $\wmink$ is a  ``Rindler-mirror'' vacuum. This is clearly
not desirable. 

What we have presented here is a ``trick'' for achieving a desired form for the vacuum in the corner. However, this
messes up the expected form at the center. The question is
whether a smooth modification of $w$ from $1$ in the center of the $t-x$ plane diamond to $\exp(-a \xi)$ at the corners
could lead to  the desired form. However, modifications of the inner product mean that the SJ vacuum is no longer
unique.

\section{Discussion}
\label{sec:conclusions}

In this chapter, we calculated the massive scalar field SJ  modes up to  fourth order of mass. The
procedure we have developed for solving the central  eigenvalue problem can be used in principle  to find the SJ modes for higher
order mass corrections.  
 
Our work shows that $\wsjc$ in the causal set  is compatible with our analytic results in the small
mass regime. The curious behavior of  $\wsjc$ with mass in the center of the diamond suggests a hidden
subtlety in the finite region, ab-initio construction, that has hitherto been missed. In particular, it shows that the massive $\wsj$ in 2d has a well defined massless limit, unlike $\wminkm$. Such a continuous behavior with mass was also seen in the calculation of $\wsjc$ in de Sitter spacetime \cite{Surya:2018byh}. A possible source for this
behavior is that $\wsj$ is built from the advanced/retarded Green's functions, which  themselves have a well defined 
 massless limit. It is surprising however  that $\wsj$ for small mass lies in the massless representation of the Poincare
 algebra rather than the expected massive representation. What this means for the uniqueness of the SJ vacuum is unclear
 and we hope to explore this in future work. 
 
In the corner of the diamond, we see that as in the massless case,  $\wsj$ resembles the  massive mirror vacuum for small masses whereas for large masses the mirror and the Rindler vacuum themselves are indistinguishable.

We end with a broad comment on the SJ formalism. It is possible to construct a $\wsj$ using a different  inner product on
$\cF(\cM,g)$, instead of the $\mL^2$ inner product adopted here. One way of doing this is to introduce a non
trivial weight function in the integral. Thus, different choices of inner product give different SJ Wightman functions
even with the same defining conditions (see Sec.~\ref{sjint.sec}).  As an almost trivial example, in Sec.~\ref{sec:rindler}  we show that the choice of inner product can yield the Rindler vacuum for massless scalar field in the corner.  In future work we hope to explore  this possibility in more detail. In the next chapter we study the SJ vacuum and its properties in de Sitter and flat FLRW spacetime which are of cosmological interest.
\vskip 0.5in
\leftline{\bf\Large{Appendices}}
\begin{subappendices}

\section{Some expressions and derivation of results used in Sec.~\ref{sec:soln}}\label{expressions}

In this appendix we show some of the details of the calculations of Sec.~\ref{sec:soln}. These details include the simplified expression of $F_{ik,n}(u,v)$ and $G_{ik,n}(u,v)$ for $n=0,1,2$ , $Z^{A/S}_l(u,v)$ and $i\hD\circ Z^{A/S}_l(u,v)$, for $l=0,1,2$ and $\PAS_n(u,v)$ for $n=0,1,2$ up to the order in $m^2$, which is required in the calculation of SJ modes up to $\cO(m^4)$. Some details of the calculations of $u^A_k(u,v)$ and $u^S_k(u,v)$ can be found in Appendix \ref{sec:dasym} and \ref{sec:dsym} respectively.

Evaluating $F_{ik,n}(u,v)$ and $G_{ik,n}(u,v)$ defined in Eqn.~(\ref{f and g}) for $n=0,1,2$, we get
\bea
F_{ik,0}(u,v)&=&v, \nonumber\\
F_{ik,1}(u,v)&=&\frac{iv^2}{2k}-\frac{1}{4}(v^2u+2v+u), \nonumber\\
F_{ik,2}(u,v)&=&-\frac{v^3}{8k^2}-\frac{i}{24k}(2v^3u+3v^2-1)+\frac{1}{48}(v^3u^2+v^3+6v^2u+3vu^2+3v+2u). \nonumber
\eea
\bea
G_{ik,0}(u,v)&=&-1, \nonumber\\
G_{ik,1}(u,v)&=&-\frac{iv}{2k}+\frac{1}{4}(v^2+2uv+1), \nonumber\\
G_{ik,2}(u,v)&=&\frac{v^2}{8k^2}+\frac{i}{24k}(2v^3+3uv^2-u)-\frac{1}{48}(2v^3u+3v^2u^2+3v^2+6uv+u^2+1). \nonumber
\eea

Next, we list $Z^A_l(u,v)$ and $Z^S_l(u,v)$ defined in Eqn.~(\ref{eq:serieszl}) and Eqn.~(\ref{eq:symasym}) for $l=0,1,2$ up to the required order of $m^2$.
 \begin{eqnarray} 
Z^A_0(u,v) =0, & \quad & Z^S_0(u,v) \approx 2-m^2uv+\frac{m^4}{8}u^2v^2, \nonumber \\
Z^A_1(u,v) \approx (u-v)-\frac{m^2}{4}uv(u-v), & \quad & Z_1^S(u,v)\approx (u+v)-\frac{m^2}{4}uv(u+v),  \nonumber \\
Z^A_2(u,v)\approx u^2-v^2,& \quad & Z^S_2(u,v) \approx u^2+v^2.
 \end{eqnarray}

Next, we list $i\hD\circ Z^A_l(u,v)$ and $i\hD\circ Z^S_l(u,v)$ for $l=0,1,2$ up to the required order of $m^2$, where $i\hD\circ Z_l(u,v)$ is described in Eqn.~(\ref{eq:ideltaz}) 
\begin{eqnarray} 
  i\hD \circ Z^A_0(u,v)&=&0,   \nonumber \\
  i\hD \circ Z^S_0(u,v) &\approx&  -i\frac{L^2}{24}(u+v) (48-12 m^2(1+uv)+m^4(3+3uv+u^2v^2))),
                                                               \nonumber \\
  i\hD \circ Z^A_1(u,v)&\approx& iL^2\left(-\frac{1}{2}(u^2-v^2)+\frac{m^2}{24}(2uv+1)(u^2-v^2)\right),  \nonumber \\
  i\hD \circ Z^S_1(u,v) &\approx&  iL^2\left(\frac{1}{2}(2-u^2-v^2)-\frac{m^2}{24}\left(6(1+2uv)+(u^2+v^2)(1-2uv)\right)\right),\nonumber \\
  i\hD \circ Z^A_2(u,v) &\approx& \frac{iL^2}{3}\left((u-v)-(u^3-v^3)\right),   \nonumber \\
i\hD \circ Z^S_2(u,v) &\approx& \frac{iL^2}{3}\left((u+v)-(u^3+v^3)\right).
  \end{eqnarray} 

$\PAS_n(u,v)$ defined in Eqn.~(\ref{eq:pas}) for $n=0,1,2$.
\bea
P_0^A(u,v) &=&0, \nonumber\\
P_1^A(u,v) &=& \left(i\left(\frac{1}{2k}-Q^A_1(k)\right)(u-v)-\frac{1}{4}(u^2-v^2)\right),\nonumber\\
P_2^A(u,v) &= & -\frac{u^2-v^2}{8k^2}-\frac{i}{24k}(u-v)(2u^2+2v^2+5uv+1)+\frac{1}{24}(1+uv)(u^2-v^2),
  \nonumber \\ &&+ Q^A_1(k)\left(\frac{u^2-v^2}{2k}+\frac{i}{4}(u-v)(uv+1)\right)-iQ^A_2(k)(u-v) ,
\label{eq:pai}
\eea
\bea
P_0^S(u,v) &=& \frac{}{}-2+2ik(u+v), \nonumber\\
P_1^S(u,v)&=& -\frac{i(u+v)}{2k}+\frac{1}{4}(u^2+v^2+4uv+2)-\left(u^2+v^2+\frac{ik}{2}(uv+3)(u+v)\right)+iQ_1^S(k)(u+v),
\nonumber\\
P_2^S(u,v)&=& \frac{u^2+v^2}{8k^2}+\frac{i}{24k}(u+v)(2u^2+2v^2+uv-1)-\frac{1}{48}((2uv+4)(u^2+v^2)+6v^2u^2+12uv+2)
\nonumber\\ &&
+2k\left(-\frac{i(u^3+v^3)}{8k^2}+\frac{1}{24k}((2uv+3)(u^2+v^2)-2)+\frac{i}{48}(u+v)(u^2v^2+u^2+v^2+8uv+5)\right)\nonumber\\
&&
+Q_{1}^S(k)\left(-\frac{u^2+v^2}{2k}-\frac{i}{4}(uv+3)(u+v)\right)+iQ_{2}^S(k)(u+v), \label{psk}
\eea
where $Q^A_n(k)$ and $Q^S_n(k)$ for $n=1,2$ can be found in Sec.~\ref{sec:dasym} and \ref{sec:dsym} respectively.

\subsection{Details of the calculations for the antisymmetric SJ modes}\label{sec:dasym}
In this section we solve Eqn.~(\ref{eq:hask}) for $H^A_k(u,v)$ by constructing each $m^{2n}P^A_n(u,v)$ out of $Z_l(u,v)$ and $i\Delta\circ Z_l(u,v)$ for different $l$. Let us start with the first non zero $P^A_n(u,v)$. It can be observed that $m^2P^A_1(u,v)$ can be constructed out of $m^2Z^A_1(u,v)$ and $m^2i\Delta\circ Z^A_1(u,v)$ up to $\cO(m^2)$ as
\be
m^2P^A_1(u,v)=\frac{im^2}{2L^2}\left(\frac{L^2}{k}\left(1-2kQ^A_1(k)\right)Z^A_1(u,v)-i\Delta\circ Z^A_1(u,v)\right)
\ee
To make the term in the bracket look like $\left(i\Delta+\frac{L^2}{k}\right)\circ Z^A_1(u,v)$, we fix
\be
Q^A_1(k)=\frac{1}{k}.
\ee
Therefore Eqn.~(\ref{eq:hask}) for $H^A_k(u,v)$ up to $\cO(m^4)$ can be written as
\bea
&&\left(i\Delta+\frac{L^2}{k}\right)\circ\left(H^A_k(u,v)+\frac{im^2\cos(k)}{2k}Z^A_1(u,v)\right)-\frac{m^4L^2\cos(k)}{k}\left(\frac{3(u^2-v^2)}{8k^2}-\frac{i}{12k}(u^3-v^3)\right.\nonumber\\
&&\quad\left.+\frac{5i}{24k}(u-v)+\frac{1}{48}(u^2-v^2)-iQ^A_2(k)(u-v)\right)=0.\nonumber\\ \label{extra term a2}
\eea
In the remaining terms, i.e., the terms which are not yet written as $Z^A_l(u,v)$ or $i\Delta\circ Z^A_l(u,v)$, the highest order of $u$ and $v$ are $u^3$ and $v^3$, which can be identified with $i\Delta\circ Z_2(u,v)$. Therefore we use
\be
-\left(i\Delta+\frac{L^2}{k}\right)\circ\frac{m^4\cos(k)}{4k^2}Z^A_2(u,v)=-\frac{m^4L^2\cos(k)}{k}\left(\frac{i}{12k}(u-v)-\frac{i}{12k}(u^3-v^3)+\frac{1}{4k^2}(u^2-v^2)\right),
\ee
to write Eqn.~(\ref{extra term a2}) as
\bea
\left(i\Delta+\frac{L^2}{k}\right)\circ\left(H^A_k(u,v)+\cos(k)\left(\frac{im^2}{2k}Z^A_1(u,v)-\frac{m^4}{4k^2}Z^A_2(u,v)\right)\right)\nonumber\\
-\frac{m^4L^2\cos(k)}{k}\left(\frac{u^2-v^2}{8k^2}+\frac{i}{8k}(u-v)+\frac{1}{48}(u^2-v^2)-iQ^A_2(k)(u-v)\right)=0. \label{extra term a3}
\eea
The remaining terms in Eqn.~(\ref{extra term a3}) can be written as
\be
\left(i\Delta+\frac{L^2}{k}\right)\circ\left(-\frac{im^4\cos(k)}{24k^3}(6+k^2)Z^A_1(u,v)\right),
\ee
by fixing
\be
Q^A_2(k)=\frac{1}{12k}-\frac{1}{4k^3}.
\ee
Finally Eqn.~(\ref{extra term a3}) can be written as
\be
\left(i\Delta+\frac{L^2}{k}\right)\circ\left(H^A_k(u,v)+\cos(k)\left(\left(\frac{im^2}{2k}-\frac{im^4(6+k^2)}{24k^3}\right)Z^A_1(u,v)-\frac{m^4}{4k^2}Z^A_2(u,v)\right)\right)=0
\ee
which implies that
\be
u^A_k(u,v)=U^A_{ik}(u,v)-\cos(k)\left(\left(\frac{im^2}{2k}-\frac{im^4(6+k^2)}{24k^3}\right)Z^A_1(u,v)-\frac{m^4}{4k^2}Z^A_2(u,v)\right)+\cO(m^6)
\ee
with eigenvalue $-\frac{L^2}{k}$, where $k$ satisfies
\be
\sin(k)=\left(\frac{m^2}{k}+\frac{m^4}{12k}\left(1-\frac{3}{k^2}\right)\right)\cos(k)+\cO(m^6)\label{f cond.}
\ee
\subsection{Details of the calculations for the symmetric SJ modes}
\label{sec:dsym}
In this section we solve Eqn.~(\ref{eq:hask}) for $H^S_k(u,v)$ by constructing each $m^{2n}P^S_n(u,v)$ out of $Z_l(u,v)$ and $i\Delta\circ Z_l(u,v)$ for different $l$. Let us start with the first non zero $P^S_n(u,v)$. It can be observed that $P^S_0(u,v)$ can be constructed out of $Z^S_0(u,v)$ and $i\Delta\circ Z^S_0(u,v)$ up to $\cO(m^0)$ as
\be
P^S_0(u,v)=\left(i\Delta+\frac{L^2}{k}\right)\circ\left(-\frac{k}{L^2}Z^S_0(u,v)\right).
\ee
Therefore Eqn.~(\ref{eq:hask}) for $H^S_k(u,v)$ up to $\cO(m^4)$ can be written as
\bea
&&\left(i\Delta+\frac{L^2}{k}\right)\circ\left(H^S_k(u,v)+Z^S_0(u,v)\cos(k)\right)-\frac{L^2\cos(k)}{k}\left(m^2\left(-\frac{3}{4}(u^2+v^2)\right.\right.\nonumber\\
&&\left.\left.+i\left(Q^S_1(k)-k-\frac{1}{2k}\right)(u+v)+\frac{1}{2}\right)+m^4\left(\frac{u^2+v^2}{8k^2}+\frac{i}{24k}(u+v)(2u^2+2v^2+uv-1)\right.\right.\nonumber\\
&&\left.\left.-\frac{1}{24}((-3uv-4)(u^2+v^2)+6uv+5)+\left(-\frac{i(u^3+v^3)}{4k}+\frac{ik}{24}(u+v)(u^2+v^2+5uv+2)\right)\right.\right.\nonumber\\
&&\left.\left.+Q^S_1(k)\left(-\frac{u^2+v^2}{2k}-\frac{i}{4}(uv+3)(u+v)\right)+iQ^S_2(k)(u+v)\right)\right)=0. \label{extra term s2}
\eea
Since the extra terms in Eqn.~(\ref{extra term s2}) has $m^2$ as a factor, we need to look for $Z^S_l$ and $i\Delta\circ Z^S_l$ only up to $\cO(m^2)$. $\cO(m^2)$ terms in Eqn.~(\ref{extra term s2}) can be written in terms of $\left(i\Delta+\frac{L^2}{k}\right)\circ Z^S_0(u,v)$ and  $\left(i\Delta+\frac{L^2}{k}\right)\circ Z^S_1(u,v)$ for
\be
Q^S_1=2k-\frac{1}{k}
\ee
as
\be
\left(i\Delta+\frac{L^2}{k}\right)\circ m^2\cos(k)\left(\frac{3i}{2k}Z^S_1(u,v)+\frac{1}{2}Z^S_0(u,v)\right).
\ee
Therefore Eqn.~(\ref{extra term s2}) can further be written as
\bea
\left(i\Delta+\frac{L^2}{k}\right)\circ\left(H^S_k(u,v)+\cos(k)\left(\left(1+\frac{m^2}{2}\right)Z^S_0(u,v)+\frac{3im^2}{2k}Z^S_1(u,v)\right)\right)+\frac{im^4L^2\cos(k)}{48k^3}\left(8ik^2\right.\nonumber\\
\left.+k\left(-34-kQ^S_2(k)+56k^2\right)(u+v)+i(30-37k^2)(u^2+v^2)+2k(4-k^2)(u^3+v^3)\right)=0. \label{extra term s3}
\eea
Remaining $\cO(m^4)$ terms in Eqn.~(\ref{extra term s3}) can be written in terms of $\left(i\Delta+\frac{L^2}{k}\right)\circ Z^S_0(u,v)$, $\left(i\Delta+\frac{L^2}{k}\right)\circ Z^S_1(u,v)$, $\left(i\Delta+\frac{L^2}{k}\right)\circ Z^S_2(u,v)$ for
\be
Q^S_2(k)=\frac{3-29k^2+28k^4}{12k^3}
\ee
as
\be
-\frac{m^4\cos(k)}{8k^2}\left((4-k^2)Z^S_2(u,v)+\frac{i(6-31k^2)}{3k}Z^S_1(u,v)+(2-9k^2)Z^S_0(u,v)\right).
\ee
Hence Eqn.~(\ref{extra term s3}) can be written as
\bea
\left(i\Delta+\frac{L^2}{k}\right)\circ\left(H^S_k(u,v)+\cos(k)\left(\left(1+\frac{m^2}{2}-\frac{m^4}{8k^2}(2-9k^2)\right)Z^S_0(u,v)\right.\right.\nonumber\\
\left.\left.+\left(\frac{3im^2}{2k}-\frac{im^4}{24k^3}(6-31k^2)\right)Z^S_1(u,v)-\frac{m^4}{8k^2}(4-k^2)Z^S_2(u,v)\right)\right)=0.
\eea
Therefore the symmetric SJ modes are
\bea
u^S_k(u,v)&=&U^S_{ik}(u,v)-\cos(k)\left(\left(1+\frac{m^2}{2}-\frac{m^4}{8k^2}(2-9k^2)\right)Z^S_0(u,v)\right.\nonumber\\
&&\left.+\left(\frac{3im^2}{2k}-\frac{im^4}{24k^3}(6-31k^2)\right)Z^S_1(u,v)-\frac{m^4}{8k^2}(4-k^2)Z^S_2(u,v)\right)+\cO(m^4),
\eea
with eigenvalue $-\frac{L^2}{k}$, where $k$ satisfies
\be
\sin(k)=\left(2k-\frac{m^2}{k}(1-2k^2)+\frac{m^4}{12k^3}(3-29k^2+28k^4)\right)\cos(k)+\cO(m^4). \label{g cond.}
\ee

\section{Summation of series with inverse powers of roots of a transcendental equation} \label{app:speigel}
In this appendix we make use of the work of \cite{speigel} to evaluate the series (Eqn.~(\ref{eq:ser1}) and Eqn.~(\ref{eq:ser2})), which involves the roots of the transcendental equation (Eqn.~(\ref{g condition m0})). They are used in Sec(\ref{sec:compl}) to determine the completeness of the SJ modes

Let us start with a brief discussion on the work of \cite{speigel}. Consider a transcendental equation of the form
\be
S(x) \equiv 1+\sum_{n=1}^\infty a_n x^n=0 \label{eqn:gentr}
\ee
with $x_1,x_2,x_3\dots$ as its roots, which means the equation can be factorised as
\be
\left(1-\frac{x}{x_1}\right)\left(1-\frac{x}{x_2}\right)\left(1-\frac{x}{x_3}\right)\dots=0 \label{eqn:gentrfac}
\ee
On comparing Eqn.~(\ref{eqn:gentr} and \ref{eqn:gentrfac}), we find that
\be
a_1= \sum_{i=1}^\infty \frac{1}{x_i}, \quad
a_2= \sum_{i<j} \frac{1}{x_ix_j}, \quad
a_3= \sum_{i<j<k} \frac{1}{x_ix_jx_k} 
\ee
and so on. It is straight forward to see that
\be
\sum_{i=1}^\infty\left(\frac{1}{x_i}\right)^2=\left(\sum_{i=1}^\infty \frac{1}{x_i}\right)^2-2\sum_{i<j} \frac{1}{x_ix_j}=a_1^2-2a_2
\ee
and similarly
\be
\sum_{i=1}^\infty\left(\frac{1}{x_i}\right)^3=3a_1a_2-3a_3-a_1^3.
\ee
Similarly we can get the sum of higher inverse powers of the roots.

Now let us come to the equation of our interest i.e. Eqn.~(\ref{g condition m0}), which on series expansion becomes
\be
S(k^2)\equiv 1-\left(1-\frac{1}{3!}\right)k^2+\left(\frac{2}{4!}-\frac{1}{5!}\right)k^4-\left(\frac{2}{6!}-\frac{1}{7!}\right)k^6\dots=0.\label{eq:expand}
\ee
The roots of Eqn.~(\ref{eq:expand}) are $k^S_0\in\cK_g$, and therefore
\bea
\sum_{k^S_0\in\cK_g}\frac{1}{{k^S_0}^2}&= &a_1 =\frac{5}{6},\nonumber\\
\sum_{k^S_0\in\cK_g}\frac{1}{{k^S_0}^4}&= &a_1^2-2a_2 =\frac{49}{90},\nonumber\\
\sum_{k^S_0\in\cK_g}\frac{1}{{k^S_0}^6}&= &3a_1a_2-3a_3-a_1^3 =\frac{377}{945}.
\eea
We are also interested in the series involving the inverse power of $4{k^S_0}^2-1$, where $k^S_0\in\cK_g$. We start with finding an equation whose solutions are given by $4{k^S_0}^2-1$. If ${k^S_0}^2$ are the solutions of $S(k^2)=0$, then $4{k^S_0}^2-1$ are the solutions of $S\left(\frac{k^2+1}{4}\right)=0$.
\be
S\left(\frac{k^2+1}{4}\right) \equiv 1-\frac{1}{4}k^2+\frac{5\cos(1/2)-9\sin(1/2)}{32\left(\cos(1/2)-\sin(1/2)\right)}k^4-\frac{53\cos(1/2)-97\sin(1/2)}{384\left(\cos(1/2)-\sin(1/2)\right)}k^6\dots=0.
\ee
Using the same method as above, we find
\bea
\sum_{k^S_0\in\cK_g}\frac{1}{4{k^S_0}^2-1}&=&\frac{1}{4},\\
\sum_{k^S_0\in\cK_g}\frac{1}{(4{k^S_0}^2-1)^2}&=&-\frac{1}{4}\left(\frac{\cos(1/2)-2\sin(1/2)}{\cos(1/2)-\sin(1/2)}\right),\\
\sum_{k_0\in\cK_g}\frac{1}{(4{k^S_0}^2-1)^3}&=&\frac{1}{64}\left(1+\frac{19\cos(1/2)-35\sin(1/2)}{\cos(1/2)-\sin(1/2)}\right).
\eea

\section{Some expressions used in Sec.~\ref{sec:wightmann}}\label{sec:wight-app}

Here we list the expressions of $\wa,\waa,\waaa$ and $\waaaa$ defined in Eqn.~(\ref{eq:aone}) in terms of Polylogarithms.
\bea
\wa&\equiv& \sum_{n=1}^\infty\frac{1}{8n\pi}\left(1-\frac{2m^2}{n^2\pi^2}+\frac{m^4}{n^2\pi^2}\left(\frac{7}{n^2\pi^2}-\frac{1}{6}\right)\right)\left(e^{-in\pi u}-e^{-in\pi v}\right)\left(e^{in\pi u'}-e^{in\pi v'}\right)\nonumber\\
&=&\frac{1}{8\pi}\left[\li_1\left(e^{-i\pi(u-u')}\right)+\li_1\left(e^{-i\pi(v-v')}\right)-\li_1\left(e^{-i\pi(u-v')}\right)-\li_1\left(e^{-i\pi(v-u')}\right)\right]\nonumber\\
&&-\frac{m^2}{4\pi^3}\left(1+\frac{m^2}{12}\right)\left[\li_3\left(e^{-i\pi(u-u')}\right)+\li_3\left(e^{-i\pi(v-v')}\right)-\li_3\left(e^{-i\pi(u-v')}\right)-\li_3\left(e^{-i\pi(v-u')}\right)\right]\nonumber\\
&&+\frac{7m^4}{8\pi^5}\left[\li_5\left(e^{-i\pi(u-u')}\right)+\li_5\left(e^{-i\pi(v-v')}\right)-\li_5\left(e^{-i\pi(u-v')}\right)-\li_5\left(e^{-i\pi(v-u')}\right)\right],
\eea
\bea
\waa&\equiv& \sum_{n=1}^\infty\frac{1}{8n\pi}\left(1-\frac{2m^2}{n^2\pi^2}\right)\left(e^{-in\pi u}-e^{-in\pi v}\right)\Psi_A^*(n,u',v')\nonumber\\
&=& \frac{1}{8\pi} \sum_{j=1}^3 f_j^*(m;u',v')\left[\li_{j+1}\left(-e^{-i\pi u}\right)-\li_{j+1}\left(-e^{-i\pi v}\right)\right]+\frac{im^4}{8\pi^4}(u'-v')\left[\li_{4}\left(-e^{-i\pi u}\right)-\li_{4}\left(-e^{-i\pi v}\right)\right]\nonumber\\
&&+\frac{1}{8\pi}\sum_{j=1}^3\left( g_j^*(m;u',v')\left[\li_{j+1}\left(e^{-i\pi(u-u')}\right)-\li_{j+1}\left(-e^{-i\pi(v-u')}\right)\right]\right.\nonumber\\
&&\left.\quad\quad\quad-g_j^*(m;v',u')\left[\li_{j+1}\left(-e^{-i\pi(u-v')}\right)-\li_{j+1}\left(-e^{-i\pi(v-v')}\right)\right]\right)\nonumber\\
&&-\frac{im^4}{8\pi^4}\left((2u'+v')\left[\li_{4}\left(-e^{-i\pi(u-u')}\right)-\li_{4}\left(-e^{-i\pi(v-u')}\right)\right]\right.\nonumber\\
&&\left.\quad\quad\quad-(2v'+u')\left[\li_{4}\left(-e^{-i\pi(u-v')}\right)-\li_{4}\left(-e^{-i\pi(v-v')}\right)\right]\right),
\eea
\bea
\waaa&\equiv& \sum_{n=1}^\infty\frac{1}{8n\pi}\left(1-\frac{2m^2}{n^2\pi^2}\right)\Psi_A(n,u,v)\left(e^{in\pi u'}-e^{in\pi v'}\right)\nonumber\\
&=& \frac{1}{8\pi} \sum_{j=1}^3 f_j(m;u,v)\left[\li_{j+1}\left(-e^{i\pi u'}\right)-\li_{j+1}\left(-e^{i\pi v'}\right)\right]-\frac{im^4}{8\pi^4}(u-v)\left[\li_{4}\left(-e^{i\pi u'}\right)-\li_{4}\left(-e^{i\pi v'}\right)\right]\nonumber\\
&&+\frac{1}{8\pi}\sum_{j=1}^3\left( g_j(m;u,v)\left[\li_{j+1}\left(e^{-i\pi(u-u')}\right)-\li_{j+1}\left(-e^{-i\pi(v-u')}\right)\right]\right.\nonumber\\
&&\left.\quad\quad\quad-g_j(m;v,u)\left[\li_{j+1}\left(-e^{-i\pi(u-v')}\right)-\li_{j+1}\left(-e^{-i\pi(v-v')}\right)\right]\right)\nonumber\\
&&+\frac{im^4}{8\pi^4}\left((2u+v)\left[\li_{4}\left(-e^{-i\pi(u-u')}\right)-\li_{4}\left(-e^{-i\pi(v-u')}\right)\right]\right.\nonumber\\
&&\left.\quad\quad\quad-(2v+u)\left[\li_{4}\left(-e^{-i\pi(u-v')}\right)-\li_{4}\left(-e^{-i\pi(v-v')}\right)\right]\right),
\eea
\bea
\waaaa&\equiv& \sum_{n=1}^\infty\frac{1}{8n\pi}\Psi_A(n,u,v)\Psi_A^*(n,u',v')\nonumber\\
&=&\frac{m^4}{32\pi^3}\left[\zeta(3)(u-v)(u'-v')-(u-v)(2u'+v')\li_3\left(-e^{i\pi u'}\right)-(2u+v)(u'-v')\li_3\left(-e^{-i\pi u}\right)\right.\nonumber\\
&&\left.+(u-v)(u'+2v')\li_3\left(-e^{i\pi v'}\right)+(u+2v)(u'-v')\li_3\left(-e^{-i\pi v}\right)\right.\nonumber\\
&&\left.+(2u+v)(2u'+v')\li_3\left(e^{-i\pi(u-u')}\right)+(u+2v)(u'+2v')\li_3\left(e^{-i\pi(v-v')}\right)\right.\nonumber\\
&&\left.-(2u+v)(u'+2v')\li_3\left(e^{-i\pi(u-v')}\right)-(u+2v)(2u'+v')\li_3\left(e^{-i\pi(v-u')}\right)\right],
\eea

Here we list the expressions of $\ws,\wss,\wsss$ and $\wssss$ defined in Eqn.~(\ref{eq:sone}) in terms of Polylogarithms.
\bea
\ws&\equiv& \frac{1}{4\pi}\sum_{n=1}^\infty \frac{1}{(2n-1)}\left(e^{-i\left(n-\frac{1}{2}\right)\pi u}+e^{-i\left(n-\frac{1}{2}\right)\pi v}\right)\left(e^{i\left(n-\frac{1}{2}\right)\pi u'}+e^{i\left(n-\frac{1}{2}\right)\pi v'}\right)\nonumber\\
&=&\frac{1}{4\pi}\left[\li_1\left(e^{-i\pi\frac{(u-u')}{2}}\right)+\li_1\left(e^{-i\pi\frac{(u-v')}{2}}\right)+\li_1\left(e^{-i\pi\frac{(v-u')}{2}}\right)+\li_1\left(e^{-i\pi\frac{(v-v')}{2}}\right)\right]\nonumber\\
&&-\frac{1}{8\pi}\left[\li_1\left(e^{-i\pi(u-u')}\right)+\li_1\left(e^{-i\pi(u-v')}\right)+\li_1\left(e^{-i\pi(v-u')}\right)+\li_1\left(e^{-i\pi(v-v')}\right)\right],
\eea
\bea
\wss &\equiv& \frac{1}{4\pi}\sum_{n=1}^\infty\frac{1}{2n-1}\left(e^{-i\left(n-\frac{1}{2}\right)\pi u}+e^{-i\left(n-\frac{1}{2}\right)\pi v}\right)\Psi_S^*(n,u',v')\nonumber\\
&=& \frac{im^2v'}{4\pi^2}\left[\li_2\left(e^{-i\pi\frac{(u-u')}{2}}\right)+\li_2\left(e^{-i\pi\frac{(v-u')}{2}}\right)-\frac{1}{4}\li_2\left(e^{-i\pi(u-u')}\right)-\frac{1}{4}\li_2\left(e^{-i\pi(v-u')}\right)\right]\nonumber\\
&&+\frac{im^2u'}{4\pi^2}\left[\li_2\left(e^{-i\pi\frac{(u-v')}{2}}\right)+\li_2\left(e^{-i\pi\frac{(v-v')}{2}}\right)-\frac{1}{4}\li_2\left(e^{-i\pi(u-v')}\right)-\frac{1}{4}\li_2\left(e^{-i\pi(v-v')}\right)\right]\nonumber\\
&&-\frac{m^4v'^2}{16\pi^3}\left[\li_2\left(e^{-i\pi\frac{(u-u')}{2}}\right)+\li_2\left(e^{-i\pi\frac{(v-u')}{2}}\right)-\frac{1}{4}\li_2\left(e^{-i\pi(u-u')}\right)-\frac{1}{4}\li_2\left(e^{-i\pi(v-u')}\right)\right]\nonumber\\
&&-\frac{m^4u'^2}{16\pi^3}\left[\li_2\left(e^{-i\pi\frac{(u-v')}{2}}\right)+\li_2\left(e^{-i\pi\frac{(v-v')}{2}}\right)-\frac{1}{4}\li_2\left(e^{-i\pi(u-v')}\right)-\frac{1}{4}\li_2\left(e^{-i\pi(v-v')}\right)\right],\nonumber\\
\eea
\bea
\wsss &\equiv& \frac{1}{4\pi}\sum_{n=1}^\infty\frac{1}{2n-1}\Psi_S(n,u,v)\left(e^{i\left(n-\frac{1}{2}\right)\pi u'}+e^{i\left(n-\frac{1}{2}\right)\pi v'}\right)\nonumber\\
&=& -\frac{im^2v}{4\pi^2}\left[\li_2\left(e^{-i\pi\frac{(u-u')}{2}}\right)+\li_2\left(e^{-i\pi\frac{(v-u')}{2}}\right)-\frac{1}{4}\li_2\left(e^{-i\pi(u-u')}\right)-\frac{1}{4}\li_2\left(e^{-i\pi(v-u')}\right)\right]\nonumber\\
&&-\frac{im^2u}{4\pi^2}\left[\li_2\left(e^{-i\pi\frac{(u-v')}{2}}\right)+\li_2\left(e^{-i\pi\frac{(v-v')}{2}}\right)-\frac{1}{4}\li_2\left(e^{-i\pi(u-v')}\right)-\frac{1}{4}\li_2\left(e^{-i\pi(v-v')}\right)\right]\nonumber\\
&&-\frac{m^4v^2}{16\pi^3}\left[\li_2\left(e^{-i\pi\frac{(u-u')}{2}}\right)+\li_2\left(e^{-i\pi\frac{(v-u')}{2}}\right)-\frac{1}{4}\li_2\left(e^{-i\pi(u-u')}\right)-\frac{1}{4}\li_2\left(e^{-i\pi(v-u')}\right)\right]\nonumber\\
&&-\frac{m^4u^2}{16\pi^3}\left[\li_2\left(e^{-i\pi\frac{(u-v')}{2}}\right)+\li_2\left(e^{-i\pi\frac{(v-v')}{2}}\right)-\frac{1}{4}\li_2\left(e^{-i\pi(u-v')}\right)-\frac{1}{4}\li_2\left(e^{-i\pi(v-v')}\right)\right],\nonumber\\
\eea
\bea
\wssss &\equiv& \frac{1}{4\pi}\sum_{n=1}^\infty\frac{1}{2n-1}\Psi_S(n,u,v)\Psi_S^*(n,u',v')\nonumber\\
&=&\frac{m^4}{4\pi^3}\left[vv'\left(\li_3\left(e^{-i\pi\frac{(u-u')}{2}}\right)-\frac{1}{8}\li_3\left(e^{-i\pi(u-u')}\right)\right)\right.+vu'\left(\li_3\left(e^{-i\pi\frac{(u-v')}{2}}\right)-\frac{1}{8}\li_3\left(e^{-i\pi(u-v')}\right)\right)\nonumber\\
&&\quad+ uv'\left(\li_3\left(e^{-i\pi\frac{(v-u')}{2}}\right)-\frac{1}{8}\li_3\left(e^{-i\pi(v-u')}\right)\right)+\left.uu'\left(\li_3\left(e^{-i\pi\frac{(v-v')}{2}}\right)-\frac{1}{8}\li_3\left(e^{-i\pi(v-v')}\right)\right)\right],\nonumber\\
\eea
\end{subappendices}

\chapter{The SJ vacuum in cosmological spacetimes}\label{ch.sjcosm}


In this chapter we study the SJ vacuum for a conformally coupled massless scalar field in cosmological spacetimes. We explicitly compute the SJ vacuum in a finite time slab of cosmological spacetimes, which are conformally related to ultrastatic slab spacetimes with a time dependent conformal factor, and then study its behaviour in the full spacetime limit. We use the work of Fewster and Verch \cite{fewster2012} on ultrastatic slab spacetimes and the work of Brum and Fredenhagen on expanding spacetimes \cite{Brum:2013bia} as our starting point to solve the integral eigenvalue equation. In particular, we consider two spacetimes of physical interest. One of which is the de Sitter spacetime, which provides a promising inflationary model for a post-Planckian very early universe. The other one is a collection of flat Friedmann-Lema{\^i}tre-Robertson-Walker (FLRW) spacetimes, which provide a model for post inflationary including the present day homogeneous, isotropic and expanding universe. In order to ensure the finiteness of the spacetime volume of flat FLRW spacetimes the spacelike hypersurface is compactified into a 3-torus. This costs spatial isotropy but local properties remain unchanged. We may note that the SJ vacuum in full and in conformal patch of a de Sitter spacetime is studied by Aslanbeigi and Buck starting from the Euclidean or the Bunch-Davies modes in respective cases \cite{Aslanbeigi:2013fga}. 
We study the SJ modes and the spectrum and its dependence on the size of the slab in both de Sitter and flat FLRW spacetimes and compare with the known vacua in respective spacetimes. We see that the SJ vacuum in a conformally ultrastatic slab spacetimes cannot be obtained merely by a conformal transformation of the SJ vacuum in the corresponding ultrastatic slab spacetime. We find that the SJ vacuum obtained here is consistent with that obtained in \cite{Aslanbeigi:2013fga} in full de Sitter spacetime limit, i.e., in odd dimensions, it reduces to the Euclidean vacuum and in even dimensions the SJ vacuum is ill defined. The results obtained here along with the one obtained in \cite{Aslanbeigi:2013fga} is inconsistent with the result of \cite{Surya:2018byh}, where the SJ vacuum has been obtained in a causal set sprinkled on a de Sitter slab spacetime and is found to be well defined but different from any of the known de Sitter vacua in full spacetime limit. In order to have a better understanding of this we compare the SJ spectrum we  obtained in a de Sitter spacetime with that obtained by the authors of \cite{Surya:2018byh} on a causal set approximated by the de Sitter spacetime. We find that in the 2d case the continuum and the causal set SJ spectrum matches with each other upto a UV cut-off associated with the sprinkling density of the causal set but there is a mismatch in the 4d continuum and the causal set de Sitter SJ spectrum even at the scales below the UV cut-off, which is yet to be understood. In the FLRW spacetime the SJ vacuum reduces to the conformal vacuum in the full spacetime limit. 

We start this chapter with a review of the SJ vacuum in ultrastatic slab spacetimes in Sec.~\ref{sec:ustatic}, followed by the SJ vacuum in a finite time slab of conformally ultrastatic slab spacetimes with a time dependent conformal factor in Sec.~\ref{sec:confustatic}. In Sec.~\ref{sec:cuss}, we in particular study the SJ vacuum in de Sitter and flat FLRW spacetimes, which are of cosmological interests. We end this chapter with a discussion of the results in Sec.~\ref{sec.discch3}.

\section{SJ vacuum in ultrastatic slab spacetimes}\label{sec:ustatic}
We start with a brief review of the work of Fewster and Verch \cite{fewster2012} on the construction of SJ modes for massive scalar field in ultrastatic slab spacetimes. Fewster and Verch studied the minimally coupled massive field i.e., $\xi=0$ but their analysis can easily be generalised to any coupling $\xi$ as long as $\xi R>0$, as we will see in this section. An ultrastatic spacetime $(\M,g)$ consists of a spacetime manifold $\M = \mr\times\spac$, where $\spac$ is a compact spacelike hypersurface and a Lorentzian metric $g$ leading to an invariant length element given by
\be
ds^2 = -d\eta^2+d\spac^2 \label{eq:usmetric}
\ee
where $\eta\in\mr$ and $d\spac$ is the length element on $\spac$, which is independent of $\eta$. The Klein-Gordon equation for the massive scalar field given by Eqn.~\eqref{kg.eq} in $(\M,g)$ takes the form
\be
(-\partial^2_\eta + \nabla_i\nabla^i -m^2 -\xi R)\Phi = 0,\label{eq:kgus}
\ee
where $\nabla_i$'s are the spatial covariant derivatives and $R$ is the Ricci scalar of the spacelike hypersurface $\spac$. Ultrastatic spacetimes admits a timelike Killing vector $\partial_\eta$, which leads to a preferred mode expansion of $\Phi$ (given that $\xi R>0$) in terms of the positive frequency modes $\up_j(\eta,\vx)\equiv\frac{1}{\sqrt{2\omega_j}}e^{-i\omega_j\eta}\psi_j(\vx)$ and negative frequency modes $\un_j(\eta,\vx)\equiv\frac{1}{\sqrt{2\omega_j}}e^{i\omega_j\eta}\psi_j(\vx)$ with respect to $\partial_\eta$. We call the vacuum associated with this preferred mode expansion as the static vacuum. Here $\vx,\vx'\in\spac$, $\{\omega_j^2\}$ and $\{\psi_j\}$ are the eigenvalues and eigenfunctions respectively of the spatial part $\spkg$ of the Klein-Gordon equation in Eqn.~\eqref{eq:esugr} and $j$ stands for the collection of indices required to specify a particular eigenfunction.
\be
\spkg\psi_j = \omega_j^2\psi_j,\quad \spkg\equiv -\nabla_i\nabla^i +m^2 +\xi R.
\ee
It is straightforward to see that for $R\geq 0$, $\spkg$ is a positive semi-definite operator (positive definite for $m>0$). The fact that $\spac$ is compact and $\spkg$ is self-adjoint ensures that the eigenfunctions $\{\psi_j\}$ of $\spkg$ form a countable basis for the space of square integrable functions on $\spac$. We now proceed to compute the SJ modes in ultrastatic slab spacetimes. The retarded Green's function of Eqn.~\eqref{eq:kgus} is known to be given by
\be
G_{R}\left(\eta,\vx;\eta',\vx'\right)= -\theta(\eta-\eta')\sum_{j}\frac{\sin\left(\omega_j(\eta-\eta')\right)}{\omega_j}\psi_j(\vx)\psi^*_j(\vx'), \label{eq:esugr}
\ee
One may note here that since $\spkg$ is real, for any of its complex eigenfunction $\psi_j$, its complex conjugate $\psi^*_j$ is also an eigenfunction of $\spkg$ with same eigenvalue $\omega_j^2$. This leads to a confirmation that the retarded/advanced Green's function in Eqn~\eqref{eq:esugr} is real and also to the fact that it is symmetric under the exchange of $\vx$ and $\vx'$. The same follows for the retarded/advanced Green's function of the Klein-Gordon equation in any conformally related spacetime as well, which we will see in Sec.~\ref{sec:confustatic}. The Pauli-Jordan function defined in Eqn.~\eqref{pj.eq} is then
\be
i\Delta(\eta,\vx;\eta',\vx') = -i\sum_{j}\frac{\sin\left(\omega_j(\eta-\eta')\right)}{\omega_j}\psi_j(\vx)\psi^*_j(\vx')
\ee
 
An ultrastatic slab spacetime $(\R,g)\subset(\M,g)$ is defined by putting bounds on $\eta$, i.e., $\eta\in[-\tau,\tau]$. We start with $u_j(\eta,\vx)\equiv \chi_j(\eta)\psi_j(\vx)$ as an ansatz for the eigenfunction of $i\hd$ in $(\R,g)$. This ansatz along with the orthonormality of $\{\psi_j\}$ on $\spac$ reduces the eigenfunction equation $i\hd\circ u_j=\lambda_j u_j$ to the following mode-wise integral eigenvalue equations for $\chi_j$.
\be
\he_j\circ\chi_j\equiv-i\int_{-\tau}^{\tau}d\eta'\frac{\sin\left(\omega_j(\eta-\eta')\right)}{\omega_j}\chi_j(\eta')=\lambda_j\chi_j(\eta). \label{eq:sjeqnus}
\ee
Now we observe that
\be
\he_j\circ\sin(\omega_j\eta) = iA_j\cos(\omega_j\eta)\quad\text{and}\quad
\he_j\circ\cos(\omega_j\eta) = -iB_j\sin(\omega_j\eta),
\ee
where
\bea
A_j&=\frac{1}{\omega_j}\int_{-\tau}^{\tau} d\eta'\sin^2(\omega_j\eta')& = \frac{\tau}{\omega_j}\left(1-\frac{\sin(2\omega_j\tau)}{2\omega_j\tau}\right),\\\quad B_j&=\frac{1}{\omega_j}\int_{\eta_1}^{\eta_2} d\eta'\cos^2(\omega_j\eta')& =\frac{\tau}{\omega_j}\left(1+\frac{\sin(2\omega_j\tau)}{2\omega_j\tau}\right).
\label{eq:abcus}\eea
This suggests us that $\chi_j(\eta)$ is of the form
\be
\chi_j(\eta)=\sin(\omega_j\eta)+\beta_j\cos(\omega_j\eta)
\ee
\be
\he_j\circ\chi_j(\eta)=-i\beta_jB_j\left(\sin(\omega_j\eta)-\frac{A_j}{\beta_jB_j}\cos(\omega_j\eta)\right)
\ee
For $\chi_j$ to satisfy Eqn~\eqref{eq:sjeqnus}, we need $\beta_j$ and $\lambda_j$ to be
\be
\beta_j=\beta_j^{\pm}\equiv \pm i\sqrt{A_j/B_j} \label{eq:betajus}
\ee
and
\be
\lambda_j=-i\beta_jB_j =\lambda_j^{\pm}\equiv\pm \sqrt{A_jB_j} \label{eq:lambdaus}
\ee
respectively. Positive eigenvalues of $i\hd$ and the corresponding eigenfunctions are $\lambda_j^+$ and
\be
u_j(\eta,\vx)=\left(\sin(\omega_j\eta)+i\sqrt{A_j/B_j}\cos(\omega_j\eta)\right)\psi_j(\vx) \label{eq:sjmodesus}
\ee
respectively with $||u_j||^2=\frac{2\left(A_jB_j-C_j^2\right)}{B_j}\omega_j$. The SJ modes are therefore
\be
\usj_j(\eta,\vx)\equiv\frac{\sqrt{\lambda_j}}{||u_j||}u_j(\eta,\vx)
=\frac{1}{\sqrt{2\omega_j}}\left(\left(\frac{B_j}{A_j}\right)^{1/4}\sin(\omega_j\eta) + i\left(\frac{A_j}{B_j}\right)^{1/4}\cos(\omega_j\eta)\right)\psi_j(\vx).\label{eq:themodesus}
\ee
These modes along with their complex conjugate form a complete basis in $\im(i\hd)$, an argument for this will be presented in Sec.~\ref{sec:completenessus}. The corresponding SJ Wightman function $\wsj = \rew + \frac{i\Delta}{2}$ where
\be
\rew(\eta,\vx;\eta',\vx')=\sum_j\frac{1}{2\omega_j}\left(\sqrt{\frac{B_j}{A_j}}\sin(\omega_j\eta)\sin(\omega_j\eta')+\sqrt{\frac{A_j}{B_j}}\cos(\omega_j\eta)\cos(\omega_j\eta')\right)\psi_j(\vx)\psi_j^*(\vx')
\ee
SJ modes $\usj_j$ obtained in Eqn.~\eqref{eq:themodesus} can be seen to be a mixture of positive and negative frequency modes $\up_j$ and $\un_j$ respectively with respect to the timelike killing vector $\partial_\eta$
\be
\usj_j = c^+_j\up_j + c^-_j\un_j,
\ee
where
\be
c^+_j = \frac{i}{2}\left(\left(\frac{A_j}{B_j}\right)^{1/4}+\left(\frac{B_j}{A_j}\right)^{1/4}\right)\quad\text{and}\quad c^-_j = \frac{i}{2}\left(\left(\frac{A_j}{B_j}\right)^{1/4}-\left(\frac{B_j}{A_j}\right)^{1/4}\right).
\ee
For large enough $\omega_j\tau$, which is the case in either $\R\rightarrow\M$ limit or for high enough frequency modes, we have
\be
A_j\approx B_j\approx \frac{\tau}{\omega_j}\quad \Rightarrow\quad |c_j^+|\approx 1\;\text{and}\;|c_j^-|\approx 0.
\ee
This suggests that the mode mixing in the SJ modes is due to the boundedness in time and in full ultrastatic spacetime, the SJ modes reduces to the purely positive frequency modes and hence the SJ vacuum reduces to the static vacuum. It would be interesting to see if this behaviour is specific to the ultrastatic slab spacetime or is true also for other spacetimes. In the next section we extend this result to a general case of a conformally ultrastatic spacetimes with time dependent conformal factor, but before getting into this we show that the SJ modes obtained in this section indeed form a complete basis in $\im(i\hd)$.

\subsection{Completeness of the SJ modes}\label{sec:completenessus}
In this section we present an argument to show that the eigenfunctions of $i\hd$ obtained in Eqn~\eqref{eq:sjmodesus} form a complete set of eigenfunctions with positive eigenvalue.

If possible, let us assume that apart from those given in Eqn~\eqref{eq:sjmodesus}, there are other eigenfunctions $\{v_{j'}\}$ of $i\hd$ with positive eigenvalues $\{\gamma^+_{j'}\}$ respectively. Let $\wsj'$ be the new Wightman function which takes both $\{u_j\}$ and $\{v_{j'}\}$ modes into account, i.e.,
\be
\wsj'(\bx;\bx')=\wsj(\bx;\bx')+\sum_{j'}\frac{\gamma^+_{j'}}{||v_{j'}||^2}v_{j'}(\bx)v_{j'}^*(\bx')
\ee
where $\wsj$ is the Wightman function which takes only $\{u_j\}$ modes into account. Using the Peierls bracket condition $\wsj'-\wsj'^*=i\hd$, we have
\be
i\hd(\bx;\bx')=\wsj(\bx;\bx)-\wsj^*(\bx;\bx')+\sum_{j'}\frac{\gamma^+_{j'}}{||v_{j'}||^2}\left(v_{j'}(\bx)v_{j'}^*(\bx')-v_{j'}^*(\bx)v_{j'}(\bx')\right). \label{eq:comp1}
\ee
But we know that $\wsj-\wsj^*=i\hd$ which implies that
\be
\sum_{j'}\frac{\gamma^+_{j'}}{||v_{j'}||^2}\left(v_{j'}(\bx)v_{j'}^*(\bx')-v_{j'}^*(\bx)v_{j'}(\bx')\right)=0.\label{eq:comp2}
\ee
Since ${v_{j'}}'s$ are linearly independent, Eqn~\eqref{eq:comp2} holds true if and only if $\gamma^+_{j'}=0$ for all $j'$, and therefore $v_{j'}\not\in \im(i\hd)$ for any $j'$. This means that $\{u_j\}$ in Eqn~\eqref{eq:sjmodesus} forms a complete set of eigenfunctions of $i\hd$ with positive eigenvalues and hence SJ modes obtained here forms a complete basis in $\im(i\hd)$.

\section{SJ vacuum in conformally ultrastatic spacetimes with time dependent conformal factor}\label{sec:confustatic}

In this section we study SJ vacuum in a slab $(\R,\gt)$ of a spacetime $(\M,\gt)$ which is conformally related to the ultrastatic spacetime $(\M,g)$ we studied in Sec.~\ref{sec:ustatic}. Unlike in Sec.~\ref{sec:ustatic}, in this section we restrict ourself to the conformally coupled massless scalar field and therefore have to deal with zero modes which may appear depending on the spacelike hypersurface $\spac$. The spacetime metric $\gt$ corresponds to the invariant length element given by
\be
d\tilde{s}^2=\conf(\eta)^2 ds^2,\label{eq:grwmetric}
\ee
where $ds$ is the invariant length element in $(\M,g)$ and is given by Eqn.~\eqref{eq:usmetric}. $\conf:\M\rightarrow\mr^+$ is a smooth non-vanishing time dependent conformal factor which describes uniform expansion or contraction of the spacetime. Klein-Gordon equation given by Eqn.~\eqref{kg.eq} then takes the form
\be
-\conf(\eta)^{-d}\partial_\eta(\conf(\eta)^{d-2}\partial_\eta\Phi)+\conf(\eta)^{-2}\nabla_i\nabla^i\Phi-\xi \Rt\Phi = 0,
\ee
where $i\in\{1,2,\dots,d-1\}$, $\nabla_i$ is the covariant derivative in $(\M,g)$ and $d$ is the spacetime dimension. Solutions of the above equation can be written as $\Phi = \conf(\eta)^{1-\frac{d}{2}}\Phit$, where $\Phit$ is the solution of the Klein-Gordon equation in $(\M,g)$ i.e.,
\be
(-\partial_\eta^2 + \nabla_i\nabla^i - \xi R)\Phit = 0.
\ee
Spacetime $(\M,\gt)$ admits a timelike conformal Killing vector associated to the Killing vector $\partial_\eta$ in $(\M,g)$, which leads to a preferred mode expansion of $\hPhi$ in terms of conformal positive frequency modes $\uconfp_j\equiv\conf(\eta)^{1-\frac{d}{2}}\up_j$ and conformal negative frequency modes $\uconfn_j\equiv\conf(\eta)^{1-\frac{d}{2}}\un_j$,
\be
\hPhi = \sum_j (\hb_j\uconfp_j + \hb_j^\dagger\uconfn_{-j}).
\ee
The associated vacuum state satisfies $\hb_j\kconf = 0$ and is called as the {\sl{conformal vacuum}}.

We now proceed to compute the SJ modes and the Wightman function for conformally coupled massless scalar field in $(\R,\gt)$ where $\eta\in[\eta_1,\eta_2]$, . As in the previous section, we start with the retarded/advanced Green's function. For any two spacetimes $(\M,g)$ and $(\M,\gt)$ which are conformally related to each other i.e. $\gt_{\mu\nu}=\conf^2 g_{\mu\nu}$, where $\conf:\M\rightarrow\mr^+$ is smooth and non-vanishing, the Pauli-Jordan functions of the Klein-Gordon operator for a conformally coupled massless scalar field are related by
\be
\tilde{\Delta}(\bx;\bx')=\conf(\bx)^{1-\frac{d}{2}}\Delta(\bx;\bx')\conf(\bx')^{1-\frac{d}{2}}, \label{eq:ccgf}
\ee
for all $\bx,\bx'\in\M$, where $\Delta$ and $\tilde{\Delta}$ are Pauli-Jordan functions in $(\M,g)$ and $(\M,\gt)$ respectively. Pauli-Jordan function $i\Delta$ in $(\M,\gt)$ is therefore
\be
\itd(\eta,\vx;\eta',\vx') = -i(\conf(\eta)\conf(\eta'))^{1-\frac{d}{2}}\sum_j\frac{\sin(\omega_j(\eta-\eta'))}{\omega_j}\psi_j(\vx)\psi_j(\vx') \label{eq:itd}
\ee
To get to the SJ modes in $(\R,\gt)$, we follow same steps as in Sec.~\ref{sec:ustatic} with $\itd$ given in Eqn.~\eqref{eq:itd}
\be
\itd\circ f(\eta,\vx) \equiv i\int_{\eta_1}^{\eta_2} \conf(\eta')^d d\eta' \int_\spac dV_{\vx'}\,\tilde{\Delta}(\eta,\vx;\eta',\vx')f(\eta',\vx').
\ee
for any $f:\R\rightarrow \mathbb{C}$. We start with $u_j(\eta,\vx)\equiv \chi_j(\eta)\psi_j(\vx)$ as an ansatz for the SJ modes. The corresponding mode-wise integral equation for $\chi_j$ is
\be
\he_j\circ\chi_j\equiv-i\int_{\eta_1}^{\eta_2}d\eta'\conf(\eta)^{1-\frac{d}{2}}\conf(\eta')^{1+\frac{d}{2}}\frac{\sin\left(\omega_j(\eta-\eta')\right)}{\omega_j}\chi_j(\eta')=\lambda_j\chi_j(\eta). \label{eq:sjeqn}
\ee
Now we observe that
\bea
\he_j\circ\left(\conf(\eta)^{1-\frac{d}{2}}\sin(\omega_j\eta)\right)&=&i\conf(\eta)^{1-\frac{d}{2}}\left(A_j\cos(\omega_j\eta)-C_j\sin(\omega_j\eta)\right)\nonumber\\
\he_j\circ\left(\conf(\eta)^{1-\frac{d}{2}}\cos(\omega_j\eta)\right)&=&i\conf(\eta)^{1-\frac{d}{2}}\left(C_j\cos(\omega_j\eta)-B_j\sin(\omega_j\eta)\right)
\eea
where
\be\begin{split}
&A_j=\frac{1}{\omega_j}\int_{\eta_1}^{\eta_2} d\eta'\conf(\eta')^2\sin^2(\omega_j\eta'),\quad B_j=\frac{1}{\omega_j}\int_{\eta_1}^{\eta_2} d\eta'\conf(\eta')^2\cos^2(\omega_j\eta')\\&\quad\text{and}\quad C_j=\frac{1}{\omega_j}\int_{\eta_1}^{\eta_2}d\eta'\conf(\eta')^2\sin(\omega_j\eta')\cos(\omega_j\eta').
\end{split}\label{eq:abc}\ee
This suggests us an ansatz for $\chi_j(\eta)$ of the form
\be
\chi_j(\eta)=\conf(\eta)^{1-\frac{d}{2}}\left(\sin(\omega_j\eta)+\beta_j\cos(\omega_j\eta)\right).
\ee
\be
\he_j\circ\chi_j(\eta)=-i\left(C_j+\beta_jB_j\right)\conf(\eta)^{1-\frac{d}{2}}\left(\sin(\omega_j\eta)-\frac{A_j+\beta_jC_j}{C_j+\beta_jB_j}\cos(\omega_j\eta)\right).
\ee
For $\chi_j$ to satisfy Eqn~\eqref{eq:sjeqn}, we need $\beta_j$ and $\lambda_j$ to be
\be
\beta_j=\beta_j^{\pm}\equiv\frac{-C_j\pm i\sqrt{A_jB_j-C_j^2}}{B_j} \label{eq:betaj}
\ee
and
\be
\lambda_j=-i\left(C_j+\beta_jB_j\right)=\lambda_j^{\pm}\equiv\pm \sqrt{A_jB_j-C_j^2} \label{eq:lambda}
\ee
respectively. By Cauchy-Schwarz inequality, $A_jB_j\geq C_j^2$ which confirms that $\{\lambda_j\}\subset\mr$. Therefore positive eigenvalues of $i\hd$ and the corresponding eigenfunctions are $\lambda_j^+$ and
\be
u_j(\eta,\vx)=\conf(\eta)^{1-\frac{d}{2}}\left(\sin(\omega_j\eta)+\frac{i\sqrt{A_jB_j-C_j^2}-C_j}{B_j}\cos(\omega_j\eta)\right)\psi_j(\vx) \label{eq:sjmodes}
\ee
respectively with $||u_j||^2=\frac{2\left(A_jB_j-C_j^2\right)}{B_j}\omega_j$. The SJ modes are therefore
\bea
\usj_j(\eta,\vx)&\equiv&\frac{\sqrt{\lambda_j}}{||u_j||}u_j(\eta,\vx)\label{eq:themodes}\\
&=&\frac{\conf(\eta)^{1-\frac{d}{2}}\sqrt{B_j}}{\sqrt{2\omega_j}\left(A_jB_j-C_j^2\right)^{1/4}}\left(\sin(\omega_j\eta)+\frac{i\sqrt{A_jB_j-C_j^2}-C_j}{B_j}\cos(\omega_j\eta)\right)\psi_j(\vx).\nonumber
\eea
We may note here that if we have a zero mode i.e. $\omega_0=0$ then $A_0\rightarrow 0$ and $B_0\rightarrow\infty$ but $P_0^2\equiv A_0B_0$ and $Q_0\equiv \omega_0B_0$ are finite, which makes $\lambda_0^+$ and $u_0(\eta,\vx)$ finite. Zero modes $\usj_0(\eta,\vx)$ takes the form
\be
\usj_0(\eta,\vx)=\frac{\conf(\eta)^{1-\frac{d}{2}}}{\sqrt{2}\left(P_0^2-C_0^2\right)^{1/4}}\left(\sqrt{Q_0}\eta+\frac{i\sqrt{P_0^2-C_0^2}-C_0}{\sqrt{Q_0}}\right)\psi_0(\vx).\label{eq:zeromodes}
\ee
Arguments on the line of that of Sec.~\ref{sec:completenessus} suggests that the SJ modes obtained in Eqn.~\ref{eq:themodes} and \ref{eq:zeromodes} along with their complex conjugates forms a complete basis in $\im(i\hd)$. We also observe that these SJ modes are not merely a conformal transformation of the SJ modes obtained in the corresponding ultrastatic slab spacetime\footnote{This is unlike the conformal modes which are just a conformal transformation of the static modes}. This is because of the non-trivial dependence of the conformal factor $\conf(\eta)$ on $A_j,B_j$ and $C_j$. Similar to the SJ modes in an ultrastatic slab spacetime, non-zero SJ modes $\usj_j$ obtained here in Eqn.~\eqref{eq:themodes} can be seen as a mixture of the positive and negative frequency conformal modes $\uconfp_j$ and $\uconfn_j$,
\be
\usj_j = c^+_j\uconfp_j + c^-_j\uconfn_j
\ee
where
\be
c^+_j = \frac{i\left(\sqrt{A_jB_j-C_j^2}+B_j\right)-C_j}{2\sqrt{B_j}(A_jB_j-C_j^2)^{1/4}}\quad\text{and}\quad c^-_j = \frac{i\left(\sqrt{A_jB_j-C_j^2}-B_j\right)-C_j}{2\sqrt{B_j}(A_jB_j-C_j^2)^{1/4}}.
\ee
For zero modes, i.e., $\omega_0=0$, the SJ modes are well defined but not the corresponding conformal modes. These zero modes are absent when $\spac=\mathbb{S}^{d-1}$ for $d>2$, i.e, in de Sitter spacetimes but are present when $\spac=\mathbb{T}^3$, i.e., in flat FLRW spacetimes. For large enough $\omega_j$ where $\conf(\eta)$ vary slowly with respect to $\sin(\omega_j\eta)$ and $\cos(\omega_j\eta)$, we have the $A_j$'s, $B_j$'s and $C_j$'s given by Eqn.~\eqref{eq:abc} goes to
\be
A_j\text{ and }B_j\approx \frac{1}{2\omega_j}\int_{\eta_1}^{\eta_2}d\eta'\conf(\eta')^2,\quad\text{and}\quad C_j\approx 0,\label{eq:ablargeomega}
\ee
and therefore
\be
|c^+_j|\approx 1,\quad\text{and}\quad |c^-_j|\approx 0.
\ee
which means that the high frequency SJ modes are same as the conformal modes and the difference comes mainly from the low frequency modes. This is expected as the high frequency modes sees only a local region and is unaffected by the global structure of the spacetime. Behaviour of $|c^+_j|$ and $|c^-_j|$ in the full spacetime limit $\R\rightarrow\M$ depends specifically on the conformal factor $\conf(\eta)$. In Sec.~\ref{sec:cuss}, we study this in two particular example of cosmological spacetimes: de Sitter and FLRW spacetimes. 

The SJ Wightman function defined in Eqn~\eqref{eq:wsj} takes the form $\wsj=\rew+\frac{\itd}{2}$ where $\rew$ is the real part of $\wsj$
\bea
\rew(\eta,\vx;\eta',\vx')&=&\sum_{j\in J}\frac{\left(\conf(\eta)\conf(\eta')\right)^{1-\frac{d}{2}}}{2\omega_j\sqrt{A_jB_j-C_j^2}}\left(B_j\sin(\omega_j\eta)\sin(\omega_j\eta')\right.\\
&&\quad\left.+A_j\cos(\omega_j\eta)\cos(\omega_j\eta')-C_j\sin(\omega_j(\eta+\eta'))\right)\psi_j(\vx)\psi_j^*(\vx').\nonumber
\eea
Now we have the form for the SJ Wightman function in spacetimes that are conformally related to an ultrastatic slab spacetimes with compact spatial hypersurface and time dependent conformal factor. In order to further specify the vacuum we now need to specify $\omega_j$ and $\psi_j$ which depends on the choice of the spacelike hypersurface $\spac$, and the constants $A_j,B_j$ and $C_j$ which depends on the choice of the conformal factor $\conf(\eta)$. In Sec.~\ref{sec:cuss}, we will do it explicitly in a finite time slab of de Sitter and flat FLRW spacetimes

\section{The SJ vacuum in cosmological spacetimes}\label{sec:cuss}
\subsection{de Sitter}\label{sec:ds}
de Sitter spacetime is defined as a submanifold of the Minkowski spacetime $\mr^{1,d}$ in one higher dimension. It is described by the following hyperboloid
\be
-x_0^2+\sum_{i=1}^d x_i^2 = \frac{1}{H^2},
\ee
where $x_i's$ are the coordinates of the Minkowski spacetime $\mr^{1,d}$. It is a special case of the class of spacetimes discussed in Sec~\ref{sec:confustatic} for which $\spac=\ms^{d-1}$ with $d$ being the spacetime dimension and $\conf(\eta)=\frac{1}{H\cos(\eta)}$ where $\eta\in(-\pi/2,\pi/2)$. The spacetime metric leads to an invariant length element given by
\be
d\tilde{s}^2=\left(\frac{1}{H\cos(\eta)}\right)^2ds^2,\quad ds^2=-d\eta^2+d\Omega_{d-1}^2.
\ee
$H$ is called as the Hubble's constant and $H^{-1}$ is the only length scale in the metric. The Klein-Gordon equation for a conformally coupled massless scalar field in de Sitter spacetime takes the form
\be
H^{2}\left(-\cos^d(\eta)\partial_\eta\left(\cos^{2-d}(\eta)\partial_\eta\Phi\right)+\cos^2(\eta)\spl\Phi\right)+\xi\Rt\Phi = 0
\ee
where $\spl$ is the spherical Laplacian in $d-1$ dimension. Solutions of the above equation can be written as $\Phi = \left(\frac{1}{H\cos(\eta)}\right)^{1-\frac{d}{2}}\Phit$, where $\Phit$ satisfies
\be
(-\partial_\eta^2+\spl-\xi R)\Phit=0,
\ee
where $R$ is the Ricci scalar of $\ms^{d-1}$. QFT in de Sitter spacetimes has been studied in \cite{Chernikov:1968zm,Mottola:1984ar,Allen:1985ux} and it is known that there exist a two parameter family of de Sitter invariant vacua for a scalar field. They are called as $\alpha$-vacua. These $\alpha$-vacua are alternatives to the Euclidean or Bunch-Davies vacuum, which corresponds to $\alpha=0$ \cite{Bunch:1978yq}. In particular the SJ formalism in de Sitter spacetime has been studied in \cite{Aslanbeigi:2013fga} starting from the Euclidean modes and it has been found that in full de Sitter spacetime, the SJ vacuum for a conformally coupled massless scalar field is in agreement with the Euclidean vacuum in odd spacetime dimensions and is ill defined in even spacetime dimensions. SJ vacuum on a causal set approximated by a slab in de Sitter spacetime has been studied in \cite{Surya:2018byh} and it has been found that the SJ vacuum in a causal set approximated by a slab of  4d de Sitter spacetime is well defined for both minimally and conformally coupled massless scalar field but does not resemble any of the known de Sitter vacua, which is inconsistent with the result of \cite{Aslanbeigi:2013fga}.

In this section, we solve for the SJ modes in a slab of a de Sitter spacetime using the method developed in Sec~\ref{sec:confustatic} and compare the SJ vacuum obtained with the results of \cite{Aslanbeigi:2013fga} and \cite{Surya:2018byh}. In $\ms^{d-1}$ eigenfunctions $\{\psi_j\}$ of the operator $\spkg=-\spl+\xi R$ are the spherical harmonics $Y_{j}$ given by
\be
Y_j\left(\theta_1,\dots,\theta_{d-2},\phi\right)\equiv \frac{1}{\sqrt{2\pi}}e^{il_{d-1}\phi}\prod_{i=1}^{d-2}{_{d-i}\bar{P}_{l_i}^{l_{i+1}}(\theta_i)}
\ee
with
\be
J\equiv\{j\}=\left\{\left(l_1,\dots,l_{d-2},l_{d-1}\right)|l_1,\dots,l_{d-2}\in\mz^+\cup\{0\}, l_{d-1}\in\mz\;\text{and}\; l_1\geq\dots\geq l_{d-2}\geq|l_{d-1}|\right\}
\ee
and
\be
{_a\bar{P}_b^c(z)}\equiv \sqrt{\frac{(a+2b-1)(a+b+c-2)!}{2(b-c)!}}\sin^{1-\frac{a}{2}}(z)P^{1-c-\frac{a}{2}}_{b+\frac{a}{2}-1}\left(\cos(z)\right)
\ee
with $P^\alpha_\beta(z)$ being the associated Legendre functions. The corresponding eigenvalues $\omega_j^2$ are
\be
\omega_j^2=\omega_{l_1}^2\equiv l_1(l_1+d-2)+\xi\Rt = \left(l_1+\frac{d}{2}-1\right)^2. \label{eq:omegal1}
\ee
Note that $\omega_{l_1}$'s are integer in even dimensions and half-integer in odd dimensions. A detailed description of spherical harmonics and its properties can be found in 
\cite{frye:2012sph}. We study the SJ vacuum in a slab of de Sitter spacetime such that $\eta\in[-\tau,\tau]$ i.e. $\eta_1=-\tau$ and $\eta_2=\tau$. In this case we have
\be
A_{l_1}=\frac{1}{H^2\omega_{l_1}}\int_{-\tau}^\tau d\eta'\frac{\sin^2(\omega_{l_1}\eta')}{\cos^2(\eta')},\quad B_{l_1}=\frac{1}{H^2\omega_{l_1}}\int_{-\tau}^\tau d\eta'\frac{\cos^2(\omega_{l_1}\eta')}{\cos^2(\eta')}\quad\text{and }C_{l_1}=0. \label{eq:dsabc}
\ee
Therefore positive eigenvalues and the corresponding eigenfunctions of the Pauli-Jordan operator in de Sitter spacetime reduces to
\be
\lambda_{l_1}^{+}=\sqrt{A_{l_1}B_{l_1}}\quad\text{and}\quad u_j(\eta,\vx)=\left(H\cos(\eta)\right)^{\frac{d}{2}-1}\left(\sin(\omega_{l_1}\eta)+i\sqrt{\frac{A_{l_1}}{B_{l_1}}}\cos(\omega_{l_1}\eta)\right)Y_{j}(\theta_1\dots\theta_{d-2},\phi) \label{eq:dsspec}
\ee
respectively. Since $\lambda_{l_1}^+$ depends on $l_1$ via $\omega_{l_1}$ only, we have degenerate SJ modes with degeneracy $\degn_{l_1,d}$ depending on the spacetime dimensions.
\be
\degn_{l_1,d} = \frac{(2l_1+d-2)(l_1+d-3)!}{l_1!(d-2)!}. \label{eq:degen}
\ee
We now take a look at the dependence of the SJ spectrum $\lambda_{l_1}^+$ of Eqn~\eqref{eq:dsspec} on $\omega_{l_1}$ for different sizes of the slab. Dependence of the SJ spectrum on $l_1$ can then be determined from Eqn~\eqref{eq:omegal1} for a given spacetime dimension.
\begin{figure}[htb]
\centering
\includegraphics[height=5cm]{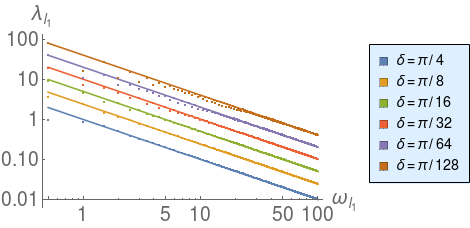}
\caption{A log-log plot of the SJ spectrum $\lambda_{l_1}$ (dotted line) for integer and half-integer values of $\omega_{l_1}$ with $H=1$ and $\tau = \pi/2-\delta$ with different values of $\delta$. Comparison with $\frac{\tan(\tau)}{H^2\omega_{l_1}}$ (solid line)}
\label{fig:dsspec}
\end{figure}
In Fig~\ref{fig:dsspec}, we plot the SJ spectrum for integer and half integer values of $\omega_{l_1}$. In even spacetime dimensions, $\omega_{l_1}$ takes on integer values and in odd spacetime dimensions it takes on half-integer values. We observe that for a fixed size of the slab $\tau<\pi/2$ and for sufficiently large $\omega_{l_1}$, the SJ spectrum of Eqn~\eqref{eq:dsspec} follows $\lambda^+_{l_1} \approx \frac{\tan(\tau)}{H^2(l_1+1)}$. This behaviour of the spectrum is expected, as for large $\omega_{l_1}$ and $\tau<\pi/2$ Eqn~\eqref{eq:dsabc} reduces to
\be
A_{l_1} = B_{l_1} = \frac{1}{H^2\omega_{l_1}}\tan(\tau)\label{eq:ablods}
\ee
as suggested by Eqn.~\eqref{eq:ablargeomega}. The bifurcation in the small $\omega_{l_1}$ region corresponds to the SJ spectrum for integer and half-integer values of $\omega_{l_1}$. The agreement between the SJ spectrum and its large $\omega_{l_1}$ limit gets better with decrease in $\tau$. Larger the $\tau$, larger has to be $\omega_{l_1}$ for the SJ spectrum to agree with its large $\omega_{l_1}$ limit.

The real part of the Wightman function is
\eq{
\rew(\eta,\vx;\eta',\vx')=\frac{1}{2}\left(H^2\cos(\eta)\cos(\eta')\right)^{\frac{d}{2}-1}\sum_{j\in J}\frac{1}{\omega_{l_1}}\left(\sqrt{\frac{B_{l_1}}{A_{l_1}}}\sin(\omega_{l_1}\eta)\sin(\omega_{l_1}\eta')\right.\nonumber\\
\left.+\sqrt{\frac{A_{l_1}}{B_{l_1}}}\cos(\omega_{l_1}\eta)\cos(\omega_{l_1}\eta')\right)Y_{j}(\theta_1,\dots,\theta_{d-2},\phi)Y_{j}^*(\theta_1',\dots,\theta_{d-2}',\phi')
}
The addition theorem of spherical harmonics suggests that for $d>2$
\be
\sum_{l_2,\dots,l_{d-2},m}Y_{j}(\theta_1,\dots,\theta_{d-2},\phi)Y_{j}^*(\theta_1',\dots,\theta_{d-2}',\phi') = \frac{\degn_{l_1,d}}{\Omega_{d-1}}\frac{l_1!(d-3)!}{(l_1+d-3)!} \gp_{l_1,\frac{d}{2}-1}(\cos(\varphi))
\ee
where degeneracy $\degn_{l_1,d}$ is given by Eqn~\eqref{eq:degen}, $\Omega_{d-1}$ is the total solid angle of $\ms^{d-1}$, $\gp_{n,\nu}$ is the Gegenbauer polynomial of degree $n$ with index $\nu$ and $\varphi$ is the angle between unit vectors pointing at $(\theta_1,\dots,\theta_{d-2},\phi)$ and $(\theta'_1,\dots,\theta'_{d-2},\phi')$. 
Therefore the real part of the Wightman function simplifies to
\be
\begin{split}
\rew(\eta,\vx;\eta',\vx')=\frac{1}{2\Omega_{d-1}}\left(H^2\cos(\eta)\cos(\eta')\right)^{\frac{d}{2}-1}\sum_{{l}=0}^{\infty}\frac{\degn_{l,d}}{\omega_{l}}\frac{l!(d-3)!}{(l+d-3)!}\left(\sqrt{\frac{B_{l}}{A_{l}}}\sin(\omega_{l}\eta)\sin(\omega_{l}\eta')\right.\\
\quad\quad\left.+\sqrt{\frac{A_{l}}{B_{l}}}\cos(\omega_{l}\eta)\cos(\omega_{l}\eta')\right)\gp_{l,\frac{d}{2}-1}(\cos(\varphi)).
\end{split}\label{eq:rewds}
\ee
In Eqn~\eqref{eq:rewds}, we have replaced $l_1$ with $l$, which we will continue with from here on. In order to completely determine the contribution of each mode to the SJ Wightman function, we need to look at the behaviour of $\sqrt{A_{l}/B_{l}}$ with $\omega_{l}$ for different $\tau$. We find this to be different in odd and even dimensions as shown in Fig.~\ref{fig:dsa/bo} and \ref{fig:dsa/be}.

\begin{figure}[htb]
\centerline{\begin{tabular}{cc}
\includegraphics[height=4cm]{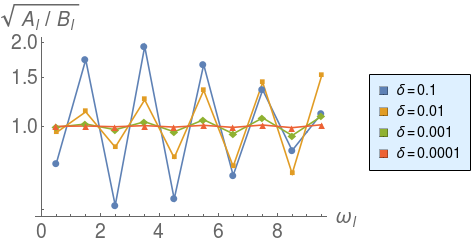} &
\includegraphics[height=4cm]{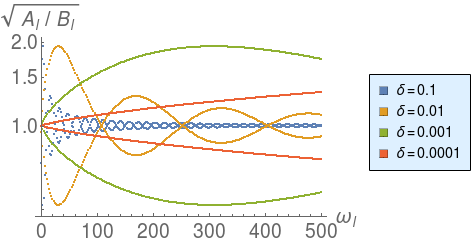} \\
\end{tabular}}
\caption{Log plots of $\sqrt{A_{l}/B_{l}}$ vs $\omega_l$ in odd dimensions for $\tau=\pi/2-\delta$ and $H=1$. In the left, we have the plots for first few $\omega_l$ to show that $\sqrt{A_{l}/B_{l}}$ follows different trends for $\omega_l$ of the form $n+1/2$ with odd and even $n$. In the right, we zoom out to include higher $\omega_l$'s to see how the trend goes.}
\label{fig:dsa/bo}
\end{figure}
In odd dimensions, where $\omega_l$ takes half-integer values, $\sqrt{A_l/B_l}$ follows different trends for $\omega_l$ of the form $(n+1/2)$  with even and odd $n$, but converges to unity in large $\omega_l$ limit as expected from Eqn.~\eqref{eq:ablods}. In the full spacetime limit $\tau\rightarrow\pi/2$, $\sqrt{A_l/B_l}\rightarrow 1$ for any given $\omega_l$. This means that in full spacetime limit, SJ modes agrees with the conformal modes,which is same as Euclidean modes for a conformally coupled massless scalar field\footnote{See chapter~5.4 of \cite{birrell} for a discussion on this.}. This is in agreement with \cite{Aslanbeigi:2013fga}.

\begin{figure}[htb]
\centerline{\begin{tabular}{cc}
\includegraphics[height=4cm]{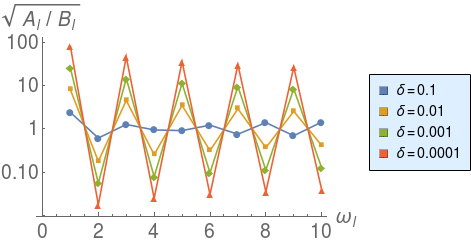} &
\includegraphics[height=4cm]{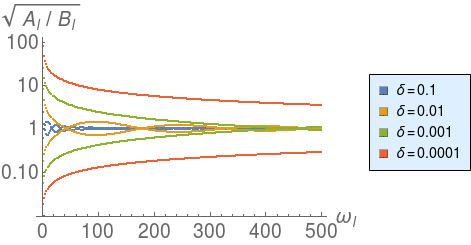} \\
\end{tabular}}
\caption{Log plots of $\sqrt{A_{l}/B_{l}}$ vs $\omega_l$ in even dimensions for $\tau=\pi/2-\delta$ and $H=1$. In the left, we have the plots for first few $\omega_l$ to show that $\sqrt{A_{l}/B_{l}}$ follows different trends for odd and even $\omega_l$. In the right, we zoom out to include higher $\omega_l$'s to see how the trend goes.}
\label{fig:dsa/be}
\end{figure}
In even dimensions, $\omega_l$ takes integer values. $\sqrt{A_l/B_l}$ follows different trends for even and odd $\omega_l$ and converge to unity in large $\omega_l$ limit. In the full spacetime limit $\tau\rightarrow\pi/2$, $\sqrt{A_l/B_l}\rightarrow 0$ for even $\omega_l$ and $\sqrt{A_l/B_l}\rightarrow\infty$ for odd $\omega_l$. In both these cases, the SJ modes are ill defined. This suggests that the SJ vacuum does not exist in even dimension de Sitter spacetimes, which is in agreement with \cite{Aslanbeigi:2013fga}. It will however be interesting to find the closed form for the series in Eqn.~\eqref{eq:rewds} for a given $\tau<\pi/2$ and then to look at the behaviour of $\wsj$ in full de Sitter limit. 
Now we have obtained the SJ vacuum in de Sitter slab spacetime of arbitrary dimensions. We now move on to compare the SJ spectrum we have obtained to that obtained in \cite{Surya:2018byh} in a causal set approximated by a slab of 2d and 4d de Sitter spacetime.

\subsubsection{Comparison with numerical results}
SJ vacuum in a causal set approximated by a slab of 2d and 4d de Sitter spacetime has been studied in detail in \cite{Surya:2018byh}. The result obtained by them is in contrast with what we obtained above. 
We showed that in full de Sitter spacetime $(\tau\rightarrow\pi/2)$ limit, the SJ modes and hence the SJ vacuum for the conformally coupled massless scalar field reduces to the Euclidean modes in odd spacetime dimensions and are not well defined in even spacetime dimensions. On the other hand, in \cite{Surya:2018byh}, the authors showed that the SJ vacuum obtained in a causal set approximated by a de Sitter slab in the full spacetime limit is well defined but does not correspond to any of the known vacua. It is therefore natural to compare the SJ spectrum and modes to look for the source of discrepancy. It is difficult to compare the SJ modes of the continuum and discrete de Sitter, as they are the functions of coordinates but comparing the SJ spectrum in the continuum and discrete case is straightforward. In this section we compare the SJ spectrum for conformally coupled massless scalar field in 2d and 4d de Sitter slab with those obtained by \cite{Surya:2018byh} in a causal set approximated by these spacetimes. We leave the comparison of the SJ modes for future work.

We start with 2d de Sitter spacetime for which we have $\xi=0$,
\be
\psi_j(x) = \frac{1}{\sqrt{2\pi}}e^{ijx}\quad\text{and}\quad \omega_{j} = |j|,
\ee\
where $j\in\mz$ and $x\sim x+2\pi$. We observed in Fig~\ref{fig:dsspec} that for a fixed size of the slab $\tau<\pi/2$ and for sufficiently large $l$, the SJ spectrum of Eqn~\eqref{eq:dsspec} follows $\lambda^+_{j} \approx \frac{\tan(\tau)}{H^2|j|}$, where the degeneracy 
\be
\degn_{|j|,2}=\left\{\begin{tabular}{cc} $2$ & for $|j|\neq 0$ \\ $1$ & for $|j|=0$.\end{tabular}\right.
\ee
\begin{figure}[htb]
\centerline{\begin{tabular}{cc}
\includegraphics[height=3.5cm]{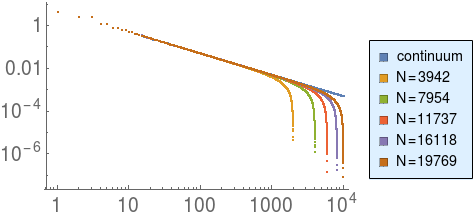} &
\includegraphics[height=3.5cm]{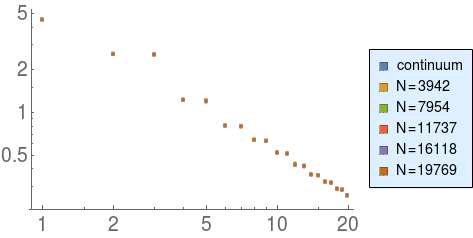}
\end{tabular}}
\caption{In the left we have a log-log plot of the continuum SJ spectrum $\lambda$ and causal set SJ spectrum $\lambda/\rho$ in a slab of 2d de Sitter spacetime with $\tau=1.2$ and $H=1$, where $\lambda$ is the positive eigenvalue of $i\hd$ and $\rho$ is the sprinkling density. In the right we have the same plot but only for the first few eigenvalues. $N$ here denotes the number of causal set elements which is proportional to the sprinkling density $\rho$.}
\label{fig:2ddscontdisc}
\end{figure}
Fig~\ref{fig:2ddscontdisc} shows us discrepancies between the continuum and the causal set SJ spectrum in a slab of 2d de Sitter spacetime. We observe a knee in the causal set SJ spectrum which occurs at smaller eigenvalue for larger $N$. This knee is also observed in the causal set SJ spectrum obtained in chapter~\ref{ch.2ddiamsj} of this thesis and other works \cite{Surya:2018byh}. This effect can clearly be attributed to the natural short distance cut-off in causal set.

Next, we move on to 4d de Sitter spacetime for which we have $\xi=1/6$,
\be
\psi_j(\vx) = Y_{l,l_2,l_3}(\theta_1,\theta_2,\phi)\quad\text{and}\quad \omega_{l} = (l+1),
\ee\
where $Y_{l,l_2,l_3}$ are the spherical harmonics on $\ms^3$ where $l,l_2\in \mz^+\cup\{0\}$ and $l_3\in \mz$ satisfying $l\geq l_2\geq |l_3|$. We observed in Fig~\ref{fig:dsspec} that for a fixed size of the slab $\tau<\pi/2$ and for sufficiently large $l$, the SJ spectrum of Eqn~\eqref{eq:dsspec} follows $\lambda^+_{l} \approx \frac{\tan(\tau)}{H^2\omega_l}$, where the degeneracy 
\be
\degn_{l,4}=\left(l+1\right)^2.
\ee
\begin{figure}[htb]
\centerline{\begin{tabular}{cc}
\includegraphics[height=3.5cm]{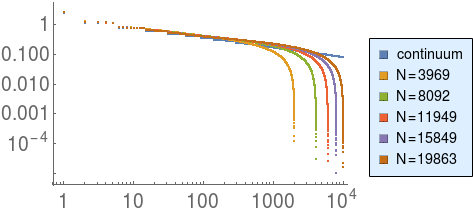} &
\includegraphics[height=3.5cm]{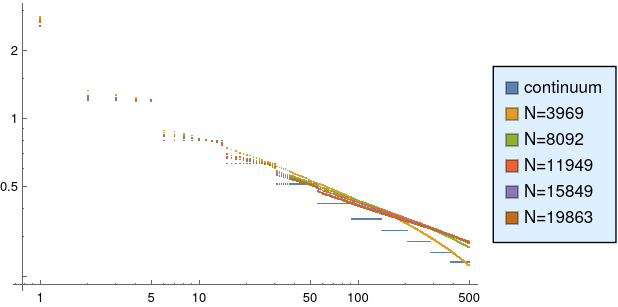}
\end{tabular}}
\caption{In the left we have a log-log plot of the continuum SJ spectrum $\lambda$ and causal set SJ spectrum $\lambda/\rho$ in a slab of 4d de Sitter spacetime with $\tau=1.2$ and $H=1$, where $\lambda$ is the positive eigenvalue of $i\hd$ and $\rho$ is the sprinkling density. In the right we have the same plot but only for the first few eigenvalues. $N$ here denotes the number of causal set elements which is proportional to the sprinkling density $\rho$.}
\label{fig:dscontdisc}
\end{figure}
Fig~\ref{fig:dscontdisc} shows us differences between the continuum and the causal set SJ spectrum in a slab of 4d de Sitter spacetime. We observe that the degeneracy of the continuum SJ spectrum is lifted in the causal set SJ spectrum. This is because of the nature of the sprinkling which doesn't allow causal sets sprinkled on a de Sitter spacetime to be foliated into causal sets sprinkled on spherical hypersurfaces. We also observe a knee in the causal set SJ spectrum which occurs at smaller eigenvalue for larger $N$, similar to the one observed in 2d SJ spectrum in Fig~\ref{fig:2ddscontdisc} which can be attributed to the natural short distance cut-off in causal set. Apart from the knee, in 4d we also observe another difference between the continuum and the discrete SJ spectrum. We see that before the knee, the discrete SJ spectrum instead of going along with the continuum SJ spectrum, goes above it. These features of the spectrum are yet to be understood and will be a subject for future work. 

\subsection{FLRW spacetime}\label{sec:flrw}
Friedmann-Lema{\^i}tre-Robertson-Walker (FLRW) spacetimes is another set of interesting cosmological spacetimes, which are exact solutions of the Einstein equation that describes homogeneous, isotropic and expanding universe. They are described by a metric $g$ leading to the invariant length element given by
\be
ds^2 = a(\eta)^2\left(-d\eta^2 + dr^2+S_\kcurv(r)^2\left(d\theta^2+\sin^2\theta d\phi^2\right)\right),\label{eq:frwsp}
\ee
where $r >0$, $\theta\in[0,\pi]$ and $\phi\sim\phi+2\pi$. $\eta$ is the conformal time which is related to the physical (cosmological) time $t$ as $dt=a(\eta)d\eta$. $S_\kcurv(r)$ is given by
\be
S_\kcurv(r)=\left\{\begin{tabular}{lc}
$\sin(r\sqrt{\kcurv})/\sqrt{\kcurv}$&for $\kcurv>0$\\
$r$&for $\kcurv=0$\\
$\sinh(r\sqrt{|\kcurv|})/\sqrt{|\kcurv|}$&for $\kcurv<0$
\end{tabular}\right.
\ee
$\kcurv = -1,0,1$ for negative, zero and positive curvature respectively. $a(\eta)$ is determined by solving the Friedmann equations Eqn.~\eqref{eq:fr} and are different for different type of matter content of the universe.
\be
\left(\frac{a'}{a^2}\right)^2 = \frac{8\pi G}{3}\rho +\frac{\Lambda}{3} - \frac{\kcurv}{a^2}\quad\text{and}\quad\frac{1}{a^2}\left(\frac{a''}{a}-\left(\frac{a'}{a}\right)^2\right) = -\frac{4\pi G}{3}(\rho+3p) +\frac{\Lambda}{3}, \label{eq:fr}
\ee
where $a'\equiv da/d\eta$, $G$ is the universal gravitational constant, $\rho$ is the matter and/or radiation density, $p$ is the pressure and $\Lambda$ is the cosmological constant. We have $p=0$ for matter (non-relativistic) and $p=\rho/3$ for radiation. 
We study here the specific case of a flat universe i.e., $\kcurv=0$, for which we have $S_{\kcurv}(r) = r$, which implies a flat spatial slice on which the length element can simply be written in terms of Cartesian coordinates as $dx_1^2+dx_2^2+dx_3^2$ and therefore Eqn.~\eqref{eq:frwsp} can be written as
\be
ds^2 = a(\eta)^2\left(-d\eta^2 + dx_1^2 + dx_2^2 + dx_3^2\right), \label{eq:frwflat}
\ee
Klein-Gordon equation for a conformally coupled massless scalar field in flat FLRW spacetimes is given by
\be
-a(\eta)^{-d}\partial_\eta(a(\eta)^{d-2}\partial_\eta\Phi)+a(\eta)^{-2}\sum_{i=1}^3\partial_i^2\Phi - \xi R\Phi = 0.
\ee
The solutions of the above equation can be written as $\Phi=a(\eta)^{1-\frac{d}{2}}\Phit$, where $\Phit$ is the solution of the Klein-Gordon equation in 4d Minkowski spacetime i.e., $\partial_\mu\partial^\mu\Phit=0$, where $\mu\in\{0,1,2,3\}$.
\vskip 0.1in
The SJ formalism for QFT is well defined in a bounded spacetime and therefore we compactify the flat spatial hypersurface $\spac$ into a 3-torus $\mt^3$ of radius $L$. This compactification leaves $\spac$ non-isomorphic globally but local properties remain the same as that of the original $\mr^3$. This ensures that the modified spacetime is the solution of the Einstein equation with same stress energy momentum tensor as the original FLRW spacetime. In the rest of this section we work with the dimensionless coordinates $\frac{\eta}{L}\rightarrow\eta$ and $\frac{x_i}{L}\rightarrow x_i$. Solving Eqn.~\eqref{eq:fr} with $\kcurv=0$ for $a(\eta)$, we find that
\be
a(\eta)=\left\{\begin{tabular}{llc}
$\frac{L^2\eta}{2\alpha_r},$&$0\leq\eta<\infty$ &Radiation dominated universe\\\\
$\frac{L^3\eta^2}{9\alpha_m^2},$&$0\leq\eta<\infty$ &Matter dominated universe\\\\
$\frac{1}{H|\eta|},$&$-\infty<\eta<0$ &$\Lambda$ dominated universe,
\end{tabular}\right.
\ee
where $H=\sqrt{\Lambda/3}$. $\alpha_m$ and $\alpha_r$ are constants with a dimension of length.

SJ vacuum in an unbounded flat FLRW spacetimes in the radiation era has been studied in \cite{aas}. In this section, we study them in a slab $\eta\in[\eta_1,\eta_2]$ of a flat FLRW spacetime with $\spac=\mt^3$. In $\mt^3$, eigenfunctions $\{\psi_j\}$ of the operator $\spkg = -\sum_{i=1}^3\partial_i^2$ are given by
\be
\psi_{j}(\vx)=\frac{1}{(2\pi)^{3/2}}\exp\left(i \vec{j}.\vx\right)\quad\text{where}\; \vec{j}\in\mz^3
\ee
and the corresponding eigenvalues are
\be
\omega_j^2 = |\vec{j}|^2. \label{eq:flatomega}
\ee
The SJ spectrum and SJ modes are then given by Eqn~\eqref{eq:lambda} and Eqn~\eqref{eq:sjmodes} respectively with $\Omega(\eta)=a(\eta)$ and constants $A_j,B_j$ and $C_j$ which are found to be as follows.

For the radiation dominated universe, we have
\bea
A_j &=& \frac{L^4}{4\alpha_r^2\omega_j}\left(\frac{\eta_2^3-\eta_1^3}{6}-\frac{\eta_2^2\sin(2\omega_j\eta_2)-\eta_1^2\sin(2\omega_j\eta_1)}{4\omega_j}-\frac{\eta_2\cos(2\omega_j\eta_2)-\eta_1\cos(2\omega_j\eta_1)}{4\omega_j^2}\right.\nonumber\\
&&\quad\quad\left.+\frac{\sin(2\omega_j\eta_2)-\sin(2\omega_j\eta_1)}{8\omega_j^3}\right)\nonumber\\
B_j &=& \frac{L^4}{4\alpha_r^2\omega_j}\left(\frac{\eta_2^3-\eta_1^3}{6}+\frac{\eta_2^2\sin(2\omega_j\eta_2)-\eta_1^2\sin(2\omega_j\eta_1)}{4\omega_j}+\frac{\eta_2\cos(2\omega_j\eta_2)-\eta_1\cos(2\omega_j\eta_1)}{4\omega_j^2}\right.\nonumber\\
&&\quad\quad\left.-\frac{\sin(2\omega_j\eta_2)-\sin(2\omega_j\eta_1)}{8\omega_j^3}\right)\\
C_j &=& \frac{L^4}{4\alpha_r^2\omega_j^2}\left(-\frac{\eta_2^2\cos(2\omega_j\eta_2)-\eta_1^2\cos(2\omega_j\eta_1)}{4}+\frac{\eta_2\sin(2\omega_j\eta_2)-\eta_1\sin(2\omega_j\eta_1)}{4\omega_j}\right.\nonumber\\
&&\quad\quad\left.+\frac{\cos(2\omega_j\eta_2)-\cos(2\omega_j\eta_1)}{8\omega_j^2}\right)\nonumber,
\eea
for the matter dominated universe, we have
\bea
A_j &=& \frac{L^6}{81\alpha_m^4\omega_j}\left(\frac{\eta_2^5-\eta_1^5}{10}-\frac{\eta_2^4\sin(2\omega_j\eta_2)-\eta_1^4\sin(2\omega_j\eta_1)}{4\omega_j}-\frac{\eta_2^3\cos(2\omega_j\eta_2)-\eta_1^3\cos(2\omega_j\eta_1)}{2\omega_j^2}\right.\nonumber\\
&&\left.+\frac{\eta_2^2\sin(2\omega_j\eta_2)-\eta_1^2\sin(2\omega_j\eta_1)}{\frac{4}{3}\omega_j^3}+\frac{\eta_2\cos(2\omega_j\eta_2)-\eta_1\cos(2\omega_j\eta_1)}{\frac{4}{3}\omega_j^4}-\frac{\sin(2\omega_j\eta_2)-\sin(2\omega_j\eta_1)}{\frac{8}{3}\omega_j^5}\right)\nonumber\\
B_j &=& \frac{L^6}{81\alpha_m^4\omega_j}\left(\frac{\eta_2^5-\eta_1^5}{10}+\frac{\eta_2^4\sin(2\omega_j\eta_2)-\eta_1^4\sin(2\omega_j\eta_1)}{4\omega_j}+\frac{\eta_2^3\cos(2\omega_j\eta_2)-\eta_1^3\cos(2\omega_j\eta_1)}{2\omega_j^2}\right.\nonumber\\
&&\left.-\frac{\eta_2^2\sin(2\omega_j\eta_2)-\eta_1^2\sin(2\omega_j\eta_1)}{\frac{4}{3}\omega_j^3}-\frac{\eta_2\cos(2\omega_j\eta_2)-\eta_1\cos(2\omega_j\eta_1)}{\frac{4}{3}\omega_j^4}+\frac{\sin(2\omega_j\eta_2)-\sin(2\omega_j\eta_1)}{\frac{8}{3}\omega_j^5}\right)\nonumber\\
C_j &=& \frac{L^6}{81\alpha_m^4\omega_j^2}\left(-\frac{\eta_2^4\cos(2\omega_j\eta_2)-\eta_1^4\cos(2\omega_j\eta_1)}{4}+\frac{\eta_2^3\sin(2\omega_j\eta_2)-\eta_1^3\sin(2\omega_j\eta_1)}{2\omega_j}\right.\nonumber\\
&&\left.+\frac{\eta_2^2\cos(2\omega_j\eta_2)-\eta_1^2\cos(2\omega_j\eta_1)}{\frac{4}{3}\omega_j^2}-\frac{\eta_2\sin(2\omega_j\eta_2)-\eta_1\sin(2\omega_j\eta_1)}{\frac{4}{3}\omega_j^3}-\frac{\cos(2\omega_j\eta_2)-\cos(2\omega_j\eta_1)}{\frac{8}{3}\omega_j^4}\right),\nonumber\\
\eea
and for the $\Lambda$ dominated universe, we have
\bea
A_j &=& \frac{1}{H^2\omega_j}\left(\frac{\sin^2(\omega_j\eta_1)}{\eta_1}-\frac{\sin^2(\omega_j\eta_2)}{\eta_2}-\omega_j\left(si(2\omega_j\eta_1)-si(2\omega_j\eta_2)\right)\right)\nonumber\\
B_j &=& \frac{1}{H^2\omega_j}\left(\frac{\cos^2(\omega_j\eta_1)}{\eta_1}-\frac{\cos^2(\omega_j\eta_2)}{\eta_2}+\omega_j\left(si(2\omega_j\eta_1)-si(2\omega_j\eta_2)\right)\right)\\
C_j &=&\frac{1}{H^2\omega_j}\left(\frac{\sin(2\omega_j\eta_1)}{2\eta_1}-\frac{\sin(2\omega_j\eta_2)}{2\eta_2}+\omega_j\left(ci(2\omega_j|\eta_1|)-ci(2\omega_j|\eta_2|)\right)\right),\nonumber
\eea
where
\be
si(z) \equiv \int_0^z dx\frac{\sin(x)}{x}\quad\text{and}\quad ci(z) \equiv -\int_z^\infty dx\frac{\cos(x)}{x}.
\ee
In particular for zero modes we have,

for the radiation dominated universe,
\be
\begin{split}
A_0 = \frac{L^4\omega_0}{20\alpha_r^2}\left(\eta_2^5-\eta_1^5\right),\quad B_0 = \frac{L^4}{12\alpha_r^2\omega_0}\left(\eta_2^3-\eta_1^3\right),\quad C_0 = \frac{L^4}{16\alpha_r^2}\left(\eta_2^4-\eta_1^4\right)\\
P_0^2\equiv A_0B_0 = \frac{L^8}{240\alpha_r^4}\left(\eta_2^5-\eta_1^5\right)\left(\eta_2^3-\eta_1^3\right), \quad Q_0\equiv\omega_0B_0 = \frac{L^4}{12\alpha_r^2}\left(\eta_2^3-\eta_1^3\right),
\end{split}
\ee

for the matter dominated universe,
\be
\begin{split}
A_0 = \frac{L^6\omega_0}{567\alpha_m^4}\left(\eta_2^7-\eta_1^7\right),\quad B_0 = \frac{L^6}{405\alpha_m^4\omega_0}\left(\eta_2^5-\eta_1^5\right),\quad C_0 = \frac{L^6}{486\alpha_m^4}\left(\eta_2^6-\eta_1^6\right)\\
P_0\equiv A_0B_0 = \frac{L^{12}}{229635\alpha_m^8}\left(\eta_2^7-\eta_1^7\right)\left(\eta_2^5-\eta_1^5\right), \quad Q_0\equiv\omega_0B_0 = \frac{L^6}{405\alpha_m^4}\left(\eta_2^5-\eta_1^5\right).
\end{split}
\ee

and for the $\Lambda$ dominated universe,
\be
\begin{split}
A_0 = \frac{\omega_0}{H^2}\left(\eta_2-\eta_1\right),\quad B_0 = \frac{1}{H^2\omega_0}\left(\eta_1^{-1}-\eta_2^{-1}\right),\quad C_0 = \frac{1}{H^2}\left(1+\gamma+\log\left|\frac{\eta_1}{\eta_2}\right|\right)\\
P_0\equiv A_0B_0 = \frac{1}{H^4}\left(\eta_2-\eta_1\right)\left(\eta_1^{-1}-\eta_2^{-1}\right), \quad Q_0\equiv \omega_0B_0 = \frac{1}{H^2}\left(\eta_1^{-1}-\eta_2^{-1}\right),
\end{split}
\ee
where $\gamma=0.577$ is the Euler constant.

Now we look at the behaviour of the SJ spectrum given by Eqn~\eqref{eq:lambda} in flat FLRW spacetime. Eqn~\eqref{eq:flatomega} suggests that the SJ spectrum is degenerate where for large enough $\omega_j$ degeneracy goes as $4\pi\omega_j^2$.
\begin{figure}[htb]
\centerline{\begin{tabular}{cc}
\includegraphics[height=4cm]{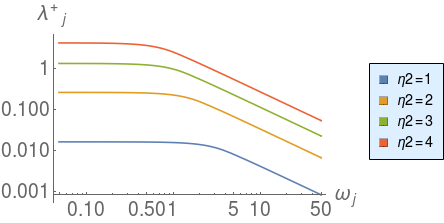} &
\includegraphics[height=4cm]{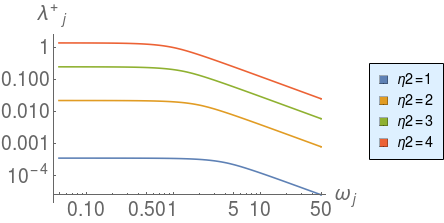}\\
(a) & (b)
\end{tabular}}
\caption{SJ eigenvalues $\lambda_j^+$ vs $\omega_j$ for $\eta_1=0$ and different values of $\eta_2$ with $L=\alpha=1$. (a) Radiation dominated universe, (b) Matter dominated universe}
\label{fig:mrlambda}
\centerline{\includegraphics[height=4cm]{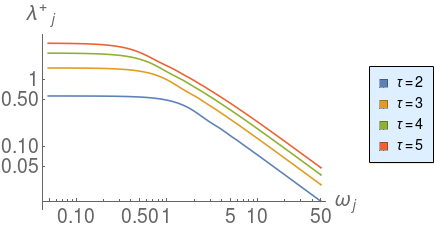}}
\caption{SJ eigenvalues $\lambda_j^+$ vs $\omega_j$ for a slab of $\Lambda$ dominated universe with $\eta_1=-\tau$ and $\eta_2=-1/\tau$ for $H=1$ and different values of $\tau$}
\label{fig:delambda}
\end{figure}

Fig~\ref{fig:mrlambda} and \ref{fig:delambda} suggests for large $\omega_j$, SJ eigenvalues $\lambda_j^+\propto\omega_j^{-1}$ where the proportionality constant depends on the conformal factor $a(\eta)$ and the temporal boundary $\eta_1$ and $\eta_2$. This behaviour is same as the one observed in de Sitter SJ spectrum in Fig~\ref{fig:dsspec}.

SJ modes and the real part of the Wightman function takes the form
\be
u_j(\eta,\vx)=\frac{a(\eta)^{-1}}{(2\pi)^{3/2}}\left(\sin(\omega_j\eta)+\frac{i\sqrt{A_jB_j-C_j^2}-C_j}{B_j}\cos(\omega_j\eta)\right)e^{i\vec{j}.\vx}
\ee
and
\be
\begin{split}
\rew(\eta,\vx;\eta',\vx')=\sum_{\vec{j}\in\mz^3}\frac{(a(\eta)a(\eta'))^{-1}}{16\pi^3\omega_j\sqrt{A_jB_j-C_j^2}}\left(B_j\sin(\omega_j\eta)\sin(\omega_j\eta')\right.\\
\quad\quad\left.+A_j\cos(\omega_j\eta)\cos(\omega_j\eta')-C_j\sin(\omega_j(\eta+\eta'))\right)e^{i\vec{j}.(\vx-\vx')}
\end{split}
\ee
respectively. In order to determine the contribution of each mode to the SJ Wightman function, we need to look at the behaviour of $\frac{A_j}{\sqrt{A_jB_j-C_j^2}}$, $\frac{B_j}{\sqrt{A_jB_j-C_j^2}}$ and $\frac{C_j}{\sqrt{A_jB_j-C_j^2}}$ with $\omega_j$.
\begin{figure}[htb]
\centerline{\begin{tabular}{ccc}
\includegraphics[height=2.6cm]{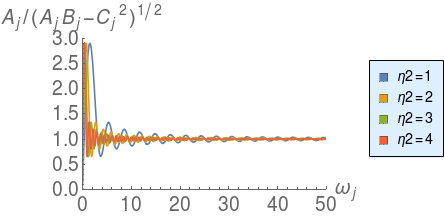} &
\includegraphics[height=2.6cm]{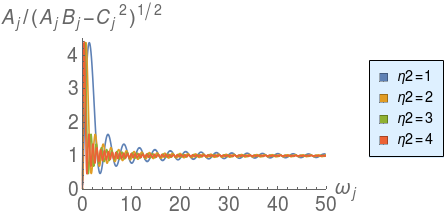} &
\includegraphics[height=2.6cm]{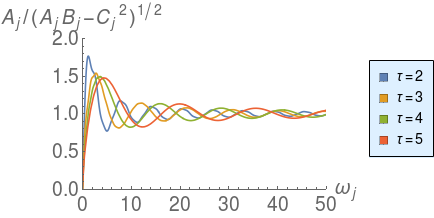} \\
\includegraphics[height=2.6cm]{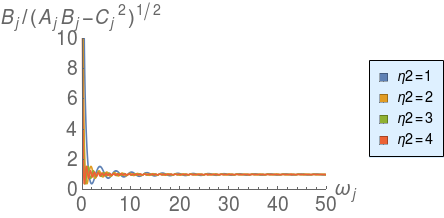} &
\includegraphics[height=2.6cm]{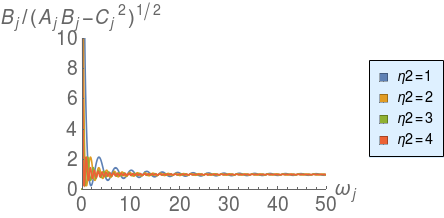} &
\includegraphics[height=2.6cm]{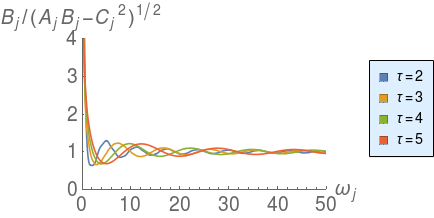} \\
\includegraphics[height=2.6cm]{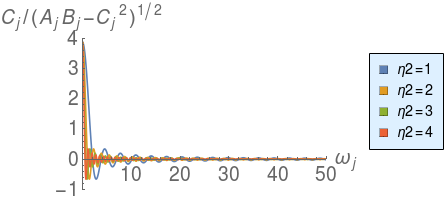} &
\includegraphics[height=2.6cm]{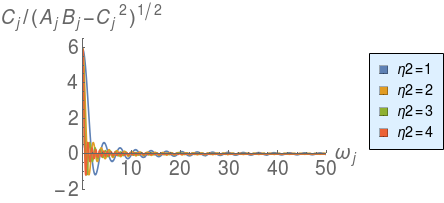} &
\includegraphics[height=2.6cm]{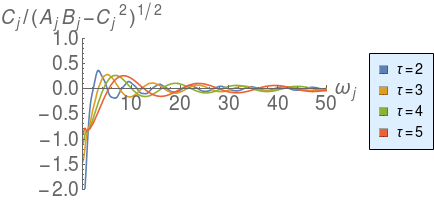}
\end{tabular}}
\caption{$\frac{A_j}{\sqrt{A_jB_j-C_j^2}}$, $\frac{B_j}{\sqrt{A_jB_j-C_j^2}}$ and $\frac{C_j}{\sqrt{A_jB_j-C_j^2}}$ vs $\omega_j$. $\eta_1=0$ and different values of $\eta_2$ with $L=\alpha=1$ for radiation (left) and matter (center) dominated universe. $\eta_1=-\tau$ and $\eta_2=-1/\tau$ for different values of $\tau$ with $H=1$ for $\Lambda$ dominated universe (right).}
\label{fig:flrwconst}
\end{figure}
In Fig~\ref{fig:flrwconst}, we see that for large enough $\omega_j$, the constants $\frac{A_j}{\sqrt{A_jB_j-C_j^2}}$, $\frac{B_j}{\sqrt{A_jB_j-C_j^2}}$ converges to one and $\frac{C_j}{\sqrt{A_jB_j-C_j^2}}$ converges to zero. The convergence gets better with increase in size of the slab. Fig~\ref{fig:flrwconst} shows only non zero modes and should be ignored for the zero mode $\omega_0$.

For large enough temporal slab i.e. $\eta_2-\eta_1>>1$ where except for first few $\omega_j$ we have $\frac{A_j}{\sqrt{A_jB_j-C_j^2}} \approx \frac{B_j}{\sqrt{A_jB_j-C_j^2}} \approx 1$ and $\frac{C_j}{\sqrt{A_jB_j-C_j^2}} \approx 0$, we can write $\rew$ as
\be
\rew(\eta,\vx;\eta',\vx') = \sum_{\vec{j}\in\mz^3}\frac{(a(\eta)a(\eta'))^{-1}}{2(2\pi)^3\omega_j}\cos(\omega_j(\eta-\eta'))e^{i\vec{j}(\vx-\vx')}+\epsilon(\eta,\vx;\eta'\vx'). \label{eq:hsjasymp}
\ee
This expression for $\rew$ in Eqn~\eqref{eq:hsjasymp} for $\eta_2-\eta_1>>1$ is that of the conformal vacuum 
for conformally coupled scalar field theory in FLRW spacetime up to the correction term $\epsilon$ whose contribution comes from the first few modes, in particular zero modes.

Here we studied the SJ vacuum in flat FLRW spacetime for the radiation, matter or $\Lambda$ dominated universe. However, we know that the SJ vacuum is global in nature, and therefore the fact that there is a transition from the radiation dominated to the matter dominated phase and from the matter dominated to the $\Lambda$ dominated phase will have an effect on the SJ vacuum in all of these individual eras. It would be interesting to study the effect of these transitions on the SJ vacuum.

\section{Discussions}\label{sec.discch3}
In this chapter, we studied the SJ vacuum for a conformally coupled massless scalar field in cosmological spacetimes. These are conformally ultrastatic spacetimes with a time dependent conformal factor. In particular we studied the SJ vacuum in de Sitter and flat FLRW spacetimes, which are of cosmological interest. In order to ensure the finiteness of the spacetimes and hence the boundedness of the integral operators, we performed our calculations in a bounded time slab of these spacetimes. The calculations of the modes in this chapter follows from that of Fewster and Verch \cite{fewster2012} in ultrastatic slab spacetimes and Brum and Fredenhagen \cite{Brum:2013bia} in expanding spacetimes. In flat FLRW, we further compactified the spatial hypersurface into a 3-torus of length $L$ to ensure the finite volume of the spacetime.

We observe that the SJ vacuum in a conformally ultrastatic slab spacetimes depends on the conformal factor non-trivially, i.e., it cannot be obtained by the conformal transformation of the SJ vacuum in the corresponding ultrastatic slab spacetime. However in the full spacetime limit, the SJ vacuum reduces to the conformal vacuum in odd dimensional de Sitter spacetimes and in flat FLRW spacetimes, whereas in even dimensional de Sitter spacetimes it turns out to be ill defined. This confirms the result of \cite{Aslanbeigi:2013fga}, where the SJ vacuum is evaluated directly in the full de Sitter spacetime. Our results, along with that of \cite{Aslanbeigi:2013fga}, disagrees with that of \cite{Surya:2018byh}, where it is observed that the SJ vacuum obtained in a causal set approximated by a 2d and 4d de Sitter slab spacetime converges with the size of the slab but does not resemble any of the known de Sitter vacua. As an attempt to understand this, we compare the SJ spectrum obtained here with that obtained by the authors of \cite{Surya:2018byh} in a causal set approximated by 2d and 4d de Sitter slab spacetime for different sprinkling density. We see that in 2d de Sitter, the causal set and the continuum SJ spectrum agrees with each other upto a UV cut-off determined by the sprinkling density of the causal set. In 4d however, there is a slight disagreement between the continuum and the causal set spectrum even before the natural cut-off. The causal set spectrum goes slightly above the continuum one. This mismatch of the SJ spectrum in the 4d continuum and the causal set de Sitter is yet to be understood. One would also like to compare the SJ modes in de Sitter slab spacetime with the one obtained in a causal set sprinkled on the de Sitter slab. However doing this is computationally very challenging.

In flat FLRW spacetime, we studied the SJ vacuum in the radiation, matter and $\Lambda$ dominated era. We considered a finite time slab of the spacetime in these era and found that in the limit of the infinite slab, the SJ vacuum in each of them reduces to the corresponding conformal vacuum. Though we studied the SJ vacuum in these eras separately, given the global nature of the SJ vacuum, it would however be interesting to study the SJ vacuum in the real universe, in which these eras are stitched together along with the inflationary phase of the universe. To do this, one need to have a better understanding of the transition from one era to another. We leave this for the future work. In the coming chapters we will instead look at another interesting aspect of the observer independent QFT in curved spacetime, which is the spacetime formulation of the entanglement entropy by Sorkin \cite{ssee}.

\chapter{Spacetime entanglement entropy of de Sitter and black hole horizons}\label{ch.sseeds}
In this chapter, we look at an important quantity of interest in the study of QFT in curved spacetimes, which is the entanglement entropy. 
 We study the spacetime formulation of the entanglement entropy proposed by Sorkin in \cite{ssee}, which is determined in a subregion $\cO$ of the spacetime manifold $\cM$ using the Wightman function of a pure state in $\cM$ restricted to $\cO$ and the Pauli-Jordan function. A brief overview of the entanglememt entropy including the Sorkin's Spacetime Entanglement Entropy (SSEE) can be found in chapter~\ref{ch.int} of the thesis.

We present an analytic calculation of the SSEE for de Sitter horizons for all $d>2$ for a massive scalar field with effective mass $\mass=\sqrt{m^2+\xi R}$, where $R$ is the Ricci scalar. The importance of studying all types of horizons was pointed out by Jacobson and Parentani~\cite{Jacobson:2003wv} who  showed that thermality and the area law are features of all causal horizons. Cosmological horizons in de Sitter (dS) spacetime are known to have thermodynamic properties similar to their black hole counterparts even though these horizons are observer dependent~\cite{gibbons}. Because of the relative simplicity of these spacetimes, they provide a useful arena to test new proposals for calculating the entanglement entropy.  Our
calculation uses the restriction of the  Bunch-Davies vacuum in the Poincare or conformal patch of de Sitter to the
static patch. Even though $\cO$ is non-compact in the time direction, we show that the generalised eigenvalue equation Eqn.~(\ref{ssee.eq}) can be explicitly solved mode  by mode. We find that  the SSEE 
is independent of the effective mass, which is
in agreement with  the results of Higuchi and Yamamoto \cite{Higuchi:2018tuk} but differs from the
result of \cite{Maldacena:2012xp,Kanno:2014lma,Iizuka:2014rua} where the  entanglement entropy was evaluated on  
the spacelike hypersurface close to the future de Sitter boundary and  found to be mass dependent.
The total SSEE can be calculated using a UV cut-off
in the {angular modes} for the Bunch-Davies vacuum and is therefore proportional to the regularised  area of the horizon.  The
other $\alpha$ vacua however need an additional momentum cut-off.

The obvious  generalisation of our calculation to static black hole
and Rindler horizons is hampered by the spatial non-compactness, except in the case of Schwarzschild de Sitter
black holes. For these spacetimes, the explicit form of the modes is not known, except in $d=2$. In \cite{Qiu:2019qgp}
certain natural boundary conditions for massless minimally coupled modes were used to analyse the thermodynamic
properties of these horizons. We employ these  same boundary conditions to find the mode-wise form for the SSEE in the
static region. Introducing the cut-off in the angular modes again gives us the requisite area dependence. 

In the special case of $d=2$, the calculation can be performed explicitly, and we find that 
the SSEE is constant for both the black  hole as well as the cosmological  horizon. Thus we do not find the logarithmic
behaviour expected from the Calabrese-Cardy formula. A key difference is that  in the earlier calculations,  $\cO$ is compact and the mixed state in $\cO$ is
{\it not}  diagonal with respect to the (Sorkin-Johnston) modes in $\cO$. 

We organise this chapter as follows. In Sec.~\ref{gf.sec} we lay out the general framework for the calculation of the
mode-dependent SSEE  for a compact region $\cO $ with respect to  a vacuum state in $\cM \supset \cO$. We find the solutions to the generalised
eigenvalue equation Eqn.~\eqref{ssee.eq} when the modes in $\cO$ are $\cL^2$ orthogonal, and the Bogoliubov
coefficients satisfy certain conditions. We then show that Eqn.~\eqref{ssee.eq} is also well posed  for static spherically symmetric spacetimes with finite
spatial extent. Assuming that the restricted vacuum $W\big|_\cO$ is block diagonal in the modes in $\cO$ we find the general
form of the mode-wise  SSEE. In Sec.~\ref{dS.sec} we review some basics of de Sitter and  Schwarzschild de Sitter
spacetimes. In Sec.~\ref{dSssee.sec} we apply the analysis of Sec.~\ref{gf.sec} to the static patches of $d=4$ de Sitter, starting with  the Bunch-Davies vacuum  in
the conformal patch. Using an angular cut-off we show that the SSEE is proportional to the regularised de Sitter horizon
area. We also extend this calculation of the SSEE to other $\alpha$ vacua in the conformal patch. We find that while the
mode-wise SSEE is still independent of the effective mass, the total SSEE needs an additional cut-off in the radial
momentum. In Sec.~\ref{bh.sec} we  calculate the SSEE for a massless minimally coupled scalar field in  the static patches of Schwarzschild de Sitter spacetimes for $d>2$ using the boundary conditions of  \cite{Qiu:2019qgp}.  An explicit calculation of the $d=2$ case  then follows. 
We discuss the implications of our results  in Sec.~\ref{discussion.sec}. In Appendix~\ref{ddimds.sec} we extend the $d=4$ analysis to all dimensions $d > 2$. 

\section{The SSEE: General Features} 
\label{gf.sec}

In this section we examine the SSEE generalised eigenvalue equation Eqn~\eqref{ssee.eq} using the two sets of modes in
the  regions $\cM, \cO$, where  $\cO \subset \cM$. We show that
when the modes in the subregion $\cO$ are $\mathcal L^2$ orthogonal, and the   Bogoliubov
transformations satisfy certain conditions, it is possible to find the general form for the SSEE. While not entirely
general, this covers a fairly wide  range of cases.

Let  $\{ \bv_{\bk}\} $ be the Klein-Gordon (KG) orthonormal modes in $(\cM,g)$, i.e.,
\eq{
(\bv_{\bk},\bv_{\bk'})_{\cM} = - (\bv_{\bk}^*,\bv_{\bk'}^*)_{\cM} = \delta_{\bk\bk'}\;\;\text{and}\,\,(\bv_{\bk},\bv_{\bk'}^*)_{\cM}=0,\label{kgortho.eq}
}
and $\{\bu_{\bpp}\} $ be those in the globally hyperbolic region $\cO \subset \cM$. Here $(.,.)_{\cM}$ denotes the KG inner product, which is defined in Eqn.~\eqref{kgip.eq}. The corresponding Wightman function in $(\cM,g)$ is
\eq{
  W(\bx,\bx')=\sum_\bk \bv_{\bk}(\bx)\bv_{\bk}^*(\bx').\label{w.eq}
}
Since $\{\bu_{\bpp}\}$ forms a complete KG orthonormal basis in $\cO$, the restriction of $\bv_\bk$ to $\cO$ can be expressed as a linear combination of $\bu_\bpp$ modes, i.e.,
\eq{
\bv_\bk(\bx)\Big|_\cO = \sum_\bpp \left(\alpha_{\bk\bpp}\bu_\bpp(\bx) + \beta_{\bk\bpp}\bu_\bpp^*(\bx)\right),
}
where $\alpha_{\bk\bpp} = (\bu_\bpp,\bv_\bk)_\cO$ and $\beta_{\bk\bpp} = -(\bu_\bpp^*,\bv_\bk)_\cO$. The restriction
of $W(\bx,\bx')$ to  $\cO$ can thus be re-expressed in terms of $\{ \bu_{\bpp}\}$ as 
\eq{
W(\bx,\bx')\Big|_\cO  =  \sum_{\bpp\bpp'}\Big(A_{\bpp\bpp'}\bu_{\bpp}(\bx)\bu_{\bpp'}^*(\bx')+
B_{\bpp\bpp'}\bu_{\bpp}(\bx)\bu_{\bpp'}(\bx') + C_{\bpp\bpp'}\bu_{\bpp}^*(\bx)\bu_{\bpp'}^*(\bx') + D_{\bpp\bpp'}\bu_{\bpp}^*(\bx)\bu_{\bpp'}(\bx')\Big),\label{wo.eq}
}
where
\eq{
A_{\bpp\bpp'} \equiv \sum_{\bk}\alpha_{\bk\bpp}\alpha_{\bk\bpp'}^*,\;\;B_{\bpp\bpp'} \equiv \sum_{\bk}\alpha_{\bk\bpp}\beta_{\bk\bpp'}^*,\;\; C_{\bpp\bpp'} \equiv \sum_{\bk}\beta_{\bk\bpp}\alpha_{\bk\bpp'}^*,\;\; D_{\bpp\bpp'} \equiv \sum_{\bk}\beta_{\bk\bpp}\beta_{\bk\bpp'}^*.\label{abcd.eq}
}
The Pauli-Jordan function $i\Delta(\bx,\bx')=[\hat \Phi(\bx), \hat \Phi(\bx')]$ can be expanded in the modes in $\cO$
to give 
\eq{
i\Delta(\bx,\bx')=\sum_\bpp \left(\bu_{\bpp}(\bx)\bu_{\bpp}^*(\bx') - \bu_{\bpp}^*(\bx)\bu_{\bpp}(\bx')\right).\label{ido.eq}
}
The generalised eigenvalue equation for the SSEE Eqn.~(\ref{ssee.eq}) thus reduces to
\eq{&\sum_{\bpp,\bpp'}\Big(A_{\bpp\bpp'}\left<\bu_{\bpp'},\chi_\br\right>_{\cO} +
  B_{\bpp\bpp'}\left<\bu_{\bpp'}^*,\chi_\br\right>_{\cO}\Big) \bu_{\bpp}(x) + \Big(C_{\bpp\bpp'}\left<\bu_{\bpp'},\chi_\br\right>_{\cO} +
  D_{\bpp\bpp'}\left<\bu_{\bpp'}^*,\chi_\br\right>_{\cO}\Big) \bu^*_{\bpp}(x) \nonumber \\
  &= \mu_\br \sum_{\bpp} \Bigl( \left<\bu_\bpp,\chi_\br\right>_{\cO} \bu_{\bpp}(x)  -
  \left<\bu^*_\bpp,\chi_\br\right>_{\cO} \bu^*_{\bpp}(x) \Bigr),  \label{redssee.eq}
}where $\left<.,.\right>_\cO$ denotes the $\cL^2$ inner product in $\cO$, which is defined in Eqn.~\eqref{l2.eq}.
Note that the coefficients in Eqn~\eqref{abcd.eq} can be evaluated using the relation
\eq{
W(\bx,\bx')\Big|_\cO - W^*(\bx,\bx')\Big|_\cO = i\Delta(\bx,\bx'),
\label{pbc.eq}}
so that
\eq{
A_{\bpp\bpp'}-D_{\bpp\bpp'}^*=\delta_{\bpp\bpp'}&\Rightarrow \sum_\bk \left(\alpha_{\bk\bpp}\alpha_{\bk\bpp'}^* - \beta_{\bk\bpp}^*\beta_{\bk\bpp'}\right) = \delta_{\bpp\bpp'},\label{eq:bc1}\\
B_{\bpp\bpp'}-C_{\bpp\bpp'}^*=0&\Rightarrow \sum_\bk \left(\alpha_{\bk\bpp}\beta_{\bk\bpp'}^* - \beta_{\bk\bpp}^*\alpha_{\bk\bpp'}\right) =0.\label{bc2.eq}
}

We now look for a special class of solutions of Eqn~(\ref{redssee.eq}).

To begin with we consider the case when the  $\mathcal L^2$ inner product is
finite  (this is the case for example if $\cO$ is compact).  We can then use the linear independence of the $\{
\bu_\bpp\}$  to obtain  the coupled equations
\eq{
\sum_{\bpp'}\Big(A_{\bpp\bpp'}\left<\bu_{\bpp'},\chi_\br\right>_{\cO} + B_{\bpp\bpp'}\left<\bu_{\bpp'}^*,\chi_\br\right>_{\cO}\Big) &= \mu_\br \left<\bu_\bpp,\chi_\br\right>_{\cO},\nonumber\\
\sum_{\bpp'}\Big(C_{\bpp\bpp'}\left<\bu_{\bpp'},\chi_\br\right>_{\cO}+ D_{\bpp\bpp'}\left<\bu_{\bpp'}^*,\chi_\br\right>_{\cO}  \Big) &= -\mu_\br \left<\bu_\bpp^*,\chi_\br\right>_{\cO}.\label{redgev.eq}
}
Next, assume that the $\{\bu_\bpp \}$  are  $\mathcal L^2$ orthogonal. Then 
\eq{\chi_\bap (\bx) = R \bu_\bap(\bx) + S \bu_\bap^*(\bx), \label{efun.eq}}
are eigenfunctions of Eqn.~\eqref{ssee.eq} if 
\eq{R A_{\bpp\bap} + S B_{\bpp\bap}  = \mu_{\bap} R \delta_{\bpp\bap} , \nonumber \\
   R C_{\bpp\bap} +  S D_{\bpp\bap} =  -\mu_{\bap} S \delta_{\bpp\bap}.  
   \label{musc.eq}}
 This has non-trivial solutions iff
 \eq{  (A_{\bpp\bap}- \mu_{\bap}\delta_{\bpp\bap} ) (D_{\bpp\bap}  +\mu_{\bap} \delta_{\bpp\bap} ) -
   B_{\bpp\bap}C_{\bpp\bap} =0. \label{RS.eq}
 }
 For $\bpp\neq \bap$ Eqns.~\eqref{eq:bc1} and \eqref{bc2.eq} this requires in particular that
\eq{|D_{\bpp\bap}|^2=|C_{\bpp\bap}|^2,\;
 \bpp\neq \bap.}
For $\bpp= \bap$, letting $A_{\bap
   \bap}=a_\bap, B_{\bap \bap}=b_\bap, C_{\bap \bap}=c_\bap, D_{\bap \bap}=d_\bap$, we see that $a_\bap, d_\bap$ are real from
 Eqn.~(\ref{abcd.eq}), so that 
 \eq{\mu_\bap^\pm=\frac{1}{2} \Biggl(1 \pm \sqrt{(1 + 2d_\bap)^2 -
     4|c_\bap|^2)}\Biggr), \label{mu.eq}}
 which is real only if \eq{(1 + 2d_\bap)^2 \geq  
   4|c_\bap|^2. \label{reality.eq}}
 This can be shown to be true using the following identity  
 \eq{\sum_{\bk} |\alpha_{\bk\bpp} - e^{i\theta} \beta_{\bk\bpp}|^2 & \geq 0 \nonumber \\
   \Rightarrow 1+2d_{\bpp} - 2 |c_{\bpp}| \cos (\theta + \theta') &\geq 0,}
 where $c_p=|c_p|e^{i \theta'}$. Taking $\theta=-\theta'$ gives us the desired relation.   
The two eigenvalues $\mu^+_\bap, \mu_\bap^-$ moreover satisfy  the relation 
 \eq{\mu^-_\bap=1-\mu^+_\bap, \label{pairs.eq}}  and therefore come in pairs $(\mu^+_\bpp,1-\mu^+_\bpp)$, as expected \cite{ssee}. 

Thus the mode-wise SSEE is 
{\eq{
\cS_\bap= \mu^+_\bap\log(|\mu^+_\bap|)+(1-\mu^+_\bap)\log(|1-\mu^+_\bap|).  \label{sp.eq}
}}
As we will see in the specific case of de Sitter and $d=2$ Schwarzschild de Sitter spacetimes, $\mu^+_\bap, \mu_\bap^- \not
\in (0,1)$ which is again consistent with the expectations of \cite{ssee}.  

In this chapter we are interested in  subregions  $\cO$ which are static and spherically symmetric. While non-compact in the
time direction we require them to be  compact in the spatial direction. Thus the  $\mathcal L^2$ inner product is $\delta$-function
orthogonal and not strictly finite. As we will see, this can still result in a finite $\cS_\bpp$.  In $d=4$ for example, 
\eq{\bu_{plm}(t,r,\theta,\phi) = N_{pl} R_{pl}(r) e^{-ipt} Y_{lm}(\theta, \phi), \quad p > 0, \label{ssph.eq}}
where $t\in(-\infty,\infty), \, r>0$ and $(\theta,\phi)\in \ms^2$,  $N_{pl}$ denotes an overall normalisation constant, and
$p$ is a continuous variable. Thus one has integrals 
over $p$  as well as summations over $l$ and $m$ in  Eqn.~(\ref{redgev.eq}). These modes are clearly $\cL^2$ orthogonal
since 
\eq{
\left<\bu_{plm},\bu_{p'l'm'}\right>_\cO = 2\pi |N_{pl}|^2||R_{pl}||^2\delta(p-p')\delta_{ll'}\delta_{mm'},
}
where $||R_{pl}||$ is the $\cL^2$ norm in the radial direction and finite by assumption. This $\delta$-function
orthogonality implies that for any function $\chi_\br$ (which can be expanded in terms of the complete $\{\bu_{plm}\}$
basis), both sides  of  Eqn.~(\ref{redgev.eq}) are finite.

If $\hW\Big|_\cO$ is block diagonal in the $\{\bu_{plm}\}$ basis
\eq{ A_{plmp'l'm'}=a_{plm} \delta(p-p')\delta_{ll'}\delta_{mm'}, \quad B_{plmp'l'm'}=b_{plm}
  \delta(p-p')\delta_{ll'}\delta_{mm'}, \nonumber\\
C_{plmp'l'm'}=c_{plm} \delta(p-p')\delta_{ll'}\delta_{mm'}, \quad D_{plmp'l'm'}=d_{plm}
\delta(p-p')\delta_{ll'}\delta_{mm'}.  \label{diag.eq} }
This simplifies   Eqn.~(\ref{redssee.eq}) considerably since the delta functions can be integrated  over $p'$ and
similarly,  summed  over $l',m'$.  
Using the ansatz 
\eq{\chi_{plm}(t,r,\theta,\phi) = R \bu_{plm}(t,r,\theta,\phi) + S \bu_{plm}^*(t,r,\theta,\phi), \label{efun.eq}}
for the eigenfunctions requires that Eqn.~\eqref{RS.eq} is satisfied, as before. This yields the same form for $\mu_{plm}^\pm$ as Eqn.~\eqref{mu.eq} and hence the SSEE Eqn.~\eqref{sp.eq}.

\section{Preliminaries}
\label{dS.sec}
We briefly review de Sitter and Schwarzschild de Sitter spacetimes.

de Sitter spacetime dS$_d$ in $d$ dimensions is a hyperboloid of ``radius'' $H^{-1}$  in $d+1$ dimensional  Minkowski
spacetime $\mr^{1,d}$. If $X_i$'s are the coordinates in $\mr^{1,d}$, it is the hypersurface defined by 
\be
-X_0^2+\sum_{i=1}^d X_i^2 = \frac{1}{H^2}.
\ee
We restrict our discussion to $d=4$ in what follows. The higher dimensional generalisation is relatively
straightforward and is discussed at the end of this section. Global dS$_4$ can be parameterized\footnote{For a detailed review of coordinate systems in dS, see \cite{stromds}.} by 4 coordinates $(\tau,\theta_1,\theta_2,\theta_3)$, where $\tau$ is the global time and $\theta_i$'s are coordinates on a 3-sphere $\ms^3$. In these coordinates the metric can be written as
\be
ds^2=-d\tau^2+\frac{1}{H^2}\cosh^2(H\tau)\,d\Omega_3^2,
\ee
where, $\tau\in\mr$, $\theta_1,\theta_2\in[0,\pi]$ and $\theta_3\in[0,2\pi]$.
The causal structure of this spacetime becomes evident if we make the  coordinate transformation $\cosh(H\tau)=1/\cos T$,
so that 
\be
ds^2=\frac{1}{H^2\cos^2T}(-dT^2+d\Omega_3^2),\quad T\in\bigg(-\frac{\pi}{2},\frac{\pi}{2}\bigg).
\label{confmetric}
\ee
\bfig[h!]
    \centering
    \includegraphics[height=5cm]{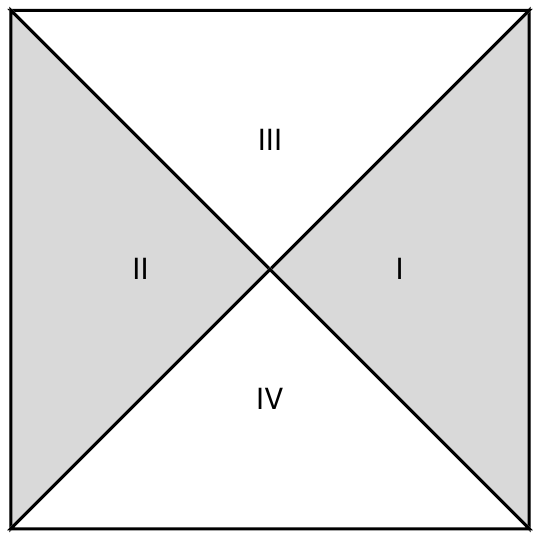}
    \caption{The Penrose diagram for dS can be deduced from the metric Eqn.~\eqref{confmetric}. Here 2 dimensions are
      suppressed so that each point represents an $\ms^2$ and each horizontal slice an $\ms^3$. dS is spatially compact,
      the left and right vertical lines correspond to $\theta_1=0,\pi$. The lower, upper horizontal
      lines correspond to $T=-\pi/2,\pi/2$ and represent the past, future null infinities respectively.} 
    \label{fig:ds}
\efig
In Fig.~\ref{fig:ds}, the region $\Righf\cup\fut$ is the right conformal patch or the Poincar\'e patch. It can be described by the metric
\be
    ds^2=\frac{1}{H^2\eta^2}\left(-d\eta^2 + dr^2 + r^2(d\theta^2+\sin^2\theta d\phi^2)\right),
\ee
where $\eta\in(-\infty,0),\;r\in[0,\infty)$ and $(\theta,\phi)\in\ms^2$. Its subregion $\Righf$ is the right static patch and is covered by the coordinates $x\in[0,1)$, $t\in\mr$, $(\theta,\phi)\in\ms^2$ which are related to the coordinates in the conformal patch by
\be
x= -\frac{r}{\eta},\quad e^{-t} = \sqrt{\eta^2-r^2},  \label{eq:conftostat}  
\ee
so that the static patch  metric is 
\be
ds^2 = \frac{1}{H^2}\left(-(1-x^2)dt^2 + \frac{dx^2}{1-x^2} + x^2d\Omega_2^2\right).
\ee

We now turn to the Schwarzschild-de Sitter spacetime, whose conformal diagram is shown in Fig \ref{fig:sds}. 
\begin{figure}[h!]
\centering
\includegraphics[height=5cm]{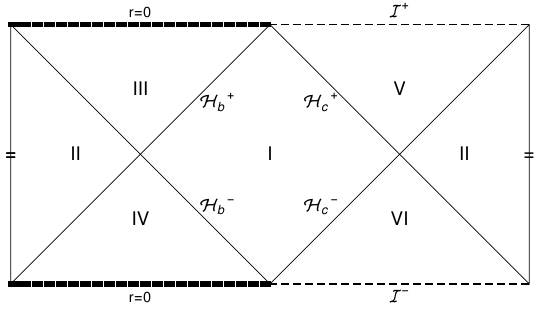}
\caption{{The Penrose diagram for the $d>2$ Schwarzschild de Sitter spacetime where each point represents an
    $\ms^{d-2}$ and each horizontal slice represents an $\ms^{d-2}\times\ms^1$. Region $I$ and $II$ are the static patches,
    and $\cH_b^{\pm}$ and $\cH_c^{\pm}$ are the black hole and the cosmological horizons respectively.}}
\label{fig:sds}
\end{figure}
It has two sets of horizons each  in regions
$I$ and $II$: the cosmological horizons $\cH_c^\pm$  and the black hole horizons $\cH_b^\pm$,  with the latter contained ``inside'' the
former.

In either of the static patches,  $I$ or $II$, the  metric of the Schwarzschild de Sitter spacetime is 
\eq{
ds^2 &= -f(r)dt^2+\frac{dr^2}{f(r)}+r^2(d\theta^2+\sin^2(\theta)d\phi^2),\quad f(r) = 1-\frac{2M}{r}-H^2r^2 \label{sdssc.eq}\\
&= -f(r)dudv +r^2(d\theta^2+\sin^2(\theta)d\phi^2),\label{sdstc.eq}
}
where $H$ is the Hubble constant and $M$ is the mass of the black hole $r\in(r_b,r_c)$, $t\in(-\infty,\infty)$ and
$(\theta,\phi)\in \ms^2$. Here $r_b<r_c$ are the real and positive solutions of $f(r)=0$, which correspond to the black
hole and the cosmological horizons $\cH_b^\pm , \cH_c^\pm$ respectively. They are related to $M$ and $H$ as
\eq{
M= \frac{r_br_c(r_b+r_c)}{2(r_b^2+r_c^2+r_br_c)},\quad H^2=\frac{1}{r_b^2+r_c^2+r_br_c}.
}
$u,v\in(-\infty,\infty)$ are the light-cone coordinates defined
as $u=t-\rs$ and $v=t+\rs$, where $d\rs=\frac{dr}{f(r)}$ \cite{Anderson:2020dim}.

As in the Schwarzschild spacetime, there is a Kruskal extension beyond the black hole and  the
cosmological horizon, given respectively by 
\eq{
U_b = -\kappa_b^{-1}e^{-\kappa_b u}\quad\text{and}\quad V_b = \kappa_b^{-1}e^{\kappa_b v},\label{schtokrusU.eq}\\ 
\quad U_c = \kappa_c^{-1}e^{\kappa_c u}\quad\text{and}\quad V_c = -\kappa_c^{-1}e^{-\kappa_c v},
\label{schtokrusV.eq}}
where $\kappa_b$ and $\kappa_c$ are the surface gravity of the black hole and the cosmological horizon
respectively \cite{Anderson:2020dim}
\eq{\kappa_b= \frac{H^2}{2r_b}(r_c-r_b)(r_c+2r_b), \quad \kappa_c=\frac{H^2}{2r_c}(r_c-r_b)(2r_c+r_b) . }
In these coordinates, the spacetime metrics in  region $\cB \equiv I\cup II\cup III\cup IV$ and  $\cC\equiv I\cup II\cup V\cup VI$ in
Fig.~\ref{fig:sds}  are, respectively 
\eq{ds_{\cB}^2 &= -f(r)e^{-2\kappa_b\rs}dU_bdV_b + r^2(d\theta^2+\sin^2(\theta)d\phi^2), \\ ds_{\cC}^2 & =
  -f(r)e^{2\kappa_c\rs}dU_cdV_c + r^2(d\theta^2+\sin^2(\theta)d\phi^2).}

\section{SSEE for Cosmological and Black hole Horizons}\label{sseedsbh.sec}
We now calculate the SSEE for the static regions in both the de Sitter and Schwarzschild de Sitter spacetimes. In both cases, since the static region is spatially finite, Eqn.~\eqref{redssee.eq} is well-defined.

\subsection{The SSEE in the de Sitter Static Patch}\label{dSssee.sec}

We are interested in finding the entanglement across the intersection sphere $\ms^2\simeq \Righf\cap\Leff$ in
Fig.~\ref{fig:ds}. The associated sub-region $\cO$ of interest to the SSEE calculation is therefore the right or left
static region. Without loss of generality we henceforth pick the right static region $R$ and take as the larger region $\cM \supset \cO$  the conformal patch
$\Righf\cup\fut$

In the larger region $\Righf\cup\fut$, we have a well known, complete, Klein-Gordon orthonormal set of modes for a free scalar field of effective mass $\mass=\sqrt{m^2+\xi R}$ (where $\xi=1/6$ and $R$ is the Ricci scalar, which is a constant for de Sitter spacetimes) called the Bunch-Davies modes~\cite{Bunch:1978yq}. These are given by $\bv_{klm}\equiv\varphi_{kl}(\eta,r)Y_{lm}(\theta,\phi)$, where
\be
\varphi_{kl}(\eta,r) = \frac{H e^{-\frac{i\pi}{2}(l+\frac{1}{2})}}{\sqrt{2k}}(-k\eta)^{\frac{3}{2}}e^{\frac{i\nu\pi}{2}}H_\nu^{(1)}(-k\eta)j_l(kr). \label{eq:bdmodes}
\ee
Here $k\in{\mathbb{R}}^+,\,l\in\{0,1,2,\dots\}$, $m\in\{-l,..,0,..,l\}$, $Y_{lm}$s are the spherical harmonics on ${\mathbb{S}^2}$, $j_l$ is the spherical Bessel function and $H_\nu^{(1)}$ is the Hankel function of the first kind with
\be
    \nu= \sqrt{\frac{9}{4}-\frac{\mass^2}{H^2}}, 
    \label{nudef}
    \ee
and satisfies the plane-wave behaviour expected at late times.

In the region $\Righf$ we have a complete set of Klein-Gordon orthonormal modes \cite{Higuchi:1986ww} given by $\bu_{plm}\equiv\psi_{pl}(t,x)Y_{lm}(\theta,\phi)$ where
\be
    \psi_{pl}(t,x) \equiv \sqrt{2\sinh(\pi p)}\,N_{pl} \,U_{pl}(x)e^{-ipt},\quad p\in\mr^+, \label{eq:stat_modes}
\ee
where
\be
    N_{pl}=\frac{H}{2\sqrt{2}\pi\Gamma(l+\frac{3}{2})}\Gamma\left(\frac{\frac{3}{2}+l-ip+\nu}{2}\right)\Gamma\left(\frac{\frac{3}{2}+l-ip-\nu}{2}\right),
    \label{npldef}
\ee
and
\be
    U_{pl}(x)=x^l(1-x^2)^{\frac{-ip}{2}}{_2F_1}\left(\frac{\frac{3}{2}+l-ip+\nu}{2},\frac{\frac{3}{2}+l-ip-\nu}{2},l+\frac{3}{2};x^2\right).
    \label{Upldef}
\ee
As shown in \cite{Higuchi:1986ww}, 
\begin{enumerate}
	\item $U_{pl}(x)=U_{-pl}(x)=U^*_{pl}(x)$, which can be shown using an identity of the Hypergeometric function i.e., ${_2F_1}(a,b,c;z)=(1-z)^{c-a-b}{_2F_1}(c-a,c-b,c;z)$.
	\item $N_{pl}=N^*_{-pl}$, which comes from the identity $\Gamma^*(z)=\Gamma(z^*)$.
        \end{enumerate}
As discussed in Sec.~\ref{gf.sec}, being static and spherically symmetric, the  $\bu_{plm}$  modes are also $\mathcal
L^2$ orthogonal in $\Righf$.

We now proceed to obtain the SSEE for the sub-region $\Righf$ with respect to the Bunch-Davies vacuum in the right conformal patch $\Righf\cup\fut$. As suggested in Sec.~\ref{gf.sec}, we begin by demonstrating that the Bogoliubov coefficients between the Bunch-Davies modes
$\bv_{klm}$ and the static modes $\bu_{plm}$ in $\Righf$ satisfy the criteria Eqn.~\eqref{diag.eq}.

Since the $(\theta,\phi)$ dependence of both sets of modes is given by $Y_{lm}(\theta,\phi)$, which themselves are
linearly independent in $\ms^2$, the Bogoliubov transformation is non-trivial only between $\varphi_{kl}$ and
$\psi_{pl}$ for each $l,m$, i.e., 
\begin{equation}
\varphi_{kl}(\eta,r) = \int_0^\infty dp\,\left(\alpha_{kp}\psi_{pl}(t,x) + \beta_{kp}\psi_{pl}^*(t,x)\right)\label{eq:vtophi}.
\end{equation}
Instead of using the Klein-Gordon inner product to calculate $\alpha_{kp}$ and
$\beta_{kp}$, we can use the $\cL^2$ orthogonality of the $\bu_{plm}$ modes as well as the $\cL^2$ inner product of $\bv_{klm}$ and $\bu_{plm}$ in $\Righf$, so that 
\begin{equation}
\alpha_{kp} = \frac{1}{n_p}\left<\bu_{plm},\bv_{klm}\right>_\Righf\quad\text{and}\quad \beta_{kp} = \frac{1}{n_p}\left<\bu_{pl-m}^*,\bv_{klm}\right>_\Righf,\label{eq:pqip}
\end{equation}
with $n_p = 4\pi\sinh(\pi p)|N_{pl}|^2||U_{pl}||^2$ being the $\cL^2$ norm of the $\bu_{plm}$ modes.
The identity \cite{magnus}
\eq{
\int_0^\infty \!\!dz\, z^\lambda H_\nu^{(1)}(az)J_\mu(bz)
&=a^{-\lambda-1}e^{i\frac{\pi}{2}(\lambda-\nu+\mu)}\frac{2^\lambda(b/a)^\mu}{\pi\Gamma(\mu+1)}\Gamma\left(\frac{\lambda+\nu+\mu+1}{2}\right)\Gamma\left(\frac{\lambda-\nu+\mu+1}{2}\right)\nonumber\\
&\times_2F_1\left(\frac{\lambda+\nu+\mu+1}{2},\frac{\lambda-\nu+\mu+1}{2},\mu+1,\left(\frac{b}{a}\right)^2\right), 
\nonumber \\  &\quad\quad\quad\quad\quad\quad\mathrm{Re}(-i(a\pm b))>0, \, \, \mathrm{Re}(\mu+\lambda+1\pm\nu)>0,
}
 can be used as in  \cite{Higuchi:2018tuk}, to show that
\eq{
\frac{1}{\sqrt{2\pi}}\int_0^\infty dk\;k^{-ip-\frac{1}{2}}\varphi_{kl}(\eta,r) &= 2^{-ip}e^{\frac{\pi p}{2}}N_{pl}\left(\eta^2-r^2\right)^{\frac{ip}{2}}U_{pl}\left(-\frac{r}{\eta}\right)\nonumber\\
&= 2^{-ip}e^{\frac{\pi p}{2}}N_{pl}U_{pl}(x)e^{-ipt},
}
where we have substituted $\lambda=-ip$, $\mu=l+1/2$, $a=-\eta$ and $b=r$. Inverting the above, 
\eq{
\varphi_{kl}(\eta,r) = \frac{1}{\sqrt{2\pi}}\int_{-\infty}^\infty dp\,2^{-ip}k^{ip-\frac{1}{2}}e^{\frac{\pi p}{2}}N_{pl}U_{pl}(x)e^{-ipt},\label{eq:vtor}
}
using which 
\eq{
\alpha_{kp} &= \frac{1}{\sqrt{2\pi}n_p}\int_{-\infty}^\infty dp'\,2^{-ip'}k^{ip'-\frac{1}{2}}e^{\frac{\pi
    p'}{2}}\sqrt{2\sinh(\pi p)}N_{p'l}N_{pl}^*
\int_0^1 dx\,x^2 U_{p'l}(x)U_{pl}(x) \int_{-\infty}^\infty dt e^{-i(p'-p)t}\nonumber\\
&=\frac{2^{-ip}k^{ip-\frac{1}{2}}}{\sqrt{2\pi(1-e^{-2\pi p})}},\label{eq:alpkp}\\
\beta_{kp} &= \frac{1}{\sqrt{2\pi}n_p}\int_{-\infty}^\infty dp'\,2^{ip'}k^{-ip'-\frac{1}{2}}e^{\frac{-\pi p'}{2}}\sqrt{2\sinh(\pi p)}N_{p'l}^*N_{pl}
\int_0^1 dx\,x^2 U_{p'l}(x)U_{pl}(x)\int_{-\infty}^\infty dt e^{i(p'-p)t}\nonumber\\
&= \frac{2^{ip}k^{-ip-\frac{1}{2}}}{\sqrt{2\pi(e^{2\pi p}-1)}}.\label{eq:betkp}
}
Notice that the Hubble constant $H$ drops out of these coefficients. Further calculation shows that 
\begin{equation}
A_{pp'}=\frac{\delta(p-p')}{1-e^{-2\pi p}},\;\;D_{pp'}=\frac{\delta(p-p')}{e^{2\pi p}-1}\;\;\text{and}\;\;B_{pp'}=C_{pp'}=0.\label{eq:cresult}
\end{equation}
This is precisely of the form Eqn.~\ref{diag.eq}, with  $b_{p}=c_{p}=0$, where we have suppressed the $l,m$ indices.
Using the ansatz 
\eq{\chi_{p}^+(t,r) = u_{p}(t,r), \quad \chi_{p}^-(t,r)  =  u_{p}^*(t,r), \label{efun.eq}}
for the generalised eigenfunctions of Eqn.~\eqref{redssee.eq}, we  find the generalised eigenvalues 
\be
\mu^{+}_{p} = \frac{1}{1-e^{-2\pi p}}\quad\text{and}\quad \mu^{-}_{p} =- \, \frac{e^{-2\pi p}}{1-e^{-2\pi p}},
\ee
respectively, for each $p\in\mr^+$.  {Note that $\mu^{+}_{p} \in [1, \infty)$ and  $\mu^{-}_{p} \in (-\infty, 0]$,
  as expected for the SSEE \cite{ssee}.} For a given $p, l, m$ the mode-wise SSEE  is therefore 
\be
\cS_p = -\log(1-e^{-2\pi p}) - \frac{e^{-2\pi p}}{1-e^{-2\pi p}}\log e^{-2\pi p}, \label{eq:hy}
\ee
which agrees with the result of \cite{Higuchi:2018tuk}. Since there is no dependence on $l,m$ there is an infinite degeneracy coming
from  the angular modes $l\in\{0,1,\ldots\}$ and $m\in\{-l,\ldots,0,\ldots,l\}$.

In order to calculate the total SSEE, therefore,  one has to sum over the $l,m$ and integrate over $p \in (0,\infty)$.  For the  Bunch-Davies vacuum the integral  over $p$ is finite. However in the absence of a cut-off in $l$, there is an
infinite degeneracy for every $p$ coming from the angular modes which  leads to an infinite factor in
the total entropy. This ``density of states''  for a given $p$ can be regulated by
introducing a cut-off  $l_\mx$, so that
\eq{
\cS =  \sum_{l=0}^{l_\mx} \sum_{m=-l}^l \int dp \, \cS_p =
\frac{\pi}{6}(l_\mx+1)^2 \simeq \frac{\pi}{6}l_\mx^2\label{eq:sseebdvac}
}
for $l_\mx >>1 $.  $l_\mx$ can in turn be interpreted as coming from the regularised area of the de Sitter horizon $\Righf\cap\Leff
  \simeq \ms^2$.  Let us for the moment suppress one of the angular variables so that the the modes on 
 an $\ms^1$ of radius $H^{-1}$  are $e^{im\phi}$. A UV cut-off  $m_\mx$  corresponds to a minimal angular
   scale $\Delta \phi=2\pi/m_\mx$ and hence a length cut-off  $\ell_c$, where $m_\mx=\frac{2 \pi}{Hl_c}$. Thus $m_\mx$
   is the  circumference of the $\ms^1$ in units of the cut-off.  A similar argument carries over to
   $\ms^2$, where we first place $\theta$ and $\phi$ on similar footing by  writing the spherical harmonics as a Fourier series\cite{Hofsommer1960}
\eq{
Y_{lm}(\theta,\phi)\propto P_l^m(\cos\theta)e^{im\phi}=\sum_{j=-l}^l\tilde{P}_{jl}^m e^{ij\theta}e^{im\phi}. 
}
Thus, we again have the angular cut-offs $\Delta\theta=2\pi/l_\mx, \Delta \phi=2 \pi/l_\mx$, so that $\l_\mx^2 = 4
\pi^2/\Delta \theta \Delta \phi$. For large $l_\mx$ the planar limit of the  region subtended by the solid angle $\Delta
\Omega= \sin \theta \Delta\theta\Delta\phi$ on $\ms^2$ can be taken near the equator, $\theta = 
\pi/2-\epsilon$, where the metric is nearly flat in $(\theta, \phi)$ coordinates: $ds^2\simeq d\epsilon^2 +
d\phi^2$. Thus,  $l_\mx^2\propto 1/\Delta\Omega$ and therefore
\eq{
\cS\propto \frac{A_{c}}{l_c^2}\label{ae.eq}
}
where we have defined a fundamental cut-off  $l_c^2=H^{-2} \Delta\Omega$ and $A_{c}=4 \pi H^{-2}$ is the area of the de
Sitter cosmological horizon. We now look at the SSEE for the $\alpha$ vacua in dS static patch.

\subsubsection{SSEE for the $\alpha$ vacua in dS static patch}\label{alphabeta.sec}

We now compute the SSEE for the $\alpha$ vacua \cite{Mottola:1984ar,Allen:1985ux} $\bv_{klm}^{(\alpha,\beta)}\equiv\varphi^{(\alpha,\beta)}_{kl}(\eta,r)Y_{lm}(\theta,\phi)$ which can be parameterized as 
\be
\varphi^{(\alpha,\beta)}_{kl}(\eta,r)\equiv \cosh(\alpha)\varphi_{kl}(\eta,r)+\sinh(\alpha)e^{i\beta}\varphi^*_{kl}(\eta,r),
\label{alphabetabdmodes}
\ee
where $\alpha\in[0,\infty)$ and $\beta\in(-\pi,\pi)$. 
Expressing these in  terms of the $\bu_{plm}$ modes in $\Righf$ as
\eq{
\varphi^{(\alpha,\beta)}_{kl}(\eta,r) = \int_0^\infty dp\, \left(\alpha^{(\alpha,\beta)}_{kp}\psi_{pl}(t,x) + \beta^{(\alpha,\beta)}_{kp}\psi_{pl}^*(t,x)\right),
}
where
\eq{
\alpha^{(\alpha,\beta)}_{kp}=\frac{1}{n_p}\left<\bu_{plm},\bv^{(\alpha,\beta)}_{klm}\right>_\Righf\quad\text{and}\quad \beta^{(\alpha,\beta)}_{kp}=\frac{1}{n_p}\left<\bu_{pl-m}^*,\bv^{(\alpha,\beta)}_{klm}\right>_\Righf.
}
Using Eqn.~\eqref{alphabetabdmodes}, \eqref{eq:pqip}, \eqref{eq:alpkp} and \eqref{eq:betkp}, we find the coefficients $\alpha^{(\alpha,\beta)}_{kp}$ and $\beta^{(\alpha,\beta)}_{kp}$ to be
\eq{
\alpha^{(\alpha,\beta)}_{kp} = \cosh(\alpha)\alpha_{kp} + \sinh(\alpha)e^{i\beta}\beta_{kp}^* = \frac{2^{-ip}k^{ip-\frac{1}{2}}}{\sqrt{2\pi(1-e^{-2\pi p})}}\left(\cosh(\alpha)+e^{-\pi p}e^{i\beta}\sinh(\alpha)\right),\\
\beta^{(\alpha,\beta)}_{kp} = \cosh(\alpha)\beta_{kp} + \sinh(\alpha)e^{i\beta}\alpha_{kp}^* = \frac{2^{ip}k^{-ip-\frac{1}{2}}}{\sqrt{2\pi(1-e^{-2\pi p})}}\left(e^{-\pi p}\cosh(\alpha)+e^{i\beta}\sinh(\alpha)\right).
}
Further calculation shows that $A^{(\alpha,\beta)}_{pp'} \equiv \int dk\, \alpha^{(\alpha,\beta)}_{kp}\alpha^{(\alpha,\beta)*}_{kp'}$, $B^{(\alpha,\beta)}_{pp'} \equiv \int dk\, \alpha^{(\alpha,\beta)}_{kp}\beta^{(\alpha,\beta)*}_{kp'}$, $C^{(\alpha,\beta)}_{pp'} \equiv \int dk\, \beta^{(\alpha,\beta)}_{kp}\alpha^{(\alpha,\beta)*}_{kp'}$ and $D^{(\alpha,\beta)}_{pp'} \equiv \int dk\, \beta^{(\alpha,\beta)}_{kp}\beta^{(\alpha,\beta)*}_{kp'}$ is of the form
\eq{
A^{(\alpha,\beta)}_{pp'} = a^{(\alpha,\beta)}_p\delta(p-p'),\;D^{(\alpha,\beta)}_{pp'} = d^{(\alpha,\beta)}_p\delta(p-p'),\;B^{(\alpha,\beta)}_{pp'} =C^{(\alpha,\beta)}_{pp'} =0,
}
where
\eq{
	a^{(\alpha,\beta)}_p &=\frac{1}{1-e^{-2\pi p}}\left(\cosh^2(\alpha)+e^{-2\pi p}\sinh^2(\alpha)+e^{-\pi p}\sinh(2\alpha)\cos(\beta)\right),\\
	d^{(\alpha,\beta)}_p &=\frac{1}{1-e^{-2\pi p}}\left(e^{-2\pi p}\cosh^2(\alpha)+\sinh^2(\alpha)+e^{-\pi p}\sinh(2\alpha)\cos(\beta)\right).
}

Generalised eigenvalues $\mu$ is then
\eq{
    \mu^{+(\alpha,\beta)}_{p} = a^{(\alpha,\beta)}_p\quad\text{and}\quad \mu^{-(\alpha,\beta)}_{p} =-d^{(\alpha,\beta)}_p,
}
from which we obtain the SSEE as
\eq{
S^{(\alpha,\beta)} &= (l_\mx+1)^2\int_0^\infty dp\left(a^{(\alpha,\beta)}_p\log(a^{(\alpha,\beta)}_p) - d^{(\alpha,\beta)}_p\log(d^{(\alpha,\beta)}_p)\right).\label{eq:sseeinstatic}
}
Unlike for the SSEE obtained from the Bunch-Davies vacuum $(\alpha=0)$, $S^{(\alpha,\beta)}$ is in general dependent on the cut-off in $p$. As an example, for $(\alpha,\beta)=(1,0)$, we evaluate the integral in Eqn.~\eqref{eq:sseeinstatic} numerically for different cut-offs in $p$ and find that the SSEE depends on both $p_\mx$ and $l_\mx$, and is of the form
\be
S^{(1,0)}= S_p(p_\mx)(l_\mx+1)^2,
\ee
where for large enough $p_\mx$, $S_p(p_\mx)$ is found to be proportional to $p_\mx$ as shown in Fig.~\ref{fig:sponezero}.
\begin{figure}[h!]
\centering\includegraphics[height=5cm]{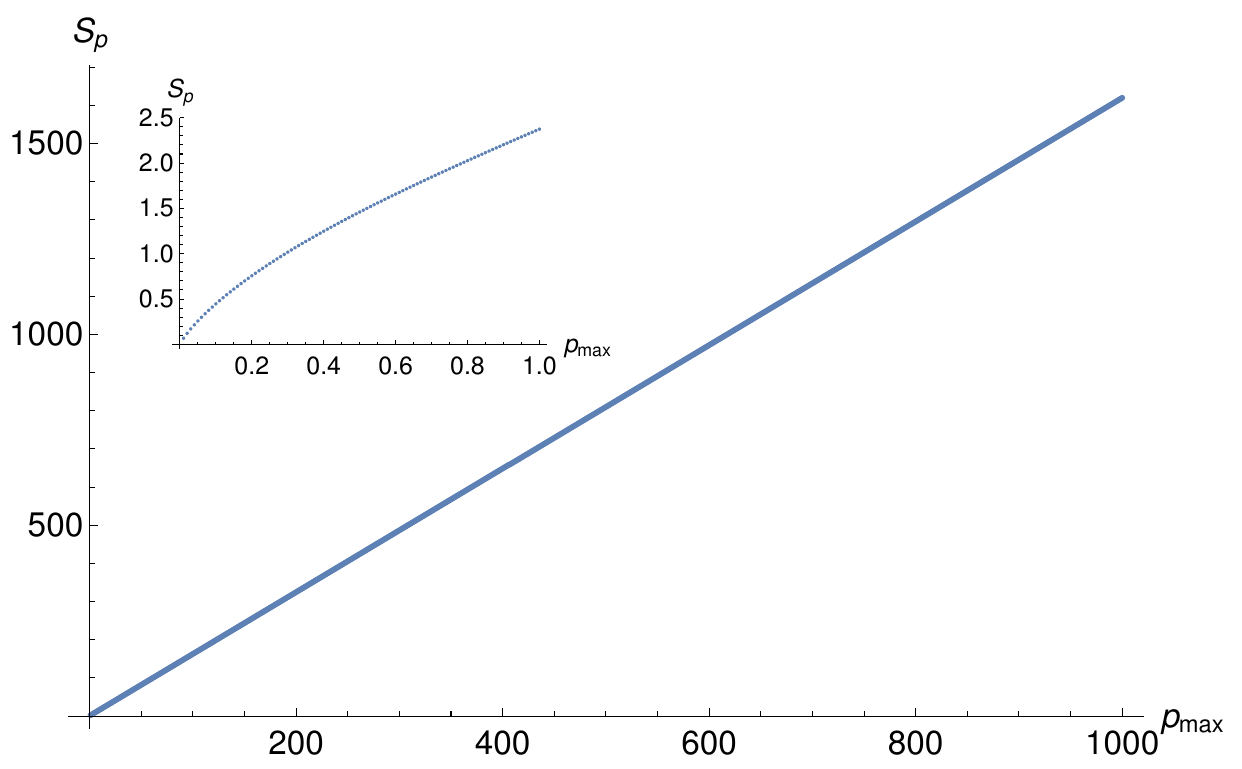}
\caption{A plot of $S_p$ vs $p_\mx$ for $(\alpha,\beta)=(1,0)$. We see that for large enough $p_\mx$, $S_p\propto p_\mx$.}
\label{fig:sponezero}
\end{figure}

As we see here for all the other $\alpha$ vacua, the integral $\int_{0}^p dp \,\cS_p$ may not be finite. This necessitates an additional cut-off $p_\mx$. Therefore the area interpretation of the SSEE (Eqn.~\eqref{ae.eq}) is not generic and is a property of the Bunch-Davies or Euclidean vacuum. The extension to general $d>2$ is shown in Appendix~\ref{ddimds.sec}. 

\subsection{The SSEE of Schwarzschild de Sitter Horizons}\label{bh.sec}

In the Schwarzschild de Sitter spacetime, regions
$I$ and $II$ are static  and spherically  symmetric, which means that the massless scalar field modes are of the form
Eqn.~\eqref{ssph.eq}. What is important for our analysis is that the spacetime is spatially bounded so that the
calculations of Sec.~\ref{gf.sec} can be applied to this case. Without loss of generality we will work with region $I$ to calculate  its
SSEE. 

Although our focus is the $d>2$ case, we begin by suppressing  the angular dependence and considering the $d=2$ case first.
The massless, Klein-Gordon orthonormal scalar field modes are then simply the plane waves in $I$
\eq{
\bu_{p}^{(1)}(u) = \frac{1}{\sqrt{4\pi p}}e^{-ipu}\quad\text{and}\quad \bu_{p}^{(2)}(v) = \frac{1}{\sqrt{4\pi p}}e^{-ipv},\;p>0,
}  as well as in regions $\cB$ and $\cC$ 
\eq{
\bv_{k}^{(1)}(U) = \frac{1}{\sqrt{4\pi k}}e^{-ikU}\quad\text{and}\quad \bv_{k}^{(2)}(V) = \frac{1}{\sqrt{4\pi k}}e^{-ikV},\;k>0,
}
where we have suppressed the $b,c$ indices in $(U_{b,c},V_{b,c}) $ for simplicity. Note that the modes in region $I$
are static, and of the form Eqn.~(\ref{ssph.eq}), with $l\in \{0,1 \}$ representing the left and right movers. This
means that the radial part is $\cL^2$, which ensures finiteness of Eqn.~\eqref{redssee.eq}.

The restriction of  $\bv_{k}^{(1,2)}$ to region I can be written in terms of $\bu_p^{(1,2)}$ as
\eq{
\bv_{k}^{(1,2)} = \int_0^\infty dp\, \left(\alpha_{kp}^{(1,2)}\bu_{p}^{(1,2)} + \beta_{kp}^{(1,2)}\bu_{p}^{(1,2)*}\right),
}
where
\eq{
\alpha_{kp}^{(1)}=\frac{1}{2\pi}\sqrt{\frac{p}{k}}\int_{-\infty}^\infty du\, e^{ipu}e^{-ikU} &= \frac{1}{2\pi
  \kappa}\sqrt{\frac{p}{k}}\left(\frac{k}{\kappa}\right)^{i\frac{p}{\kappa}}e^{\frac{\pi p}{2 \kappa}}\Gamma\left(-i\frac{p}{\kappa}\right)\label{eq:alpha2d},\\
\beta_{kp}^{(1)}=\frac{1}{2\pi}\sqrt{\frac{p}{k}}\int_{-\infty}^\infty du\, e^{-ipu}e^{-ikU}
&=\frac{1}{2\pi\kappa}\sqrt{\frac{p}{k}}\left(\frac{k}{\kappa}\right)^{-i\frac{p}{\kappa}}e^{-\frac{\pi p}{2 \kappa}}\Gamma\left(i\frac{p}{\kappa}\right)\label{eq:beta2d},
}
$\alpha_{kp}^{(2)}=\alpha_{kp}^{(1)*}$ and $\beta_{kp}^{(2)}=\beta_{kp}^{(1)*}$
for the black hole horizon. For the cosmological horizon, they  are complex conjugates of
Eqns.~\eqref{eq:alpha2d} and \eqref{eq:beta2d}. 
Thus, we find that for both $(1,2)$ modes,
\eq{
A_{pp'}=a_p\delta(p-p'),\;\;D_{pp'}=d_p\delta(p-p')\;\;\text{and}\;B_{pp'}=C_{pp'}=0, \label{2ddiag.eq}
}
with 
\eq{
  a_p=\frac{1}{1-e^{-2\pi \frac{p}{\kappa}}}\quad\text{and}\quad d_p=\frac{e^{-2\pi \frac{p}{\kappa}}}{1-e^{-2\pi \frac{p}{\kappa}}}
  .}
Using Eqn.~\eqref{mu.eq} and the dimension-free $\tip\equiv p\kappa^{-1}$, 
we see that 
\eq{
  \cS_\tip
  = -\log\left(1-e^{-2\pi \tip}\right)
  - \frac{e^{-2\pi \tip}}{1-e^{-2\pi \tip}}\log\left(e^{-2\pi \tip}\right).
}
The total entropy is then 
\eq{
\cS=2\int_0^\infty d\tip \cS_\tip =-\frac{2}{\pi }\int_0^1 dz \frac{\log(z)}{1-z} = \frac{\pi}{3},\label{eq:rindssee}
}
where $z=e^{-2\pi \tip}$ and the factor of two comes from the fact that the total entropy is the sum of the entropy of
the  $(1,2)$ modes.  $\cS$ is therefore the same for both horizons.

{We now consider the $d=4$ case by using the boundary conditions of \cite{Qiu:2019qgp}.  As mentioned earlier, the full
  modes are not known, but the boundary conditions suffice to calculate the Bogoliubov coefficients. For our purposes it
  suffices to use the past boundary conditions, since this defines the Klein Gordon norm on the 
limiting initial null surface $\cH_b^- \cup \cH_c^-$ in Region I.  
For the static patch modes, which are of the form Eqn.~\ref{ssph.eq}, these  boundary conditions are 
\eq{
\bu_{plm} = \begin{cases}\frac{1}{\sqrt{4\pi p}r_b}e^{-ipu}Y_{lm}(\theta,\phi)\quad &\text{on} \;\cH_b^-\\ 0 \quad &\text{on}\; \cH_c^-
\end{cases},\quad p>0,\label{statmodes4d.eq}
}
while for the Kruskal modes across the black hole horizon, they are 
\eq{
\bv_{klm} = \begin{cases}\frac{1}{\sqrt{4\pi k}r_b}e^{-ikU_b}Y_{lm}(\theta,\phi)\quad &\text{on} \;\cH_b^-\\ 0 \quad &\text{on}\; \cH_c^-
\end{cases},\quad k>0,\label{krusmodes4d.eq}
}
where $U_b$ is related to $u$ as in Eqn.~\eqref{schtokrusU.eq} \cite{Qiu:2019qgp}. Note that our normalisation differs
  from that of \cite{Qiu:2019qgp} and comes from the KG norm on $\cH_b^-\cup \cH_c^-$ or equivalently $\cH_b^-$ for these
  boundary conditions. The factor $r_b^{-1}$ is dimension
  dependent and comes from the 
  normalisation of the modes along $\cH_b^-$ where $r=r_b$, and the angular measure is $r_b^2 d\Omega$. Thus for any $d>2$,
  one must include a  factor $r_b^{-\frac{d-2}{2}}$ to normalise the modes. Importantly,  these boundary conditions 
    are not appropriate for  $d=2$, since the left and right movers are independent in that case. Setting the modes to
    zero on $\cH_c^-$ in $d=2$ would thus lead to an incomplete set of modes in region I. This is not the case for
    $d>2$, where there is a ``mixing'' or scattering of the left movers on $\cH_b^-$ in region I.  
  
Since the modes vanish along $\cH_c^-$, the KG norm can be defined using only $\cH_b^-$ in region $I$ of
Fig.~\ref{fig:sds}, where $u\in(-\infty,\infty)$ and $U_b\in(-\infty,0)$. As in the de Sitter calculation, the angular
modes for $\bv_{klm}$ and  $\bu_{plm}$  are the same, so that the calculation reduces to the $d=2$ case described
above, with the Bogoliubov coefficients given by Eqn.~\eqref{eq:alpha2d} and \eqref{eq:beta2d}. Note that unlike  $d=2$,
there is only {\it one} set of complete modes, which corresponds in our case to the set $(1)$.

Thus, the SSEE  is given by the $d=2$ SSEE for one mode,  multiplied as in the de Sitter case, by the angular cut-off
term, $(l_\mx+1)^2$ coming from the degeneracy of the generalised eigenfunctions. A similar calculation can be done for
the cosmological horizon, so that we have  
\eq{\cS_\cB\propto  \frac{A_\cB}{\ell_c^2}, \quad \cS_\cC\propto  \frac{A_\cC}{\ell_c^2}.}}

We note that a  calculation of the Rindler and Schwarzschild horizons with similar boundary conditions should in
principle be  { possible}
if
one employs a suitable  radial IR cut-off  to regulate the radial $\cL^2$ norm,  so that Eqn.~\eqref{redssee.eq} is well defined.

\section{Discussion}
\label{discussion.sec} 

In this chapter we  began with an analysis of  the SSEE, using the two sets of modes in $\cM$ and $\cO
\subset \cM$.  We found that  when the Bogoliubov
transformations satisfy certain conditions in both the finite as well as the static, spatially finite cases, there
are real solutions to the eigenvalue equations which come in pairs $(\mu, 1-\mu)$. We then calculated the SSEE
for de Sitter horizons in $d>2$ as well as Schwarzschild de Sitter horizons in $d>2$.
We found that in both cases, the eigenvalues also satisfy the condition  $\mu\not\in (0,1)$, as
expected from the arguments given in \cite{ssee}.  In both spacetimes, we used the cut-off in the angular
modes to demonstrate that  $\cS \propto A$ for $d>2$.  This is as expected, and is a further confirmation that the SSEE
is a good measure of  entanglement entropy. 

When we restrict to $d=2$, however, we find that the SSEE is constant and thus {\it not} of the Calabrese-Cardy form. This
differs from the results of earlier $d=2$ calculations of the SSEE  both in the continuum and using causal set discretisations \cite{saravani2014spacetime,Sorkin:2016pbz,Belenchia:2017cex,Surya_2021} and also from our work in \cite{Mathur:2021zzl}, where the Calabrese-Cardy form
was obtained.

An obvious difference stems from  the non-compactness of $\cO$ in the de Sitter cases studied here. For the  nested causal diamonds in $d=2$ Minkowski spacetime as well as the causal diamond on the finite cylinder spacetime, $\cO$ is chosen to be the domain of dependence of a finite interval, and is therefore compact \cite{saravani2014spacetime,Mathur:2021zzl}.  In de Sitter spacetime, the domain of dependence of the half circle is the static patch which is not compact. We have shown  that despite the temporal non-compactness, the SSEE
equation Eqn.~\eqref{ssee.eq} is well defined for the static patch. On the other hand, the numerical calculation for de Sitter causal sets
\cite{Surya_2021} necessitated an  IR cut-off, so that the regions $(\cM,g)$ as well as $\cO$ differ from those used here. After a suitable truncation in the discrete spectrum,  the  Calabrese-Cardy form for the
causal set SSEE was recovered.

We also note that in these calculations, the angular modes tranform trivially. Thus, the 
generalised eigenvalues are dimension independent, which makes the $d=2$ calculation a simple dimensional
restriction. Hence the conclusions we draw in higher dimensions -- namely that $\cS$ has an area dependence -- also implies that the SSEE is constant in $d=2$.  In higher dimensions the density
of states comes from the degeneracy of the angular modes on $\ms^{d-2}$  which necessitates a cut-off, while that 
in $d=2$ comes from the two ``angular modes'' on $\ms^0$.

Ultimately, the use of the SSEE lies in its covariant formulation and its applicability to systems where Hamiltonian
methods are not at hand. This is the case with causal set quantum gravity, since the analogues of spatial hypersurfaces
allow for  a certain ``leakage''  of information.  As shown in \cite{Sorkin:2016pbz,Surya_2021} the calculation of the SSEE for QFT on
causal sets throws up some unexpected behaviour, due to the non-local but covariant nature of the UV cut-off.  It is of
course not obvious that entanglement entropy plays a fundamental role in
quantum gravity, but the effects of the latter can be non-trivial when discussing emergent phenomena.

The SSEE approach to entanglement entropy  is compatible with that of  algebraic quantum field theory, where entanglement measures are
state functionals which measure  the  entanglement of a  mixed state $\hW_{\cO}$ obtained by restricting the pure
state $\hW$ in $\cM \supset \cO$. The  SSEE was motivated by the study of systems with finite degrees of freedom, but
has been shown to give the expected results for systems with infinite degrees of freedom, as is the case here and the
$d=2$ examples discussed above. Defining entanglement entropy for systems with infinite degrees of
freedom is however known to be non-trivial. 
Although we have several QFT examples for  which the SSEE is a good entanglement measure, an important open question is whether it can be rigorously derived using methods from algebraic quantum field theory.

One of the features that simplified our calculations was the diagonal form Eqn.~\eqref{diag.eq},\eqref{2ddiag.eq}. This, as we shall see in the next chapter, is not true for the $d=2$ cylinder calculation. Re-examining our calculation we see that a temporal IR cut-off in  $\cO$
would destroy this diagonal property. Whether this could restore the logarithmic behaviour or not would be difficult to
establish analytically, but given the causal set example, it suggests that this may indeed be the case.  This in turn
suggests new subtleties in the nature of $d=2$ entanglement in curved spacetime, which should be explored.
\vskip 1in
\leftline{\bf\Large{Appendices}}
\begin{subappendices}

\section{SSEE of general $d$-dimensional dS horizon}\label{ddimds.sec}
In this section, we extend our calculation of SSEE in four dimensional de Sitter to a general $d$-dimensional de Sitter spacetime with $d>2$ and show that the entropy depends on the spacetime dimension solely due to the dimension dependent degeneracy of the spherical harmonics.

We start with showing that the Bunch-Davies modes $\{\bv_{kL}\}$ in the conformal patch of $d$-dimensional de Sitter spacetime is given by $\bv_{kL}=\varphi_{kl}(\eta,r)Y_L(\Omega_{d-2})$, where
\begin{equation}
\varphi_{kl}(\eta,r) = \frac{H^{d/2-1}}{\sqrt{2k}}(-k\eta)^{\frac{d-1}{2}}H^{(1)}_\nd(-k\eta)(kr)^{2-d/2}j_{l+\frac{d}{2}-2}(kr),\;k>0.\label{eq:ddimbdmodes}
\end{equation}
Here $L$ represents a collection of indices $\{l,l_1,\ldots,l_{d-4},m\}$ such that $l,l_1,\ldots l_{d-4} \in \{0,1,2,\ldots\}$, $m\in\mathbb{Z}$ and $l\geq l_1\geq \ldots l_{d-4}\geq |m|$. $\Omega_{d-2}$ represents a collection of angular coordinates on $\mathbb{S}^{d-2}$. We can clearly see that for $d=4$, these modes reduces to the Bunch-Davies modes given by Eqn.~\eqref{eq:bdmodes}. For $\{\bv_{kL}\}$ to qualify for the QFT modes they have to be Klein-Gordon orthonormal solutions of the Klein-Gordon equation, which we will show now.

Klein-Gordon equation for the massive scalar field with effective mass $\mass$ in de Sitter spacetime is
\begin{equation}
-\eta^d\partial_\eta(\eta^{2-d}\partial_\eta\phi) + \frac{\eta^2}{r^{d-2}}\partial_r(r^{d-2}\partial_r\phi) + \frac{\eta^2}{r^2}\nabla_{\Omega_{d-2}}^2\phi = \frac{\mass^2}{H^2}\phi. \label{eq:ddimkgeqn}
\end{equation}
For $\varphi_{kl}$ given by Eqn.~\eqref{eq:ddimbdmodes} we find
\begin{eqnarray}
-\eta^d\partial_\eta(\eta^{2-d}\partial_\eta\varphi_{kl}) &=& \frac{1}{4}\left(1+d(d-2)+4k^2\eta^2-4\nu^2\right)\varphi_{kl},\\
\frac{\eta^2}{r^{d-2}}\partial_r(r^{d-2}\partial_r\varphi_{kl}) &=& \frac{\eta^2}{r^2}\left(l(l+d-3)-k^2r^2\right)\varphi_{kl},\\
\frac{\eta^2}{r^2}\nabla_{\Omega_{d-2}}^2 Y_L &=& -\frac{\eta^2}{r^2}l(l+d-3)Y_L.
\end{eqnarray}
Therefore $\varphi_{kl}$ given by Eqn.~\eqref{eq:ddimbdmodes} solves the Klein-Gordon equation for
\begin{equation}
\nd = \sqrt{\left(\frac{d-1}{2}\right)^2-\frac{\mass^2}{H^2}}.
\end{equation}
The modes $\{\bv_{kL}\}$ are Klein-Gordon orthonormal: 
\begin{equation}
\begin{split}
(\bv_{kL},\bv_{k'L'})_{\Righf\cup\fut} = i\frac{kk'}{2}|\eta|\left(H_\nd^{{(1)}*}(-k\eta)\partial_\eta H_\nd^{{(1)}}(-k'\eta) - H_\nd^{{(1)}}(-k'\eta)\partial_\eta H_\nd^{{(1)}*}(-k\eta)\right)\\
\int_0^\infty dr\,r^2 j_{l+\frac{d}{2}-2}(kr)j_{l'+\frac{d}{2}-2}(k'r)\int_{\mathbb{S}^{d-2}} d\Omega_{d-2}Y_L^*(\Omega_{d-2})Y_{L'}(\Omega_{d-2}).
\end{split}
\end{equation}
Here the volume element on the constant $\eta$ surface $\spac$ is $d\spac = (H\eta)^{1-d}r^{d-2}dr d\Omega_{d-2}$ and the future pointing unit vector normal to $\spac$ is $\hn^\mu\partial_\mu = H|\eta|\partial_\eta$. Using the fact that spherical harmonics are $L^2$ orthonormal on $\mathbb{S}^{d-2}$ and
\begin{equation}
\int_0^\infty dr\,r^2 j_{n}(kr)j_{n}(k'r) = \frac{\pi}{2k^2}\delta(k-k'),
\end{equation}
for $n>-1$, we can write
\begin{equation}
(\bv_{kL},\bv_{k'L'})_{\Righf\cup\fut} = \frac{i\pi}{4}|\eta| \left(H_\nd^{{(1)}*}(-k\eta)\partial_\eta H_\nd^{{(1)}}(-k\eta) - H_\nd^{{(1)}}(-k\eta)\partial_\eta H_\nd^{{(1)}*}(-k\eta)\right)\delta(k-k')\delta_{LL'}.\label{eq:kgbdf}
\end{equation}
Since the Klein-Gordon inner product is independent of the choice of the spacelike hypersurface, we will evaluate it at the surface $\eta\rightarrow-\infty$, where 
\begin{equation}
H^{(1)}_\nd(-k\eta) \rightarrow \sqrt{\frac{-2}{\pi k\eta}}e^{-i\left(k\eta+\frac{\pi\nd}{2}+\frac{\pi}{4}\right)}.\label{eq:hankelapprox}
\end{equation}
Substituting Eqn.~\eqref{eq:hankelapprox} in Eqn.~\eqref{eq:kgbdf}, we see that
\begin{equation}
(\bv_{kL},\bv_{k'L'})_{\Righf\cup\fut} = \delta(k-k')\delta_{LL'}.
\end{equation}
Similarly we can show that
\begin{equation}
(\bv_{kL}^*,\bv_{k'L'}^*)_{\Righf\cup\fut} = -\delta(k-k')\delta_{LL'}\quad\text{and}\quad (\bv_{kL},\bv_{kL}^*)_{\Righf\cup\fut} = 0.
\end{equation}

As in the case of $d=4$, we show that in region $\Righf$, we have a Klein-Gordon orthonormal set of modes given by $\bu_{pL}=\psi_{pl}(t,x)Y_L(\Omega_{d-2})$, where
\eq{
\psi_{pl}(t,x) =\sqrt{2\sinh(\pi p)}\,\npl^{(d,\nd)} \,U_{pl}^{(d,\nd)}(x)e^{-ipt},\quad p>0, \label{eq:stat_modes_ddim}
}
with
\eq{
U_{pl}^{(d,\nd)}&=x^l(1-x^2)^{\frac{-ip}{2}}{_2F_1}\left(\frac{\frac{d-1}{2}+l-ip+\nd}{2},\frac{\frac{d-1}{2}+l-ip-\nd}{2},l+\frac{d-1}{2};x^2\right)\nonumber\\
&= x^{2-\frac{d}{2}}U_{p\tl}^{(4,\nd)}(x),
\label{eq:Upl}\\
\npl^{(d,\nd)} &= \frac{H^{\frac{d}{2}-1}}{2\sqrt{2}\pi\Gamma\left(l+\frac{d-1}{2}\right)}\Gamma\left(\frac{l+\frac{d-1}{2}-ip+\nd}{2}\right)\Gamma\left(\frac{l+\frac{d-1}{2}-ip-\nd}{2}\right)\nonumber\\
&=H^{\frac{d}{2}-2}N_{p\tl}^{(4,\nd)},\label{eq:Npl}
}
where $\tl = l+\frac{d}{2}-2$, and the $U_{p\tl}^{(4,\nd)}(x)$ and $ N_{p\tl}^{(4,\nd)}$ carry the extra label
$\nu_d(\mass) \neq \nu_4(\mass)$. The  Klein-Gordon inner product  
\begin{eqnarray}
(\bu_{pL},\bu_{p'L'})_{\Righf} &=& 2(p+p')H^{2-d}\sqrt{\sinh(\pi p)\sinh(\pi p')}e^{i(p-p')t}\npl^{(d,\nd)*}N_{p'l}^{(d,\nd)}\nonumber\\
&&\times\int_0^1 dx\,\frac{x^{d-2}}{1-x^2}U_{pl}^{(d,\nd)}(x)U_{p'l}^{(d,\nd)}(x)\delta_{LL'},\label{eq:kgstatddim}
\end{eqnarray}
where $d\spac=H^{1-d}(1-x^2)^{-1/2}x^{d-2}$  on  the Cauchy hypersurface $\spac_t$ and the future
pointing unit vector normal to the $\spac$ is $\hn^\mu\partial_\mu = H(1-x^2)^{-1/2}\partial_t$. Using the relations Eqn.~\eqref{eq:Upl} and \eqref{eq:Npl}, we see that the $\{\bu_{pL}\}$ are Klein-Gordon orthogonal as in the $d=4$ case,
\begin{equation}
(\bu_{pL},\bu_{p'L'})_{\Righf} = \delta(p-p')\delta_{LL'}.
\end{equation}
We can similarly show that
\begin{equation}
(\bu_{pL}^*,\bu_{p'L'}^*)_{\Righf} = -\delta(p-p')\delta_{LL'}\quad\text{and}\quad (\bu_{pL}^*,\bu_{p'L'})_{KG} = 0.
\end{equation}
Using 
\eq{
\int_0^\infty dk\; k^{-ip-\frac{1}{2}}\varphi_{kl}(\eta,r) = 2^{-ip}e^\frac{\pi p}{2}\npl^{(d,\nd)}U_{pl}^{(d,\nd)}(x)e^{-ipt},\label{eq:phetouint}
} 
we see that the Bogoliubov transformation between $\{\bv_{kL}\}$  and $\{\bu_{pL}\}$  in  $\Righf$ are given by
Eqn.~\eqref{eq:alpkp} and  ~\eqref{eq:betkp} and are the same for all dimensions. This immediately implies that the
mode-wise entropy is given by Eqn.~\eqref{eq:hy}, with an infinite degeneracy coming from the angular modes. Integrating
over $p \in (0, \infty)$ gives us a finite answer as before, but we need to impose an angular cut-off  $l_\mx$ as we did in
$d=4$. The regulated  SSEE  is then 
\eq{ \cS &= \sum_{L, l=0}^{l_\mx} \frac{\pi}{6} =
  \frac{\pi}{6}\frac{(2l_\mx+d-2)(l_\mx+d-3)!}{l_\mx!(d-2)!} \nonumber \\
  & \simeq \frac{\pi}{6} \frac{2}{(d-2)!} l_\mx^{d-2}, \quad l_\mx>>1.}
As in $d=4$ 
using the approximate flatness of the metric at the equator, $d\Omega \simeq (2\pi/l_\mx)^{d-2} \simeq
(l_c H)^{d-2}$,  which means that $ \cS \propto \frac{A_c}{l_c^{d-2}}$.  
\end{subappendices}

\chapter{Spacetime entanglement entropy of a diamond in a 2d cylinder}\label{ch.sseecyl}

In chapter \ref{ch.sseeds} we studied the SSEE in one of the static patches of de Sitter and Schwarzschild de Sitter spacetimes. We find that in de Sitter static patch the SSEE is consistent with the von Neumann entropy evaluated by Higuchi and Yamamoto in \cite{Higuchi:2018tuk}. In these cases, we observe that the restriction of the Wightman function $W\big|_\cO$ to $\cO\subset\cM$ is block diagonal in terms of the Klein-Gordon orthonormal modes $\{\bu_\bpp\}$ in $\cO$. Due to this, we could easily solve the generalised eigenvalue equation Eqn.~\eqref{ssee.eq} and compute the SSEE. In this chapter, we will look at a particular case for which this is not true.

We study the SSEE for a massless scalar field which is in the SJ vacuum state in a finite time slab of 2d cylinder spacetime and is restricted to a causal diamond in it, using a combination of analytical and numerical methods. We observe that the SSEE obtained in a diamond and its complimentary region are equal, which is as expected. We compare our results with the entanglement entropy obtained by Calabrese and Cardy of a CFT which is defined on a circle and is restricted to an arc of it \cite{cc}. We find that in the full cylinder spacetime limit, the SSEE matches exactly with the entanglement entropy obtained by Calabrese and Cardy.

We start this chapter with a brief discussion on the Calabrese-Cardy entropy formula in Sec.~\ref{cc.sec}. In Sec.~\ref{sccc.sec} we setup the calculation of the spacetime version of the Calabrese-Cardy entropy, which is basically the SSEE for a massless scalar field in a cylinder slab restricted to a diamond. We then discuss our results in Sec.~\ref{numres.sec}. We end this chapter with a discussion on the results and some future questions in Sec.~\ref{dis.sec}.

\section{The Calabrese-Cardy entropy formula}\label{cc.sec}
The Calabrese-Cardy formula for the entanglement entropy of a CFT for an interval $\cI_{s}$ of length $s$  in a circle $\cC_{\ell}$ of
circumference  $\ell$ is given by
\begin{equation}
	S = \frac{c}{3} \ln\bigg(\frac{\ell}{\pi \epsilon}\bigg)+ \frac{c}{3} \ln(\sin(\alpha
	\pi))+c_1
	\label{cc.eq}
\end{equation}  
where $\alpha={s}/{\ell}$, $c$ is the CFT central charge, $\epsilon$ is a UV cut-off and $c_1$ is a non-universal
constant.  This formula has been shown to apply to a diverse range of two dimensional systems which fall within the same universality class,  including a geometric realisation by  Ryu and Takayanagi \cite{ryu2006holographic} and others \cite{headrick}. Since entanglement entropy is proposed as a possible contributor to black hole entropy, understanding Eqn.~(\ref{cc.eq}) from a spacetime perspective is of broad interest. 

As a follow up to their earlier work, Calabrese and Cardy studied
the unitary time evolution of the entanglement entropy for  an interval $\cI_s$
inside  a larger interval $\cI \supset
\cI_s$. Starting with a pure state, which is an eigenstate of a "pre-quench" Hamiltonian, and then quenching the system at $t=0$,  they used  path  integral techniques  to show that the entanglement entropy increases with time.
It then saturates  after the ``light-crossing'' time, in keeping with causality \cite{evolcc}. This corresponds to the ``time''  required for the domain of dependence of $\cI_s$ to be fully
defined. Seeking out a covariant formulation of entanglement entropy  is therefore of interest to understanding 
the results of \cite{evolcc}  in a spacetime language.

\section{Spacetime calculation of the Calabrese-Cardy entropy}\label{sccc.sec}

In \cite{saravani2014spacetime} the SSEE for nested causal diamonds $\diam_s \subset \diam_S$ was shown to yield the first, cut-off dependent term of Eqn.~(\ref{cc.eq})  with $c=1$ when $s<<S$.  Since $\diam_s$ is the domain of dependence of $I_s$, this is the natural spacetime analogue of $\cI_s \subset \cI_S$. In this chapter we calculate the SSEE for the spacetime analogue of $\cI_s \subset \cC_\ell$ for finite $\ell$ and additionally, find the same $\alpha$-dependence as Eqn.~(\ref{cc.eq}), thus explicitly demonstrating  complementarity. 
A natural spacetime analogue of $\cC_{\ell} $ is its (zero momentum) Cauchy completion, which is the $d=2$ cylindrical spacetime $(\cM,g)$ with $ds^2=-dt^2 + dx^2, \,  x+\ell \sim x$. The domains of dependence of $\cI_{s} $ and its complement $\cI_{\ell-s}$ in $(\cM,g)$ are the  causal diamonds $\diam_{s}$ and $\diam_ {\ell -s}$ respectively, as shown in Fig \ref{cyl.fig}.

\begin{figure}[htb]
	\centering{\includegraphics[height=6.6cm]{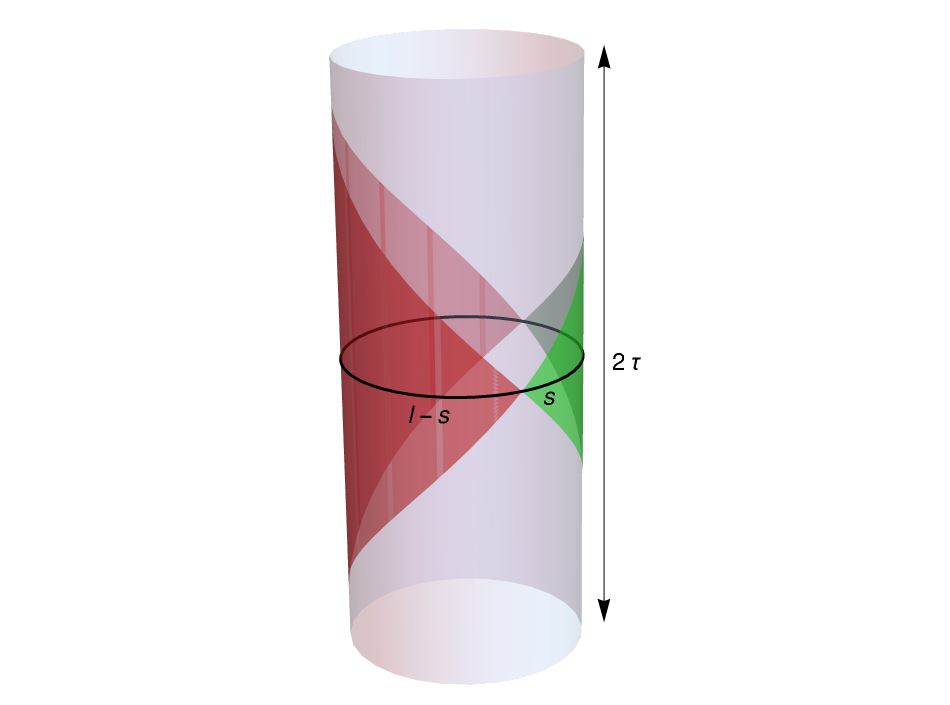}}
	\caption{The spacetime analogues of $\cI_{s}, \cI_{\ell-s} \subset \cC_\ell$ are their domains of dependence  $\diam_s$ and $\diam_{\ell-s}$ in $(\cM,g)$ shown in green and red respectively.}
	\label{cyl.fig} 
\end{figure}
 
In what follows we use a mixture of analytical and numerical methods to solve the SSEE eigenvalue problem Eqn.~\eqref{ssee.eq}. We will find it convenient to work with the Sorkin-Johnston (SJ) formulation 
 discussed in previous chapters, where the SJ spectrum  provides the covariant UV cut-off with which to calculate $\ssee$, as was done in \cite{saravani2014spacetime}.

For our calculation of $\ssee$ we will consider the free massless scalar field to be in the SJ vacuum state in  a  slab  $(\mtau, g)$ of
height $2 \tau$  in the cylinder  spacetime \cite{fewster2012}, and its
restriction to $\diam_s \subset \mtau$. The SJ Wightman function in $\mtau$ is
\begin{equation}
	\wtau(x,t; x',t') = \sum_{m\in \mathbb Z} \cylev_{m} \psi_m(x,t) \psi^*_m(x',t'),  \label{cylsj.eq}
\end{equation} 
where $\{ \psi_m, \cylev_m\}$ are the $\ccL^2$ normalised positive frequency SJ eigenmodes and eigenvalues in $\mtau$
\cite{fewster2012}: 
\begin{align} 
	&\psi_m(x,t) = \biggl(\!\frac{\!(1 \!- \!\zeta_m\!)}{2\sqrt{2\ell}\cm}\!e^{i \!\frac{2 \pi |m| t}{\ell}} \! +  \!\frac{\!(\!1 \!+ \!\zeta_m\!)}{2\sqrt{2\ell}\cm} \!e^{-i \!\frac{2 \pi |m|t}{\ell}}\!\biggr)e^{i \!\frac{2 \pi m x}{\ell}} \nonumber  \\ 
	&\cylev_m  =   \ell
	\frac{\sm\cm }{2 \pi |m|}, \;    \zeta_m=\frac{\cm}{\sm}, \;  \taubyl=\frac{2\tau}{\ell}, \quad m \in \mathbb Z, \nonumber\\
	&\cm^2\! = \! \tau\left(1+\sinc(2 |m| \pi \taubyl)\right), \; \sm^2\! =\! \tau\left(1-\sinc( 2 |m| \pi \taubyl)\right).
	\label{cylsjef.eq}
\end{align} 

The $m=0$ ``zero mode''  in particular takes the  form\footnote{Note that ``$m$" here denotes the quantum numbers labelling the SJ modes in the cylinder slab spacetime. It should not be confused with mass of the field, which is zero throughout this chapter.}
\begin{equation} 
	\psi_0(t)= \frac{1}{2\sqrt{\tau l}}\left(1 - i \frac{\sqrt{3}}{\tau}t\right),  \quad \cylev_0=
	\frac{2}{\sqrt{3}}\tau^2. \label{zm.eq}
\end{equation}
Unlike  the standard vacuum on the cylinder, $\wtau$ is $\tau$-dependent. Each  $\wtau$  can however be viewed as  a pure (non-vacuum) state in $\mtaup$  for any $\taup >
\tau$, as we will later show. To accommodate both  $\diam_{s}$ and 
$\diam_{\ell-s}$ in our calculations,  we require $ 2 \tau  \geq s, \ell-s$.

The  SJ modes in $\diam_{s}$ are naturally expressed in terms of
the light cone coordinates $u=\frac{1}{\sqrt{2}}(t+x), v=\frac{1}{\sqrt{2}}(t-x)$ and   come in the two mutually $\cL^2(\diam_s)$ orthogonal series   \cite{johnston}\footnote{A discussion on this can also be found in chapter \ref{ch.2ddiamsj} of this thesis.}
\begin{align}
	f_k \!&\!=\! e^{-iku}\!-\!e^{-ikv},  \quad  k = 2 \sqrt{2} n \pi/s \nonumber \\ 
	g_\kappa \!&\!=\!  e^{-i\kappa u}\!+\!e^{-i\kappa v}\!-\!2\cos\bigg(\frac{\kappa s}{2 \sqrt{2}}\bigg),\quad \tan\bigg(\frac{\kappa s
	}{2\sqrt{2}}\bigg)=\frac{\kappa s}{\sqrt{2}}\, \, 
	\label{diamsj.eq}
\end{align}
with eigenvalues $\lambda_{k}= \dfrac{s}{2\sqrt{2}k}$ and $ \lambda_{\kappa}=
\dfrac{s}{2 \sqrt{2}\kappa}$, respectively, and with $\ccL^2$ norm in  $\diam_{s}$ 
\begin{equation}
	{||f_k||^2 =s^2  , \quad ||g_\kappa||^2 = s^2\left(1-2\cos^2\left(\frac{\kappa s}{2\sqrt{2}}\right)\right)}. 
\end{equation} 
Since $i \hD $  is diagonal in this basis we will use it to transform Eqn.~(\ref{ssee.eq}) to the matrix form 
\begin{equation} 
	\hW_\tau|_{\diam_s} X   =  \mu \Lambda X,  
\end{equation} 
where $\Lambda$ is the diagonal matrix $\{\lambda_k, \lambda_\kappa \}$.  
For $X \not\in
\kker(i\hD)$, we can invert this to suggestively write 
\begin{equation}
	\hr X =  \Lambda^{-1} \hW_\tau|_{\diam_s} X     =\mu X,\label{hrho.eq} 
\end{equation} 
so that $\ssee$ can be viewed as  the
von-Neumann entropy of $\hr$.  The spectrum of $\hr$ is  unbounded and hence needs 
a UV cut-off. As in \cite{saravani2014spacetime}
we use the  covariant UV-cut off with respect to  the SJ spectrum $\{\lambda_k,\lambda_\kappa\}$. For large $\kappa$  the condition $\tan({\kappa s}/{2\sqrt{2}})=\kappa s/\sqrt{2}$ can be approximated by $\kappa \sim\sqrt{2}(2n+1)\pi/s$, so that a consistent choice of cut-off for both sets of eigenvalues is $\epsilon=k_\mx^{-1}= s/(2 \sqrt{2} n_\mx \pi)$.  We also  need to ensure that this same cut-off is used in the causal complement, i.e., $k_\mx=2 \sqrt{2} n'_\mx \pi/(\ell-s)$, where $n'$ denotes the quantum number for the SJ spectrum in $\diam_{\ell -s}$, 
so that 
\eq{\epsilon \!=\!\dfrac{\ell \alpha}{2\sqrt{2}\pi n_\mx}\!=\! \dfrac{\ell (1-\alpha)}{2\sqrt{2}\pi n_\mx'},\quad \alpha\equiv s/l.}

We define the matrix elements of $\hwdt$ in the $\{f_k\}$ and $\{g_\kappa\}$ basis as $\hW_{kk'}\equiv\langle f_k,\hwdt f_{k'}\rangle$ and $\hW_{\kappa\kappa'}\equiv\langle g_\kappa,\hwdt g_{\kappa'}\rangle$, where $\langle.,.\rangle$ denotes the $\cL^2$ inner product in $\diam_s$. Using the series form of $\hW_\tau$ given by Eqn.~\eqref{cylsj.eq}, for general $\alpha$ and $\gamma$ we find
\begin{align}
	\hW_{kk'}\!&=\! \frac{s^4}{32\pi} \sum_{m> 0}\frac{1}{|m|\zeta_m} \biggl(\eta^-_m\sinc(x_+)-\eta^+_m \sinc(x_-)\biggr) \times \biggl(\eta_m^-\sinc(x_+') - \!\eta_m^+ \sinc(x_-') \biggr) \nonumber \\
	\hW_{\kappa\kappa'}\!&=\! \frac{s^4}{32\pi} \sum_{m> 0}\! \frac{1}{|m|\zeta_m} 
	\biggl(\eta_m^-\sinc(z_+) +\eta_m^+ \sinc(z_-)\biggr) \times \biggl(\eta_m^-\sinc(z_+') + \eta_m^+ \sinc(z_-')
	\biggr) \nonumber\\
	&\hspace{8cm}+\wz_{\kappa\kappa'}  \label{wmatrix.eq} 
\end{align}
where $\eta_m^\pm = 1\pm \zeta_m$, $x_\pm=(n\pm\alpha m)\pi$,  $x_\pm'=(n'\pm \alpha m)\pi$,  $z_\pm=\kappa s/2\sqrt{2}\pm \alpha m\pi$,   $z_\pm'=\kappa' s/2\sqrt{2}\pm \alpha m\pi$,  and  the contribution from the zero mode is
\begin{align} 
	\wz_{\kappa\kappa'} &= \frac{s^4}{2 \sqrt{3}} \frac{\tau}{\ell} \cos(\kappa s/(2\sqrt{2})) \cos(\kappa' s/(2\sqrt{2}))\nonumber \\ 
	& \times \biggl(1+ \sqrt{\frac{3}{2}} \frac{1}{\kappa \tau} \biggr)   \biggl(1+ \sqrt{\frac{3}{2}} \frac{1}{\kappa'
		\tau}\biggr).  \label{wzeromode.eq} 
\end{align}

Our strategy is to  construct $\hr\,$ from these matrix elements and to solve for its eigenvalues using a numerical matrix solver. However, each matrix elements in Eqn.~(\ref{wmatrix.eq}) is an infinite sum over the quantum number $m$ and hence not amenable to explicit calculation. 
We therefore need to find a closed form expression for the above matrix elements. 

We notice that when $\taubyl$ takes  half-integer values (for which the SJ vacuum is Hadamard \cite{fewster2012}), $\zeta_m=1$ for $ m \neq 0$, which leads to a considerable simplification. Further, let   $\alpha$ be rational,  so that we can write  $\alpha=\frac{p}{q}$, with   $p, q \in \mathbb Z,$ and $  p, q  >0$  being relatively prime. For these choices of $\alpha$ and $\gamma$, the infinite sums of Eqn.~(\ref{wmatrix.eq}) reduce to the following finite sums over Polygamma functions $\pg(x)$ and $\pg^{(1)}(x)$
	
	\begin{align}
		\hW_{kk'}& =  \frac{ s^4}{8 \pi n}  \Biggl[ \delta_{n,n'} \biggl( \alpha  \Theta(n) \sum_m
		\delta_{n,m\alpha} + 
		\frac{1}{\pi^2 \alpha q^2 n} \sum_{r=1}^{q-1} \sin^2(r\alpha \pi) \biggl[-\alpha q  \pg\!\Bigl(\frac{r}{q}\Bigr) + \alpha q
		\pg\!\Bigl(\frac{\alpha r -n}{\alpha q}\Bigr)\nonumber\\
		&+ n \pg^{(1)}\!\Bigl(\frac{\alpha r  -n}{\alpha q}\Bigr)\biggr]\biggr) 
		+ (1-\delta_{n,n'}) \frac{(-1)^{n+n'}}{\pi^2 n'(n-n') q}  \sum_{r=1}^{q-1} \sin^2(r \alpha \pi) \biggl[ (n'-n)
		\pg\!\Bigl(\frac{r}{q}\Bigr)\nonumber\\
		& - n' \pg\!\Bigl(\frac{\alpha r -n}{\alpha q}\Bigr) + n \pg\!\Bigl(\frac{\alpha r
			-n'}{\alpha q} \Bigr)\biggr] \Biggr]  \nonumber \\
		\hW_{\kappa\kappa'} &=s^4\cos\left(\frac{\kappa
			s}{2\sqrt{2}}\right)\cos\left(\frac{\kappa's}{2\sqrt{2}}\right)\Biggl[\frac{\tau}{2\sqrt{3}
			\ell}\biggl(1+\sqrt{\frac{3}{2}} \frac{1}{\tau\kappa}\biggr) \biggl(1+\sqrt{\frac{3}{2}} \frac{1}{\tau\kappa'}\biggr)\nonumber\\
		& +  \delta_{\kappa,\kappa'} \frac{1}{\alpha q^2s^2\kappa^2\pi^2}\biggr(
		\sum_{r=1}^{q-1}\Omega(\kappa,\kappa',\alpha,r)\biggl[\alpha q \pi\biggl(\pg\!\Bigl(\frac{r}{q}-\frac{\kappa
			s}{\eta}\Bigr) - \pg\!\Bigl(\frac{r}{q}\Bigr)\biggr) +\frac{\kappa
			s}{2\sqrt{2}}\pg^{(1)}\!\Bigl(\frac{r}{q}-\frac{\kappa s}{\eta}\Bigr)\biggr] \nonumber\\
		& +   \frac{s^2\kappa\kappa'}{2} \biggl[ \alpha q  \pi\biggl(\eul + \pg\!\Bigl(1-\frac{\kappa s}{\eta}\Bigr)\biggr) + \frac{\kappa
			s}{2 \sqrt{2}}\pg^{(1)}\!\Bigl(1-\frac{\kappa s}{\eta}\Bigr)\biggr] \biggr)  +  (1-\delta_{\kappa,\kappa'}) \frac{1}{s^2 q \kappa\kappa' (\kappa-\kappa')}\nonumber\\
		&\times \biggl( 
		\sum_{r=1}^{q-1}\Omega(\kappa,\kappa',\alpha,r) \biggl[\kappa\pg\!\Bigl(\frac{r}{q}-\frac{\kappa's}{\eta}\Bigr) - \kappa'\pg\!\Bigl(\frac{r}{q}-\frac{\kappa s}{\eta}\Bigr) - (\kappa-\kappa')\pg\!\Bigl(\frac{r}{q}\Bigr)\biggr]\nonumber\\
		&+ \frac{s^2\kappa\kappa'}{2} \biggl[\eul(\kappa-\kappa') + \kappa\pg\!\Bigl(1-\frac{\kappa's}{\eta}\Bigr)
		-\kappa'\pg\!\Bigl(1-\frac{\kappa s}{\eta}\Bigr) \biggr] \biggr)\Biggr], \qquad \eta=2\sqrt{2} \alpha q
		\pi  \label{wmatrixgg.eq} 
	\end{align} 

where $k$ and $n$ are related as in Eqn.~\eqref{diamsj.eq}, $\eul$ denotes the Euler-Mascheroni constant and
\begin{align}
	\Omega(\kappa,\kappa',\alpha,r)&=\kappa\kappa' \frac{s^2}{2} \cos^2(\alpha r\pi) + \sin^2(\alpha r\pi)\nonumber \\
	&-(\kappa+\kappa')\frac{s}{2\sqrt{2}}\sin(2\alpha r\pi).
\end{align}

We now have the finite series expressions for the matrix elements of $\hwdt$ in the $\{f_k,g_\kappa\}$ basis. We are therefore in a position to construct $\hr$ upto a covariant UV cut-off with respect to the SJ spectrum, and solve for its eigenvalues numerically, which we will do next.

\subsection{Numerical results}\label{numres.sec}
We now solve for the eigenvalues of $\hr$ using Mathematica's numerical eigenvalue solver. We consider a range of values of $\alpha, \gamma$ and the cut-off $n_\mx/\alpha$ given in the table below. 

\begin{table}[h]
	\begin{center}
		{\renewcommand{\arraystretch}{1.2}
			\begin{tabular}{|c|c|}
				\hline
				\rule{0pt}{3ex}$\alpha$&
				$\frac{1}{10},\frac{1}{5},\frac{1}{4},\frac{1}{3},\frac{1}{2},\frac{2}{3},\frac{3}{4},\frac{4}{5},\frac{9}{10}$ 
				\\ [1ex]\hline \rule{0pt}{3ex}
				$\gamma$& $1,2,4,6,8,16,21.5,32,40.3,100,200,1000,2000$\\  [1ex] \hline \rule{0pt}{3ex}
				$\frac{n_\mx}{\alpha}$& $1000,1200,1400,1600,1800,2000,2200,2400,2600$\\[1ex] 
				\hline
			\end{tabular}
		}
	\end{center} 
\end{table}
In the list of $\gamma$ values, we have also included the specific non-half-integer value  of
$\gamma =40.3 $ for which $\zeta_m \sim 1$ even for $m =1$.   In general, we note that   $\zeta_m \sim 1$ for  $m
>> \gamma^{-1}$. The error coming from small $m$ terms
has been explicitly calculated in this case as a function of $m$ and seen to be small.
For the special case $\alpha=0$,  $\ssee$ is trivially zero, while for
$\alpha=1$, the domain of dependence of $\cC_\ell$ is no longer a causal diamond, but all of $\mtau$. Since
$\hwt$ is the SJ vacuum  and therefore pure, $\ssee=0$.

Fig.~\ref{salpha.fig}  shows the results of simulations for these  various $\alpha$ and $\gamma$ values, for cut-off $n_\mx/\alpha=1200, 2000$ and $2600$. It is clear that $\ssee$ satisfies complementarity. 
The SSEE can be fitted to the form $\ssee = a_1\log(\sin(\alpha\pi))+b_1$. 
The values of $a_1$ and $b_1$ along with their errors are given in the tables in Fig.~\ref{salpha.fig}, where $a_1$ can be seen to be independent of $n_{\mx}/\alpha$. It is however dependent on $\gamma$ and asymptotes to the universal value of $1/3$ in the Calaberse-Cardy formula for $\gamma>>1$.
\begin{figure*}
	\centering
	\subfloat[$n_{\mx}/\alpha=1200$]{\includegraphics[height=5.4cm]{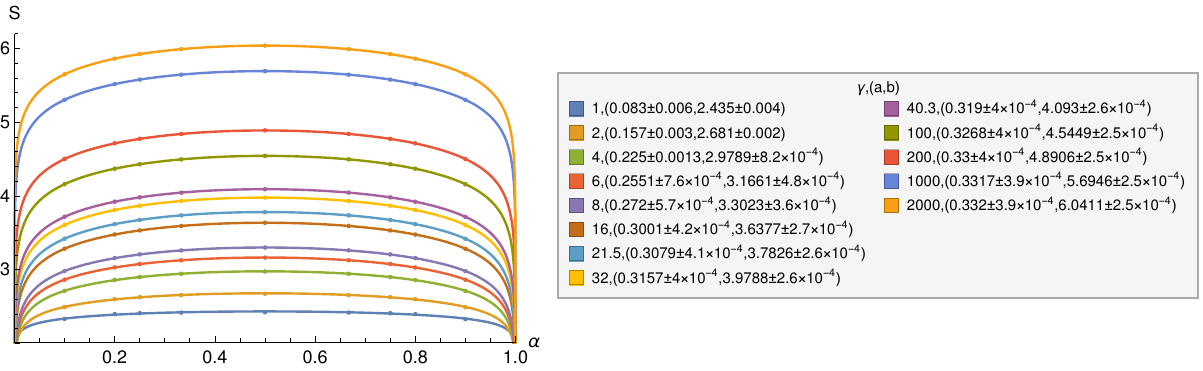}}\hskip 0.1in
	\subfloat[$n_{\mx}/\alpha=2000$]{\includegraphics[height=5.4cm]{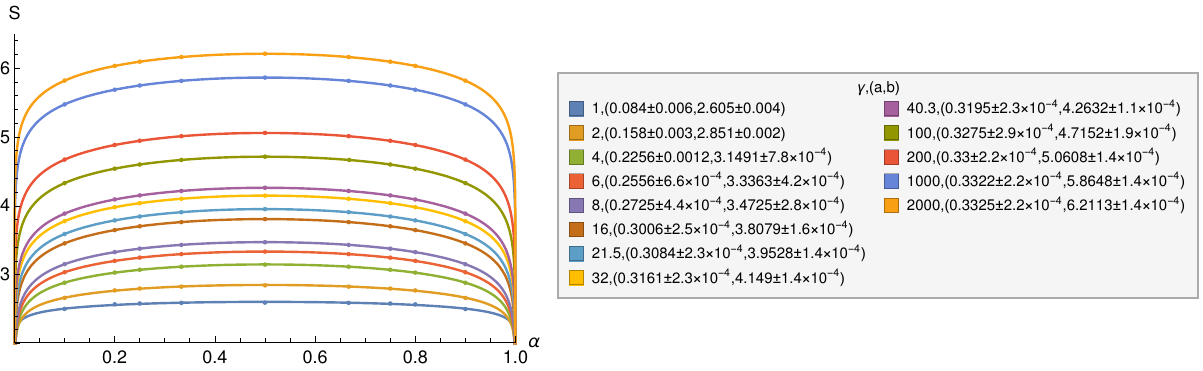}}\hskip 0.1in
	\subfloat[$n_{\mx}/\alpha=2600$]{\includegraphics[height=5.4cm]{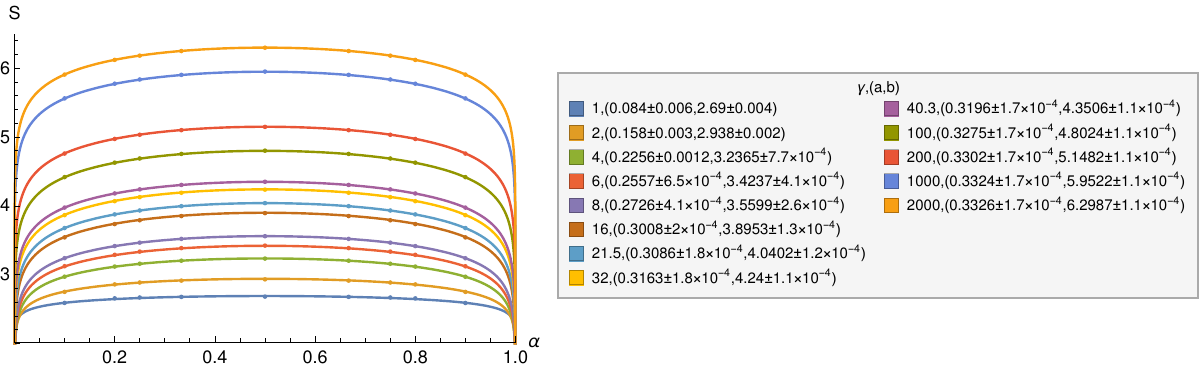}}
	\caption{$\ssee$ vs $\alpha$ for different $\gamma$ with $n_\mx/\alpha=1200,\, 2000$ and $2600$ fitted to 
		$\ssee=a\log(\sin(\pi\alpha))+b$. The fit parameters are shown in the table. }
	\label{salpha.fig}
\end{figure*}

Fig.~\ref{Compl.fig} shows the dependence of $\ssee$ on $n_{\mx}/\alpha$ for different $\alpha$ and with three different values of $\gamma$ $(16, 200\text{ and }1000)$. Here SSEE can be fitted to the form, $\ssee = a_2\log(n_{\mx}/\alpha) + b_2$. As is clear from the tables in this figure $a_2\approx 0.33\approx 1/3$ for all $\alpha$ and $\gamma$ with the order of error given in the table. $b_2$ however depends on $\alpha$ and $\gamma$.
\begin{figure*}
	\centering
	\subfloat[$\gamma=16$]{\includegraphics[height=5cm]{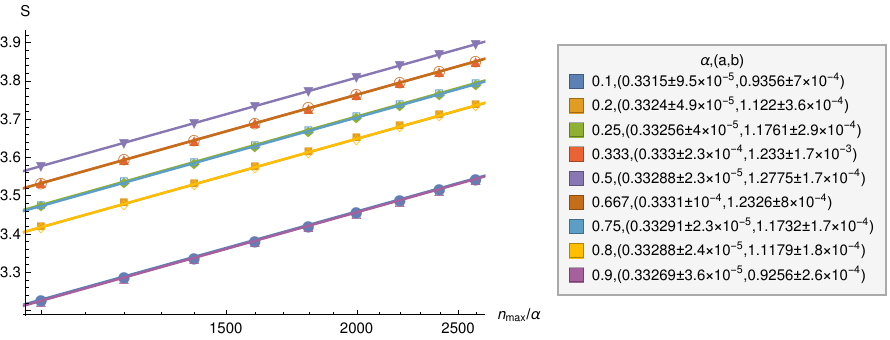}}\\
	\subfloat[$\gamma=200$]{\includegraphics[height=5cm]{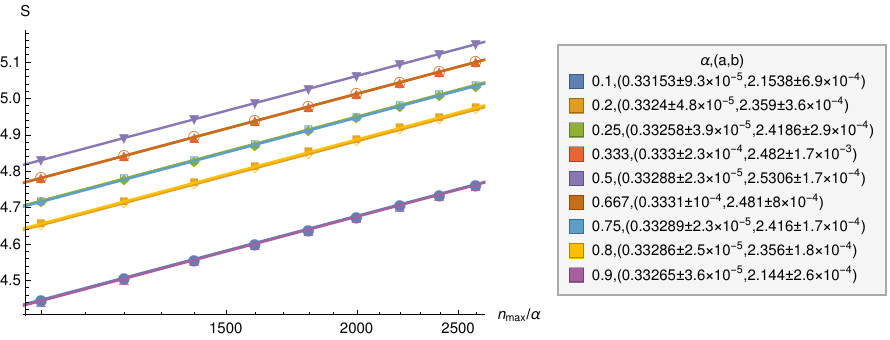}}\\
	\subfloat[$\gamma=1000$]{\includegraphics[height=5cm]{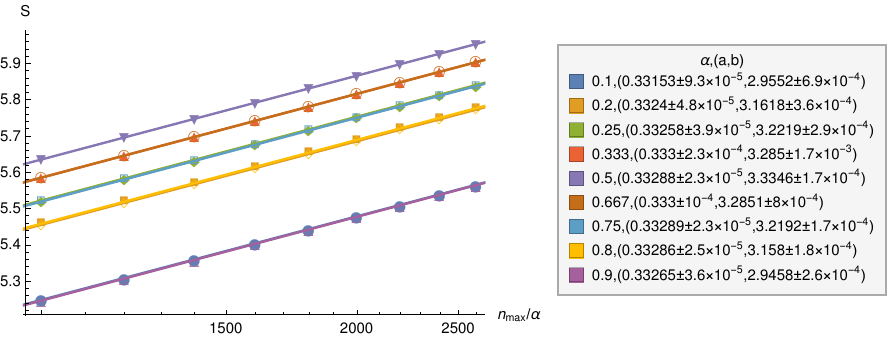}}
	\caption{A log-linear plot of  $\ssee$ vs $n_\mx/\alpha$ for different $\alpha$ with $\gamma=16,\,200$ and $1000$ fitted to $\ssee=a\log\,(n_\mx/\alpha)+b$. The fit parameters are shown in the table which show $a \sim 1/3$. This is also true for other values of $\gamma$. The curves for  complementary values of $\alpha$ are indistinguishable.} \label{Compl.fig}
\end{figure*} 

The results in Fig.~\ref{salpha.fig} and \ref{Compl.fig} suggest that  $\ssee$ takes the general  form
\begin{equation}
	\ssee  = \frac{c(\gamma)}{3} \ln\bigg(\frac{\ell}{\pi \epsilon}\bigg)+ {f(\gamma)}\ln(\sin(\alpha
	\pi))+c_1(\gamma).  \label{ourresult.eq}
\end{equation}
where $c(\gamma)/3\equiv a_2$ and $f(\gamma)\equiv a_1$. Using the values of $a_2$ given in the tables of Fig.~\ref{Compl.fig}, we find that $c(\gamma) \sim 1$. We fit $f(\gamma)$ values to the form
\begin{equation}
f(\gamma)=0.33 + a_3/\gamma + b_3/\gamma^2
\end{equation}
and find that $a_3\approx-0.48$ and $b_3\approx0.23$ with the error given in the tables of Fig.~\ref{fgamma.fig}.

In order to extract $c_1(\gamma)$ we subtract the first term in Eqn.~\eqref{ourresult.eq} (which depends on $n_{\mx}/\alpha$) using $c(\gamma)/3$ given by the values of $a_2$ in the table of Fig.~\ref{Compl.fig} from the values of $b_1$ in the table of Fig.~\ref{salpha.fig} for $n_{\mx}/\alpha = 1200, 2000$ and $2600$. We find that the difference (or $c_1(\gamma)$) is independent of the choice of $n_{\mx}/\alpha$ which is as expected. We fit the dependence on $\gamma$ by 
\begin{equation}
c_1(\gamma)=a_4\log(\gamma)+b_4
\end{equation}
and values of the coefficients $a_4$ and $b_4$ are given in the table in the Fig.~\ref{conegamma.fig}.


\begin{figure}[h!]
	\centering
	\includegraphics[scale=0.8]{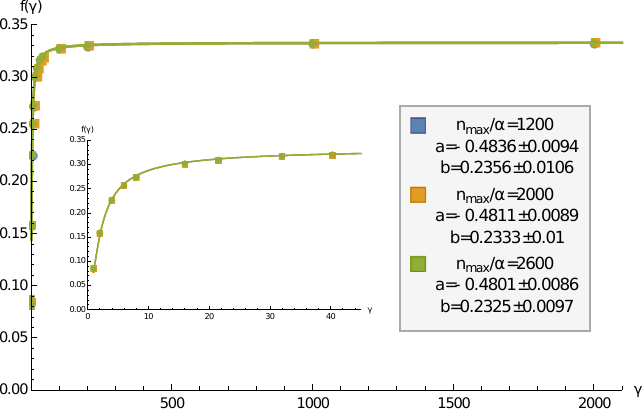}
	\caption{A plot of $f(\gamma)$ vs. $\gamma$ for different values of $n_\mx/\alpha$, fitted to 
		$0.33+a/\gamma+b/\gamma^2$.  
		The inset figure shows the smaller $\gamma$ values.} \label{fgamma.fig} 
\end{figure}
\begin{figure}[!h]
	\centering
	\includegraphics[scale=0.8]{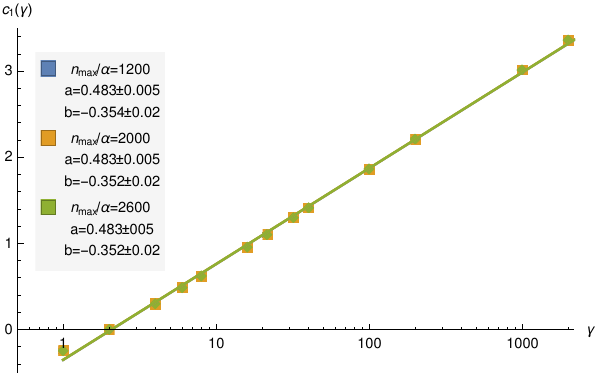}
	\caption{A log-linear plot of $c_1(\gamma)$ vs $\gamma$ for different values of $n_\mx/\alpha$ 
		fitted to  $a\log\gamma+b$. }\label{conegamma.fig}
\end{figure}

Thus, the first term of Eqn.~(\ref{cc.eq}) is reproduced for any choice of $\alpha$ and $\gamma$.  
This  generalises the results of \cite{saravani2014spacetime}  where this was shown in the limit of $\alpha \ll 1$. The dependence on $\alpha$, i.e., the second term of  Eqn.~(\ref{cc.eq}) is also reproduced and hence exhibits complementarity for any $\alpha$ 
(see Fig.~\ref{salpha.fig}, ~\ref{Compl.fig}).  Its  coefficient however is {\it not}
universal and depends on $\gamma$ as shown in Fig.~\ref{fgamma.fig}.  However, as $\gamma >>1$,  
$f(\gamma) $ does asymptote to the universal value $1/3$. Finally,  the non-universal constant
$c_1(\gamma)$  diverges logarithmically  with $\gamma$ as shown in 
Fig.~\ref{conegamma.fig}. This can be traced to the IR divergence in the zero modes of the massless theory.

\subsection{A look at the generalised eigenvalues}
We now take a look at the eigenvalues of $\hr$ defined in Eqn.~\eqref{hrho.eq}. We find that the eigenvalues (which always come in pairs $(\mu,1-\mu)$) exhibit the  surprising feature
that all but one pair hovers around the values $0$ and $1$ and
hence contributes significantly to $\ssee$. In Fig.~\ref{ev.fig} we also show the comparison of the eigenvalues obtained in the two complementary regions, we find that they differ only in the numbers of $(0,1)$ pairs. Further, if we calculate $\ssee$ for the largest pairs of eigenvalues, we find that the
error is small, as shown in Fig~\ref{evtwo.fig}.
\begin{figure*}[!ht]
	\centering{\subfloat{\includegraphics[height=5cm]{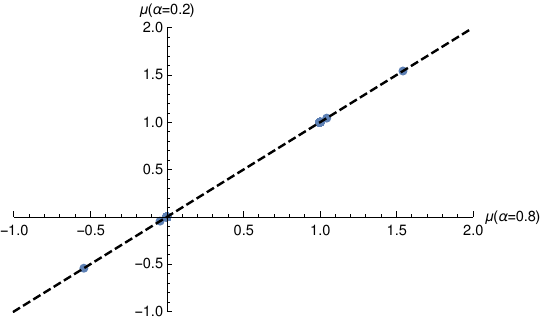}}\,\subfloat{\includegraphics[height=5cm]{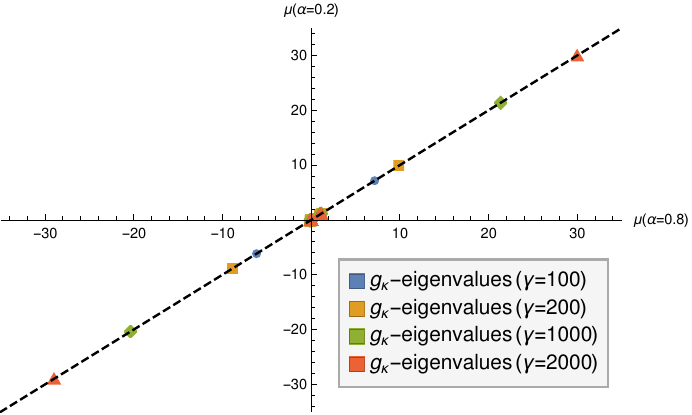}}}
	\caption{A plot comparing the eigenvalues $\mu$ of the entropy equation for one choice of complementary regions
		with $\alpha=0.2,\,0.8$ for $n_\mx/\alpha=2600$ and different choices of $\gamma$. On the left are the eigenvalues
		associated with  the $f_k$ matrix elements which are independent of $\gamma$ and on the right are those associated
		with the $g_\kappa$, which are $
		\gamma$ dependent.  We note that
		the number of eigenvalues differ in both regions but only in the  number of $(1, 0)$ pairs which  leads to the
		equality of the SSEE in these complementary regions. Further, the significant contribution comes from the $g_\kappa$
		matrix elements of which there are precisely {\it two} which are substantially different from $(1,0)$. These increase
		with $\gamma$ and are the main contributors  to $c_1(\gamma)$} 
	\label{ev.fig} 
\end{figure*}

\begin{figure*}[!ht]
	\centering
	\subfloat[$\alpha=0.5$, $n_\mx/\alpha=2600$]{
		\includegraphics[width=5cm]{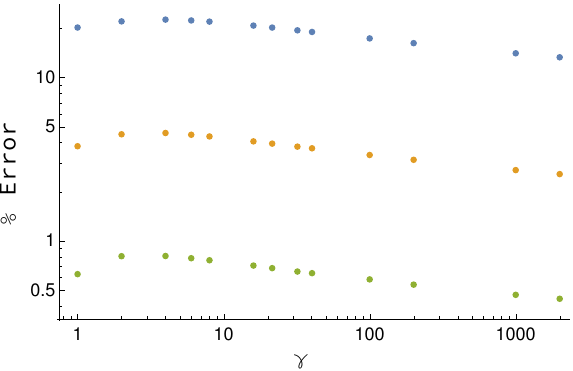} }\, 
	\subfloat[$\alpha=0.5$, $\gamma=1000$]{\includegraphics[width=5cm]{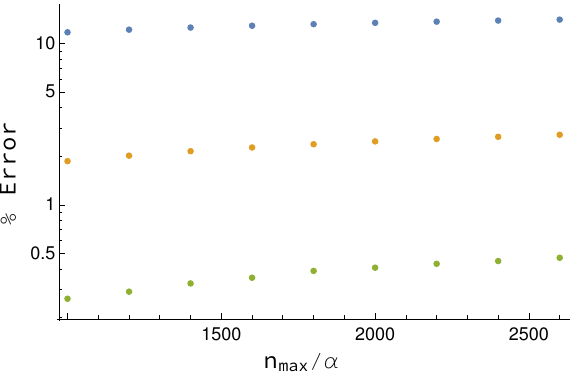}} \, 
	\subfloat[$n_\mx/\alpha=2600$, $\gamma=1000$]{\includegraphics[width=5cm]{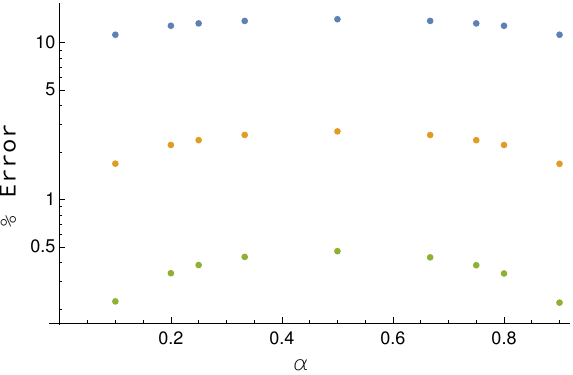}} 
	\caption{In order to estimate the contribution of the pairs $(\mu,1-\mu)$,  we plot the percentage error in the SSEE
		when only the largest pairs (one, two and three represented in blue, orange and green respectively) of eigenvalues are considered, as a function of the different
		parameters $\gamma, n_\mx/\alpha$ and $\alpha$. In each case, we see that the error goes down to $< 1\%$ even when
		only the  3 largest eigenvalues are retained.}
	\label{evtwo.fig} 
\end{figure*}

\section{Discussions}\label{dis.sec}
In this chapter, we studied the SSEE of a massless scalar field in the SJ vacuum state in a finite time slab of height $2\tau$ of a cylinder spacetime of circumference $l$, restricted to a causal diamond $\diam_s$ in the center of the slab. The region of interest $\diam_s$ is the domain of dependence of an interval $\cI_s$ of length $s$ in a circle of circumference $l$. It is therefore natural to compare the SSEE obtained for the massless scalar field in $\diam_s$ with the entanglement entropy of a CFT in the interval $\cI_s$, which is given by Calabrese and Cardy and is of the form given by Eqn.~\eqref{cc.eq}. We compute the SSEE using a combination of analytical and numerical techniques in $\diam_s$ and find that the cut-off dependent term (using the covariant cut-off in the SJ spectrum) is identical to the one obtained by Calabrese and Cardy. The coefficient of the $\alpha$ dependent term is however dependent on $\gamma$, which is the ratio of the height $2\tau$ and the circumference $l$ of the cylinder. However it asymptotes to the universal value of $1/3$ in full cylinder limit, i.e., $\gamma>>1$. We find the non-universal term $c_1(\gamma)$ to have a log dependence on $\gamma$. The contribution to this term comes primarily from the zero modes given by Eqn.~\eqref{wzeromode.eq}.

The behaviour of $f(\gamma)$ can be viewed as a
dependence on the choice of the pure state $W_\tau$ in $\mtaup$, for $\mtaup \supset \mtau$. From  Eqn.~(\ref{cylsj.eq}) we see that $\hW_\tau$ is a state in $\mtaup$, i.e., 
$\hW_\tau=\hR_\tau+ i \hD/2$, where $\hR_\tau$ is real and symmetric.  

Expanding $i\hD$ in    the SJ modes $\{ \psitp_m \}$ of $\mtaup$ and $\hW_\tau$  in $\{ \psit_m\}$, and inserting in Eqn.~(\ref{ssee.eq}) 
we see that term by term 
\be
\cylev^{\tau}_m\psit_m(x,t) A_m= \mu
\cylev^\taup_m[\psitp_m(t,x) \Ap_m\ \!-\! \psistp_m(t,x)\Bp_m],  \nonumber 
\ee
where $A_m \!=\! \langle\psit_m,\chi\rangle_\taup$, $\Ap_m\!=\!\langle\psitp_m,\chi\rangle_\taup$,  $\Bp_m\!=\!\langle\psistp_m,\chi\rangle_\taup$ and $\langle.,.\rangle_\taup$ is the $\cL^2$ inner product in $\mtaup$. 
Expanding  $\psit_m=a_m\psitp_m + b_m\psistp_m$,   $a_m=\frac{\sm^{\tau'}}{2 \cm^{\tau}} (\zeta_m^ {(\tau')}
+\zeta_m^ {(\tau)}) $,  $b_m=\frac{\sm^{\tau'}}{2 \cm^{\tau}} (\zeta_m^ {(\tau')}
-\zeta_m^ {(\tau)}) $  this simplifies to
\begin{align}
	\cylev^{\tau}_m a_m\!\left(\! a_m \Ap_m+ b_m \Bp_m\!\right) 
	\!&=\! \mu\cylev^{\taup}_m \Ap_m \nonumber\\
	\cylev^{\tau}_m b_m\!\left(\!a_m \Ap_m + b_m \Bp_m\!\right)\!&=\!-\mu\cylev^{\taup}_m \Bp_m.                               \end{align}
The solutions for this are  either $a_m \Ap_m + b_m \Bp_m =
0 \Rightarrow \mu=0$, or 
$
a_m \Bp_m + b_m \Ap_m = 0 \Rightarrow \mu
=\frac{\cylev^{\tau}_m}{\cylev^{\tau'}_m}\left(a_m^2-b_m^2\right)=1,   
$
which means that $\hW_\tau$ is a pure state in $\mtaup$. Thus $f(\gamma)$ can be viewed as the dependence  on the choice of  pure state in  $\mtaup$ for $\mtau \subset\mtaup$. 

We end with some remarks. While we have demonstrated complementarity for certain rational values of $\alpha$,  an analytic  demonstration using Eqn.~(\ref{wmatrixgg.eq}) seems non-trivial, in part because the UV regulated matrices $\hr_\alpha$ and $\hr_{1-\alpha}$  are of different
dimensions.  
Conversely, complementarity implies that if $n_\mx>n_\mx'$,   $\hr_\alpha=\hr_{1-\alpha}
\oplus \mathbf{1}_{N} \oplus \mathbf{0}_{N}$, where $\mathbf{0}$ is the zero matrix and $N=(n_\mx-n_\mx')/2$.

In our computations we find that the eigenvalues of $\hr$ (which always come in pairs $(\mu,1-\mu)$) exhibit the  surprising feature
that all but one pair hovers around the values $0$ and $1$, 
thus contributing most significantly to $\ssee$. Indeed, the  $\ssee$ calculated using the largest few pairs of eigenvalues accounts for most of the entropy, as shown in Fig.~\ref{evtwo.fig}.

Finally, it would be interesting to study the SSEE for the non-zero mass case which is IR divergence free. While the small mass approximation of the SJ modes in $\diam_s$ is known \cite{Mathur:2019yvl}, the challenge will be to obtain closed form expressions for the matrix elements of $\hW$ as we have done here for the massless case.

\chapter{Conclusion}\label{ch.concl}

The theme of this thesis is a study of a real scalar quantum field theory in different curved spacetimes in a covariant way which can readily be extended to quantum gravity theories like Causal Set Theory (CST). The work in this thesis can be broadly divided into two parts. The first of which is a study of the SJ vacuum proposed by Sorkin and Johnston in different bounded globally hyperbolic spacetimes \cite{sorkin,Johnston:2009fr} and the second is a study of a spacetime formulation of entanglement entropy (SSEE) proposed by Sorkin \cite{ssee}. SJ formalism turns out to be a useful way of picking out a vacuum uniquely, especially in a spacetime which lacks any symmetry to act as a guide towards a particular choice of vacuum. Though one can define the SJ vacuum in an arbitrary spacetime, calculating the SJ Wightman function is far from simple in an arbitrary spacetime. As a result we have very few cases in which the SJ vacuum has explicitly been obtained. These cases of explicit calculations are referred to in the thesis. In \cite{aas,Aslanbeigi:2013fga}, the SJ vacuum has been studied for different unbounded (infinite volume) spacetimes. However we, following the work of \cite{Afshordi:2012ez,fewster2012}, study the behaviour of the SJ vacuum in a bounded (finite volume) spacetime, where the SJ formalism is known to be well defined. We then extend it to an unbounded spacetime by an appropriate limiting procedure as discussed in the thesis.

In the first part of the thesis we studied the SJ vacuum for a massive scalar field in a 2d causal diamond and for a conformally coupled massless scalar field in a slab of conformally ultrastatic spacetimes with compact spatial hypersurfaces like de Sitter spacetimes and flat FLRW spacetimes. We find that in the center of the causal diamond the SJ vacuum behaves like the massive Minkowski vacuum for large masses $(m=1,2)$ which is as expected, whereas for small masses $(m=0.2,0.4)$ it behaves like the massless Minkowski vacuum with a fixed infrared cut-off which depends on the size of the diamond. In order to have a better understanding of this transition in the behaviour of the SJ vacuum (from massless to massive Minkowski vacuum) with an increase in mass, we need to have a full expression of the massive SJ modes and the SJ Wightman function. We hope to achieve this goal in the near future. In the corner of the diamond we look at correlation plots between the SJ Wightman function, the mirror Wightman function and the Rindler Wightman function for different masses. We see that for small masses, the SJ Wightman function has a better correlation with the mirror Wightman function than with the Rindler Wightman function, which is as expected from the continuum results. However with an increase in mass, the correlation between the mirror and the Rindler Wightman function improves, surpassing that of the SJ Wightman function with any of them. Therefore for a large mass we cannot associate the SJ vacuum with either the mirror or the Rindler vacuum on the basis of correlation plots. We need to come up with another way of understanding the causal set SJ vacuum in the corner. We hope that the knowledge of the SJ modes for an arbitrary mass will lead to a better understanding of the behaviour of the SJ vacuum in the corner of the diamond. We also studied a modification of the SJ vacuum for massless scalar field in a causal diamond, obtained by introducing a suitable weight function to the $\cL^2$ measure on the space of test functions. This modified SJ vacuum reduces to the Rindler vacuum in an appropriate limit. It would be interesting to study the family of all such possible modifications to the SJ vacuum and see how these modified vacua are related to known vacua in any given spacetime. It is shown in \cite{wingham} that one can obtain a Hadamard state by this procedure.

In de Sitter spacetimes we find that in odd dimensions the SJ vacuum agrees with the standard Euclidean vacuum in the full spacetime limit, whereas in even dimensions the SJ vacuum in full spacetime limit turns out to be ill defined. This result is consistent with that of \cite{Aslanbeigi:2013fga}, and is therefore inconsistent with the result of \cite{Surya:2018byh}, where the SJ vacuum is obtained in a causal set sprinkled on a slab of de Sitter spacetime, which is found to be well defined in full de Sitter limit but disagrees with any of the known de Sitter vacua. With a hope of understanding this disagreement we compare the SJ spectrum we obtained in slabs of 2d and 4d de Sitter spacetime with the one obtained by the authors of \cite{Surya:2018byh} in a causal set approximated by these spacetimes. In 2d de Sitter we find that these SJ spectrums are in agreement with each other upto a UV cut-off, which is determined by the sprinkling density. In 4d de Sitter however the SJ spectrums obtained by us vary slightly from the one obtained by the authors of \cite{Surya:2018byh} even below the UV cut-off. We are yet to understand the reasons for this discrepancy. We would also like to compare the SJ modes in de Sitter slab spacetime with the one obtained in a causal set sprinkled on the de Sitter slab. However doing this is computationally very challenging. In flat FLRW spacetimes we evaluated the SJ vacuum in matter, radiation and $\Lambda$ dominated era. We find that in the full spacetime limit, the SJ vacuum agrees with the conformal vacuum. The study of SJ vacuum in a real universe, in which these different cosmological eras are stitched together along with the inflationary phase may provide some useful insight into the physics of early universe. To study this one need to have a better understanding of the transition from one era to another. The global nature of the SJ vacuum however makes the understanding of the cosmological particle creation difficult. For a vacuum which is well defined at a moment of time, one can study the cosmological particle creation in an expanding spacetime by looking at the Bogoliubov transformation between the initial and the final vacua\footnote{See chapter~3.4 of \cite{birrell} for an example showing cosmological particle creation.}. This is not directly possible in the SJ construction, as it provides us with just a unique spacetime vacuum.

In the second part of the thesis we studied the SSEE in two different settings. One of which is a spatially compact and temporally non compact static patch. In particular we studied the SSEE of a massive scalar field in a static patch of de Sitter spacetimes and of a massless scalar field in a static patch of Schwarzschild de Sitter spacetimes. In de Sitter spacetime we find that the SSEE is independent of the mass of the field and agrees with the von Neumann entropy evaluated by Higuchi and Yamamoto \cite{Higuchi:2018tuk}. We also argue that the SSEE satisfies the area law. This is an important confirmation of the validity of the SSEE as a candidate for the de Sitter horizon entropy. However in 2d case, the SSEE instead of depending logarithmically on the UV cut-off turns out to be constant which is independent of a UV cut-off. The mass independence of the entanglement entropy is true only for the horizon and not for an arbitrary entangling surface, as it is shown in \cite{Maldacena:2012xp,Kanno:2014lma,Iizuka:2014rua}, that the entanglement entropy is mass dependent for a superhorizon (entangling surface of radius much larger than that of the de Sitter horizon) at a late time hypersurface near the future boundary of the spacetime. We also studied the SSEE of a massless scalar field restricted to a causal diamond in the center of a 2d cylinder slab spacetime. We find the following: $(i)$ the SSEE is UV cut-off dependent\footnote{Here the cut-off is taken in the SJ spectrum in the diamond.}, $(ii)$ the SSEE of the complimentary diamonds are equal which is expected from any definition of the entropy for a bipartite system in a pure state and $(iii)$ in the large height limit of the cylinder\footnote{The SJ vacuum in a cylinder slab depends on its height, so different heights means different vacua.}, the SSEE we obtain exactly matches with the Calabrese-Cardy entropy which further confirms the equivalence between the SSEE and the von Neumann entropy. Here we notice some strange features of the generalised eigenvalues. We find that all, except one, pairs of eigenfunctions hovers around $0,1$, which means that the largest few eigenvalues accounts for most of the entropy. We hope to have a better understanding of this feature of SSEE in the future. Our results along with the one obtained by \cite{saravani2014spacetime} suggests that the SSEE is equivalent to the von Neumann entropy in cases where both of them are well defined. 
Also the SSEE is shown to be consistent with the von Neumann entropy for a system in a Gaussian state. We don't yet have an understanding of how the SSEE formula gets modified for a non-Gaussian state.

We end this thesis with listing down some of the open questions discussed above
\begin{itemize}
\item We need to solve the SJ eigenvalue problem for a scalar field with arbitrary mass in a causal diamond. This will help us to study the behaviour of the SJ vacuum in the corner of the diamond for larger masses.
\item We need to understand the discrepancy between the continuum and the causal set SJ spectrum in 4d de Sitter slab spacetimes. Does it has anything to do with the 4d causal set Green's function?
\item We would like to understand the transition from one cosmic era to other and how it affects the SJ vacuum state.
\item How to define thermality and cosmological particle creation in the SJ formalism? 
\item The SJ formalism is developed only for a free real scalar field. We still need to have an observer independent formulation of an interacting field theory. A work in this direction was initiated by Johnston in \cite{johnston}.
\item We would like to have an understanding of the conditions in which the SSEE agrees with the von Neumann entropy. We would also like to understand the SSEE in algebraic QFT formalism.
\item SSEE, unlike the von Neumann entropy, appears to be applicable in contexts where the Cauchy hypersurface is not well defined like in Topology changing spacetimes, AdS etc. It would be interesting to study SSEE in those contexts.
\item We would like to understand how the SSEE formula gets modified for a non-Gaussian state.
\end{itemize}




\bibliography{reference}
\bibliographystyle{ieeetr} 

\end{document}